\pdfoutput=1
\documentclass[12pt]{article}
\usepackage[utf8]{inputenc}
\usepackage[english]{babel}
\usepackage{amsmath,amsthm,amssymb}
\usepackage{graphicx,psfrag,epsf}
\usepackage{enumerate}
\usepackage{natbib}
\usepackage{url} 
\usepackage{xcolor}
\usepackage{tikz}
\usetikzlibrary{automata}
\usepackage{bm}
\usepackage{subfig}
\usepackage{xr-hyper}
\usepackage{amsfonts}
\usepackage{adjustbox}
\usepackage{booktabs}
\usepackage{fancyhdr}
\usepackage{siunitx}
\usepackage{etoc}
\usepackage[colorlinks=true,linkcolor=red,citecolor=blue]{hyperref}
\usepackage{multirow}
\usepackage{appendix}
\usepackage{mathtools}
\usepackage{algorithm}
\usepackage{algpseudocode}
\usepackage{array}
\usepackage{tabularx}
\newcommand{\blind}{1}

\makeatletter
\newcommand*{\addFileDependency}[1]{
\typeout{(#1)}
\@addtofilelist{#1}
\IfFileExists{#1}{}{\typeout{No file #1.}}
}\makeatother

\newcommand{\mP}{\mathbb{P}}
\newcommand{\mG}{\mathbb{G}}

\newcommand{\calX}{\mathcal{X}}
\newcommand{\bbeta}{\boldsymbol{\beta}}

\newcommand{\bGamma}{\boldsymbol{\Gamma}}
\newcommand{\bX}{\mathbf{X}}
\newcommand{\bb}{\mathbf{b}}

\newcommand{\bx}{\mathbf{x}}

\newcommand{\bH}{\mathbf{H}}

\newcommand{\bh}{\mathbf{h}}

\newcommand{\bZ}{\mathbf{Z}}

\newcommand{\mO}{\mathcal{O}}

\newcommand{\bM}{\mathbf{M}}

\newcommand{\bv}{\mathbf{v}}

\newcommand{\lb}{\left\{}
\newcommand{\rb}{\right\}}
\newcommand{\lsb}{\left(}
\newcommand{\rsb}{\right)}

\newcommand{\calL}{\mathcal{L}}
\newcommand{\bE}{\mathbb{E}}
\newcommand{\bP}{\mathbf{P}}
\newcommand{\bG}{\mathbf{G}}
\newcommand{\lessp}{\lesssim_\mP}
\newcommand{\textop}{\text{op}}
\newtheorem{assumption}{Assumption}%
\newtheorem{remark}{Remark}%
\newtheorem{theorem}{Theorem}%
\newtheorem{prop}{Proposition}%
\newtheorem{example}{Example}%

\definecolor{orange}{rgb}{1,0.5,0}

\newcommand{\bt}{\boldsymbol{t}}

\newcommand{\lbb}{\left[}
\newcommand{\rbb}{\right]}
\newcommand{\calT}{\mathcal{T}}
\newcommand{\calG}{\mathcal{G}}

\newtheorem{lemma}{Lemma}

\makeatother
\begin{document}

\def\spacingset#1{\renewcommand{\baselinestretch}%
{#1}\small\normalsize} \spacingset{1}


\if1\blind
{
  \title{\bf Nonparametric heterogeneous causal mediation with orthogonal machine learning}
 \author{
   Jiaqi Tong$^{1}$, Yi Zhao$^{2}$, Bhramar Mukherjee$^{1}$, and Fan Li$^{1,*}$
   \vspace{0.2cm} 
   \\
   $^{1}$Department of Biostatistics, Yale School of Public Health, New\\ Haven, CT, USA \\
   $^{2}$Department of Biostatistics and Health Data Science, Indiana\\ University School of Medicine\\
   $^{*}$\emph{email}: fan.f.li@yale.edu
 }

  \maketitle
} \fi

\if0\blind
{
  \bigskip
  \bigskip
  \bigskip
  \begin{center}
    {\LARGE\bf Nonparametric heterogeneous causal mediation with orthogonal machine learning}
\end{center}
  \medskip
} \fi

\bigskip
\begin{abstract}
Causal mediation analysis decomposes the total effect of an intervention on an outcome into a direct pathway and an indirect pathway transmitted through a mediator, but standard methods typically summarize these pathways using population average effects. In many applications, however, the indirect effect may vary substantially across individual profiles. We propose an orthogonal statistical learning framework for estimating heterogeneous causal mediation effects conditional on individual characteristics. The method constructs a class of weighted Neyman orthogonal losses motivated by influence function representations of weighted population average effects. These losses directly target conditional mediation estimands whose minimizers are locally insensitive to nuisance estimation errors. We implement the resulting learners under a two-stage meta-learning framework with regularized linear sieves as second-stage smoothers, and introduce a combination of targeted learning and orthogonal learning designed to improve stability when mediator density ratios are unstable. We establish $L^2$ and uniform limit theory and develop pointwise and uniform confidence bands. Simulation studies show that the proposed orthogonal learners reduce the mean integrated squared error by more than $50\%$ compared with existing model-based methods and provide computationally efficient inference in nonlinear settings. The CARDIA, PSACR, and STAR analyses reveal heterogeneous mediated effects across cardiometabolic, psychological, and educational settings.
\end{abstract}

\noindent%
{\it Keywords:} Causal mediation, conditional natural indirect effect, precision medicine, targeted learning, overlap weights, weighted orthogonal learning
\vfill

\newpage
\spacingset{1.8} 

\section{Introduction}

\subsection{Background and related literature}

Causal mediation analysis decomposes the total effect of an intervention on an outcome into an indirect effect transmitted through an intermediate variable, or \emph{mediator}, and a direct effect operating through pathways that bypass the mediator. It is widely used in applied studies to assess, for example, whether an exposure affects a later health outcome partly through an intermediate biomarker, such as whether cigarette smoking influences subsequent systolic blood pressure through abdominal intermuscular adipose tissue volume. The counterfactual outcomes framework has provided estimand definitions, identification formulas, and semiparametric efficiency theory for direct and indirect mediation effects \citep{imai2010identification,tchetgen2012semiparametric}. Much of this literature targets population average mediation effects. In many scientific applications, however, the mediated pathway itself may vary across individual profiles. Across different research contexts, exploring such variation can help identify which individuals are most affected through a particular mechanism and inform the design of more personalized and targeted interventions.

For estimating treatment effect heterogeneity without mediators, debiased machine learning and orthogonal statistical learning provide flexible two-stage meta-learning procedures for conditional average treatment effects by constructing Neyman orthogonal losses that directly target the causal contrast while reducing sensitivity to first-stage nuisance estimation error \citep{nie2021quasi,kennedy2023towards,foster2023orthogonal}. In contrast, optimal estimators for heterogeneous mediation analysis remain less developed. Existing work has primarily adopted model-based or Bayesian approaches. Among model-based methods, \cite{zhao2025estimation} and \cite{xue2022heterogeneous} proposed LASSO-regularized ordinary least squares estimators and fused-LASSO-regularized estimators, respectively, based on linear structural equation models (LSEMs), with an emphasis on high-dimensional covariates. In addition, \cite{wang2021causal} considered mixture models with parameters estimated using EM algorithms, whereas \cite{li2026modeling} considered a Cox proportional hazards model for survival outcomes. Within the Bayesian framework, \cite{ting2025estimating} and \citet{liu2026heterogeneous} employed Bayesian additive regression trees to estimate conditional mediation effects. Other related methods consider either a simpler heterogeneous mediation estimand conditional on a single categorical variable \citep{qin2017weighting} or a plug-in approach with specific nuisance estimation methods \citep{huan2024individualized}. However, the performance of model-based and simple plug-in approaches may depend heavily on model assumptions and the quality of nuisance estimation, and can deteriorate when nuisance functions are misspecified or too complex to estimate efficiently. This issue is especially pronounced in mediation analysis because the estimands depend on several nuisance components, including the treatment propensity score, the conditional mediator densities, and the conditional outcome regressions. Moreover, Bayesian methods may lack strong asymptotic guarantees, and computationally efficient pointwise and simultaneous inference procedures without strong model assumptions generally remain unavailable.

\subsection{Our contributions}

We develop a flexible two-stage causal machine learning framework for estimating heterogeneous causal mediation effects conditional on baseline characteristics. Our primary contributions are threefold. First, we propose a class of weighted orthogonal learners for conditional mediation effects, formulated as minimizers of Neyman orthogonal loss functions whose directional derivatives with respect to the target estimand are locally insensitive to perturbations of the nuisance functions around their true values \citep{foster2023orthogonal}. Consequently, the orthogonal learners avoid linear error propagation and depend only on second-order first-stage nuisance estimation errors. We exemplify this class through three interpretable choices of the weight function---inverse probability weight, treated weight and overlap weight; the latter two choices can improve stability when positivity is empirically weak. Second, to improve the stability of inverse mediator density weighting, we propose combining targeted learning with the orthogonal learners, which can be viewed as a nonparametric extension of commonly used stabilization techniques for studying population average mediation effects. Third, we establish the asymptotic theory and computationally efficient statistical inference procedures for the more complex mediation estimands. Leveraging regularized least squares sieves \citep{chen2007large}, we establish $L^2$ and uniform limit theory for the proposed class of orthogonal learners and show that they can achieve oracle efficiency, as if the nuisance functions were known. Moreover, we show that the empirical loss tailored to nonparametric targeted learning converges uniformly to the orthogonal loss; therefore, targeted learning improves stability without breaking insensitivity to first-stage nuisance function estimation.
To quantify uncertainty, we further construct pointwise and uniform confidence bands, where {the former are} {based on a closed-form calculation} and {the latter are} {based on a computationally efficient Gaussian bootstrap}.

Our simulation studies demonstrate that the proposed orthogonal learners are substantially less sensitive to nuisance model complexity and more robust to working model misspecification than existing model-based learners. In nonlinear settings, they reduce mean integrated squared error by more than $50\%$ and deliver reliable pointwise and uniform coverage. The primary CARDIA application examines heterogeneity in the pathway from cigarette smoking through intermuscular adipose tissue volume to later systolic blood pressure, while supplementary PSACR-002 and Tennessee STAR applications extend the analysis to psychological and educational settings.

The remainder of this article is organized as follows. Section~\ref{sec:notation-assumptions-identification} introduces the estimands and identification results and 
provides a review of a simple plug-in T-learner. Section~\ref{sec:OSL} develops the proposed weighted orthogonal learning framework. Section~\ref{sec:targeted-orthogonal-learning} introduces targeted orthogonal learners designed to improve finite-sample stability. Section~\ref{sec:6-repre-ortho-learners} summarizes representative learners induced by different weighting choices. Section~\ref{sec:asymptotic-inference} establishes asymptotic theory and develops pointwise and uniform inference. Sections~\ref{sec:simulation} and~\ref{sec:data-application} present simulation studies and an empirical application, respectively. Section~\ref{sec:conclusion} concludes.

\section{Notation, assumptions, and identification}\label{sec:notation-assumptions-identification}
We consider a study with $n$ individuals and denote the treatment indicator as $A$, where $A=1$ represents the treated condition and $A=0$ represents the control condition. Let $\bX\in\mathcal{X}\subseteq\mathbb{R}^{p}$ denote the vector of measured pre-treatment covariates. 
Suppose a scalar outcome $Y\in\mathbb{R}$ is observed, and a scalar mediator $M\in\mathcal{M}\subseteq\mathbb{R}$ is measured prior to the final outcome. Thus, the analyst collects $n$ i.i.d. copies of the observed data vector $\mO=(\bX^\top,A,M,Y)^\top$. Let $M(a)$ denote the counterfactual mediator that would have been observed under treatment value $A=a$, and let $Y(a,m)$ denote the counterfactual outcome that would have been observed under mediator value $M=m$ and treatment value $A=a$. The observed mediator and outcome satisfy the consistency assumption: $M(a)=M$ if $A=a$ and $Y(a,m)=Y$ if $A=a$ and $M=m$. We invoke the composition assumption such that $Y(a):=Y(a,M(a))$; that is, the counterfactual outcome under treatment value $A=a$ coincides with the potential outcome that would have been observed under treatment value $A=a$ and the mediator taking the natural value it would have been observed to take under the same treatment value. Causal mediation analysis studies the decomposition of the total effect (TE) as the sum of the natural direct effect (NDE) that bypasses the mediator and the natural indirect effect (NIE) that operates through the mediator \citep{imai2010identification}. Formally, we define $\text{TE}=\bE\{Y(1)-Y(0)\}$, $\text{NDE}=\bE\{Y(1,M(0))-Y(0,M(0))\}$, $\text{NIE}=\bE\{Y(1,M(1))-Y(1,M(0))\}$,
and obtain the canonical decomposition as $\text{TE}=\text{NDE}+\text{NIE}$.

To quantify treatment effect heterogeneity in causal mediation analysis, we focus on the following conditional analogues: $\text{CTE}(\bX)=\bE\{Y(1)-Y(0)|\bX\}$, $\text{CNDE}(\bX)=\bE\{Y(1,M(0))-Y(0,M(0))|\bX\}$, and $\text{CNIE}(\bX)=\bE\{Y(1,M(1))-Y(1,M(0))|\bX\}$.
Here, CTE represents the conditional total effect and is of central interest in treatment effect heterogeneity estimation \citep{kennedy2023towards}; CNDE represents the conditional natural direct effect, quantifying the personalized natural direct effect conditional on individual profiles $\bX$; and CNIE represents the conditional natural indirect effect, quantifying the personalized natural indirect effect that operates through the mediator. Finally, the canonical decomposition carries over as $\text{CTE}(\bX)=\text{CNDE}(\bX)+\text{CNIE}(\bX)$.
To point-identify the conditional causal mediation effects defined above, we invoke the following set of standard assumptions.

\begin{assumption}[\emph{Sequential Ignorability}]\label{assump:SI}
For all $a,a'\in\{0,1\}$ and all $m\in\mathcal{M}$, (i) $\{Y(a',m),M(a)\}\perp A|\bX$ and (ii) $Y(a',m)\perp M(a)|A=a,\bX$ hold. 
\end{assumption}


\begin{assumption}
\label{assump:positivity}
For every $\bx\in\mathcal X$, $0<\pi(\bx):=\Pr(A=1|\bx)<1$. Moreover, the conditional distribution of the mediator under control is
absolutely continuous with respect to its conditional distribution under
treatment: $P_{M\mid A=0,\bX=\bx}
\ll
P_{M\mid A=1,\bX=\bx}$. 
\end{assumption}

Sequential ignorability (SI) consists of two parts. SI(i) assumes that treatment is as if randomized given measured baseline covariates, and {SI(ii) rules out post-treatment mediator-outcome confounding, including that induced by treatment}.
Assumption \ref{assump:positivity} requires (i) treatment positivity, ensuring treatment overlap, and (ii) absolute continuity of the control mediator distribution with respect to the treated distribution, permitting a change of measure from treatment to control. The following proposition shows that, under Assumptions \ref{assump:SI}-\ref{assump:positivity}, the conditional causal mediation effects are point-identifiable.

\begin{prop}\label{prop:identification-CCME}
    Under Assumptions \ref{assump:SI}-\ref{assump:positivity}, the structural parameters $\theta_{a_1a_2}(\bX):=\bE\{Y(a_1,M(a_2))|\bX\}$ with $(a_1,a_2)\in\{(0,0),(1,0),(1,1)\}$ are identified by $\theta_{a_1a_2}(\bX)=\int \mu_{a_1}(m,\bX)f(m|a_2,\bX)dm$, where $\mu_a(M,\bX)=\bE\{Y|A=a,M,\bX\}$ is the conditional outcome mean and $f(m|a,\bX)$ is the conditional density of the mediator given treatment and baseline covariates. The conditional causal mediation effects are additive combinations of these parameters.
\end{prop}

For simplicity but with a slight abuse of notation, we use $g$ and $g(\bX)$ to denote the average and conditional causal mediation effects, respectively, where $g \in \{\text{TE (CTE)}, \text{NDE (CNDE)},\allowbreak \text{NIE (CNIE)}\}$.

\subsection{Some technical notation}\label{sec:technical-notation}
To proceed, the following technical notation is defined for subsequent use. Let $[n]=\{1,\ldots,n\}$ denote the set of positive integers up to $n$. Let $\bullet^c$ denote the set complement, and let $A \backslash B = A \cap B^c$ denote the set difference. Let $\| \bv \|_q=(\sum_{l=1}^L|v_l|^q)^{1/q}$ denote the usual $\ell^q$ norm. Let $\| \bullet \|_{\mathrm{op}}$ denote the matrix operator or spectral norm and $\| f \|_{\mP,q} = \left( \int |f|^q d\mP \right)^{1/q}$ the $L^q(\mP)$ norm. In particular, $\| f \|_{\mP,\infty}=\sup_\bx |f(\bx)|$ is the uniform norm. Let $d_u(\bullet)=\max_{1\leq q\leq Q}\| \widehat{\bullet}^q-\bullet\|_{\mathbb{P},u}$ denote the nuisance estimation error in the $L^u(\mathbb{P})$ norm. Let $\psi_k(\bullet)$ be the $k$th eigenvalue of a generic matrix $\bullet\in\mathbb{R}^{K\times K}$ such that $\psi_1\leq\ldots\leq\psi_K$. We write $a_n\lesssim b_n$ if $a_n\leq cb_n$ for some constant $c$ independent of $n$, and $a_n\asymp b_n$ if $a_n\lesssim b_n$ and $b_n\lesssim a_n$. We use $X =_d Y$ to denote that two random variables $X$ and $Y$ have the same distribution, and write $X =_d Y + o_\mP(a_n)$ if $(X - Y)/a_n$ converges to zero in probability. We also use $a_n \lesssim_\mP b_n$ to denote that the stochastic sequence $a_n$ is of order at most $b_n$ in probability, i.e., $a_n = O_\mP(b_n)$. The notations $a_n \lesssim_\mP b_n$ and $a_n = O_\mP(b_n)$ are used interchangeably for convenience. Let $\mP_n(V)=n^{-1}\sum_{i=1}^n V_i$ denote the empirical mean of a generic random object $V$, and let $\mG_n\{f(V)\}=n^{-1/2}\sum_{i=1}^n [f(V_i)-\bE\{f(V_i)\}]$ denote the empirical process indexed by a generic function $f$. Let $\text{diam}(\calX):=\sup_{\bx_1,\bx_2\in\calX}\|\bx_1-\bx_2\|_2$ denote the diameter of the covariate support \(\calX\). Let $a \wedge b:=\min(a,b)$ and $a \vee b:=\max(a,b)$.

\subsection{A simple nonorthogonal learner: the T-learner}\label{sec:T-learner}

Following \cite{2019tlearnerkunzel}, the $g$-computation identification formulas in Proposition \ref{prop:identification-CCME} suggest a simple T-learner that estimates $g(\bX)$ by taking the difference between two estimated nuisance components. Define $\eta_{a_1a_2}(\bX)=\int\mu_{a_1}(m,\bX)f(m|a_2,\bX)dm$, and, in particular, $\eta_{a_1a_1}(\bX)=\bE(Y|A=a_1,\bX)$. Formally, the T-learner is given by $\widehat{\text{CTE}}(\bX)=\widehat{\eta}_{11}(\bX)-\widehat{\eta}_{00}(\bX)$, $\widehat{\text{CNDE}}(\bX)=\widehat{\eta}_{10}(\bX)-\widehat{\eta}_{00}(\bX)$, and $\widehat{\text{CNIE}}(\bX)=\widehat{\eta}_{11}(\bX)-\widehat{\eta}_{10}(\bX)$,
where, for $a_1,a_2\in\{0,1\}$, $\widehat{\eta}_{a_1a_2}(\bX)=\int \widehat{\mu}_{a_1}(m,\bX)\widehat{f}(m|a_2,\bX)\,dm$, with the integral evaluated by either Monte Carlo simulation or Gaussian quadrature for continuous mediators. \cite{zhao2025estimation} proposed a parametric T-learner under the LSEM assumptions that $\mu_a(M,\bX)$ is linear in $(M,\bX^\top)^\top$ given $a\in\{0,1\}$, and that $M=\underline{\mu}_a(\bX)+\text{error}$ with mean-zero error, where $\underline{\mu}_a(\bX)=\bE(M\mid A=a,\bX)$ is also linear in $\bX$ given $a$. \cite{zhao2025estimation} recommended estimation and inference based on OLS with model-based standard errors for low-dimensional $\bX$, and LASSO for high-dimensional $\bX$. In contrast, \cite{liu2026heterogeneous} considered a Bayesian analogue of the T-learner using flexible BART models for the mediators and outcomes that incorporate clever covariates. Although the calculation of the T-learner is straightforward, it has several potential disadvantages. First, the T-learner does not directly target $g(\bX)$ but instead focuses on learning the building blocks $\theta_{a_1a_2}(\bX)$. That is, the T-learner may achieve an optimal bias-variance trade-off for $\theta_{a_1a_2}(\bX)$, but not for $g(\bX)$ itself, which may worsen performance when the contrast $g(\bX)$ is much simpler than the building blocks $\theta_{a_1a_2}(\bX)$. Second, the T-learner is not Neyman orthogonal in the sense of \cite{foster2023orthogonal}, with the precise meaning given later in Section \ref{sec:OSL}. As a result, low-quality nuisance estimation may carry over to the T-learner and lead to a slower convergence rate. In addition, for the parametric T-learner, when the linear outcome and mediator mean models are misspecified, the nuisance estimates, and hence the corresponding parametric T-learner, are not consistent. Third, computationally efficient pointwise and uniform confidence bands for the T-learner are generally unavailable without imposing strong model assumptions; neither \citet{zhao2025estimation} nor \citet{liu2026heterogeneous} provided methods to construct uniform confidence bands for their estimators.

\section{Introducing a class of weighted orthogonal learners}\label{sec:OSL}
For illustration, we focus mainly on the optimal learning procedures for the CNIE, because the primary scientific question in the motivating data example concerns the indirect pathway.
Alternatively, if the primary interest lies in the direct pathway, a similar procedure can be constructed for optimally learning the CNDE; details are provided in Supplementary Material Section \ref{supp;sec:orthogonal-cnde}. Of note, the separate optimal learners for $\text{CTE}$, $\text{CNDE}$, and $\text{CNIE}$ may not numerically satisfy the effect decomposition, because naive addition or subtraction of the optimal learners for any two of these estimands does not necessarily yield the optimal bias-variance trade-off for the remaining one. We leave the joint learning of $\text{CTE}$, $\text{CNDE}$, and $\text{CNIE}$ under the effect-decomposition constraint for future research. To aid interpretation when the effect decomposition is of scientific interest, we recommend a simpler strategy by optimally estimating any two of these estimands and obtaining the remaining one by addition or subtraction.


To improve upon the T-learner, we propose a weighted orthogonal statistical learning framework for constructing optimal learners for heterogeneous causal mediation effects. To proceed, we first focus on a family of weighted population average causal mediation effects, which can be formally defined as $\check{g}_w = {\bE\{w(\bX)Y_g\}}/{\bE\{w(\bX)\}}$,
where $w$ denotes the weight function and $Y_g$ represents the corresponding potential outcome contrast for the estimand $g$. Specifically, (i) for $g = \text{TE}$, $Y_g = Y(1) - Y(0)$ and $\check{g}_w$ represents the weighted TE; (ii) for $g = \text{NDE}$, $Y_g = Y(1, M(0)) - Y(0, M(0))$ and $\check{g}_w$ represents the weighted NDE; and likewise (iii) for $g = \text{NIE}$, $Y_g = Y(1, M(1)) - Y(1, M(0))$ and $\check{g}_w$ represents the weighted NIE. Following the balancing weight framework of \cite{li2018balancing}, (i) the target population is the combined population when $w = 1$; (ii) the target population is the treated population when $w = \pi(\bX)$; and (iii) the target population is the overlap population when $w^\ast_{\text{TE}} = \pi(\bX)\{1-\pi(\bX)\}$. 

\cite{morzywolek2023weighted} constructed a class of weighted orthogonal learners for the $\text{CTE}$ that includes the R-learner \citep{nie2021quasi} and the DR-learner \citep{kennedy2023towards} as special cases. 
We generalize their methods to the optimal learning of $\text{CNIE}$; in what follows, $g$ denotes the $\text{CNIE}$ or NIE. We first consider the following family of weighted least squares loss functions:
\begin{align}\label{eq:weighted-LS-LOSS}
   \calL(w) = \bE[w(\bX)\{Y(1,M(1))-Y(1,M(0)) - g(\bX)\}^2].
\end{align}
We assume that the weight function $w(\bX)$ is almost surely positive with $\Pr\{w(\bX)\neq0\}=1$ and can depend on the propensity score $\pi(\bX)$ only through $w(\bX)=\omega\{\pi(\bX)\}$, where $\omega:[0,1]\to\mathbb{R}_{\geq 0}$ is a twice continuously differentiable function. 
Consequently, under Assumptions \ref{assump:SI}–\ref{assump:positivity}, 
the minimizer of the loss function defined in Equation \eqref{eq:weighted-LS-LOSS}, denoted by $g_{\min}(\bX; w)$, is $\text{CNIE}(\bX)$; that is, $g_{\min}(\bX;w)=\text{CNIE}(\bX)=\kappa(\bX)$,
where $\kappa(\bX)=\eta_{11}(\bX)-\eta_{10}(\bX)$. Interestingly, the choice of weighting function does not alter the minimizer but can motivate a class of orthogonal learners, which is the core idea of \cite{morzywolek2023weighted} for studying CTE. 

To construct an orthogonal learner, we leverage the heuristic correspondence between the uncentered efficient influence function (EIF) under the nonparametric model for the chosen finite-dimensional smoothed parameter $\check{g}_w$ and the orthogonal learner, a connection recognized in the literature for estimating CTE~\cite[e.g.,][]{semenova2021debiased,kennedy2023towards,morzywolek2023weighted}. To facilitate the derivation of the nonparametric EIF for $\check{g}_w$, we define, for $a\in\{0,1\}$, $\phi_{10}(\mO)=\pi(\bX)^{-1}{A}r(M,\bX)\{Y-\mu_1(M,\bX)\}+\{1-\pi(\bX)\}^{-1}(1-A)\{\mu_1(M,\bX)-\eta_{10}(\bX)\}+\eta_{10}(\bX)$ and $\phi_{aa}(\mO)=\pi(\bX)^{-a}\{1-\pi(\bX)\}^{a-1}{\mathbb{I}(A=a)}\{Y-\eta_{aa}(\bX)\}+\eta_{aa}(\bX)$,
where $r(M,\bX)=f(M|0,\bX)/f(M|1,\bX)$ and $\mathbb{I}(\bullet)$ denotes the indicator function. Here, $\zeta(\mO):=\phi_{11}(\mO)-\phi_{10}(\mO)$ and $\phi_{10}(\mO)-\phi_{00}(\mO)$ are the uncentered nonparametric EIFs for $\text{NIE}$ and $\text{NDE}$, respectively \citep{tchetgen2012semiparametric}. Theorem \ref{thm:EIF-smoothed-checkg} below derives the nonparametric EIF for $\check{g}_w$ when $g=\text{CNIE}$.

\begin{theorem}\label{thm:EIF-smoothed-checkg}
Under Assumptions \ref{assump:SI}-\ref{assump:positivity}, the nonparametric EIF for $\check{g}_w$ can be written as $\varphi_{\text{NIE}}=\bE\lb w(\bX)\rb^{-1}(\phi_n^{\text{NIE}}-\check{g}_w\phi_d^{\text{NIE}})$,
where $\phi_d^{\text{NIE}}=w(\bX)+\omega'\{\pi(\bX)\}\lb A-\pi(\bX)\rb$ and $\phi_n^{\text{NIE}}=\kappa(\bX)\lb \phi_d^{\text{NIE}}-w(\bX)\rb+w(\bX)\zeta(\mO)$.

\end{theorem}
Theorem \ref{thm:EIF-smoothed-checkg} motivates the following orthogonal loss function for learning $\text{CNIE}$: $\widetilde{\calL}(w)=\bE\{l_w(\mO;\bGamma,g)\}$, where
\begin{align}\label{eq:loss-define-learner}
l_w(\mO;\bGamma,g)=&\phi_d^{\text{NIE}}\lb \kappa(\bX)-g(\bX)\rb^2-2w(\bX)\lb \zeta(\mO)-\kappa(\bX)\rb g(\bX).
\end{align}
Several remarks are in order based on the orthogonal loss function. First, one can verify that the minimizers of $\mathcal{L}(w)$ and $\widetilde{\calL}(w)$ are both $\text{CNIE}(\bX)$. Second, $\widetilde{\calL}(w)$ is Neyman orthogonal in the sense of \cite{foster2023orthogonal}, whereas the original weighted least squares loss $\mathcal{L}(w)$ is not. In other words, the learner constructed by minimizing the orthogonal loss $\widetilde{\calL}(w)$ is locally insensitive to nuisance function error propagation and has the mixed-bias property in the sense of \cite{rotnitzky2021characterization} and \citet{cheng2025inverting}. To see this, one can verify that $\partial l_w/\partial g\propto \varphi_{\text{NIE}}$ by treating $g=\check{g}_w$ as a scalar. Finally, our construction of the loss function refines the work of \cite{morzywolek2023weighted}; see Remark \ref{remark:refine-loss} below.
\begin{remark}\label{remark:refine-loss}
    Our proposed loss function in Equation \eqref{eq:loss-define-learner} refines the intuitive construction in \cite{morzywolek2023weighted}, namely, setting $l_w = \phi_d^{\text{NIE}}\{\phi_n^{\text{NIE}}/\phi_d^{\text{NIE}} - g(\bX)\}^2$. The intuitive loss $l_w$ can be interpreted as a weighted least squares loss with weight $\phi_d^{\text{NIE}}$ and pseudo-outcome $\phi_n^{\text{NIE}}/\phi_d^{\text{NIE}}$. However, it requires special care when $\phi_d^{\text{NIE}}$ can equal zero. In contrast, our refinement avoids division by the potentially zero $\phi_d^{\text{NIE}}$ by adding a shift term independent of the target function $g$ while preserving the Neyman orthogonal property with $\partial l_w/\partial g\propto \varphi_{\text{NIE}}$, and therefore does not change the minimizer of the loss. Such a refinement also applies to the CTE estimation considered in \cite{morzywolek2023weighted}. This refinement is useful because $\phi_d^{\text{NIE}}$ can equal zero, for example, when $w(\bX)=\pi(\bX)$ and $\phi_d^{\text{NIE}}=A$.
\end{remark}

Although the orthogonal loss function $\widetilde{\calL}(w)$ can in theory motivate many different nonparametric estimators, we proceed with the regularized linear sieve method \citep{chen2007large} because it is computationally efficient, enables feasible pointwise and uniform inference, and nests many familiar parametric submodels. Using a lower-case $\bx \in \mathcal{X}$, the expression $g(\bx)$ denotes the CNIE evaluated at a fixed point $\bx$, distinguishing it from the random function $g(\bX)$. Suppose that the function $g$ belongs to a specified function class $\mathcal{G}$, such as a Hölder class of smoothness order $s$, denoted by $\mathcal{H}_s$, or a Sobolev space. The sieve method uses a sequence of approximating spaces $\mathcal{G}_n$ such that $\mathcal{G}_n$ becomes asymptotically dense in $\mathcal{G}$ as the sample size $n \to \infty$, in the sense that functions in $\mathcal{G}_n$ can approximate any function in $\mathcal{G}$ arbitrarily well \citep{chen2007large}. Linear sieves typically refer to sieve spaces constructed as the linear span of a set of basis functions, that is, $\mathcal{G}_n=\{\bb(\bx)^\top\bbeta:\bb(\bx)=(b_1(\bx),\ldots,b_K(\bx))^\top,\bbeta=(\beta_1,\ldots,\beta_K)^\top,K=K_n\}$, where the dimension $K_n$ grows slowly with the sample size \citep{belloni2015some}. Common well-studied choices of basis functions include polynomials, splines, and wavelets. The approximation quality of the sieve space $\mathcal{G}_n$ is controlled by $\xi_K := \sup_{\bx\in\mathcal{X}} \| \bb(\bx) \|_2$, and the optimal rate $\xi_K\asymp\sqrt{K}$ is attained by many basis functions, including spline, wavelet, and Fourier series. Finally, the proposed class of weighted orthogonal learners for estimating the CNIE is constructed through a two-stage procedure with sample cross-fitting, as shown in Algorithm \ref{alg:cnie_cross_fitting}. Specifically, let $\bGamma=\{\pi,f_0,f_1,\allowbreak\mu_1,\eta_{11}\}$ denote the collection of nuisance functions. In Stage 1, the full sample is randomly partitioned into $Q$ folds of approximately equal size, up to rounding. Let $\mathcal{F}_q \subseteq [n]$, $q\in[Q]$, denote the set of indices in the $q$th fold. For each individual $i\in\mathcal{F}_q$, the predicted values $\widehat{\phi}_d^{\text{NIE}}(A_i,\bX_i)={\phi}_d^{\text{NIE}}(A_i,\bX_i;\widehat{\bGamma}^q)$ and $\widehat{\phi}_n^{\text{NIE}}(\mO_i)={\phi}_n^{\text{NIE}}(\mO_i;\widehat{\bGamma}^q)$ are obtained using the estimated nuisance functions $\widehat{\bGamma}^q$ trained on the disjoint training sample $\mathcal{F}_q^c = [n] \backslash \mathcal{F}_q$. In Stage 2, the proposed class of weighted orthogonal learners is defined as the empirical minimizer of the orthogonal loss $\widetilde{\calL}(w)$ under a generalized ridge penalty over the linear sieve $\mathcal{G}_n$. That is, it can be expressed as $\widehat{g}(\bx)=\bb(\bx)^\top\widehat{\bbeta}$, where
\begin{align}\label{eq:plss-hat-beta}
\widehat{\bbeta}=&\arg\min_{\bbeta} \frac{1}{2}\left[\mP_n\lb l_w(\mO;\widehat{\bGamma},\bb(\bX)^\top\bbeta)\rb+\lambda\bbeta^\top\bP\bbeta\right].
\end{align}
Moreover, it is computationally efficient because $\widehat{\bbeta}$ admits the closed-form solution $\widehat{\bbeta}=(\widehat{\bH}+\lambda \bP)^{-1}\widehat{\bh}$, where $\widehat{\bH}=\mP_n\{\widehat{\phi}_d^{\text{NIE}}\bb(\bX)\bb(\bX)^\top\}$ and $\widehat{\bh}=\mP_n\{\bb(\bX)\widehat{\phi}_n^{\text{NIE}}\}$. 

\begin{algorithm}[htbp]
\caption{Weighted orthogonal learning under sample cross-fitting}
\label{alg:cnie_cross_fitting}
\begin{algorithmic}[1]
\fontsize{11}{9}\selectfont
\Require Observed data $\{\mathcal{O}_i\}_{i=1}^n$ and $Q$ folds 
\Ensure Weighted orthogonal learner for the CNIE

\State Randomly partition the full sample into $Q$ disjoint folds $\mathcal{F}_q, q\in[Q]$.

\State \textbf{Stage 1: Nuisance function estimation}
\For{$q = 1$ to $Q$}
    \State Obtain the estimated nuisance functions $\widehat{\boldsymbol{\Gamma}}^q$ using the training sample $\mathcal{F}_q^c = [n] \setminus \mathcal{F}_q$.
    \State Compute predicted values $\widehat{\phi}_{d,i}^{\text{NIE}}={\phi}_d^{\text{NIE}}(\mO_i,\widehat{\boldsymbol{\Gamma}}^q)$ and $\widehat{\phi}_{n,i}^{\text{NIE}}={\phi}_n^{\text{NIE}}(\mO_i;\widehat{\boldsymbol{\Gamma}}^q)$ for all $i \in \mathcal{F}_q$.
\EndFor

\State \textbf{Stage 2: Orthogonal loss minimization}
\State Empirically minimize the penalized loss $\widetilde{\mathcal{L}}(w)$ over the pooled sample.

\end{algorithmic}
\end{algorithm}

\section{Fusing targeted learning with orthogonal learning}\label{sec:targeted-orthogonal-learning}
Although choosing appropriate weights, such as $w(\bX)=\pi(\bX)$ or $w(\bX)=\pi(\bX)\{1-\pi(\bX)\}$, can remove the explicit division by the propensity score in the orthogonal loss function in \eqref{eq:loss-define-learner}, and thus reduce vulnerability to practical instability in the propensity score, for example, when $\widehat{\pi}(\bX)$ is near 0 or 1, the loss function may still be affected by practical violations of the boundedness of $f(M| 1,\bX)$ (see the term $r(M,\bX)$ in $\phi_{10}$). To further alleviate this issue, we generalize the idea of combining targeted learning and orthogonal learning, in the spirit of the EP learner in \cite{van2024combining} and the i-learner in \cite{vansteelandt2025orthogonal}. We do not restrict attention to the penalized least squares series estimator defined in Equation \eqref{eq:plss-hat-beta}; the following discussion accommodates arbitrary approximation methods beyond the linear sieves $\mathcal{G}_n$, such as nonlinear artificial neural networks, local regression, and kernel methods.

The primary objective is to carry out empirical minimization of the following targeted loss function: $\check{l}_w(\mO;\bGamma,g)=l_w(\mO;\bGamma,g)-2\Delta(\mO;\bGamma,g)$,
where the debiasing drift term $\Delta(\mO;\bGamma,g)$ is defined by 
\begin{align*}
   \Delta(\mO;\bGamma,g)=w(\bX)\frac{A}{\pi(\bX)}r(M,\bX)\{Y-\mu_1(M,\bX)\}g(\bX).
\end{align*}
Here, the debiasing drift term $\Delta(\mO;\bGamma,g)$, despite having mean zero, may not be sufficiently close to zero empirically when $r(M,\bX)$ is numerically unstable. In this case, the targeted loss function $\check{l}_w$ completely removes dependence on the potentially unstable density ratio term $r(M,\bX)$. However, the targeted loss $\check{l}_w$ is not exactly Neyman orthogonal. To render $\check{l}_w$ approximately orthogonal asymptotically, we seek to make empirical minimization of $\mP_n\{l_w(\mO;\widehat{\bGamma},g)\}$ as close as possible to empirical minimization of $\mP_n\{\check{l}_w(\mO;\widehat{\bGamma},g)\}$. Because the debiasing drift term $\Delta(\mO;\bGamma,g)$ depends on an infinite-dimensional function $g(\bX)$, it is generally not possible to establish an exact equivalence between the two empirical minimization problems over $g(\bX)$ based on $l_w$ and $\check{l}_w$. To ensure the learners minimizing $\mP_n\{\check{l}_w(\mO;\widehat{\bGamma},g)\}$ are approximately orthogonal, we follow \cite{van2024combining} and apply infinite-dimensional targeted learning using the method of sieves. This sieve method differs from the previous one in Equation \eqref{eq:plss-hat-beta} and serves solely to ensure that the empirical loss $\mP_n\{\check{l}_w(\mO;\widehat{\bGamma},g)\}$ closely approximates the empirical orthogonal loss $\mP_n\{l_w(\mO;\widehat{\bGamma},g)\}$. 
For simplicity, we still refer to the resulting targeted and approximately orthogonal learners as orthogonal learners.

To proceed, we propose a refined nuisance estimator for $\mu_1$, denoted by $\widehat{\mu}_1^\ast$, obtained by updating the initial estimator $\widehat{\mu}_1$ using an asymptotically dense linear sieve $\check{\mathcal{G}}_n=\{\check{\bb}(\bX)^\top\check{\boldsymbol{\epsilon}}:\dim(\check{\boldsymbol{\epsilon}})=\check{K}_n\}$. As noted earlier, $\check{\mathcal{G}}_n$ need not coincide with $\mathcal{G}_n$. 
We summarize the targeted learning procedure with $\check{\mathcal{G}}_n$ in Algorithm \ref{alg:targeted-learner}. 
We then formalize the preceding claim that empirical minimization of the targeted loss provides a close approximation to empirical minimization of the orthogonal loss. To proceed, let $N(\rho,\mathcal{G},\|
\bullet\|_{\mP,\infty})$ be the $\rho$-covering number for $\mathcal{G}$ equipped with the uniform norm, i.e., the smallest number of balls of $L^\infty$-radius $\rho$ needed to cover $\mathcal{G}$, and $J(\delta,\mathcal{G},\|
\bullet\|_{\mP,\infty}):=\int_{0}^{\delta} \sqrt{\log N(\rho,\mathcal{G},\|
\bullet\|_{\mP,\infty})}d\rho$ be the induced entropy integral. Provided that the complexity of $\mathcal{G}$ is properly controlled through the entropy integral, Theorem \ref{prop:uniform-convergence-drift-term-targeted-learning} below shows that the sample average debiasing drift term $\mathbb{P}_n\{\Delta(\mathcal{O}; \widehat{\pi}, \widehat{f}_0, \widehat{f}_1, \widehat{\mu}_1^\ast, g)\}$ converges to zero uniformly over $g \in \mathcal{G}$, and derives its uniform convergence rate. In other words, the empirical targeted loss converges uniformly to the empirical orthogonal loss, indicating that minimization over either loss is asymptotically equivalent and therefore preserves Neyman orthogonality.

\begin{algorithm}[htbp]
\caption{Targeted learning with linear sieves}
\label{alg:targeted-learner}
\begin{algorithmic}[1]
\fontsize{10}{9}\selectfont
\Require Observed data $\{\mathcal{O}_i\}_{i=1}^n$ and initial cross-fitted predictions $\{\widehat{\mu}_1(M_i,\bX_i),\widehat{\pi}(\bX_i),\widehat{r}(M_i,\bX_i)\}_{i=1}^n$
\Ensure Refined predictions $\{\widehat{\mu}_1^\ast(M_i,\bX_i)\}_{i=1}^n$ for $\mu_1$

\State Conditional on $A_i=1$, fit a weighted ordinary least squares regression of the outcome $Y_i$ on the regressors $\check{\bb}(\bX_i)$, with an offset of $\widehat{\mu}_1(M_i,\bX_i)$ and weights $\omega\{\widehat{\pi}(\bX_i)\}\widehat{r}(M_i,\bX_i)/\widehat{\pi}(\bX_i)$. Obtain $\widehat{\boldsymbol{\epsilon}}$ as the estimated coefficients.

\State Obtain the refined predictions as $\widehat{\mu}_1^\ast(M_i,\bX_i)=\widehat{\mu}_1(M_i,\bX_i)+\widehat{\boldsymbol{\epsilon}}^\top\check{\bb}(\bX_i)$.
\State Output the final learner by minimizing $\mP_n\{\check{l}_w(\mO)\}$ based on the updated predictions.

\end{algorithmic}
\end{algorithm}

\begin{theorem}\label{prop:uniform-convergence-drift-term-targeted-learning}
Assume that (i) there exists a stable linear projection $\Pi(g)$ onto $\check{\mathcal{G}}_n$, with the stability constant $\Lambda_{\check{K}}:=\sup_{\|g\|_{\mP,\infty}\neq0}\|\Pi(g)\|_{\mP,\infty}/\|g\|_{\mP,\infty}$; (ii) there exist finite constants $c_{\check{K}}$ and $l_{\check{K}}$ satisfying, for all $g\in\mathcal{G}$, $\|g-\Pi(g)\|_{\mP,2}\leq c_{\check{K}}$ and $\|g-\Pi(g)\|_{\mP,\infty}\leq l_{\check{K}}c_{\check{K}}$; (iii) $J(\delta,\mathcal{G},\|
\bullet\|_{\mP,\infty})\lesssim\delta^{1-1/(2\tau)}$ for some $\tau>1/2$; (iv) $d_2(f_0)\vee d_2(f_1)\vee d_2(\pi)=o_{\mP}(1)$; (v) Assumptions \ref{asp:regularityL2}(a) and \ref{asp:regularityL2}(d) hold, with $\bb(\bX)$ and $K_n$ replaced by $\check{\bb}(\bX)$ and $\check{K}_n$, respectively, and with the additional uniform boundedness of $\widehat{\mu}_1^\ast$; (vi) $\xi_{\check K}^2\log(\check K\vee2)/n=o(1)$; and (vii) $\|A\{Y-\mu_1(M,\bX)\}\|_{\mP,\infty}\leq \epsilon_3$ for some constant $\epsilon_3>0$.
Then $\|\widehat{\mu}_1^\ast-\mu_1\|_{\mP,2}\lessp r_{1n}:= \sqrt{\check{K}/n}+(1+\xi_{\check{K}}/\sqrt{n})d_2(\mu_1)$ and 
\begin{align}
    &\sup_{g\in\mathcal{G}}|\mP_n\{l_w(\mO;\widehat{\mu}_1^\ast,g)\}-\mP_n\{\check{l}_w(\mO;\widehat{\mu}_1^\ast,g)\}|=2\sup_{g\in\mathcal{G}}\left|\mP_n\{\Delta(\mO;\widehat{\pi},\widehat{f}_0,\widehat{f}_1,\widehat{\mu}_1^\ast,g)\}\right| \lessp c_{\check{K}}r_{1n}+n^{-\frac{1}{2}}\times\nonumber\\
    &{\{(c_{\check{K}}+ n^{-\frac{1}{2}})^{1-\frac{1}{2\tau}}(1+r_{1n}^{\frac{1}{2\tau}})+{d_\infty^\dagger}^{\frac{1}{2\tau}}(r_{2n}+n^{-\frac{1}{2}})^{1-\frac{1}{2\tau}}\}(1+\Lambda_{\check{K}})^{\frac{1}{2\tau}}}+l_{\check{K}}c_{\check{K}}
\sqrt{\frac{\check{K}\log n}{n}}
\lsb
r_{1n}\vee
\sqrt{\frac{\check{K}\log n}{n}}
\rsb,\nonumber
\end{align}
where $r_{2n}:=(d_\infty^\dagger c_{\check{K}})\wedge (d_2^\dagger l_{\check{K}}c_{\check{K}})$ and $d_q^\dagger:=d_q(\pi)+\sum_{a=0,1}d_q(f_a)$.
\end{theorem}

Several additional remarks are in order based on Theorem \ref{prop:uniform-convergence-drift-term-targeted-learning}. First, the uniform convergence rate for the empirical loss function is general, in the sense that it is expressed in terms of high-level quantities. This rate can be made more explicit when one restricts attention to specific choices of the basis functions $\check{\bb}(\bX)$ and the function class $\mathcal{G}$. For example, under the more specific additional regularity conditions listed in \cite{van2024combining}, the rate can be $o_{\mP}(n^{-1/2})$. Second, Theorem \ref{prop:uniform-convergence-drift-term-targeted-learning} also provides an upper bound for the $L^2$ estimation error of the refined nuisance $\widehat{\mu}_1^\ast$ in terms of the first-stage nuisance estimation errors through $r_{1n}$. This helps clarify how the targeting step affects nuisance estimation quality. Finally, because the targeting step is implemented on the pooled sample and the fluctuation parameter $\boldsymbol{\epsilon}$ has dimension $\check{K}_n$ growing with the sample size, this step may sometimes practically overfit $\widehat{\mu}_1^\ast$, since it is no longer cross-fitted. To mitigate this issue, one may follow \cite{vansteelandt2025orthogonal} by adding a LASSO penalty to the weighted least squares regression model in Algorithm \ref{alg:targeted-learner}.

\section{Six representative orthogonal learners}\label{sec:6-repre-ortho-learners}
In this section, we illustrate the class of weighted orthogonal learners introduced in Section \ref{sec:OSL} by presenting six orthogonal learners defined by three explicit choices of weight function and whether a targeting step is incorporated, as summarized in Table \ref{tab:possible-weight}. For illustration, Examples \ref{ss:TR-learner}–\ref{ss:OW-learner} detail three orthogonal learners without a targeting step. Of note, the choice of weighting function is not limited to the three considered below. For example, one may also use $w(\bX)=1-\pi(\bX)$, which corresponds to the control weight.

\begin{table}[htpb]
\centering
\caption{Summary of the six representative orthogonal learners presented in Examples \ref{ss:TR-learner}–\ref{ss:OW-learner}, their corresponding weight functions, and whether a targeting step is incorporated.}\label{tab:possible-weight}
\begingroup
\setlength{\tabcolsep}{2.8pt}
\renewcommand{\arraystretch}{0.70}
\setlength{\aboverulesep}{0.20ex}
\setlength{\belowrulesep}{0.21ex}
\begin{tabular}{lcccccc}
\toprule
 & TR & TW & OW & TTR & TTW &TOW\\
\midrule
Weight $w(\bX)$ 
& $1$ 
& $\pi(\bX)$ 
& $\pi(\bX)\{1-\pi(\bX)\}$ & $1$ 
& $\pi(\bX)$ 
& $\pi(\bX)\{1-\pi(\bX)\}$\\

Population 
& Combined 
& Treated 
& Overlap& Combined 
& Treated 
& Overlap\\

Targeted &$\times$ &$\times$ &$\times$ &$\checkmark$ &$\checkmark$ &$\checkmark$\\
\bottomrule
\end{tabular}
\endgroup
\end{table}

\begin{example}[\emph{TR learner}]\label{ss:TR-learner}
    We begin by setting the weight function to unity. The resulting learner directly generalizes the DR-learner, or two-stage pseudo-outcome regression method. To see this, note that the proposed TR learner can alternatively be viewed as regressing the pseudo-outcome $\widehat{\phi}_n^{\text{NIE}}$ on $\bX$ using the penalized least squares sieve method in Stage 2. We refer to this learner as the {triply robust (TR)} learner because the corresponding semiparametric one-step estimator for the natural indirect effect is triply robust \citep{tchetgen2012semiparametric}. Under $w=1$, $\check{g}_w$ is the natural indirect effect and its nonparametric EIF is given by $\varphi_{\text{NIE}}=\zeta(\mO)-\check{g}_w$. Moreover, we have that $\widehat{\bH}=\mP_n\{\bb(\bX)\bb(\bX)^\top\}$ and $\widehat{\bh}=\mP_n\{\bb(\bX)\widehat{\zeta}(\mO)\}$. Notably, the Gram matrix $\widehat{\bH}$ for the TR learner does not depend on the estimated nuisance functions. 
\end{example}

\begin{example}[\emph{TW learner}]\label{ss:PSW-learner}
Following \cite{tchetgen2012semiparametric}, the finite-sample performance of the TR learner can be sensitive to instability in the inverse propensity score weights. To mitigate this issue, a natural choice is $w=\pi(\bX)$, since only the treated population is weighted for estimating NIE. We refer to the orthogonal learner with $w=\pi(\bX)$ as the {treatment weighted  (TW)} learner. Under $w=\pi(\bX)$, the EIF components for $\check{g}_w$ simplify to $\phi^{\text{NIE}}_d=A$ and $\phi^{\text{NIE}}_n=\kappa(\bX)\{A-\pi(\bX)\}+\pi(\bX)\zeta(\mO)$. Moreover, we have that $\widehat{\bH}=\mP_n\{A\bb(\bX)\bb(\bX)^\top\}$ and $\widehat{\bh}=\mP_n[\bb(\bX)[\widehat{\kappa}(\bX)\{A-\widehat{\pi}(\bX)\}+\widehat{\pi}(\bX)\widehat{\zeta}(\mO)]]$. Of note, the Gram matrix $\widehat{\bH}$ for the TW learner still does not depend on the nuisance functions, but is averaged only over the treated individuals.
   
\end{example}

\begin{example}[\emph{OW learner}]\label{ss:OW-learner}
Although the weighting estimator for the NIE weights only the treated population, the semiparametrically efficient one-step estimator in \cite{tchetgen2012semiparametric} also includes a correction term for control individuals weighted by $\{1-\pi(\bX)\}^{-1}$; see the form of $\phi_{10}$. Therefore, one may also consider using overlap weights with $w=\pi(\bX)\{1-\pi(\bX)\}$, which we refer to as the {overlap weighted (OW) learner}. In this case, the EIF components for $\check{g}_w$ simplify to $\phi^{\text{NIE}}_d=\{A-\pi(\bX)\}^2$ and $\phi^{\text{NIE}}_n=\kappa(\bX)[\{A-\pi(\bX)\}^2-\pi(\bX)\{1-\pi(\bX)\}]+\pi(\bX)\{1-\pi(\bX)\}\zeta(\mO)$. Of note, the Gram matrix $\widehat{\bH}$ for the OW learner depends on the unknown propensity score. Interestingly, under overlap weights, $\phi_d^{\text{NIE}}=\{A-\pi(\bX)\}^2$ corresponds to residualizing the treatment and has a similar flavor to the R-learner for CTE estimation \citep{nie2021quasi}.
\end{example}

For ease of reference, we denote the six representative learners for the CNIE in Table \ref{tab:possible-weight} by $\widehat{g}^{\text{TR}}(\bX)$, $\widehat{g}^{\text{TW}}(\bX)$, $\widehat{g}^{\text{OW}}(\bX)$, $\widehat{g}^{\text{TTR}}(\bX)$, $\widehat{g}^{\text{TTW}}(\bX)$, and $\widehat{g}^{\text{TOW}}(\bX)$, respectively.

\section{Asymptotic theory and statistical inference}\label{sec:asymptotic-inference}
In this section, we derive the $L^2$ and uniform limit theory for the proposed class of orthogonal learners. Specifically, we (i) establish rates of convergence showing that our estimator can be oracle efficient, achieving the fast rate as if the first-stage nuisance functions were known; and (ii) construct computationally efficient and asymptotically honest pointwise and uniform confidence bands for statistical inference. For ease of exposition, we assume that the nuisance functions are cross-fitted, as for the weighted orthogonal learners in Section~\ref{sec:OSL}. However, the refined nuisance estimator $\widehat{\mu}_1^\ast$ used by the targeted learners in Section~\ref{sec:targeted-orthogonal-learning} does not preserve cross-fitting because the weighted least squares procedure in Algorithm~\ref{alg:targeted-learner} computes $\widehat{\mu}_1^\ast$ using the pooled sample. The confidence band construction is adjusted for the targeted learners and performs well in the simulation studies, although its asymptotic analysis may yield weaker guarantees and require additional technical arguments under stronger entropy conditions. We return to a discussion of this point in Section \ref{sec:conclusion}.

\subsection{$L^2$ convergence}\label{sec:asymptotic;ss:L^2}
We first study the $L^2$ limit theory. Let $g^\ast(\bx)=\bb(\bx)^\top\bbeta^\ast$ denote the weighted $L^2$ projection of the truth $g$ onto the sieve space $\mathcal{G}_n$, that is, $\bbeta^\ast=\underset{\bbeta}{\arg}\min 2^{-1}\bE\left[ w(\bX)\lb\alpha(\bX;g,\bbeta)\rb^2\right]$,
where $\alpha(\bX;g,\bbeta)=g(\bX)-\bb(\bX)^\top\bbeta$ denotes the approximation error for a given true function $g$ and sieve coefficients $\bbeta$. Then the total error $\widehat{g}-g$ can be decomposed into $\widehat{g}-g^\ast$, which captures the estimation and regularization errors, and $g^\ast-g$, which captures the approximation error from approximating $\mathcal{G}$ by asymptotically dense $\mathcal{G}_n$. To facilitate the asymptotic analysis, we introduce the regularity conditions commonly used in the debiased machine learning literature \citep{dml} and the least squares series literature \citep{chen2007large}. 

\begin{assumption}\label{asp:regularityL2} 
(a) The eigenvalues of the unweighted Gram matrix $\bG=\bE\{\bb(\bX)\bb(\bX)^\top\}$ are uniformly bounded above and away from zero. (b) For all $n$ and $K$, there exist finite constants $c_{K}$ and $l_{K }$ such that $\| \alpha(\bX;g,\bbeta^\ast)\|_{\mP,2}\leq c_{K}$ and $\| \alpha(\bX;g,\bbeta^\ast)\|_{\mP,\infty}\leq l_{K }c_{K }$. Moreover, $\bE\{g(\bX)^4\}<\infty$. (c) The complexity of the basis functions satisfies 
$m^{\bH}_n:=\sqrt{\xi_{K}^2 \log K / n} +\xi_K^2d_2(\pi)^2\wedge\xi_Kd_4(\pi)^2\wedge\allowbreak d_\infty(\pi)^2= o_{\mP}(1)$.
(d) There exist strictly positive constants $\epsilon_1\in(0,1/2]$ and $\epsilon_2>0$ such that for all $a\in\{0,1\}$, $q\in[Q]$, $\bX\in\mathcal{X}$, and $M\in\mathcal{M}$, $\sup_{x\in[0,1]}\omega(x)\vee|\omega'(x)|\vee|\omega''(x)| \leq \epsilon_2$, $\inf_\bX w(\bX)\geq \epsilon_1$, 
$-\omega(t)/(1-t)\leq\omega'(t)\leq\omega(t)/t$ for all $t\in[\epsilon_1,1-\epsilon_1]$,
$\epsilon_1\leq\{\pi(\bX),\widehat{\pi}^q(\bX)\}\leq1-\epsilon_1$, $\epsilon_1\leq\{f(M|a,\bX),\widehat{f}^q(M|a,\bX)\}\leq\epsilon_2$, $|\mu_1(M,\bX)|\vee|\widehat{\mu}^q_1(M,\bX)|\leq\epsilon_2$, $\bE[\{Y-\mu_1(M,\bX)\}^2|1,M,\bX]\leq\epsilon_2$, $\bE(\Omega^2|\bX)\leq \epsilon_2$, $\bE(\check{\Omega}^2|\bX)<\epsilon_2$, and  $\sup_k\sup _{\bX}|b_k(\bX)|\leq \epsilon_2$, where $\Omega:=\phi_n^{\text{NIE}} - w(\bX)g(\bX)$ and $\check{\Omega}:=\phi_d^{\text{NIE}}-w(\bX)$.

\end{assumption}


Assumption \ref{asp:regularityL2}(a) ensures that the basis regressors are not overly collinear by assuming the condition number of $\bG$ is bounded. Assumption \ref{asp:regularityL2}(b) provides the $L^2$ approximation error rate $c_K$ for approximating $\mathcal{G}$ by $\mathcal{G}_n$ and the modulus of continuity $l_K$, which relates the uniform approximation error rate to the $L^2$ rate. In addition, we assume that the true CNIE is fourth-order integrable. Assumption \ref{asp:regularityL2}(c) restricts the growth rate of the basis complexity $K$ and ensures that the weighted Gram matrix is consistent, satisfying $\|\widehat{\bH}-\bH\|_\textop\lessp m^{\bH}_n=o_\mP(1)$, where $\bH=\bE\{\phi_d^{\text{NIE}}\bb(\bX)\bb(\bX)^\top\}=\bE\{w(\bX)\bb(\bX)\bb(\bX)^\top\}$ denotes the population analogue of $\widehat{\bH}$. Further discussion of the rate $m^{\bH}_n$ is provided in Remark \ref{remark:rate-weighted-gram}. Finally, Assumption \ref{asp:regularityL2}(d) imposes a set of boundedness conditions commonly invoked in the literature for estimating CTE \citep{belloni2015some,dml}.  Under these regularity conditions, Theorem \ref{thm:L2-rate-of-convergence} below establishes the $L^2$ error bounds for the proposed weighted orthogonal learners.

\begin{remark}\label{remark:rate-weighted-gram}
Importantly, when the weight function $w(\bX)$ is constant, the unweighted Gram matrix $\widehat{\bG}=\mP_n\{\bb(\bX)\bb(\bX)^\top\}$ converges to $\bG$ at the rate $m_n^{\bG}:=\sqrt{\xi_{K}^2 \log K / n}$, which is the standard rate in the sieve literature \citep{chen2007large}. The same oracle rate holds for the TW and TTW learners because $\widehat{\phi}_d^{\mathrm{NIE}}=A$, so their Gram matrices do not depend on estimated nuisance functions. However, for other general propensity-score-dependent weights, we derive a new convergence rate for the weighted Gram matrix $\widehat{\bH}$, which additionally depends on the fastest one of the following three rates: the squared $L^2$ convergence rate of the propensity score estimator multiplied by $\xi_K^{2}$, the squared $L^4$ convergence rate multiplied by $\xi_K$, or the squared uniform convergence rate. This rate exploits the Neyman orthogonality of $\phi_d$ when bounding $\|\widehat{\bH}-\bH\|_\textop\lessp m^\bH_n$; see Theorem \ref{thm:rate-gram-norm} in the Supplementary Material and its proof based on the matrix Bernstein inequality. Finally, the convergence rate for the weighted Gram matrix attains the oracle rate $\sqrt{\xi_{K}^2 \log K / n}$ when $\xi_Kd_2(\pi)^2=o_{\mP}(n^{-1/2})$, $d_4(\pi)=o_{\mP}(n^{-1/4})$, or $d_{\infty}(\pi)^2=o_\mP(\xi_K/\sqrt{n})$, which is similar to the rate conditions in the debiased machine learning literature.

\end{remark}

\begin{theorem}\label{thm:L2-rate-of-convergence}
Under Assumptions \ref{assump:SI}–\ref{asp:regularityL2}, the error of the weighted orthogonal learner in the $L^2(\mathbb{P})$ norm is bounded by
\begin{align*}
&\| \widehat{g}-g\|_{\mP,2}\lessp  \frac{2}{\epsilon_1+2\lambda\psi_{\min}}\left(\frac{\xi_K}{\sqrt{n}}+\sqrt{\frac{K}{n}}+\sum_{j=0}^3m_{jn}+\lambda\psi_{\max}\right)+c_K,
\end{align*}
where $\psi_{\min}=\psi_1(\bP)$, $\psi_{\max}=\psi_K(\bP)$, $m_{0n}=(l_Kc_K\sqrt{{K}/{n}})\wedge({\xi_Kc_K}/{\sqrt{n}})$, $d^{\Sigma}_q:=d_q(\pi)+d_q(\mu_1)+\sum_{a=0}^1d_q(f_a)$, $m_{1n}=\sqrt{K} \{d_2(\pi)d^{\Sigma}_2+d_2(\mu_1)\sum_{a=0}^1d_2(f_a)\}\wedge \{d_4(\pi)d_4^{\Sigma}+d_4(\mu_1)\sum_{a=0}^1d_4(f_a)\}$, $m_{2n}=\xi_K/\sqrt{n}\{d_2^{\Sigma}+d_{4}(\pi)\}$, and $m_{3n}=\sqrt{K}d_4(\pi)^2$. 
\end{theorem}

The $L^2$ convergence rate in Theorem \ref{thm:L2-rate-of-convergence} can be interpreted as follows. First, $\sqrt{K/n}+c_K$ is the oracle convergence rate, attained when $c_K=o(1)$, $\xi_K\asymp\sqrt{K}$, $\lambda\psi_{\max}=O(n^{-1/2})$, and $\sum_{j=1}^3m_{jn}=O(m_{0n})$. For the Hölder class with smoothness order $s$, i.e., when $\mathcal{G}=\mathcal{H}_s$, the oracle convergence rate is minimax optimal when $K \asymp n^{p/(p+2s)}$. Second, the term $\sum_{j=1}^3m_{jn}$ accounts for first-stage nuisance estimation errors and has a structure similar to error analyses in the debiased machine learning literature \citep{kennedy2022semiparametric}, where the variance term $m_{2n}$ is typically dominated by the bias terms $m_{1n}+m_{3n}$. Importantly, Neyman orthogonality induced by the EIF ensures that the bias term has a second-order product error rate, and hence can be of the same order as the variance term. Third, the regularization bias is controlled by requiring the eigenvalues of the scaled penalty matrix $\lambda\bP$, namely $\lambda\psi_{\min}$ and $\lambda\psi_{\max}$, to converge to zero at appropriate rates. Here, the regularization bias shrinks independently of the other components. Therefore, under the optimal bias-variance tradeoff, it does not affect the best achievable rate and only needs to shrink faster than the slower one of the approximation and estimation errors. Finally, we typically require $c_K=o(1)$, which requires the basis to be carefully constructed. In some settings, the chosen basis may be too restrictive for the approximation error to vanish, for example when the basis consists only of separately additive functions but the true function contains higher-order interactions \citep{belloni2015some}.

\subsection{Uniform convergence}\label{sec:asymptotic;ss:uniform}
We then derive uniform convergence rates for the proposed weighted orthogonal learners. As expected, these results require stronger assumptions than those required for $L^2$ convergence. We summarize the additional assumptions below.
\begin{assumption}\label{assump:uniform-further-boundedness} 
Suppose that $\text{diam}(\calX)$ is bounded above uniformly over $n$. Define $\widetilde{\Omega}:=\phi_n^{\text{NIE}} - \phi_d^{\text{NIE}}g(\bX)$. For some $\nu\geq3$, $\bE\{|g^\ast(\bX)|^\nu\}=O(1)$, $\sup_{\bx}\bE(|\Omega|^\nu|\bX=\bx)\vee\bE(|\check{\Omega}|^\nu|\bX=\bx)\vee\bE(|\widetilde{\Omega}|^\nu|\bX=\bx)=O(1)$, $\xi_K^{2\nu/(\nu-2)}\log K/n=O(1)$, $\log \xi_K^L=O(\log K)$, and $\log \xi_K =O(\log K)$, where $\widetilde{\bb}(\bx) := \bb(\bx) / \| \bb(\bx) \|_2$ denotes the normalized basis functions, and the corresponding Lipschitz constant is defined as $\xi_K^L:= \sup_{\bx\neq\bx'}{\| \widetilde{\bb}(\bx)-\widetilde{\bb}(\bx')\|_2}/{\| \bx-\bx'\|_2}$.
\end{assumption}
Assumption \ref{assump:uniform-further-boundedness} imposes additional uniform boundedness conditions on the conditional $\nu$th moments of the regression errors, requires the projection to be $\nu$th-integrable, and further restricts the growth rate of the basis functions. Under this additional assumption, Theorem \ref{thm:uniform-rate} establishes the uniform convergence rate for the proposed orthogonal learners.
\begin{theorem}\label{thm:uniform-rate}
Under Assumptions \ref{assump:SI}–\ref{assump:uniform-further-boundedness}, the error of the weighted orthogonal learner in the uniform norm is bounded by
\begin{align*}
\| \widehat{g}-g\|_{\mP,\infty}\lessp &\frac{\xi_K}{\sqrt{n}}\Bigg[\frac{\sqrt{n}\sum_{j=1}^3m_{jn}+\lambda \psi_{\max}(\sqrt{n}+\sqrt{K}+\sqrt{n}m_{0n})}{\epsilon_1+\lambda\psi_{\min}}+\\
&\frac{m^{\bH}_n\lb \sqrt{n}\sum_{j=0}^3m_{jn}+m_{4n}+\sqrt{n}\lambda\psi_{\max}\rb}{\lb\epsilon_1+\lambda\psi_{\min}\rb\lb\epsilon_1/2+\lambda\psi_{\min}\rb}+\sqrt{\log K}(1+l_Kc_K)\Bigg]+l_Kc_K,
\end{align*}
where $m_{4n}=n^{1/\nu}\sqrt{\log K}+\sqrt{K}l_Kc_K$.
\end{theorem}
Several remarks on the uniform convergence rate are in order. First, unlike its independent role in the $L^2$ error bound in Theorem \ref{thm:L2-rate-of-convergence}, the regularization bias enters the uniform error bound through $\lambda\psi_{\max}(\sqrt{n}+\sqrt{K}+\sqrt{n}m_{0n}+\sqrt{n}m_n^{\bH})$, which must be controlled by choosing $\lambda\psi_{\max}$ at an appropriate rate. Second, the oracle uniform convergence rate, ${\xi_K}/{\sqrt{n}}\{m_n^{\bG}m_{4n}+\sqrt{\log K}(1+l_Kc_K)\}+l_Kc_K$, is attainable when the conditions in Remark \ref{remark:rate-weighted-gram} hold, $c_{K}=o(1)$, $\xi_K\asymp\sqrt{K}$, $\sqrt{n}\sum_{j=0}^3m_{jn}=O(m_n^{\bG}m_{4n})$, $\lambda\psi_{\max}=o(1)$, $\sqrt{n}\lambda\psi_{\max}=O(m_n^{\bG}m_{4n})$, and $\lambda\psi_{\max}\sqrt{K}=O(\sqrt{n}\sum_{j=1}^3m_{jn})$. This oracle uniform rate is minimax optimal, with order $(\log n / n)^{s/(2s+p)}$ for $\mathcal{G}=\mathcal{H}_s$ when $K\asymp(\log n / n)^{-p/(2s+p)}$ \citep{belloni2015some}.

\subsection{Pointwise and uniform confidence bands}
To facilitate statistical inference, we construct computationally efficient pointwise and uniform confidence bands based on linear sieves. To motivate our results, Supplementary Material Propositions \ref{thm:pointwise-linearization} and \ref{thm:uniform-linearization} show that $\widehat{\bbeta}$ is asymptotically linear up to some vanishing remainder terms. That is, for any unit vector $\widetilde{\bb}$, $    \sqrt{n}\widetilde{\bb}^\top(\widehat{\bbeta}-\bbeta^\ast)=\widetilde{\bb}^\top\bH^{-1}\mG_n[\bb(\bX)\lb\phi_n^{\text{NIE}}-\phi_d^{\text{NIE}}g^\ast(\bX)\rb]+\text{Rem}$.
Using the above limiting linear representation, we define the following $t$-statistic, with covariance matrix $\mathbb{V}=\bH^{-1}\bE[\lb\phi_n^{\text{NIE}}-\phi_d^{\text{NIE}}g^\ast(\bX)\rb^2\bb(\bX)\bb(\bX)^\top]\bH^{-1}$: for $\bx\in\mathcal{X}$, $    T_n(\bx)=\sqrt{n}\{\widehat{g}(\bx)-g(\bx)\}/{\| \mathbb{V}^{1/2}\bb(\bx)\|_2}$.
The following theorem establishes the pointwise asymptotic normality of the $t$-statistic.

\begin{theorem}
\label{thm:pointwise-normal}
Suppose Assumptions \ref{assump:SI}-\ref{asp:regularityL2} and the conditions in Remark \ref{remark:rate-weighted-gram} hold. Furthermore, assume (i) $c_K=o(1)$, $\xi_K\asymp\sqrt{K}$, $\lambda\psi_{\max}=o(n^{-1/2})$,  $\sum_{j=1}^3m_{jn}=O(m_{0n})$, $\sqrt{K/n}=o(1)$, $\sqrt{n}m_{0n}=o(1)$, and $m_n^{\bG}(\xi_K+\sqrt{n}m_{0n})=o(1)$; (ii) the 
Lindeberg condition holds such that $\sup_{\bx} \bE\{\widetilde{\Omega}^2\mathbb{I}(|\widetilde{\Omega}|>\iota)|\bX=\bx\}\to 0$ as $\iota\to\infty$; (iii) $1\lesssim\inf_{\bx} E(\widetilde{\Omega}^2|\bX=\bx)$; and (iv) $\sqrt{n}\alpha(\bx;g,\bbeta^\ast)=o(\| \mathbb{V}^{1/2}\bb(\bx)\|_2)$. Then for any given design point $\bx\in\mathcal{X}$, $\lim_{n\to\infty}\sup_{t\in\mathbb{R}}|\Pr(T_n(\bx)<t)-\Phi(t)|=0$,
where $\Phi$ is the standard normal cumulative distribution function. 

\end{theorem}

Of note, Theorem \ref{thm:pointwise-normal} typically requires the undersmoothing condition in (iv), which is well known in the nonparametric statistics literature \citep{belloni2015some} and essentially requires the approximation error (bias) to be of smaller order than the standard error. Therefore, it is natural to construct the following pointwise confidence bands at significance level $\gamma \in (0,1)$: for any $\bx\in\mathcal{X}$, $[\mathfrak{L}_n,\mathfrak{U}_n]:=[\widehat{g}(\bx)-c_n(1-\gamma)\| \widehat{\mathbb{V}}^{1/2}\bb(\bx)\|_2/\sqrt{n},\widehat{g}(\bx)+c_n(1-\gamma)\| \widehat{\mathbb{V}}^{1/2}\bb(\bx)\|_2/\sqrt{n}]$,
where the covariance matrix estimator is given by $\widehat{\mathbb{V}}=\widehat{\bH}^{-1}\mP_n[\{\widehat{\phi}_n^{\text{NIE}}-\widehat{\phi}_d^{\text{NIE}}\widehat{g}(\bX)\}^2\bb(\bX)\bb(\bX)^\top]\widehat{\bH}^{-1}$ and $c_n(1-\gamma)$ is the $(1-\gamma/2)$-quantile of the standard normal distribution. For the targeted learners, the covariance matrix estimator is adjusted as $\widehat{\mathbb{V}}_{\mathrm{tar}}
=
\widehat{\bH}^{-1}
\mP_n[
\widehat{\boldsymbol{U}}_{\mathrm{prof}}(\mO)
\widehat{\boldsymbol{U}}_{\mathrm{prof}}(\mO)^\top
]
\widehat{\bH}^{-1}$, where $\widehat{\boldsymbol{U}}_{\mathrm{prof}}(\mO)$ is defined in Section~\ref{supp;sec:band-targeting-adjustment} of the Supplementary Material and denotes the profiled score accounting for the estimating equation in Algorithm~\ref{alg:targeted-learner}. Moreover, to enable simultaneous inference, we can set the threshold $c_n(1-\gamma)$ to the $(1-\gamma)$-quantile of the supremum of the $t$-statistic, $\sup_{\bx}|T_n(\bx)|$. However, the exact distribution of $\sup_{\bx}|T_n(\bx)|$ is not analytically tractable, so we follow \cite{chernozhukov2014gaussian} and use the Gaussian bootstrap to approximate it. To proceed, we first show in Theorem \ref{thm:uniform-strong-gaussian} that the $t$-statistic can be well approximated by a limiting Gaussian process, up to an error of order $o_\mP(a_n^{-1})$, and then use the derived limiting Gaussian process to approximate the threshold $c_n(1-\gamma)$ via the Gaussian bootstrap in Algorithm~\ref{alg:gaussian-threshold} of the Supplementary Material.

\begin{theorem}
\label{thm:uniform-strong-gaussian}
Suppose Assumptions \ref{assump:SI}–\ref{assump:uniform-further-boundedness} and the conditions in Remark \ref{remark:rate-weighted-gram} hold. Furthermore, assume (i) $c_{K}=o(1)$, $\xi_K\asymp\sqrt{K}$, $\sqrt{n}\sum_{j=0}^3m_{jn}=O(m_n^{\bG}m_{4n})$, $\lambda\psi_{\max}=o(1)$, $\sqrt{n}\lambda\psi_{\max}=o(a_n^{-1})$, and $\lambda\psi_{\max}\sqrt{K}=O(\sqrt{n}\sum_{j=1}^3m_{jn})$; (ii) $m_n^{\bG}m_{4n}=o_{\mP}(a_n^{-1})$; (iii) $1\lesssim\inf_{\bx} E\{\widetilde{\Omega}^2|\bX=\bx\}$; (iv) $a_n^6K^4\xi_K^2(1+l_K^3c_K^3)^2(\log n)^2/n=o(1)$; and (v) $\sup_{\bx}\sqrt{n}|\alpha(\bx;g,\bbeta^\ast)|\allowbreak /\| \mathbb{V}^{1/2}\bb(\bx)\|_2=o(a_n^{-1})$. Then the following strong Gaussian process approximation holds in the $L^\infty$ norm: $T_n(\bx)=_d{\bb(\bx)^\top\mathbb{V}^{1/2}}/{\| \mathbb{V}^{1/2}\bb(\bx)\|_2}\mathcal{N}(0,\boldsymbol{I}_K)+o_{\mP}(a_n^{-1})$.
\end{theorem}

Similar to pointwise normality, the strong Gaussian process approximation also requires the undersmoothing condition in (v), which ensures that the approximation error (bias) shrinks faster than the standard error. Finally, Theorem \ref{thm:honest-cov-gauss-boot-ci} shows that the proposed uniform confidence bands are asymptotically honest, achieving exact simultaneous $(1-\gamma)$-coverage under appropriate conditions.

\begin{theorem}
\label{thm:honest-cov-gauss-boot-ci}
Suppose that all the conditions in Theorem~\ref{thm:uniform-strong-gaussian} hold, with some $\nu \geq 4$ and $a_n \asymp \sqrt{\log K}$. In addition, assume that $\|g\|_{\mP,\infty}\lesssim1$, $d_\infty^{\Sigma}=o_\mP(1)$, $m_n^{\bG}m_{4n}+\sqrt{\log K}l_Kc_K\lesssim \sqrt{\log K}$, $(m_n^{\bG}+d_\infty^{\Sigma})n^{1/\nu}=o_{\mP}(1)$, and $\xi_K(\log K)^2/n^{1/2-1/\nu}=o(1)$. Then $\Pr\{g(\bx)\in[\mathfrak{L}_n,\mathfrak{U}_n]\text{ for all $\bx\in\mathcal{X}$}\}\to1-\gamma$. 

\end{theorem}

\section{Simulation experiments}\label{sec:simulation}
To demonstrate the performance of the proposed methods, we consider a sample of $n=3000$ individuals and simulate $p=3$ uniformly distributed covariates. For illustration, we assume an additive mediator model with an unspecified mean and Gaussian error, and consider a scenario with highly nonlinear mediator and outcome mean models. Due to space limitations, the complete setup is provided in Supplementary Material Section \ref{supp;sec:supporting-info-simulation}. We use the parametric T-learner in \cite{zhao2025estimation} under the LSEM, denoted by $\widehat{g}^{\text{pT}}(\bX)$, as a benchmark comparator, and compare it with the nonparametric meta-learners proposed in this manuscript, including (i) the nonorthogonal T-learner, denoted by $\widehat{g}^{\text{T}}(\bX)$, and (ii) the six representative orthogonal learners presented in Section \ref{sec:6-repre-ortho-learners}. For the nonparametric meta-learners, the nuisance functions are estimated using SuperLearner \citep{luedtke2016super} with an ensemble consisting of \texttt{"SL.glm"}, \texttt{"SL.earth"}, \texttt{"SL.glmnet"}, \texttt{"SL.nnet"}, and \texttt{"SL.rpart"}. For each orthogonal learner, we consider six Stage 2 sieve smoothers implemented using the \texttt{mgcv} package \citep{wood2017generalized}. These methods are defined by penalized spline bases using the full factorial combination of three choices of $(K_1,K_2)$, with $(K_1,K_2)\in\{(4,3),(5,3),(7,3)\}$, where $K_1$ is the univariate basis complexity and $K_2$ is the basis complexity for two-way interactions, and two penalty specifications: a penalty selected by minimizing the generalized cross-validation (GCV) score and no penalty. We use the same sieve smoothers for the targeting step in Algorithm \ref{alg:targeted-learner}; that is, we set $\check{\mathcal{G}}_n=\mathcal{G}_n$. To ensure a fair comparison, we simulate an independent test sample of size $n_{\text{out}}=10^4$ to evaluate out-of-sample performance across $10^3$ Monte Carlo replications. The performance metrics include (i) the integrated squared error (ISE), defined as $\int \{\widehat{g}(\bX)-g(\bX)\}^2d\mP(\bX)\approx n_{\text{out}}^{-1}\sum_{i=1}^{n_{\text{out}}}\{\widehat{g}(\bX_i)-g(\bX_i)\}^2$; (ii) the pointwise bias, Monte Carlo standard deviation (MCSD), average estimated standard error (AESE), and empirical pointwise coverage probability at four randomly chosen representative points; 
and (iii) the empirical uniform confidence band coverage probability based on a $25^3$ equally spaced grid in the cube $[-1,1]^3$. Notably, uniform inference is generally not available for the T-learner, and pointwise inference for the T-learner is available when a working model is specified for the nuisance functions, such as the linear models in \cite{zhao2025estimation}.

\begin{figure}[ht!]
    \centering
    \includegraphics[width=0.93\linewidth]{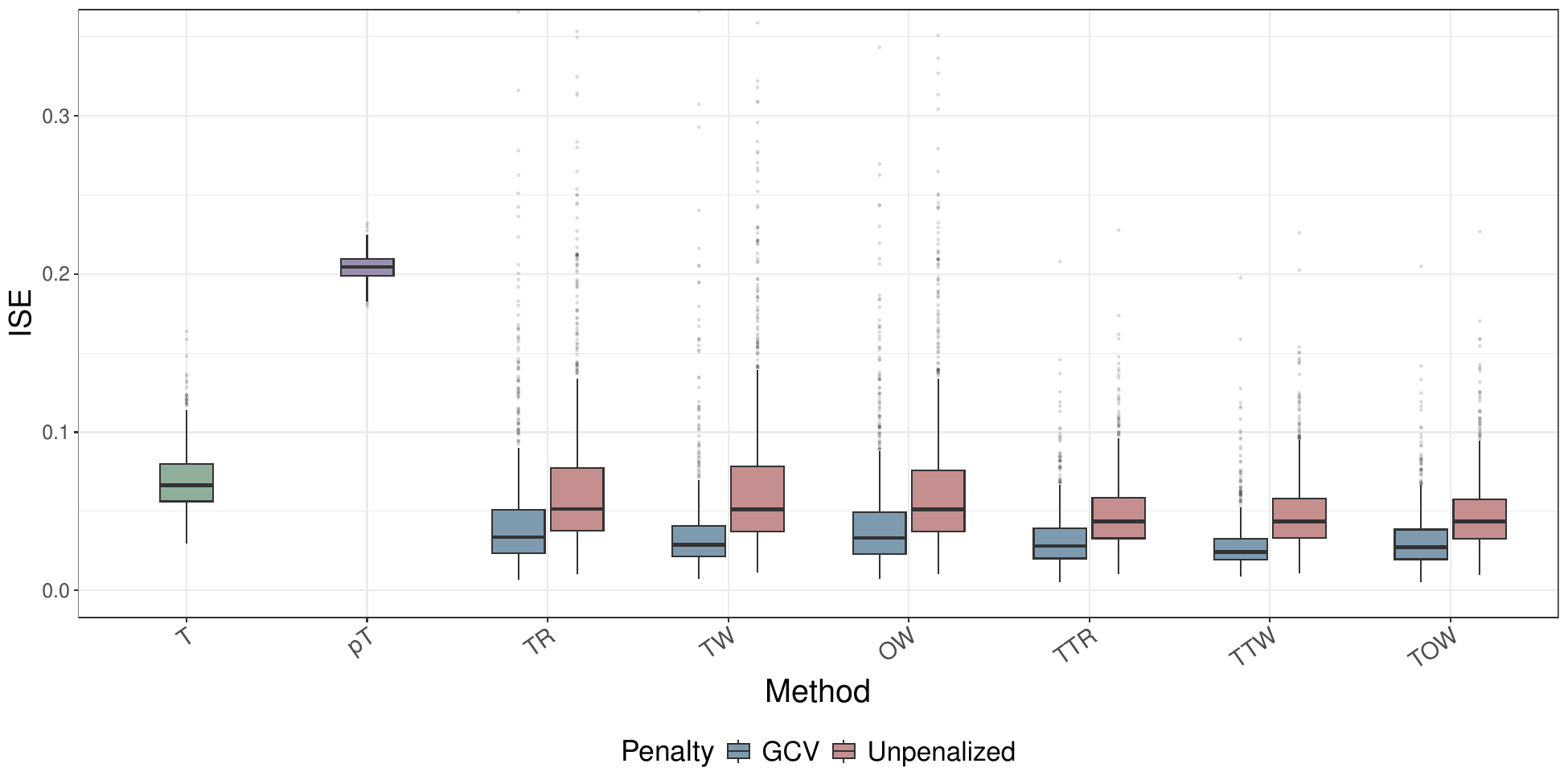}
    \caption{Simulation results presenting box plots of the integrated squared error (ISE) across $10^3$ iterations for the parametric T-learner in \cite{zhao2025estimation}, the nonparametric T-learner, and the six orthogonal learners in Section \ref{sec:6-repre-ortho-learners}. Each orthogonal learner is implemented using $(K_1,K_2)=(4,3)$ and two penalty strategies.}
    \label{fig:sim-results-ise-7,3}
\end{figure}

Several findings are worth noting. First, Figure \ref{fig:sim-results-ise-7,3} and Supplementary Material Figures \ref{fig:sim-results-ise-11,4}--\ref{fig:sim-results-ise-14,5} show box plots of the ISE using basis complexities $(K_1,K_2)=(4,3),(5,3),(7,3)$, respectively. Across all basis complexities, the parametric T-learner has the largest ISE, and the nonparametric T-learner has smaller ISE than the parametric T-learner. However, all nonorthogonal T-learners generally perform substantially worse than the orthogonal learners. This demonstrates that the T-learner is more sensitive to nuisance estimation quality. With misspecified parametric nuisance models, the T-learner is generally not consistent, whereas the nonparametric T-learner based on SuperLearner still suffers from the difficulty of learning complex nuisance functions. In contrast, the orthogonal learners remain locally insensitive to nuisance estimation errors, and among them, the TTW learner appears to have the smallest ISE. The targeted orthogonal learners have more concentrated ISE distributions, with fewer outliers, compared with their nontargeted variants.

Second, Table \ref{tab:pointwise-TPSW} and Tables \ref{tab:pointwise-TR}--\ref{tab:pointwise-TOW} in the Supplementary Material report pointwise inferential performance at four randomly selected representative points for the TTW, TR, TW, OW, TTR, and TOW learners, respectively. The pointwise inference procedures for all proposed orthogonal learners show satisfactory performance across all smoothers. This is reflected in the negligible pointwise bias and empirical coverage probabilities close to the nominal 95\% level. Moreover, the penalized smoother is generally more efficient than the unpenalized variant. Penalization selected by GCV does not appear to induce oversmoothing and therefore maintains at least conservative pointwise coverage. In contrast, the model-based pointwise inference for the parametric T-learner based on \cite{zhao2025estimation} is provided in Supplementary Material Table \ref{tab:pointwise-parametric-T}, which exhibits the highest bias and zero empirical coverage probabilities throughout. This is because the nuisance models are highly nonlinear but are misspecified as linear for the parametric T-learner.

Finally, Supplementary Material Table \ref{tab:uniform-coverage} shows the empirical performance of the uniform inference procedures for the orthogonal learners. Overall, the GCV penalty tuning strategy provides conservative uniform coverage, whereas the unpenalized variant can lead to slight undercoverage. It is worth noting that empirical uniform confidence band coverage is sensitive to localized bias, compared with the generally nominal pointwise coverage. For example, although some unpenalized variants have empirical uniform coverage slightly below 90\%, the average grid noncoverage is generally close to zero ($<0.1\%$), indicating that localized spikes at a small number of points can substantially reduce empirical coverage.

\begin{table}[ht!]
\centering
\caption{Pointwise bias, Monte Carlo standard deviation (MCSD), average estimated standard error (AESE), and empirical pointwise coverage for the TTW learner at four representative points. 
}\label{tab:pointwise-TPSW}
\begingroup
\setlength{\tabcolsep}{6.8pt}
\renewcommand{\arraystretch}{0.53}
\setlength{\aboverulesep}{0.20ex}
\setlength{\belowrulesep}{0.21ex}
\begin{tabular}[t]{ccrrrrrrrr}
\toprule
\multicolumn{2}{c}{ } & \multicolumn{4}{c}{GCV} & \multicolumn{4}{c}{Unpenalized} \\
\cmidrule(l{3pt}r{3pt}){3-6} \cmidrule(l{3pt}r{3pt}){7-10}
$(K_1,K_2)$ & $\bx$ & BIAS & MCSD & AESE & CP(\%) & BIAS & MCSD & AESE & CP(\%)\\
\midrule
\multirow{4}{*}{$(4,3)$} & $\bx_1$ & -0.03 & 0.15 & 0.24 & 99.7 & -0.00 & 0.28 & 0.24 & 92.0\\
 & $\bx_2$ & -0.00 & 0.15 & 0.21 & 99.3 & 0.02 & 0.20 & 0.21 & 96.6\\
 & $\bx_3$ & -0.04 & 0.06 & 0.12 & 99.3 & 0.02 & 0.12 & 0.12 & 95.3\\
 & $\bx_4$ & 0.02 & 0.11 & 0.18 & 99.8 & 0.01 & 0.19 & 0.18 & 93.4\\
\cmidrule{1-10}
\multirow{4}{*}{$(5,3)$} & $\bx_1$ & -0.04 & 0.15 & 0.27 & 99.8 & -0.00 & 0.30 & 0.27 & 92.4\\
 & $\bx_2$ & 0.00 & 0.15 & 0.22 & 99.4 & 0.01 & 0.21 & 0.22 & 96.3\\
 & $\bx_3$ & -0.04 & 0.06 & 0.13 & 99.6 & 0.01 & 0.13 & 0.13 & 95.2\\
 & $\bx_4$ & 0.01 & 0.11 & 0.19 & 99.9 & 0.01 & 0.20 & 0.19 & 93.6\\
\cmidrule{1-10}
\multirow{4}{*}{$(7,3)$} & $\bx_1$ & -0.04 & 0.15 & 0.32 & 100.0 & 0.03 & 0.36 & 0.32 & 92.2\\
 & $\bx_2$ & 0.00 & 0.15 & 0.24 & 99.8 & 0.02 & 0.23 & 0.24 & 96.2\\
 & $\bx_3$ & -0.04 & 0.06 & 0.16 & 99.9 & -0.00 & 0.16 & 0.15 & 94.8\\
 & $\bx_4$ & 0.01 & 0.11 & 0.21 & 100.0 & 0.01 & 0.22 & 0.21 & 95.1\\
\bottomrule
\end{tabular}
\endgroup
\end{table}


\section{Empirical application}\label{sec:data-application}

We illustrate the proposed method using data from the Coronary Artery Risk Development in Young Adults (CARDIA) study \citep{friedman1988cardia}. The supplementary analyses of the PSACR and STAR experiments, presented in Sections~\ref{subsec:data-application-PSACR-002}--\ref{subsec:data-application-STAR} of the Supplementary Material, further illustrate the application of the framework in distinct psychological and educational settings. For illustration, we analyze \(n=2{,}396\) participants who met the eligibility criteria and passed prespecified plausibility screening. The treatment \(A\) is current cigarette smoking at Year 20, with former or never smoking as the control condition. The mediator \(M\) is the natural logarithm of abdominal intermuscular adipose tissue (IMAT) volume at Year 25. The outcome \(Y\) is calibrated systolic blood pressure at Year 30, measured in mm Hg. We consider Year-15 demographic, socioeconomic, smoking history, inflammatory, cardiometabolic, and physical activity characteristics as plausible effect modifiers and mediator--outcome confounders. 
Treatment overlap is weak, with 68.3\% of the cross-fitted propensity score estimates below 0.05, as shown in Figure~\ref{fig:cardia-imat-propensity-overlap} of the Supplementary Material.

Using debiased machine learning based on the triply robust estimator of \citet{tchetgen2012semiparametric}, we estimate a population average TE of \(0.774\) mm Hg (SE, \(2.112\); 95\% CI, \([-3.365, 4.913]\)) and an NIE through IMAT of \(0.544\) mm Hg (SE, \(0.383\); 95\% CI, \([-0.207, 1.295]\)). The estimated mediation proportion is 70.3\%. Clinically, these estimates imply a modest \(0.77\) mm Hg increase in Year-30 systolic blood pressure under current versus former or never smoking, with \(0.54\) mm Hg mediated through smoking-induced changes in Year-25 IMAT. This pattern is consistent with existing evidence that smoking may increase IMAT and that greater IMAT, reflecting greater intermuscular adiposity, is associated with higher blood pressure and hypertension. Because the population average NIE may mask variation across individuals in both the magnitude and direction of the mediated pathway, estimating the CNIE surface is of interest.

For illustration, we implement all six orthogonal learners with sieve smoothers satisfying $\check{\mathcal G}_n=\mathcal G_n$. The nuisance functions are estimated using 10-fold cross-fitting and Super Learner with a library comprising \texttt{SL.mean}, \texttt{SL.glm}, \texttt{SL.glmnet}, \texttt{SL.earth}, and \texttt{SL.gam}. The second-stage sieve models the 11 standardized continuous covariates using thin plate regression splines of dimension three and includes main effects for the eight categorical covariates. To capture plausible heterogeneity while limiting overfitting, it includes a tensor product interaction between baseline BMI and systolic blood pressure and an interaction between sex and baseline waist circumference. We emphasize the TTW learner because it has the smallest average ISE and the most concentrated ISE distribution in the simulation study, and its treatment weighting aligns with the one-sided lack of overlap in CARDIA: 68.3\% of the estimated propensity scores are below 0.05. Supplementary Section~\ref{supp:cardia-additional-information} further compares all six learners.

\begin{figure}[ht!]
   \centering
   \includegraphics[width=0.95\linewidth]
   {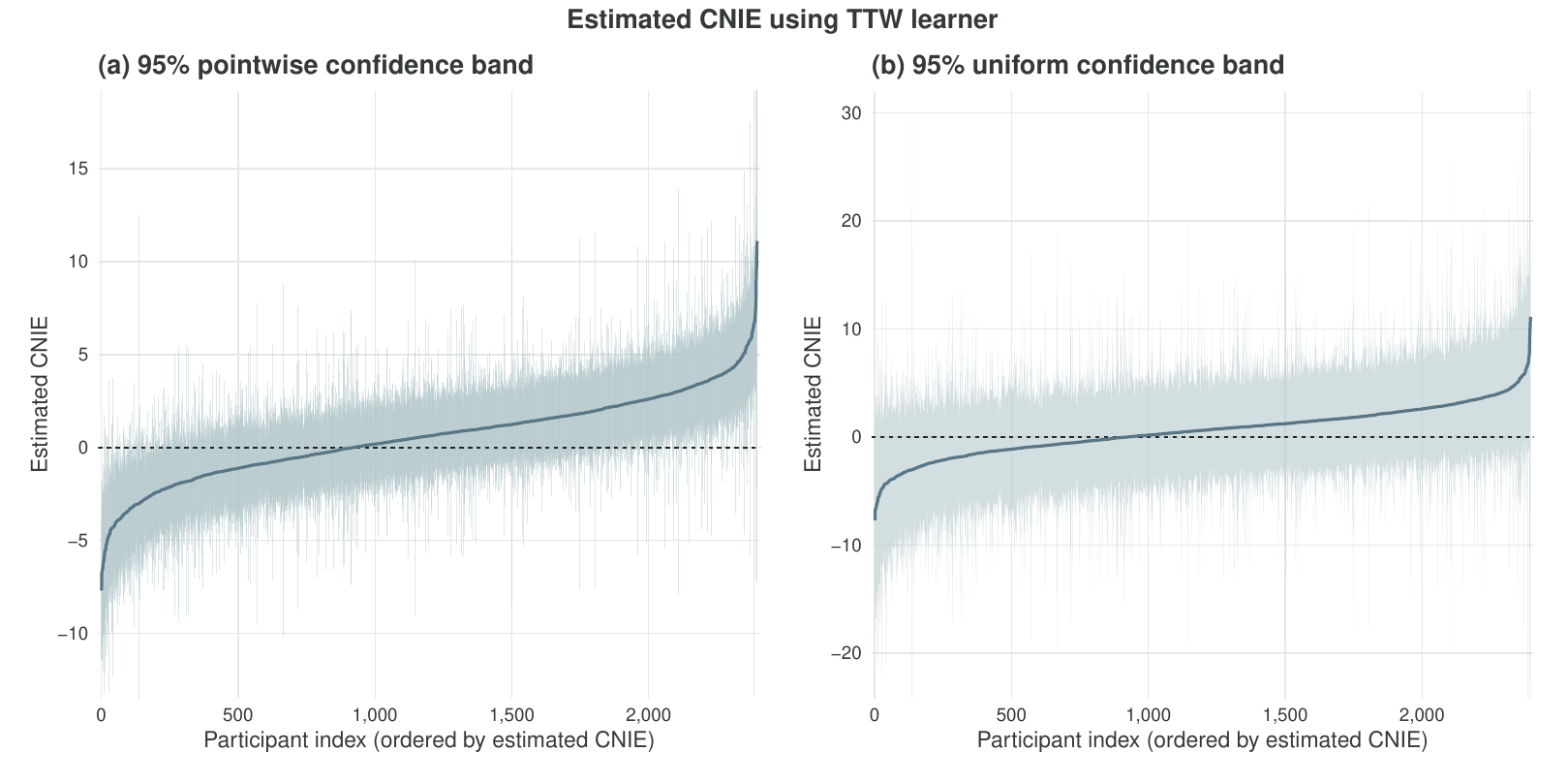}
   \caption{CNIE estimates at the observed covariate profiles in the CARDIA application, obtained using the TTW learner, along with 95\% pointwise and uniform confidence bands.}
   \label{fig:cardia-imat-prediction-ttw}
\end{figure}

Figure \ref{fig:cardia-imat-prediction-ttw} and Supplementary Figures \ref{fig:cardia-imat-prediction-tr}--\ref{fig:cardia-imat-prediction-tow} display the ordered estimated CNIEs, together with pointwise and uniform confidence bands over the 2,396 observed profiles, for the TTW, TR, TW, OW, TTR, and TOW learners. The Gaussian supremum critical value widens the uniform bands relative to the pointwise bands but does not affect their smoothness. For the TTW learner, the estimated CNIEs have a mean of \(0.562\), a median of \(0.623\), and an interquartile range of \([-0.855,1.945]\), with 61.6\% being positive. Targeting reduces mean pointwise band widths by 45.8\% to 60.9\% and mean uniform band widths by 43.9\% to 59.9\% relative to the corresponding untargeted learners.

To interpret the fitted heterogeneity, Figure~\ref{fig:cardia-imat-tree-ttw} presents a shallow regression tree, or fit-the-fit summary, of the TTW-estimated CNIE surface. The first split separates former or current smokers at Year 15 from never smokers, indicating prior smoking status as a prominent modifier. Among former or current smokers, the tree further separates former smokers, with a mean estimated CNIE of \(-1.351\) mm Hg, from current smokers, with a mean near zero (\(0.031\) mm Hg). Among never smokers, household income below versus above \$35,000 yields mean estimated CNIEs of \(0.214\) and \(1.504\) mm Hg, respectively. Supplementary Figures~\ref{fig:cardia-imat-tree-tr}--\ref{fig:cardia-imat-tree-tow} present corresponding summaries for the other learners, identifying baseline hypertension history, household income, age, BMI, high-density lipoprotein cholesterol, and prior smoking as additional descriptive modifiers.

Overall, the CARDIA analysis illustrates the value of estimating heterogeneous mediation effects. The population average NIE estimate suggests a modest positive mediated pathway from current smoking through IMAT to later systolic blood pressure. The individual estimates reveal a more nuanced pattern: the fitted pathway is most negative among participants who were former smokers at Year 15 and most positive among those who had never smoked at Year 15 and had household incomes of at least \$35,000.

\begin{figure}[ht!]
   \centering
   \includegraphics[width=\linewidth]
   {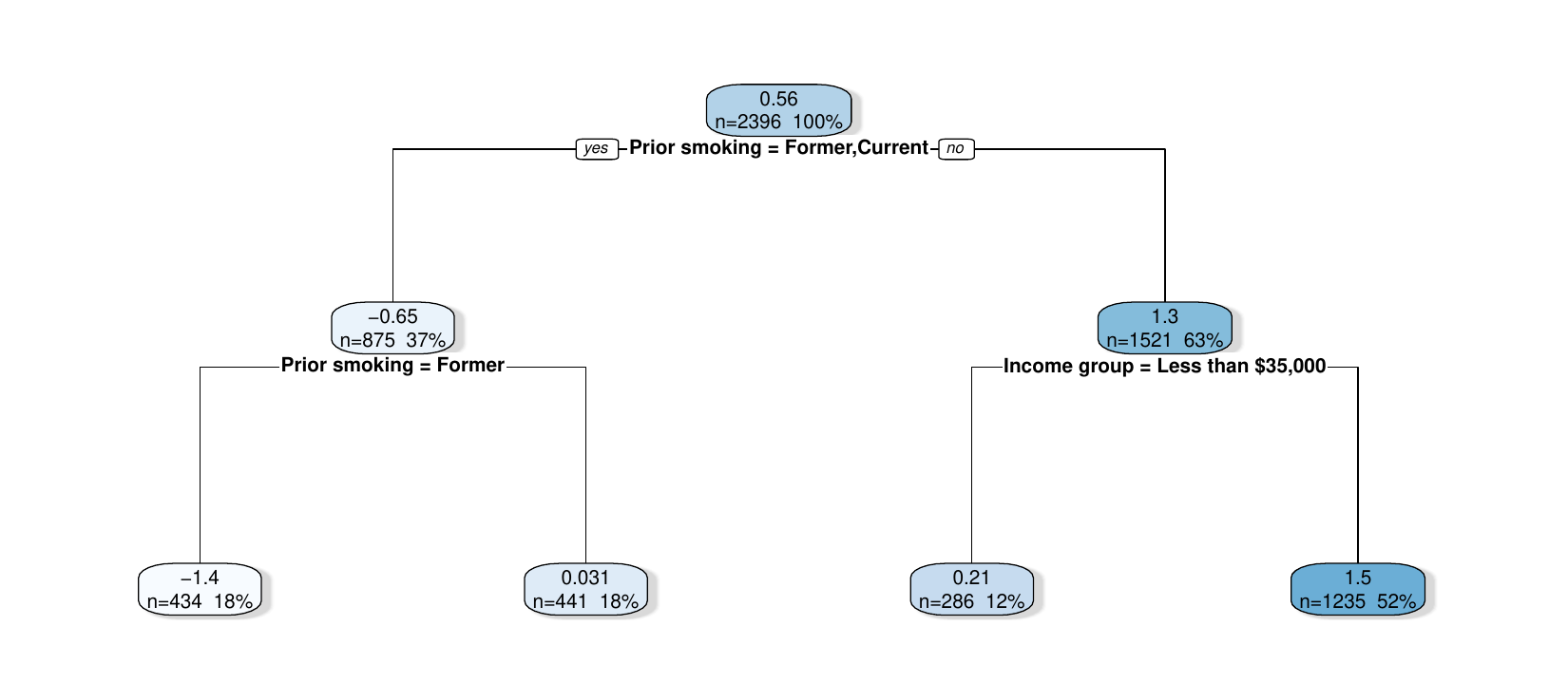}
   \caption{Fit-the-fit plot based on a decision tree summary of the estimated CNIEs from the TTW learner in the CARDIA application.}
   \label{fig:cardia-imat-tree-ttw}
\end{figure}

\section{Concluding remarks}\label{sec:conclusion}

We highlight several directions for future work. First, the targeting step in Algorithm~\ref{alg:targeted-learner} is implemented using the pooled sample and therefore does not preserve cross-fitting. The asymptotic analysis must account for the resulting dependence using suitable empirical process arguments and maximal inequalities, which typically require stronger conditions and may yield weaker guarantees. An alternative procedure following \citet{vansteelandt2025orthogonal} performs targeting only on the training sample, thereby preserving cross-fitting. This approach may reduce the risk of overfitting but provide less finite-sample stabilization. Further investigation of both the targeted learners and this alternative targeting strategy is of interest. Second, the proposed methods rely on the untestable sequential ignorability assumption, which may remain questionable even in randomized experiments. Further sensitivity analysis could therefore be developed.



\section*{Data Availability Statement}
CARDIA data are available from NHLBI BioLINCC (\url{https://biolincc.nhlbi.nih.gov/studies/cardia/}), subject to approval and a data-use agreement. PSACR-002 data are openly available on the Open Science Framework at \url{https://osf.io/jeu73/}. STAR data are openly available in Harvard Dataverse at \url{https://doi.org/10.7910/DVN/SIWH9F}.

\bibliographystyle{chicago}
\bibliography{ref}

\clearpage
\pdfbookmark[0]{Supplementary Material}{supplementary-material}

\setcounter{page}{1}
\renewcommand{\thepage}{S\arabic{page}}
\setcounter{section}{0}
\setcounter{subsection}{0}
\setcounter{subsubsection}{0}
\renewcommand{\thesection}{S\arabic{section}}
\renewcommand{\thesubsection}{\thesection.\arabic{subsection}}
\renewcommand{\thesubsubsection}{\thesubsection.\arabic{subsubsection}}

\setcounter{equation}{0}
\renewcommand{\theequation}{S\arabic{equation}}
\setcounter{figure}{0}
\renewcommand{\thefigure}{S\arabic{figure}}
\setcounter{table}{0}
\renewcommand{\thetable}{S\arabic{table}}
\setcounter{theorem}{0}
\renewcommand{\thetheorem}{S\arabic{theorem}}
\setcounter{prop}{0}
\renewcommand{\theprop}{S\arabic{prop}}
\setcounter{assumption}{0}
\renewcommand{\theassumption}{S\arabic{assumption}}
\setcounter{lemma}{0}
\renewcommand{\thelemma}{S\arabic{lemma}}
\setcounter{algorithm}{0}
\renewcommand{\thealgorithm}{S\arabic{algorithm}}
\setcounter{remark}{0}
\renewcommand{\theremark}{S\arabic{remark}}
\setcounter{coro}{0}
\renewcommand{\thecoro}{S\arabic{coro}}
\setcounter{example}{0}
\renewcommand{\theexample}{S\arabic{example}}
\setcounter{rc}{0}
\renewcommand{\therc}{S\arabic{rc}}
\setcounter{inequality}{0}
\renewcommand{\theinequality}{S\arabic{inequality}}

\renewcommand{\theHsection}{supp.\arabic{section}}
\renewcommand{\theHsubsection}{\theHsection.\arabic{subsection}}
\renewcommand{\theHsubsubsection}{\theHsubsection.\arabic{subsubsection}}
\makeatletter
\@ifundefined{theHequation}
  {\newcommand{\theHequation}{supp.\arabic{equation}}}
  {\renewcommand{\theHequation}{supp.\arabic{equation}}}
\@ifundefined{theHfigure}
  {\newcommand{\theHfigure}{supp.\arabic{figure}}}
  {\renewcommand{\theHfigure}{supp.\arabic{figure}}}
\@ifundefined{theHtable}
  {\newcommand{\theHtable}{supp.\arabic{table}}}
  {\renewcommand{\theHtable}{supp.\arabic{table}}}
\@ifundefined{theHtheorem}
  {\newcommand{\theHtheorem}{supp.\arabic{theorem}}}
  {\renewcommand{\theHtheorem}{supp.\arabic{theorem}}}
\@ifundefined{theHprop}
  {\newcommand{\theHprop}{supp.\arabic{prop}}}
  {\renewcommand{\theHprop}{supp.\arabic{prop}}}
\@ifundefined{theHassumption}
  {\newcommand{\theHassumption}{supp.\arabic{assumption}}}
  {\renewcommand{\theHassumption}{supp.\arabic{assumption}}}
\@ifundefined{theHlemma}
  {\newcommand{\theHlemma}{supp.\arabic{lemma}}}
  {\renewcommand{\theHlemma}{supp.\arabic{lemma}}}
\@ifundefined{theHalgorithm}
  {\newcommand{\theHalgorithm}{supp.\arabic{algorithm}}}
  {\renewcommand{\theHalgorithm}{supp.\arabic{algorithm}}}
\@ifundefined{theHremark}
  {\newcommand{\theHremark}{supp.\arabic{remark}}}
  {\renewcommand{\theHremark}{supp.\arabic{remark}}}
\@ifundefined{theHcoro}
  {\newcommand{\theHcoro}{supp.\arabic{coro}}}
  {\renewcommand{\theHcoro}{supp.\arabic{coro}}}
\@ifundefined{theHexample}
  {\newcommand{\theHexample}{supp.\arabic{example}}}
  {\renewcommand{\theHexample}{supp.\arabic{example}}}
\@ifundefined{theHrc}
  {\newcommand{\theHrc}{supp.\arabic{rc}}}
  {\renewcommand{\theHrc}{supp.\arabic{rc}}}
\@ifundefined{theHinequality}
  {\newcommand{\theHinequality}{supp.\arabic{inequality}}}
  {\renewcommand{\theHinequality}{supp.\arabic{inequality}}}
\makeatother

\allowdisplaybreaks
\spacingset{1.5}

\if1\blind
{
\begin{center}
  {\large\bfseries Supplementary Material for ``Nonparametric heterogeneous causal mediation with orthogonal machine learning'' \mbox{by Tong et al.}\par}
  \vspace{0.55em}
  {\normalsize\bfseries Jiaqi Tong$^{1}$, Yi Zhao$^{2}$, Bhramar Mukherjee$^{1}$, and Fan Li$^{1,*}$\par}
  \vspace{0.15em}
  {\small $^{1}$Department of Biostatistics, Yale School of Public Health, New Haven, CT, USA\par}
  {\small $^{2}$Department of Biostatistics and Health Data Science, Indiana University School of Medicine\par}
\end{center}
}
\fi

\if0\blind
{
  \bigskip
  \bigskip
  \bigskip
  \begin{center}
    {\Large\bf Supplementary Material for ``Nonparametric heterogeneous causal mediation with orthogonal machine learning''}
  \end{center}
  \medskip
}
\fi

\smallskip

\hypersetup{linkcolor=black,linktoc=all}
\setcounter{tocdepth}{2}
\etocsettocstyle{\section*{\contentsname}}{}
\etocsetlocaltop{part}
\makeatletter
\renewcommand*{\l@subsection}{\@dottedtocline{2}{1.5em}{3.5em}}
\makeatother
\footnotesize
\localtableofcontents
\normalsize
\hypersetup{linkcolor=red}
\clearpage

\section{Summary}
\label{sec:intro_supple}
This Supplementary Material is organized as follows.

Section \ref{supp;sec:notation-summary} summarizes the technical notation used in the subsequent statistical analysis.

Section~\ref{supp;sec:band-targeting-adjustment} presents the covariance matrix adjustment used to construct confidence bands for the targeted learners.

Section \ref{supp;sec:orthogonal-cnde} presents the construction of weighted orthogonal learners for the conditional natural direct effect, $\text{CNDE}(\bX)$.

Section \ref{ss:lemma-additional-technical-results-proof} presents useful lemmas with proofs and additional technical results.

Section \ref{sec:proofs-all-technical} presents the proofs of all technical results in the main manuscript and this Supplementary Material.

Section \ref{supp;sec:supporting-info-simulation} presents supporting information for the simulation experiments.

Section~\ref{sec:two-more-data-application} presents two additional empirical applications.

Section \ref{supp;sec:additional-figure} presents additional tables and figures.

Throughout, let $\bH=\bE\{\phi_d^{\text{NIE}}\bb(\bX)\bb(\bX)^\top\}=\bE\{w(\bX)\bb(\bX)\bb(\bX)^\top\}$ and $\bh=\bE\{\bb(\bX)\phi_n^{\text{NIE}}\}=\bE\{\bb(\bX)w(\bX)g(\bX)\}$ denote the population analogues of $\widehat{\bH}$ and $\widehat{\bh}$, respectively.

\section{Some technical notation}\label{supp;sec:notation-summary}
The following technical notation is employed in the main manuscript and is reproduced below for use in the Supplementary Material. Let $[n]=\{1,\ldots,n\}$ denote the set of positive integers up to $n$. Let $\bullet^c$ denote the set complement, and let $A \backslash B = A \cap B^c$ denote the set difference. Let $\lVert \bv \rVert_q=(\sum_{l=1}^L|v_l|^q)^{1/q}$ denote the usual $\ell^q$ norm. Let $\lVert \bullet \rVert_{\mathrm{op}}$ denote the matrix operator or spectral norm and $\lVert f \rVert_{\mP,q} = \left( \int |f|^q d\mP \right)^{1/q}$ the $L^q(\mP)$ norm. Let $d_u(\bullet)=\max_{1\leq q\leq Q}\lVert \widehat{\bullet}^q-\bullet\rVert_{\mathbb{P},u}$ denote the nuisance estimation error in the $L^u(\mathbb{P})$ norm based on the training sample $\mathcal{F}_q^c$ (see Section \ref{sec:OSL}). In particular, $\lVert f \rVert_{\mP,\infty}=\sup_\bx |f(\bx)|$ is the uniform norm. Let $\psi_k(\bullet)$ be the $k$th eigenvalue of a generic matrix $\bullet\in\mathbb{R}^{K\times K}$ such that $\psi_1\leq\ldots\leq\psi_K$. We write $a_n\lesssim b_n$ if $a_n\leq cb_n$ for some constant $c$ independent of $n$, and $a_n\asymp b_n$ if $a_n\lesssim b_n$ and $b_n\lesssim a_n$. We use $X =_d Y$ to denote that two random variables $X$ and $Y$ have the same distribution, and write $X =_d Y + o_\mP(a_n)$ if $(X - Y)/a_n$ converges to zero in probability. We also use $a_n \lesssim_\mP b_n$ to denote that the stochastic sequence $a_n$ is of order at most $b_n$ in probability, i.e., $a_n = O_\mP(b_n)$. The two notations $a_n \lesssim_\mP b_n$ and $a_n = O_\mP(b_n)$ are used interchangeably for convenience. Let $\mP_n(V)=n^{-1}\sum_{i=1}^n V_i$ denote the empirical mean of a generic random object $V$, and let $\mG_n\{f(V)\}=n^{-1/2}\sum_{i=1}^n [f(V_i)-\bE\{f(V_i)\}]$ denote the empirical process indexed by a generic function $f$. Let $\text{diam}(\calX):=\sup_{\bx_1,\bx_2\in\calX}\|\bx_1-\bx_2\|_2$ denote the diameter of the covariate support \(\calX\). Let $a \wedge b:=\min(a,b)$ and $a \vee b:=\max(a,b)$.

\section{Confidence band construction for targeted learners}\label{supp;sec:band-targeting-adjustment}
For the untargeted learners, the covariance matrix estimator is
given by
\[
\widehat{\mathbb{V}}
=
\widehat{\bH}^{-1}
\mP_n\left[
\left\{
\widehat{\phi}_n^{\text{NIE}}
-
\widehat{\phi}_d^{\text{NIE}}\widehat{g}(\bX)
\right\}^2
\bb(\bX)\bb(\bX)^\top
\right]
\widehat{\bH}^{-1}.
\]
For the targeted learners, the fluctuation parameter
$\widehat{\boldsymbol{\epsilon}}$ in Algorithm
\ref{alg:targeted-learner} is estimated using the pooled sample and
therefore contributes to the first-order uncertainty. To account for
this contribution, define $\widehat{w}(\bX)=\omega\{\widehat{\pi}(\bX)\}$ and
\[
\widehat{d}_w(\mO)
=
\frac{A\widehat{w}(\bX)}{\widehat{\pi}(\bX)}
\widehat{r}(M,\bX)
\left\{
Y-\widehat{\mu}_1^\ast(M,\bX)
\right\}.
\]
Let $\widehat{\phi}_{n,\mathrm{tar}}^{\text{NIE}}
=
\widehat{\phi}_{n,\ast}^{\text{NIE}}
+
\widehat{d}_w$ denote the numerator induced by the targeted loss, where
$\widehat{\phi}_{n,\ast}^{\text{NIE}}$ is the original orthogonal
numerator evaluated using $\widehat{\mu}_1^\ast$. Define the
Stage 2 and targeting scores by
\[
\begin{aligned}
\widehat{\boldsymbol{U}}_g(\mO)
=&
\bb(\bX)
\left\{
\widehat{\phi}_{n,\mathrm{tar}}^{\text{NIE}}
-
\widehat{\phi}_d^{\text{NIE}}\widehat{g}(\bX)
\right\},\\
\widehat{\boldsymbol{U}}_{\epsilon}(\mO)
=&
\check{\bb}(\bX)\widehat{d}_w(\mO).
\end{aligned}
\]
Then the corresponding second-order derivatives are
\[
\begin{aligned}
\widehat{\boldsymbol{J}}_{g\epsilon}
=&
\mP_n\left[
\frac{A\widehat{w}(\bX)}{\widehat{\pi}(\bX)}
\bb(\bX)\check{\bb}(\bX)^\top
\right],\\
\widehat{\boldsymbol{J}}_{\epsilon\epsilon}
=&
\mP_n\left[
\frac{A\widehat{w}(\bX)}{\widehat{\pi}(\bX)}
\widehat{r}(M,\bX)
\check{\bb}(\bX)\check{\bb}(\bX)^\top
\right].
\end{aligned}
\]
Profiling the pooled fluctuation parameter out of the joint estimating
equations gives the adjusted score
\[
\widehat{\boldsymbol{U}}_{\mathrm{prof}}(\mO)
=
\widehat{\boldsymbol{U}}_g(\mO)
-
\widehat{\boldsymbol{J}}_{g\epsilon}
\widehat{\boldsymbol{J}}_{\epsilon\epsilon}^{-1}
\widehat{\boldsymbol{U}}_{\epsilon}(\mO).
\]
The covariance matrix estimator for a targeted learner is therefore
\[
\widehat{\mathbb{V}}_{\mathrm{tar}}
=
\widehat{\bH}^{-1}
\mP_n\left[
\widehat{\boldsymbol{U}}_{\mathrm{prof}}(\mO)
\widehat{\boldsymbol{U}}_{\mathrm{prof}}(\mO)^\top
\right]
\widehat{\bH}^{-1}.
\]
When
$\check{\mathcal{G}}_n=\mathcal{G}_n$ with $\check{\bb}(\bX)=\bb(\bX)$, we obtain 
$\boldsymbol{J}_{g\epsilon}
=\boldsymbol{J}_{\epsilon\epsilon}$ because $\bE\{r(M,\bX)\mid A=1,\bX\}=1$.
The profile score then reduces asymptotically to
\[
\bb(\bX)
\left\{
\widehat{\phi}_{n,\mathrm{tar}}^{\text{NIE}}
-
\widehat{d}_w
-
\widehat{\phi}_d^{\text{NIE}}\widehat{g}(\bX)
\right\}
=
\bb(\bX)
\left\{
\widehat{\phi}_{n,\ast}^{\text{NIE}}
-
\widehat{\phi}_d^{\text{NIE}}\widehat{g}(\bX)
\right\}.
\]
Thus, when constructing the covariance matrix under $\check{\mathcal{G}}_n=\mathcal{G}_n$, the adjustment restores the debiasing drift term removed from the targeted loss. Finally, the pointwise and uniform confidence bands are constructed in the same manner as those for the untargeted learners, with the adjustment applied to the covariance matrix estimator. Finally, Algorithm~\ref{alg:gaussian-threshold} below describes the calculation of the uniform confidence band threshold \(c_n(1-\gamma)\).

\begin{algorithm}[ht!] 
\caption{Uniform confidence band threshold $c_n(1-\gamma)$.}
\label{alg:gaussian-threshold}
\begin{algorithmic}[1]
\fontsize{11}{9}\selectfont
\Require Covariance matrix estimate $\widehat{\mathbb{V}}$, dense grid $\{\bx_j\}_{j=1}^J$, number of iterations $B$
\Ensure Threshold $c_n(1-\gamma)$

\For{$b = 1$ to $B$}
    \State Generate a standard normal random vector $\boldsymbol{Z}_b \sim \mathcal{N}(0, \boldsymbol{I}_K)$
    \State Compute $\mathfrak{S}_b = \max_{1\leq j\leq J} | {\bb(\bx_j)^\top \widehat{\mathbb{V}}^{1/2} \boldsymbol{Z}_b}|/{\| \widehat{\mathbb{V}}^{1/2}\bb(\bx_j)\|_2} $
\EndFor
\State \Return $c_n(1-\gamma)$ as the empirical conditional $(1-\gamma)$-quantile of $\{\mathfrak{S}_b\}_{b=1}^B$ given the data.
\end{algorithmic}
\end{algorithm}

\section{A class of weighted orthogonal learners for the conditional natural direct effect}\label{supp;sec:orthogonal-cnde}
When the primary scientific interest lies in the direct pathway that bypasses the mediator, {the same approach can be applied to estimate the conditional natural direct effect}. Specifically, the corresponding weighted orthogonal learners minimize the following loss function:
\begin{align*}
&l_w(\mO;\bGamma,g)\\
=&\phi_d^{\text{NDE}}\lb \eta_{10}(\bX)-\eta_{00}(\bX)-g(\bX)\rb^2-2w(\bX)\lb \phi_{10}(\mO)-\phi_{00}(\mO)- \eta_{10}(\bX)+\eta_{00}(\bX)\rb g(\bX),
\end{align*}
where $\phi_d^{\text{NDE}}=\phi_d^{\text{NIE}}=w(\bX)+\omega'\{\pi(\bX)\}\lb A-\pi(\bX)\rb$ and 
\begin{align*}
    \phi_n^{\text{NDE}}=\{\eta_{10}(\bX)-\eta_{00}(\bX)\}\lb \phi_d^{\text{NDE}}-w(\bX)\rb+w(\bX)\{\phi_{10}(\mO)-\phi_{00}(\mO)\}.
\end{align*}
Using regularized linear sieves, the resulting weighted orthogonal learners have the following closed form:
\begin{align*}
    \widehat{\text{CNDE}}(\bx)=\bb(\bx)^\top\widehat{\bbeta}=\bb(\bx)^\top(\widehat{\bH}+\lambda \bP)^{-1}\widehat{\bh},
\end{align*}
where $\widehat{\bH}=\mP_n\{\widehat{\phi}_d^{\text{NDE}}\bb(\bX)\bb(\bX)^\top\}$ and $\widehat{\bh}=\mP_n\{\bb(\bX)\widehat{\phi}_n^{\text{NDE}}\}$. Furthermore, targeted learning can be implemented exactly as outlined in Algorithm \ref{alg:targeted-learner}, because $\phi_{00}$ has a form analogous to $\phi_{11}$ and the nuisance function $\widehat{\mu}_0$ does not require targeted refinement. Finally, the asymptotic analysis remains largely the same after imposing on $\mu_0$ regularity conditions analogous to those imposed on $\mu_1$.

\section{Some useful lemmas and additional technical results}\label{ss:lemma-additional-technical-results-proof}
Before presenting the proofs of the $L^2$ and uniform limit theory, we state a useful lemma, the \emph{matrix Bernstein concentration inequality} (Theorem 6.1.1 in \cite{tropp2015introductionmatrixconcentrationinequalities}), which provides a non-asymptotic tail bound for mean-zero random matrices and is required to derive the convergence rates for $\widehat{\bH}$. For simplicity, we drop the superscripts in $\phi_d^{\text{NIE}}$ and $\phi_n^{\text{NIE}}$ and write them as $\phi_d$ and $\phi_n$, respectively.
\begin{lemma}[\emph{Matrix Bernstein}]\label{lemma:bernstein-matrix-concentration-inequality}
Let $\{\bM_i\}_{i=1}^n$ be a sequence of independent, mean-zero, symmetric $K\times K$ random matrices such that $\|\bM_i\|_\textop\leq\tau$ almost surely for some $\tau>0$ and all $i\in\{1,\ldots,n\}$. Let $\sigma^2=\|\sum_{i=1}^n\bE(\bM_i^2)\|_\textop$ be the matrix variance statistic of the sum. Then
\begin{align*}
    \bE\left(\left\|\sum_{i=1}^n\bM_i\right\|\right)\leq \sqrt{2\sigma^2\log(2K)}+\frac{1}{3}\tau\log(2K).
\end{align*}
Furthermore, for all $t\geq0$,
\begin{align*}
    \Pr\left( \left\|\sum_{i=1}^n\bM_i\right\|_\textop\geq t\right)\leq 2K\exp\lb\frac{-t^2/2}{\sigma^2+\tau t/3}\rb.
\end{align*}
\end{lemma}

Lemma \ref{lemma:bernstein-matrix-concentration-inequality} is a well-known concentration inequality for matrices that extends the classical Bernstein inequality for scalar random variables. Theorem \ref{thm:rate-gram-norm} below establishes the convergence rate for $\widehat{\bH}$ under the spectral norm. To proceed, recall that $\bH=\bE\{\phi_d\bb(\bX)\bb(\bX)^\top\}=\bE\{w(\bX)\bb(\bX)\bb(\bX)^\top\}$ and $\bh=\bE\{\bb(\bX)\phi_n\}=\bE\{\bb(\bX)w(\bX)g(\bX)\}$ are defined as the population analogues of $\widehat{\bH}$ and $\widehat{\bh}$ evaluated at the true nuisance functions, respectively. Moreover, we define $\widetilde{\bH}=\mP_n\{\phi_d\bb(\bX)\bb(\bX)^\top\}$ and $\widetilde{\bh}=\mP_n\{\bb(\bX)\phi_n\}$. 
\begin{theorem}\label{thm:rate-gram-norm}
Under Assumptions \ref{assump:SI}–\ref{asp:regularityL2}, the convergence rate for the weighted Gram matrix estimator $\widehat{\bH}=\mP_n\{\widehat{\phi}_d\bb(\bX)\bb(\bX)^\top\}$ is given by
\begin{align*}
    \left\|\widehat{\bH}-\bH\right\|_\textop\lessp&\sqrt{\frac{\xi_K^2\log K}{n}}+\min\lb\xi_K^2d_2(\pi)^2,\xi_Kd_4(\pi)^2,d_\infty(\pi)^2\rb:=m^{\bH}_n.
\end{align*}    
In particular, when $w(\bX)=1$ or $w(\bX)=\pi(\bX)$, it follows that $\left\|\widehat{\bH}-\bH\right\|_\textop\lessp \sqrt{{\xi_K^2\log K}/{n}}$.
\end{theorem}

\begin{proof}
    \textbf{Step 0: }The triangle inequality implies
\begin{align*}
    \|\widehat{\bH}-\bH\|_{\text{op}}\leq&\|\widetilde{\bH}-\bH\|_{\text{op}}+\|\widehat{\bH}-\widetilde{\bH}\|_{\text{op}}.
\end{align*}
We first show that $\phi_d$ is almost surely bounded above. To see this, we have
\begin{align}
    \sup _{A,\bX}|\phi_d|=&\sup_{A,\bX}|w(\bX)+\omega'\{\pi(\bX)\}\{A-\pi(\bX)\}|\nonumber\\
    \leq&\sup_\bX [w(\bX)+2|\omega'\{\pi(\bX)\}|]\leq3\epsilon_2.\label{eq:bounded-root-phid}
\end{align}
where the last inequality is due to Assumption \ref{asp:regularityL2}(d). 

\textbf{Step 1: }Now, we apply Lemma \ref{lemma:bernstein-matrix-concentration-inequality} to derive the bound for $\|\widetilde{\bH}-\bH\|_{\text{op}}$. We let $\bZ=\phi_d\bb(\bX)\bb(\bX)^\top$ and note 
\begin{align*}
    \|\widetilde{\bH}-\bH\|_{\text{op}}=\|(\mP_n-\mP)(\bZ)\|_{\text{op}}=n^{-1}\left\|\sum_{i=1}^n \lb \bZ_i-\bE(\bZ)\rb\right\|_{\text{op}},
\end{align*}
which implies that one can set $\bM_i=\bZ_i-\bE(\bZ)$. Clearly, $\{\bM_i\}_{i=1}^n$ are independent, mean-zero, symmetric $K\times K$ random matrices. Also, we have
\begin{align*}
    \left\|\bM\right\|_\textop\leq& \left\|\bZ\right\|_\textop+\left\|\bE(\bZ)\right\|_\textop\\
    \leq&\left\|\bZ\right\|_\textop+\bE(\left\|\bZ\right\|_\textop)\\
    \leq&6\epsilon_2\left\|\bb(\bX)\bb(\bX)^\top\right\|_\textop=6\epsilon_2\|\bb(\bX)\|_2^2\leq6\epsilon_2\xi_K^2:=\tau,
\end{align*}
where the first inequality follows from the triangle inequality, the second from Jensen's inequality, and the third from \eqref{eq:bounded-root-phid} and the identity $\|\bv\bv^\top\|_\textop=\|\bv\|_2^2$. Moreover, we have
\begin{align*}
    \|\bE(\bM_i^2)\|_\textop\leq&\|\bE(\bZ^2)\|_\textop\\
    \leq&\|\bE\{\phi_d^2\|\bb(\bX)\|_2^2\bb(\bX)\bb(\bX)^\top\}\|_\textop\\
    \leq&9\epsilon_2^2\xi_K^2\|\bE\{\bb(\bX)\bb(\bX)^\top\}\|_\textop\\
    \leq&9\epsilon_2^2\xi_K^2\|\bG\|_\textop,
\end{align*}
where the first inequality follows because $\bE(\bM_i^2)\preceq\bE(\bZ_i^2)$ ($\preceq$ denotes the Loewner order; $\|\mathbf{A}\|_\textop\leq\|\mathbf{B}\|_\textop$ for any positive semi-definite matrices $\mathbf{A}$ and $\mathbf{B}$ such that $\mathbf{A}\preceq \mathbf{B}$) for the symmetric matrix $\bZ_i$ and the third from the fact that $\phi_d^2\|\bb(\bX)\|_2^2\bb(\bX)\bb(\bX)^\top\preceq9\epsilon_2^2\xi_K^2\bb(\bX)\bb(\bX)^\top$ and \eqref{eq:bounded-root-phid}. Thus, we obtain $\sigma^2\leq n\|\bE(\bM_i^2)\|_\textop\leq9\epsilon_2^2n\xi_K^2\|\bG\|_\textop$ by the triangle inequality. Finally, using Lemma \ref{lemma:bernstein-matrix-concentration-inequality} gives
\begin{align*}
    \bE\left(\|\widetilde{\bH}-\bH\|_{\text{op}}\right)=n^{-1}\bE\left(\left\|\sum_{i=1}^n\bM_i\right\|_\textop\right)\leq& n^{-1}\lb\sqrt{18\epsilon_2^2n\xi_K^2\|\bG\|_\textop\log(2K)}+2\epsilon_2\xi_K^2\log(2K)\rb\\
    \lesssim& \sqrt{\frac{\xi_K^2\|\bG\|_\textop\log K}{n}}+\frac{\xi_K^2 \log K}{n}\\
    \lesssim&\sqrt{\frac{\xi_K^2\log K}{n}},
\end{align*}
where the last two inequalities are due to the fact that $K=K_n\to\infty$, Assumption \ref{asp:regularityL2}(a), and Assumption \ref{asp:regularityL2}(c).

\textbf{Step 2.1: }We then analyze the second term. We further decompose the second term into two components as follows:
\begin{align}
    \|\widehat{\bH}-\widetilde{\bH}\|_{\text{op}}=&\left\|\sum_{q=1}^Q\frac{n_q}{n}\mP_{n_q}\{(\widehat{\phi}_d-\phi_d)\bb(\bX)\bb(\bX)^\top\}\right\|_{\text{op}}\nonumber\\
    \leq&\left\|\sum_{q=1}^Q\frac{n_q}{n}\mP\lb(\widehat{\phi}_d-\phi_d)\bb(\bX)\bb(\bX)^\top|\widehat{\bGamma}^q\rb\right\|_{\text{op}}+\label{eq:matrixhatH-tildeH-expectation}\\
    &\left\|\sum_{q=1}^Q\frac{n_q}{n}\left(\mP_{n_q}-\mP\right)\lb(\widehat{\phi}_d-\phi_d)\bb(\bX)\bb(\bX)^\top|\widehat{\bGamma}^q\rb\right\|_{\text{op}}.\label{eq:matrixhatH-tildeH-empirical-process}
\end{align}
To bound the term in \eqref{eq:matrixhatH-tildeH-expectation}, it is straightforward to verify, via a pointwise Taylor expansion, that
\begin{align}
    &\mP\{\widehat{\phi}_d^q-\phi_d|\bX,\widehat{\bGamma}^q\}=-\frac{1}{2}\omega''\{\widetilde{\pi}^q(\bX)\}\lb\widehat{\pi}^q(\bX)-\pi(\bX)\rb^2,\label{eq:expected-bias-phid}
\end{align}
where $\widetilde{\pi}^q(\bX)$ lies between $\widehat{\pi}^q(\bX)$ and $\pi(\bX)$. In total, we propose three methods to bound the term in \eqref{eq:matrixhatH-tildeH-expectation}. The first method is
\begin{align*}
    \eqref{eq:matrixhatH-tildeH-expectation}\leq&Q^{-1}\sum_{q=1}^Q\left\|\bE\lb(\widehat{\phi}_d^q-\phi_d)\bb(\bX)\bb(\bX)^\top|\widehat{\bGamma}^q\rb\right\|_{\text{op}}\\
    \leq&\max_q\bE\left[\left\|\frac{1}{2}\omega''\{\widetilde{\pi}^q(\bX)\}\lb\widehat{\pi}^q(\bX)-\pi(\bX)\rb^2\bb(\bX)\bb(\bX)^\top\right\|_{\text{op}}|\widehat{\bGamma}^q\right]\\
    \leq&\frac{\epsilon_2}{2}\max_q\bE\left[\lb\widehat{\pi}^q(\bX)-\pi(\bX)\rb^2\left\|\bb(\bX)\bb(\bX)^\top\right\|_{\text{op}}|\widehat{\bGamma}^q\right]\\
    \leq&\xi_K^2\max_q\|\widehat{\pi}^q-\pi\|_{\mP,2}^2,
\end{align*}
where the second inequality follows from the law of total expectation (LOTE) and Equation \eqref{eq:expected-bias-phid}, the third inequality follows from Assumption \ref{asp:regularityL2}(d), and the fourth inequality follows from the identity $\|\bv \bv^\top\|_{\mathrm{op}}=\|\bv\|_2^2$ for any vector $\bv$ and Assumption \ref{asp:regularityL2}(d). Second, we can also bound the term in \eqref{eq:matrixhatH-tildeH-expectation} as
\begin{align*}
    \eqref{eq:matrixhatH-tildeH-expectation}\leq&Q^{-1}\sum_{q=1}^Q\left\|\bE\lb(\widehat{\phi}_d^q-\phi_d)\bb(\bX)\bb(\bX)^\top|\widehat{\bGamma}^q\rb\right\|_{\text{op}}\\
    =&Q^{-1}\sum_{q=1}^Q\sup_{\|v\|_2=1}\left|\bE\left[ (\widehat{\phi}_d^q-\phi_d)\{v^\top\bb(\bX)\}^2|\widehat{\bGamma}^q\right]\right|\\
    \leq&\max_q\sup_{\|v\|_2=1}\bE\left[ \left|\frac{1}{2}\omega''\{\widetilde{\pi}^q(\bX)\}\right|\lb\widehat{\pi}^q(\bX)-\pi(\bX)\rb^2\{v^\top\bb(\bX)\}^2|\widehat{\bGamma}^q\right]\\
    \leq&\frac{\epsilon_2}{2}\max_q\sup_{\|v\|_2=1}\bE\left[ \lb\widehat{\pi}^q(\bX)-\pi(\bX)\rb^2\{v^\top\bb(\bX)\}^2|\widehat{\bGamma}^q\right]\\
    \leq&\frac{\epsilon_2}{2}\max_q\sup_{\|v\|_2=1}\|\widehat{\pi}^q-\pi\|_{\mP,4}^2\left[\bE\left[ \{v^\top\bb(\bX)\}^4\right]\right]^{1/2}\\   
    \leq&\frac{\epsilon_2}{2}\xi_K\max_q\|\widehat{\pi}^q-\pi\|_{\mP,4}^2\|\bG\|_\textop^{1/2}\\
    \lesssim&\xi_K\max_q\|\widehat{\pi}^q-\pi\|_{\mP,4}^2,
\end{align*}
where the fourth inequality follows from the Cauchy-Schwarz inequality and the fifth from the fact that $|v^\top\bb(\bX)|\leq\|v\|_2\|\bb(\bX)\|_2$. Third, we can also bound the term in \eqref{eq:matrixhatH-tildeH-expectation} as follows:
\begin{align*}
     \eqref{eq:matrixhatH-tildeH-expectation}\leq&Q^{-1}\sum_{q=1}^Q\left\|\bE\lb(\widehat{\phi}_d^q-\phi_d)\bb(\bX)\bb(\bX)^\top|\widehat{\bGamma}^q\rb\right\|_{\text{op}}\\  
    =&Q^{-1}\sum_{q=1}^Q\sup_{\|v\|_2=1}\left|\bE\left[ \frac{1}{2}\omega''\{\widetilde{\pi}^q(\bX)\}\lb\widehat{\pi}^q(\bX)-\pi(\bX)\rb^2\{v^\top\bb(\bX)\}^2|\widehat{\bGamma}^q\right]\right|\\
    \leq&\frac{\epsilon_2}{2}\max_q \|\hat{\pi}^q-\pi\|_{\mP,\infty}^2\|\bG\|_\textop\lesssim\max_q \|\hat{\phi}_d^q-\phi_d\|_{\mP,\infty}.
\end{align*}

\textbf{Step 2.2: }We then apply Lemma \ref{lemma:bernstein-matrix-concentration-inequality} to bound the empirical process term in \eqref{eq:matrixhatH-tildeH-empirical-process} by constructing new matrices $\bM_i$ (different from those considered in Step 1) such that the lemma applies. Specifically, we let 
\begin{align*}
    \bM_i=&\left[\{\widehat{\phi}_{d}^q(A_i,\bX_i)-\phi_d(A_i,\bX_i)\}\bb(\bX_i)\bb(\bX_i)^\top-\right.\\
    &\left.\bE\left[ \{\widehat{\phi}_d^q(A_i,\bX_i)-\phi_d(A_i,\bX_i)\}\bb(\bX_i)\bb(\bX_i)^\top|\widehat{\bGamma}^q\right]\right].
\end{align*}
We then obtain 
\begin{align*}
    \left\|\bM_i\right\|_\textop\leq&\left\|\{\widehat{\phi}_{d}^q(A_i,\bX_i)-\phi_d(A_i,\bX_i)\}\bb(\bX_i)\bb(\bX_i)^\top\right\|_\textop+\\
    &\left\|\bE\left[ \frac{1}{2}\omega''\{\widetilde{\pi}^q(\bX)\}\{\widehat{\pi}^q(\bX)-\pi(\bX)\}^2\bb(\bX_i)\bb(\bX_i)^\top|\widehat{\bGamma}^q\right]\right\|_\textop\\
    \leq&\left\|\widehat{\phi}_{d}^q-\phi_d\right\|_{\mP,\infty}\xi_K^2+\epsilon_2/2\bE\left[ \{\widehat{\pi}^q(\bX)-\pi(\bX)\}^2\left\|\bb(\bX_i)\bb(\bX_i)^\top\right\|_\textop|\widehat{\bGamma}^q\right]\\
    \leq&\left\|\widehat{\phi}_{d}^q-\phi_d\right\|_{\mP,\infty}\xi_K^2+\epsilon_2/2\xi_K^2\|\widehat{\pi}^q-\pi\|_{\mP,2}^2:=\tau,
\end{align*}
where the first inequality follows from \eqref{eq:expected-bias-phid}, the second from Jensen's inequality and Assumption \ref{asp:regularityL2}(d), and the identity $\|\bv \bv^\top\|_{\mathrm{op}}=\|\bv\|_2^2$. Next, we compute $\sigma^2=\|\sum_{i=1}^{n_q}\bE(\bM_i^2|\widehat{\bGamma}^q)\|_\textop$. We have
\begin{align*}
    \bE(\bM_i^2|\widehat{\bGamma}^q)\preceq &\bE\left[\{\widehat{\phi}_{d}^q(A_i,\bX_i)-\phi_d(A_i,\bX_i)\}^2\bb(\bX_i)\bb(\bX_i)^\top\bb(\bX_i)\bb(\bX_i)^\top|\widehat{\bGamma}^q\right]\\
    \preceq&\xi_K^2\|\widehat{\phi}^q_d-\phi_d\|_{\mP,\infty}^2\bG,
\end{align*}
where the first inequality follows because the variance is bounded by the second moment. Thus, we obtain
\begin{align*}
    \sigma^2\leq&n_q\xi_K^2\|\widehat{\phi}^q_d-\phi_d\|_{\mP,\infty}^2\|\bG\|_\textop.
\end{align*}
Finally, applying Lemma \ref{lemma:bernstein-matrix-concentration-inequality} yields
\begin{align*}
    &\bE\left\{\eqref{eq:matrixhatH-tildeH-empirical-process}\right\}\\\leq&\max_qn^{-1}\bE\left(\left\|\sum_{i=1}^{n_q}\bM_i\right\|_\textop|\widehat{\bGamma}^1,\ldots,\widehat{\bGamma}^Q\right)\\
    \leq& \max_qn^{-1}\lb\sqrt{2n_q\xi_K^2\|\widehat{\phi}^q_d-\phi_d\|_{\mP,\infty}^2\|\bG\|_\textop\log(2K)}+\frac{1}{3}\left(\left\|\widehat{\phi}_{d}^q-\phi_d\right\|_{\mP,\infty}+\frac{\epsilon_2}{2}\left\|\widehat{\pi}^q-\pi\right\|^2_{\mP,2}\right)\xi_K^2\log(2K)\rb\\
    \lesssim& \sqrt{\frac{\xi_K^2\log K}{n}},
\end{align*}
where the last inequality is due to Assumption \ref{asp:regularityL2}(a), Assumption \ref{asp:regularityL2}(c), and Assumption \ref{asp:regularityL2}(d). When $w(\bX)=1$ or $w(\bX)=\pi(\bX)$, we have $\widehat{\phi}_d=\phi_d$ (equal to $1$ or $A$, respectively), so $\widehat{\bH}=\widetilde{\bH}$ and the second term in Step 0 vanishes. By Markov's inequality, combining the results from Steps 0, 1, 2.1, and 2.2 completes the proof.
\end{proof}

By Theorem \ref{thm:rate-gram-norm}, we further establish the following lemma for the rate of convergence of the inverse of the penalized Gram matrix, $(\widehat{\bH}+\lambda \bP)^{-1}$.

\begin{lemma}\label{lemma:rate-penalized-gram}
Under Assumptions \ref{assump:SI}–\ref{asp:regularityL2}, it follows that
\begin{equation*}
    \left\|(\widehat{\bH}+\lambda \bP)^{-1}-(\bH+\lambda \bP)^{-1}\right\|_\textop\lessp\frac{\sqrt{{\xi_K^2\log K}/{n}}+\min\lb\xi_K^2d_2(\pi)^2,\xi_Kd_4(\pi)^2,d_\infty(\pi)^2\rb}{\lb\epsilon_1+\lambda\psi_1(\bP)\rb\lb\epsilon_1/2+\lambda\psi_1(
    \bP
    )\rb}.
\end{equation*}
\end{lemma}

\begin{proof}
We have
\begin{align}
    \left\|(\widehat{\bH}+\lambda \bP)^{-1}-(\bH+\lambda \bP)^{-1}\right\|_\textop=&\left\| (\widehat{\bH}+\lambda\bP)^{-1}(\bH-\widehat{\bH})(\bH+\lambda\bP)^{-1}\right\|_\textop\nonumber\\
    \leq&\left\| (\widehat{\bH}+\lambda\bP)^{-1}\right\|_\textop\left\|\bH-\widehat{\bH}\right\|_\textop\left\|(\bH+\lambda\bP)^{-1}\right\|_\textop\nonumber\\
    \leq&\frac{1}{\epsilon_1+\lambda\psi_1(
    \bP
    )}\left\| (\widehat{\bH}+\lambda\bP)^{-1}\right\|_\textop\left\|\bH-\widehat{\bH}\right\|_\textop\nonumber\\
    \lessp&\frac{1}{\lb\epsilon_1+\lambda\psi_1(\bP)\rb\lb\epsilon_1/2+\lambda\psi_1(
    \bP
    )\rb}\left\|\widehat{\bH}-\bH\right\|_\textop\nonumber\\
    \lessp&\displaystyle\frac{\sqrt{{\xi_K^2\log K}/{n}}+\min\lb\xi_K^2d_2(\pi)^2,\xi_Kd_4(\pi)^2,d_\infty(\pi)^2\rb}{\lb\epsilon_1+\lambda\psi_1(\bP)\rb\lb\epsilon_1/2+\lambda\psi_1(
    \bP
    )\rb},\nonumber
\end{align}
where the first equality is due to the matrix identity $\mathbf{A}^{-1}-\mathbf{B}^{-1}=\mathbf{A}^{-1}(\mathbf{B}-\mathbf{A})\mathbf{B}^{-1}$ for any invertible matrices $\mathbf{A}$ and $\mathbf{B}$, the first inequality is submultiplicativity of the operator norm, and the second inequality follows from \eqref{eq:bounded-eigven-gram-H}, and the third inequality follows from the fact that $\left\| (\widehat{\bH}+\lambda\bP)^{-1}\right\|_\textop=1/\psi_1(\widehat{\bH}+\lambda\bP)\leq 1/\lb\psi_1(\widehat{\bH})+\lambda\psi_1(\bP)\rb\lessp 1/\lb\epsilon_1/2+\lambda\psi_1(\bP)\rb$ ($\psi_1(\widehat{\bH})\geq\epsilon_1/2$ with probability approaching one, as shown in Step 2 of the proof of Theorem \ref{thm:L2-rate-of-convergence}), and the last inequality follows from Theorem \ref{thm:rate-gram-norm}.
\end{proof}

The next two propositions give linearizations of $\widehat{\bbeta}$ with pointwise and uniform error bounds, which serve as the basis for deriving the convergence rates and Gaussian approximations for the proposed orthogonal learners. Their proofs are given in Sections \ref{ss:proof-propositon2} and \ref{ss:proof-proposition-uniform-linearization}, respectively.

\begin{prop}[\emph{Pointwise linearization}]\label{thm:pointwise-linearization}
Under Assumptions \ref{assump:SI}–\ref{asp:regularityL2}, for any $\widetilde{\bb}$ in the unit sphere $\{\widetilde{\bb}\in\mathbb{R}^K:\| \widetilde{\bb}\|_2=1\}$, the estimator $\widehat{\bbeta}$ is asymptotically linear with 
\begin{align*}
    &\sqrt{n}\widetilde{\bb}^\top\left(\widehat{\bbeta}-\bbeta^\ast\right)=\widetilde{\bb}^\top\bH^{-1}\mG_n\left[\bb(\bX)\lb\phi_n^{\text{NIE}}-\phi_d^{\text{NIE}}g^\ast(\bX)\rb\right]+\text{Rem}_{1n}(\widetilde{\bb}),    
\end{align*}
where the remainder term $\text{Rem}_{1n}(\widetilde{\bb})$ 
is bounded as follows: 
\begin{align*}
    \text{Rem}_{1n}(\widetilde{\bb})\lessp&\frac{\sqrt{n}\sum_{j=1}^3m_{jn}+\lambda \psi_{\max}(\sqrt{n}+\sqrt{K}+\sqrt{n}m_{0n})}{\epsilon_1+\lambda\psi_{\min}}+\\
    &\frac{m^{\bH}_n\lb \xi_K+\sqrt{n}\sum_{j=0}^3m_{jn}+\sqrt{n}\lambda\psi_K(\bP)\rb}{\lb\epsilon_1+\lambda\psi_{\min}\rb\lb\epsilon_1/2+\lambda\psi_1(
    \bP
    )\rb}.
\end{align*}
\end{prop}

\begin{prop}[\emph{Uniform linearization}]\label{thm:uniform-linearization}
Under Assumptions \ref{assump:SI}–\ref{assump:uniform-further-boundedness}, the estimator $\widehat{\bbeta}$ is asymptotically linear with 
\begin{align*}
    &\sqrt{n}\widetilde{\bb}(\bx)^\top\left(\widehat{\bbeta}-\bbeta^\ast\right)=\widetilde{\bb}(\bx)^\top\bH^{-1}\mG_n\left[\bb(\bX)\lb\phi_n^{\text{NIE}}-\phi_d^{\text{NIE}}g^\ast(\bX)\rb\right]+\text{Rem}_{1n}\{\widetilde{\bb}(\bx)\},\\ 
    &\sqrt{n}\widetilde{\bb}(\bx)^\top\lsb\widehat{\bbeta}-\bbeta^\ast\rsb=\widetilde{\bb}(\bx)^\top\bH^{-1}\mG_n\left[\bb(\bX)\lb\phi_n^{\text{NIE}}-\phi_d^{\text{NIE}}g(\bX)\rb\right]+\text{Rem}_{1n}\{\widetilde{\bb}(\bx)\}+\text{Rem}_{2n}\{\widetilde{\bb}(\bx)\},   
\end{align*}
where the remainder terms $\text{Rem}_{1n}\{\widetilde{\bb}(\bx)\}$ and $\text{Rem}_{2n}\{\widetilde{\bb}(\bx)\}$ capture the impact of the unknown design and first-stage nuisance estimation errors, and the impact of approximation error, respectively, and are bounded as follows. Define $m_{4n}=n^{1/\nu}\sqrt{\log K}+\sqrt{K}l_Kc_K$. Then
\begin{align*}
    &\sup_{\bx}|\text{Rem}_{1n}\{\widetilde{\bb}(\bx)\}|\\
    \lessp&\frac{\sqrt{n}\sum_{j=1}^3m_{jn}+\lambda \psi_K(\bP)(\sqrt{n}+\sqrt{K}+\sqrt{n}m_{0n})}{\epsilon_1+\lambda\psi_1(\bP)}+\frac{m^{\bH}_n\lb \sqrt{n}\sum_{j=0}^3m_{jn}+m_{4n}+\sqrt{n}\lambda\psi_K(\bP)\rb}{\lb\epsilon_1+\lambda\psi_1(\bP)\rb\lb\epsilon_1/2+\lambda\psi_1(
    \bP
    )\rb},\\
   &\sup_\bx|\text{Rem}_{2n}\{\widetilde{\bb}(\bx)\}|\lessp\sqrt{\log K}l_Kc_K.
\end{align*}
\end{prop}

The next lemma is Proposition 6.1 of \cite{belloni2015some}, which provides a sharper bound for the empirical processes that appear in the proofs of uniform linearization and uniform convergence rates. For ease of reference, we state it without proof.

\begin{lemma}\label{lemma:tight-bound-empirical-process-uniform}
Let $(\epsilon_1,\bX_1),\ldots,(\epsilon_n,\bX_n)$ be i.i.d. random vectors taking values in $\mathbb{R}^{p+1}$. Suppose that $\bE(\epsilon_i | \bX_i)=0$ and $\sigma^2:=\sup_{\bx\in\mathcal{X}}\bE(\epsilon_i^2| \bX_i=\bx)<\infty$, where $\mathcal{X}$ denotes the support of $\bX_1$. Let $\mathcal{F}$ be a class of functions on $\mathbb{R}^p$ such that $\bE\{f(\bX_1)^2\}=1$ and $\|f\|_\infty\le b$ for all $f\in\mathcal{F}$. Define $\mathcal{G}:=\{(\epsilon,\bx)\mapsto \epsilon f(\bx): f\in\mathcal{F}\}$. Suppose that there exist constants $A>e^2$ and $V\ge 2$ such that $\sup_{\mathbb{Q}} N\{\mathcal{G},L^2(\mathbb{Q}),\varepsilon\|G\|_{L^2(\mathbb{Q})}\}\le (A/\varepsilon)^V$ for all $0<\varepsilon\le 1$, where the envelope is $G(\epsilon,\bx):=|\epsilon|b$. If $\bE(|\epsilon_1|^\nu)<\infty$ for some $\nu>2$, then $\bE\bigl[\sup_{f\in\mathcal{F}}|\sum_{i=1}^n \epsilon_i f(\bX_i)|\bigr]\le C\bigl[(\sigma+\sqrt{\bE|\epsilon_1|^\nu})\sqrt{nV\log(Ab)}+Vb^{\nu/(\nu-2)}\log(Ab)\bigr]$, where $C$ is a universal constant.

\end{lemma}

Finally, the next two lemmas, Lemmas 4 and 6 of \cite{van2024combining}, provide local maximal inequalities used to bound the empirical processes in the analysis of targeted learning in Theorem \ref{prop:uniform-convergence-drift-term-targeted-learning}. We omit the proofs, which can be found in \cite{van2024combining}, which generalizes the results in \cite{van2011local}. Define $\|\mathcal{F}\|_{\mP,q}:=\sup_{f\in\mathcal{F}}\|f\|_{\mP,q}$ as the envelope for the function class $\mathcal{F}$ under the $L^q(\mP)$ norm.

\begin{lemma}\label{lemma:local-maximam-inequality-1}
Suppose $J(\infty,\mathcal{F},\|\cdot\|_{\mP,\infty}) < \infty$. Then
\[
\bE\left[
\sup_{f\in\mathcal{F}} |(\mP_n-\mP)(f)|
\right]
\lesssim
n^{-1/2}J(\delta,\mathcal{F},\|\cdot\|_{\mP,\infty}),
\]
for any $\delta \geq \|\mathcal{F}\|_{\mP,2} + n^{-1/2}$.   
\end{lemma}

\begin{lemma}\label{lemma:local-maximam-inequality-2}
Let $\mathcal{H}$ be a uniformly bounded function class satisfying
$J(\infty,\mathcal{H},\|\cdot\|_{\mP,\infty}) < \infty$, and let $\mathcal{G}$ be
a function class with $J(\delta,\mathcal{G},\|\cdot\|_{\mP,2})
\lesssim
\delta \sqrt{\check{K}\log(1/\delta)}$ where $\log(1/\|\mathcal{G}\|_{\mP,2}) + \log(1/\|\mathcal{H}\|_{\mP,2})
\lesssim \log n$. Then
\begin{align*}
 &\bE\lb\sup_{f\in\mathcal{H}\mathcal{G}}|\mG_n(f)|\rb
\\
&\lesssim
\|\mathcal{G}\|_{\mP,2}
J
\left(
\frac{\max\{\|\mathcal{H}\|_{\mP,2},n^{-1/2}\}}
{\|\mathcal{G}\|_{\mP,2}},
\mathcal{H},\|\cdot\|_{\mP,\infty}
\right)
+
\|\mathcal{H}\|_{\mP,\infty}
\sqrt{\check{K}\log n}
\max\left\{
\|\mathcal{G}\|_{\mP,2},
\sqrt{\frac{\check{K}\log n}{n}}
\right\},   
\end{align*}
where $\mathcal{H}\mathcal{G}:=\{hg:h\in\mathcal{H},g\in\mathcal{G}\}$.
\end{lemma}

\section{Proofs for all technical results}\label{sec:proofs-all-technical}
For simplicity, we refer to the sequential ignorability assumption as SI.

\subsection{Proof of Proposition \ref{prop:identification-CCME}}
\begin{proof}
{For any $(a_1,a_2)\in\{(0,0),(1,0),(1,1)\}$, the LOTE gives
\begin{align*}
\theta_{a_1a_2}(\bX)
=&\int \bE\{Y(a_1,m)|  M(a_2)=m,\bX\}
 f_{M(a_2)| \bX}(m| \bX)\,dm.
\end{align*}
By Assumption \ref{assump:SI}(i), the pair $\{Y(a_1,m),M(a_2)\}$ is independent of $A$ conditional on $\bX$, and Assumption \ref{assump:SI}(ii) implies $Y(a_1,m)\perp M(a_2)|  A=a_2,\bX$. Therefore,
\begin{align*}
\bE\{Y(a_1,m)|  M(a_2)=m,\bX\}
=&\bE\{Y(a_1,m)|  M(a_2)=m,A=a_2,\bX\}\\
=&\bE\{Y(a_1,m)|  A=a_2,\bX\}\\
=&\bE\{Y(a_1,m)| \bX\},
\end{align*}
where the last equality again follows from Assumption \ref{assump:SI}(i). Moreover,
\begin{align*}
\bE\{Y(a_1,m)| \bX\}
=&\bE\{Y(a_1,m)|  A=a_1,\bX\}\\
=&\bE\{Y(a_1,m)|  M(a_1)=m,A=a_1,\bX\}\\
=&\bE(Y|  M=m,A=a_1,\bX)\\
=&\mu_{a_1}(m,\bX),
\end{align*}
where the first equality follows from Assumption \ref{assump:SI}(i), the second from Assumption \ref{assump:SI}(ii), and the third from consistency. Similarly, Assumption \ref{assump:SI}(i) and consistency yield
\begin{align*}
f_{M(a_2)| \bX}(m| \bX)
=f_{M(a_2)|  A,\bX}(m|  a_2,\bX)
=f(m|  a_2,\bX).
\end{align*}
Substitution gives
\begin{align*}
\theta_{a_1a_2}(\bX)
=\int \mu_{a_1}(m,\bX)f(m|  a_2,\bX)\,dm,
\end{align*}
as required.
}
\end{proof}

\subsection{Proof of Theorem \ref{thm:EIF-smoothed-checkg}}
\subsubsection{Some preliminaries for deriving the EIF}
We present some preliminaries standard in the semiparametric literature \citep{tsiatis2006semiparametric} that are useful for deriving the nonparametric EIFs in the next subsection.

Denote $\theta_d=\bE\{w(\bX)\}$ and $\theta_n=\theta_d\check{g}_{w}=\bE[w(\bX)\{\eta_{11}(\bX)-\eta_{10}(\bX)\}]$. 
We consider a submodel $\mathcal{P}$ parameterized by a univariate parameter $\epsilon$, such that $\mathcal{P}=\{f_\epsilon(\mO):\epsilon\in \mathcal{T}\subseteq\mathbb{R}\}$, with the true distribution attained at $\epsilon=0$. Let $\theta_{r,\epsilon},r\in\{d,n\}$ be the value of $\theta_r$ within the submodel $\mathcal{P}$, with $\theta_r=\theta_{r,0}$ at the truth. Consider the following orthogonal factorization for the joint density $f_\epsilon(\mO)$ within the submodel $\mathcal{P}$
\begin{align*}
    f_\epsilon(\mO)=f_\epsilon(Y|M,A,\bX)f_\epsilon(M|A,\bX)f_\epsilon(A|\bX)f_\epsilon(\bX),
\end{align*}
and the corresponding score functions
\begin{align*}
    S(\mO)=&\frac{\partial \log f_\epsilon(\mO)}{\partial \epsilon}|_{\epsilon=0},\\
    S(Y|M,A,\bX)=&\frac{\partial \log f_\epsilon(Y|M,A,\bX)}{\partial \epsilon}|_{\epsilon=0},\\
    S(M|A,\bX)=&\frac{\partial \log f_\epsilon(M|A,\bX)}{\partial \epsilon}|_{\epsilon=0},\\
    S(A|\bX)=&\frac{\partial \log f_\epsilon(A|\bX)}{\partial \epsilon}|_{\epsilon=0},\\
    S(\bX)=&\frac{\partial \log f_\epsilon(\bX)}{\partial \epsilon}|_{\epsilon=0}.
\end{align*}
Under the nonparametric model, the score $S(\mO)$ lies in the full Hilbert space of mean-zero functions with finite second moments, with $\mathcal{H}:=\{S(\mO):\bE\{S(\mO)\}=0,\bE\{S(\mO)^2\}<\infty\}$. Following \cite{tsiatis2006semiparametric,kennedy2022semiparametric}, the nonparametric EIF, denoted as $\varphi_r^{\text{NIE}}(\mO;\theta_r),r\in\{d,n\}$, for $\theta_r$ is the unique solution, if it exists, to the following differential equation:
\begin{align*}
    \nabla_{\epsilon=0}\theta_{r,\epsilon}=\bE\{\varphi_r^{\text{NIE}}(\mO;\theta_r)S(\mO)\},
\end{align*}
where $\nabla_{\epsilon=0} h_\epsilon=\partial h_\epsilon/\partial \epsilon|_{\epsilon=0}$ denotes the partial derivative with respect to $\epsilon$ evaluated at $\epsilon=0$.

\subsubsection{Main proof of Theorem \ref{thm:EIF-smoothed-checkg}}
\begin{proof}
To begin, we derive the nonparametric EIF for the denominator $\theta_d$. By the chain rule, we have that 
\begin{align*}
    \dot{\theta}_d=T_1+T_2,
\end{align*}
where 
\begin{align*}
    T_1=&\bE\lb \omega(\pi(\bX))S(\bX)\rb,\\
    T_2=&\bE\lb \omega'(\pi(\bX))\dot{\pi}(\bX)\rb.
\end{align*}
For the first term, we have that 
\begin{align*}
    T_1=\bE\left[\lb\omega(\pi(\bX))-\theta_d\rb S(\bX)\right]=\bE\left[\lb\omega(\pi(\bX))-\theta_d\rb S(\mO)\right],
\end{align*}
because $\bE\lb S(\bX)\rb=0$. For the second term, we have that 
\begin{align*}
    T_2=&\bE\lb \omega'(\pi(\bX))\bE\lb AS(A|\bX)|\bX\rb\rb\\
    =&\bE\lb \omega'(\pi(\bX))AS(A|\bX)\rb\\
    =&\bE\left[ \omega'(\pi(\bX))\lb A-\pi(\bX)\rb S(A|\bX)\right]\\
    =&\bE\left[ \omega'(\pi(\bX))\lb A-\pi(\bX)\rb S(\mO)\right].
\end{align*}
To summarize, the nonparametric EIF for $\theta_d$ is given by
\begin{align*}
    \varphi_d^{\text{NIE}}=-\theta_d+\phi_d^{\text{NIE}},
\end{align*}
where
\begin{align*}
    \phi_d^{\text{NIE}}=&\omega(\pi(\bX))+\omega'(\pi(\bX))\lb A-\pi(\bX)\rb.
\end{align*}

Next, we derive the nonparametric EIF for the numerator $\theta_n$. By Proposition \ref{prop:identification-CCME}, we can express $\theta_n$ using the following integral:
\begin{align*}
    \theta_n=&\int \omega(\pi(\bx)) \lb\eta_{11}(\bx) -\eta_{10}(\bx)\rb f(\bx)d\bx.
\end{align*}
Therefore, the chain rule implies that
\begin{align*}
    \dot{\theta}_n=T_4+T_5+T_6-T_7,
\end{align*}
where
\begin{align*}
    T_4=&\int \omega(\pi(\bx)) \lb\eta_{11}(\bx) -\eta_{10}(\bx)\rb \dot{f}(\bx)d\bx,\\
    T_5=&\int \omega'(\pi(\bx))\dot{\pi}(\bx) \lb\eta_{11}(\bx) -\eta_{10}(\bx)\rb f(\bx)d\bx,\\
    T_6=&\int \omega(\pi(\bx)) \dot{\eta}_{11}(\bx) f(\bx)d\bx,\\
    T_7=&\int \omega(\pi(\bx)) \int \lb \dot{\mu}_1(m,\bx)f(m|0,\bx)+\mu_1(m,\bx)\dot{f}(m|0,\bx)\rb dm f(\bx)d\bx.
\end{align*}
We analyze the term $T_4$ as follows: 
\begin{align*}
    T_4=&\int \omega(\pi(\bx)) \lb\eta_{11}(\bx) -\eta_{10}(\bx)\rb S(\bx)f(\bx)d\bx\\
    =&\bE\left[  w(\bX)\lb\eta_{11}(\bX)-\eta_{10}(\bX)\rb S(\bX)\right]\\
    =&\bE\left[  \left[w(\bX)\lb\eta_{11}(\bX)-\eta_{10}(\bX)\rb-\theta_n\right] S(\bX)\right],
\end{align*}
where the last equality follows from the fact that $\bE\lb S(\bX)\rb=0$. 
The analysis of the term $T_2$ implies that
\begin{align*}
    T_5=&\bE\left[ \omega'(\pi(\bX))\lb\eta_{11}(\bX) -\eta_{10}(\bX)\rb\lb A-\pi(\bX)\rb S(\mO)\right].
\end{align*}
We then analyze the term $T_6$ as follows:
\begin{align*}
    T_6=&\int\int \omega(\pi(\bx))  y \lb S(y,1,\bx)-S(1,\bx)\rb f(y|1,\bx) f(\bx)dyd\bx\\
    =&\int\int \omega(\pi(\bx))  y S(y,1,\bx) f(y|1,\bx) f(\bx)dyd\bx-\int \omega(\pi(\bx)) \eta_{11}(\bx) S(1,\bx) f(\bx) d\bx\\
    =&\bE\lb \frac{A}{\pi(\bX)}w(\bX)Y S(Y,A,\bX)\rb-\bE\lb \frac{A}{\pi(\bX)}w(\bX)\eta_{11}(\bX)S(A,\bX)\rb\\
    =&\bE\lb \frac{A}{\pi(\bX)}w(\bX)Y S(\mO)\rb-\bE\lb \frac{A}{\pi(\bX)}w(\bX)\eta_{11}(\bX)S(\mO)\rb\\
    =&\bE\left[ \frac{A}{\pi(\bX)}w(\bX)\lb Y-\eta_{11}(\bX)\rb S(\mO)\right].
\end{align*}
where the fourth equality follows from the fact that $S(Y,A,\bX)=\bE\{S(\mO)|Y,A,\bX\}$. We finally analyze the term $T_7$ as follows: 
\begin{align*}
    T_7=&\int  \int  \int w(\bx) y \lb S(y,m,1,\bx)-S(m,1,\bx)\rb f(y|m,1,\bx) f(m|0,\bx)f(\bx)dy  dm d\bx+\\
    &\int \int w(\bx) \mu_1(m,\bx)\lb S(m,0,\bx)-S(0,\bx)\rb f(m|0,\bx) f(\bx)dm d\bx\\
    =&\bE\lb \frac{A}{\pi(\bX)}w(\bX) r(M,\bX) YS(\mO)\rb-\bE\lb \frac{A}{\pi(\bX)}\mu_1(M,\bX)w(\bX) r(M,\bX) S(\mO)\rb+\\
    &\bE \lb \frac{1-A}{1-\pi(\bX)}w(\bX)\mu_1(M,\bX) S(\mO)\rb-\bE \lb \frac{1-A}{1-\pi(\bX)}w(\bX)\eta_{10}(\bX) S(\mO)\rb\\
    =&\bE\left[w(\bX)\left[\frac{A}{\pi(\bX)}r(M,\bX)\lb Y-\mu_1(M,\bX)\rb+\frac{1-A}{1-\pi(\bX)}\lb \mu_1(M,\bX)-\eta_{10}(\bX)\rb\right]S(\mO)\right].
\end{align*}
Finally, the nonparametric EIF for $\theta_n$ is given by 
\begin{align*}
    \varphi_n^{\text{NIE}}=-\theta_n+\phi_n^{\text{NIE}},
\end{align*}
where
\begin{align*}
\phi_n^{\text{NIE}}=&
\lb\eta_{11}(\bX) -\eta_{10}(\bX)\rb\lb \phi_d^{\text{NIE}}-w(\bX)\rb+w(\bX)\lb \phi_{11}(\mO)-\phi_{10}(\mO)\rb.
\end{align*}

Finally, the nonparametric EIF for $\check{g}_w$ follows from the product rule for the influence-functions operator; see, for example, Trick 2a in Section 3.4.3 of \cite{kennedy2022semiparametric}.

\end{proof}

\subsection{Proof of Theorem \ref{prop:uniform-convergence-drift-term-targeted-learning}}
\begin{proof}
\textbf{Step 0: }We consider the following decomposition:
\begin{align}
    &\sup_{g\in\calG}|\mP_n\{\Delta(\mO;\widehat{\pi},\widehat{f}_0,\widehat{f}_1,\widehat{\mu}_1^\ast,g)\}|\nonumber\\
    \leq&\sup_{g\in\calG}|\mP_n\{\Delta(\mO;\widehat{\pi},\widehat{f}_0,\widehat{f}_1,\widehat{\mu}_1^\ast,\Pi(g))\}|+\sup_{g\in\calG}|\mP_n\{\Delta(\mO;\widehat{\pi},\widehat{f}_0,\widehat{f}_1,\widehat{\mu}_1^\ast,g-\Pi(g))\}|\nonumber\\
    \leq&\sup_{g\in\calG}|\mP_n\{\Delta(\mO;\widehat{\pi},\widehat{f}_0,\widehat{f}_1,\widehat{\mu}_1^\ast,\Pi(g))\}|+\label{eq:targeted-decomposition-term1}\\
    &\sup_{g\in\calG}|\mP_n\{\Delta(\mO;\pi,f_0,f_1,\mu_1,g-\Pi(g))\}|\label{eq:targeted-decomposition-term2}+\\
    &\sup_{g\in\calG}|\mP_n\{\Delta(\mO;\widehat{\pi},\widehat{f}_0,\widehat{f}_1,\widehat{\mu}_1^\ast,g-\Pi(g))-\Delta(\mO;\pi,f_0,f_1,\mu_1,g-\Pi(g))\}|\label{eq:targeted-decomposition-term3}.
\end{align}  
We analyze these three terms in \eqref{eq:targeted-decomposition-term1}-\eqref{eq:targeted-decomposition-term3}.

\textbf{Step 1: }We show that $\eqref{eq:targeted-decomposition-term1}=0$. To see this, we define the empirical projection, characterized by the weighted least squares loss in Algorithm \ref{alg:targeted-learner}, as $\check{\bb}(\bX)^\top\widehat{\boldsymbol{\epsilon}}$, where
\begin{align}\label{eq:definition-hat-epsilon-targeted-flucuation}
    \widehat{\boldsymbol{\epsilon}}=\arg\min_{\boldsymbol{\epsilon}}\mP_n\lbb \frac{A\omega\{\widehat{\pi}(\bX)\}\widehat{r}(M,\bX)}{\widehat{\pi}(\bX)}\lb Y-\widehat{\mu}_1(M,\bX)-\check{\bb}(\bX)^\top\boldsymbol{\epsilon}\rb^2\rbb.
\end{align}
The corresponding score equation for the empirical projection is given by 
\begin{align*}
    \mathbf{0}=&\mP_n\lbb \frac{A\omega\{\widehat{\pi}(\bX)\}\widehat{r}(M,\bX)}{\widehat{\pi}(\bX)}\lb Y-\widehat{\mu}_1(M,\bX)-\check{\bb}(\bX)^\top\widehat{\boldsymbol{\epsilon}}\rb\check{\bb}(\bX)\rbb\\
    =&\mP_n\lbb \frac{A\omega\{\widehat{\pi}(\bX)\}\widehat{r}(M,\bX)}{\widehat{\pi}(\bX)}\lb Y-\widehat{\mu}^\ast_1(M,\bX)\rb\check{\bb}(\bX)\rbb.
\end{align*}
Suppose that $\Pi(g)(\bx)=\check{\bb}(\bx)^\top\boldsymbol{\epsilon}^\ast$ for some fixed  coordinate vector $\boldsymbol{\epsilon}^\ast$. Therefore, it follows that
\begin{align}
    \mathbf{0}^\top\boldsymbol{\epsilon}^\ast=0=&\mP_n\lbb \frac{A\omega\{\widehat{\pi}(\bX)\}\widehat{r}(M,\bX)}{\widehat{\pi}(\bX)}\lb Y-\widehat{\mu}^\ast_1(M,\bX)\rb\check{\bb}(\bX)^\top\boldsymbol{\epsilon}^\ast\rbb\nonumber\\
    =&\mP_n\lb \Delta(\mO;\widehat{\pi},\widehat{f}_0,\widehat{f}_1,\widehat{\mu}_1^\ast,\Pi(g))\rb.\label{eq:score-equation-empirical-projection}
\end{align}
By \eqref{eq:score-equation-empirical-projection}, we obtain $\eqref{eq:targeted-decomposition-term1}=0$. 

\textbf{Step 2: }We analyze the term in \eqref{eq:targeted-decomposition-term2}. Now, we treat $\Delta:g\mapsto \Delta(\mO;\pi,f_0,f_1,\mu_1,T(g))$ as an operator taking values in $\mathcal{G}$. Assumption~\ref{asp:regularityL2}(d) and condition (vii) of Theorem~\ref{prop:uniform-convergence-drift-term-targeted-learning} imply
\[
\left\|\frac{Aw(\bX)r(M,\bX)}{\pi(\bX)}
\{Y-\mu_1(M,\bX)\}\right\|_{\mP,\infty}
\leq \epsilon_3\epsilon_2^2/\epsilon_1^2\leq C_0,
\]
where $C_0:=\max\{1,C_e\epsilon_2^2/\epsilon_1^2\}$. Consider the following class of functions:
\begin{align*}
    \mathcal{R}:=\{\mO\mapsto \Delta(\mO;\pi,f_0,f_1,\mu_1,T(g)):g\in\mathcal{G}\},
\end{align*}
where $T(g)=g-\Pi(g)$ denotes the residual operator induced by $\Pi$. We derive the entropy for $\mathcal{R}$. We note, for any $g_1,g_2\in\mathcal{G}$, $\|\Delta(\mO;\pi,f_0,f_1,\mu_1,T(g_1))-\Delta(\mO;\pi,f_0,f_1,\mu_1,T(g_2))\|_{\mP,\infty}\leq C_0\|T(g_1)-T(g_2)\|_{\mP,\infty}\leq C_0(1+\Lambda_{\check{K}})\|g_1-g_2\|_{\mP,\infty}$ under Assumption \ref{asp:regularityL2}(d), which implies the operator $\Delta$ is $C_0(1+\Lambda_{\check{K}})$-Lipschitz under the sup norm. Therefore, an $\rho/\{C_0(1+\Lambda_{\check{K}})\}$-cover of $\mathcal{G}$ is an $\rho$-cover of $\mathcal{R}$, which suggests that
\begin{align}\label{eq:inequality-convering-number}
    N(\rho,\mathcal{R},\|
\cdot\|_{\mP,\infty})\leq N(\rho/\{C_0(1+\Lambda_{\check{K}})\},\mathcal{G},\|
\cdot\|_{\mP,\infty}).
\end{align}
Further, the inequality \eqref{eq:inequality-convering-number} implies that
\begin{align}
    J(\delta,\mathcal{R},\|
\cdot\|_{\mP,\infty})=&\int_0^\delta \sqrt{\log N(\rho,\mathcal{R},\|
\cdot\|_{\mP,\infty})} d\rho\nonumber\\
\leq&\int_0^\delta \sqrt{\log N(\rho/\{C_0(1+\Lambda_{\check{K}})\},\mathcal{G},\|
\cdot\|_{\mP,\infty})}d\rho\nonumber\\
=&C_0(1+\Lambda_{\check{K}})\int_0^{\delta/\{C_0(1+\Lambda_{\check{K}})\}} \sqrt{\log N(\rho',\mathcal{G},\|
\cdot\|_{\mP,\infty})}d\rho'\nonumber\\
=&C_0(1+\Lambda_{\check{K}})J(\delta/\{C_0(1+\Lambda_{\check{K}})\},\mathcal{G},\|
\cdot\|_{\mP,\infty})\nonumber\\
\lesssim&(1+\Lambda_{\check{K}})\{\delta/(1+\Lambda_{\check{K}})\}^{1-1/(2\tau)}=\delta^{1-1/(2\tau)}(1+\Lambda_{\check{K}})^{1/(2\tau)},\label{eq:bound-on-entropy-integral-delta-R-this}
\end{align}
where the second equality is due to the change of variables $\rho'=\rho/\{C_0(1+\Lambda_{\check{K}})\}$ and the second inequality is due to $J(\delta,\mathcal{G},\|
\cdot\|_{\mP,\infty})\lesssim\delta^{1-1/(2\tau)}$.
Moreover, for $\mathcal{R}$, under Assumption \ref{asp:regularityL2}(d), we have $\sup_{f\in\mathcal{R}}\|f\|_{\mP,2}\leq C_0 c_{\check{K}}$. 
Since $C_0\geq1$, the choice $\delta=C_0(c_{\check K}+n^{-1/2})$ satisfies
$\delta\geq\sup_{f\in\mathcal R}\|f\|_{\mP,2}+n^{-1/2}$.
Lemma~\ref{lemma:local-maximam-inequality-1} therefore gives
\begin{align*}
    \bE\{\eqref{eq:targeted-decomposition-term2}\}=&\frac{1}{\sqrt{n}}\bE\lbb\sup_{g\in\calG}|\mG_n\{\Delta(\mO;\pi,f_0,f_1,\mu_1,g-\Pi(g))\}|\rbb\\
    \lesssim&n^{-1/2}J\{C_0(c_{\check{K}}+n^{-1/2}),\mathcal{R},\|
\cdot\|_{\mP,\infty}\}\\
\lesssim&n^{-1/2}(c_{\check{K}}+n^{-1/2})^{1-1/(2\tau)}(1+\Lambda_{\check{K}})^{1/(2\tau)},
\end{align*}
where the first equality follows from the fact that, by the LOTE, 
\begin{align*}
    &\bE\{\Delta(\mO;\pi,f_0,f_1,\mu_1,g-\Pi(g))\}\\
    =&\bE\lbb  w(\bX)\frac{A}{\pi(\bX)}r(M,\bX)\{Y-\mu_1(M,\bX)\}\lb g(\bX)-\Pi(g)(\bX)\rb\rbb=0.
\end{align*}
Finally, by Markov's inequality, we obtain $\eqref{eq:targeted-decomposition-term2}\lessp n^{-1/2}(c_{\check{K}}+n^{-1/2})^{1-1/(2\tau)}(1+\Lambda_{\check{K}})^{1/(2\tau)}$.

For use in Step 3.2.1, define the class
$\mathcal R_0:=\{\bX\mapsto T(g)(\bX):g\in\mathcal G\}$. Similar arguments imply 
\begin{align}\label{eq:bound-on-entropy-integral-residual-class}
J(\delta,\mathcal R_0,\|\cdot\|_{\mP,\infty})
\leq&(1+\Lambda_{\check K})J\left(\frac{\delta}{1+\Lambda_{\check K}},\mathcal G,\|\cdot\|_{\mP,\infty}\right)\nonumber\\
\lesssim&\delta^{1-1/(2\tau)}(1+\Lambda_{\check K})^{1/(2\tau)}.
\end{align}

\textbf{Step 3: }We analyze the term in \eqref{eq:targeted-decomposition-term3}. We consider a further decomposition of \eqref{eq:targeted-decomposition-term3} as follows:
\begin{align}
    \eqref{eq:targeted-decomposition-term3}\leq&\sup_{g\in\calG}|\mP\{\Delta(\mO;\widehat{\pi},\widehat{f}_0,\widehat{f}_1,\widehat{\mu}_1^\ast,g-\Pi(g))-\Delta(\mO;\pi,f_0,f_1,\mu_1,g-\Pi(g))\}|+\label{eq:targeted-decomposition-term3-1}\\
    &\sup_{g\in\calG}|(\mP_n-\mP)\{\Delta(\mO;\widehat{\pi},\widehat{f}_0,\widehat{f}_1,\widehat{\mu}_1^\ast,g-\Pi(g))-\Delta(\mO;\pi,f_0,f_1,\mu_1,g-\Pi(g))\}|\label{eq:targeted-decomposition-term3-2}.
\end{align}
We first bound the estimation error for the refined nuisance $\widehat{\mu}^\ast_1$ using the estimation error for the first-stage nuisance $\widehat{\mu}_1$. To do this, we define $\underline{\bG}:=\bE\{\check{\bb}(\bX)\check{\bb}(\bX)^\top\}$, $\underline{\widehat{\bH}}:=\mP_n\{A\widehat{w}(\bX)\widehat{r}(M,\bX)/\widehat{\pi}(\bX)\check{\bb}(\bX)\check{\bb}(\bX)^\top\}$, and $\underline{\widehat{\bh}}:=\mP_n[A\widehat{w}(\bX)\allowbreak\widehat{r}(M,\bX)/\widehat{\pi}(\bX)\check{\bb}(\bX)\{Y-\hat{\mu}_1(M,\bX)\}]$. Then we have
\begin{align}
    \|\mu_1-\widehat{\mu}_1^\ast\|_{\mP,2}\leq&\|\widehat{\mu}_1-\mu_1\|_{\mP,2}+\|\widehat{\mu}_1^\ast-\widehat{\mu}_1\|_{\mP,2}\nonumber\\
    =&\|\widehat{\mu}_1-\mu_1\|_{\mP,2}+\|\check{\bb}(\bX)^\top\widehat{\boldsymbol{\epsilon}}\|_{\mP,2}\nonumber\\
    \leq&d_2(\mu_1)+\{\widehat{\boldsymbol{\epsilon}}^\top\underline{\bG}\widehat{\boldsymbol{\epsilon}}\}^{1/2},\nonumber\\
    \lesssim&d_2(\mu_1)+\|\widehat{\boldsymbol{\epsilon}}\|_2,\nonumber\\
    \leq&d_2(\mu_1)+\|\underline{\widehat{\bH}}^{-1}\|_\textop\|\underline{\widehat{\bh}}\|_2\nonumber\\
    =&d_2(\mu_1)+\frac{\|\underline{\widehat{\bh}}\|_2}{\psi_1(\underline{\widehat{\bH}})}\nonumber\\
    \lessp&d_2(\mu_1)+\|\underline{\widehat{\bh}}\|_2\nonumber\\
    \leq&d_2(\mu_1)+\nonumber\\
    &\|\mP_n[A\widehat{w}(\bX)\allowbreak\widehat{r}(M,\bX)/\widehat{\pi}(\bX)\check{\bb}(\bX)\{Y-\mu_1(M,\bX)\}]\|_2+\label{eq:bound-mu1ast-term1}\\
    &\|\mP_n[A\widehat{w}(\bX)\allowbreak\widehat{r}(M,\bX)/\widehat{\pi}(\bX)\check{\bb}(\bX)\{\mu_1(M,\bX)-\widehat{\mu}_1(M,\bX)\}]\|_2.\label{eq:bound-mu1ast-term2}
\end{align}
Here, the first inequality follows from the triangle inequality, and the first equality follows from
$\widehat{\mu}_1^\ast(M,\bX)-\widehat{\mu}_1(M,\bX)
=\check{\bb}(\bX)^\top\widehat{\boldsymbol{\epsilon}}$.
The second inequality uses the definition of $d_2(\mu_1)$ and the identity $\|\check{\bb}(\bX)^\top\widehat{\boldsymbol{\epsilon}}\|_{\mP,2}
=
\left\{
\widehat{\boldsymbol{\epsilon}}^\top
\underline{\bG}
\widehat{\boldsymbol{\epsilon}}
\right\}^{1/2}$, where the fitted coefficient vector is held fixed when taking the population norm.
The third inequality follows from
$\widehat{\boldsymbol{\epsilon}}^\top\underline{\bG}
\widehat{\boldsymbol{\epsilon}}
\leq
\psi_{\check K}(\underline{\bG})
\|\widehat{\boldsymbol{\epsilon}}\|_2^2$
and Assumption~\ref{asp:regularityL2}(a), applied to $\check{\bb}$.
The fourth inequality follows from the normal equations
$\widehat{\boldsymbol{\epsilon}}
=\underline{\widehat{\bH}}^{-1}\underline{\widehat{\bh}}$
and the operator norm inequality.
The second equality uses
$\|\underline{\widehat{\bH}}^{-1}\|_\textop
=1/\psi_1(\underline{\widehat{\bH}})$
on the event that $\underline{\widehat{\bH}}$ is positive definite.
The fifth inequality uses
$\psi_1(\underline{\widehat{\bH}})^{-1}=O_{\mP}(1)$,
as justified below.
Finally, the last inequality follows from the triangle inequality after substituting
$Y-\widehat{\mu}_1=(Y-\mu_1)+(\mu_1-\widehat{\mu}_1)$
into the definition of $\underline{\widehat{\bh}}$.

To justify the fifth inequality and the preceding inverse matrix identities, define
$\widehat a_i
=A_i\widehat w(\bX_i)\widehat r(M_i,\bX_i)/
\widehat\pi(\bX_i)$.
The derivative bounds in Assumption~\ref{asp:regularityL2}(d) imply, for $t,p\in[\epsilon_1,1-\epsilon_1]$,
\[
\omega(t)\geq
\omega(p)\min\left\{\frac{t}{p},\frac{1-t}{1-p}\right\}.
\]
Taking $t=\widehat\pi(\bX_i)$ and $p=\pi(\bX_i)$, and using the bounds on the propensity scores, mediator densities, and $w$ in the same assumption, gives
$cA_i\leq\widehat a_i\leq CA_i$
for some positive constants $c,C$.
Consequently, $\underline{\widehat{\bH}}
\succeq
c\mP_n\{A\check{\bb}(\bX)\check{\bb}(\bX)^\top\}$ and $\bE\{A\check{\bb}(\bX)\check{\bb}(\bX)^\top\}
\succeq
\epsilon_1\underline{\bG}$.
The growth condition (vi) ensures that $\left\|
\mP_n\{A\check{\bb}(\bX)\check{\bb}(\bX)^\top\}
-
\bE\{A\check{\bb}(\bX)\check{\bb}(\bX)^\top\}
\right\|_\textop
=o_{\mP}(1)$, by the same matrix concentration argument used for the unweighted Gram matrix.
Since Assumption~\ref{asp:regularityL2}(a), applied to $\check{\bb}$, bounds
$\psi_1(\underline{\bG})$ away from zero, the preceding inequalities imply that
$\underline{\widehat{\bH}}$ is positive definite with probability tending to one and that
$\psi_1(\underline{\widehat{\bH}})^{-1}=O_{\mP}(1)$.

For each fold $q$, write, for $i\in\mathcal F_q$,
\[
\boldsymbol Z_i^q=
\frac{A_i\widehat w^q(\bX_i)\widehat r^q(M_i,\bX_i)}{\widehat\pi^q(\bX_i)}
\check{\bb}(\bX_i)\{Y_i-\mu_1(M_i,\bX_i)\}.
\]
Conditional on the training sample $\mathcal F_q^c$, these vectors are independent and centered. Jensen's inequality across folds and the conditional second moment identity within each fold imply
\begin{align*}
\bE\left\|\sum_{q=1}^Q\frac{n_q}{n}\mP_{n_q}\boldsymbol Z^q\right\|_2^2
\leq&\sum_{q=1}^Q\frac{n_q}{n}\bE\|\mP_{n_q}\boldsymbol Z^q\|_2^2\\
=&\sum_{q=1}^Q\frac{n_q}{n}\frac{1}{n_q}\bE\|\boldsymbol Z^q\|_2^2
\lesssim\frac{Q\check K}{n}\lesssim\frac{\check K}{n}.
\end{align*}
Thus, Markov's inequality gives $\eqref{eq:bound-mu1ast-term1}\lessp\sqrt{\check K/n}$.

By Step 2.1 of the proof in Section \ref{ss:proof-of-L2-rate}, we obtain
\begin{align*}
    \eqref{eq:bound-mu1ast-term2}\lessp d_2(\mu_1)+\frac{\xi_{\check{K}}}{\sqrt{n}}d_2(\mu_1).
\end{align*}
Combining the preceding bounds yields
\begin{align*}
    \|\mu_1-\widehat{\mu}_1^\ast\|_{\mP,2}\lessp \lsb1+\frac{\xi_{\check{K}}}{\sqrt{n}}\rsb d_2(\mu_1)+\sqrt{\frac{\check{K}}{n}}:=r_{1n}.
\end{align*}
It remains to analyze the terms in \eqref{eq:targeted-decomposition-term3-1} and \eqref{eq:targeted-decomposition-term3-2}.

\textbf{Step 3.1: }For the term in \eqref{eq:targeted-decomposition-term3-1}, we have
\begin{align*}
    &\eqref{eq:targeted-decomposition-term3-1}\\
    \leq&\sup_{g\in\calG}\left|\bE\lbb A\lb \frac{\widehat{w}(\bX)}{\widehat{\pi}(\bX)}\widehat{r}(M,\bX)-\frac{w(\bX)}{\pi(\bX)}r(M,\bX)\rb\{Y-\mu_1(M,\bX)\}\{g-\Pi(g)\}\rbb\right|+\\
    &\sup_{g\in\calG}\left|\bE\lbb A \frac{\widehat{w}(\bX)}{\widehat{\pi}(\bX)}\widehat{r}(M,\bX)\{\mu_1(M,\bX)-\widehat{\mu}_1^\ast(M,\bX)\}\{g-\Pi(g)\}\rbb\right|\\
    =&\sup_{g\in\calG}\left|\bE\lbb A \frac{\widehat{w}(\bX)}{\widehat{\pi}(\bX)}\widehat{r}(M,\bX)\{\mu_1(M,\bX)-\widehat{\mu}_1^\ast(M,\bX)\}\{g-\Pi(g)\}\rbb\right|\\
    \leq&\sup_{g\in\calG}\bE\lbb A \frac{\widehat{w}(\bX)^2}{\widehat{\pi}(\bX)^2}\widehat{r}(M,\bX)^2\{\mu_1(M,\bX)-\widehat{\mu}_1^\ast(M,\bX)\}^2\rbb^{1/2}\left\|g-\Pi(g)\right\|_{\mP,2}\\
    \lesssim&\sup_{g\in\calG}\|\mu_1-\widehat{\mu}_1^\ast\|_{\mP,2}\left\|g-\Pi(g)\right\|_{\mP,2},
\end{align*}
where the first equality follows from the LOTE, the second inequality follows from the Cauchy-Schwarz inequality, and the last inequality follows from Assumption \ref{asp:regularityL2}(d). 

\textbf{Step 3.2: }Due to sample splitting, we consider $\eqref{eq:targeted-decomposition-term3-2}\leq\sum_{q=1}^Q(n_q/n)Z_q$, where
\begin{align*}
   Z_q:=&\sup_{g\in\calG}|(\mP_{n_q}-\mP)\{\Delta(\mO;\widehat{\pi}^q,\widehat{f}_0,\widehat{f}_1,\widehat{\mu}_1^\ast,g-\Pi(g))-\Delta(\mO;\pi,f_0,f_1,\mu_1,g-\Pi(g))|\mathcal{F}_q^c\}|\\
    \leq&\sup_{g\in\calG}|(\mP_{n_q}-\mP) [\{V^{(1)}_{n,q}(\mO;\widehat{\boldsymbol{\epsilon}})+V^{(2)}_{n,q}(\mO)\}\{g-\Pi(g)\}|\mathcal{F}_q^c]|,
\end{align*}
where 
\begin{align*}
    V^{(1)}_{n,q}(\mO;\widehat{\boldsymbol{\epsilon}}):=&A\frac{\widehat{w}^q(\bX)}{\widehat{\pi}^q(\bX)}\widehat{r}^q(M,\bX)\{\mu_1(M,\bX)-\widehat{\mu}^q_1(M,\bX)-\check{\bb}(\bX)^\top\widehat{\boldsymbol{\epsilon}}\},\\
    V^{(2)}_{n,q}(\mO):=&A\lb \frac{\widehat{w}^q(\bX)}{\widehat{\pi}^q(\bX)}\widehat{r}^q(M,\bX)-\frac{w(\bX)}{\pi(\bX)}r(M,\bX)\rb\{Y-\mu_1(M,\bX)\}.
\end{align*}

\textbf{Step 3.2.1: }To bound the term involving $V^{(1)}_{n,q}$, we note that, conditional on $\mathcal{F}_q^c$, the randomness of $V^{(1)}_{n,q}$ comes from $\widehat{\boldsymbol{\epsilon}}$. To remove the randomness induced by the sieve coefficients $\widehat{\boldsymbol{\epsilon}}$, we embed $V^{(1)}_{n,q}$ into the following deterministic (conditional on $\mathcal{F}_q^c$) class of functions: for some sufficiently large constant $M_{\upsilon}$ that depends on $\upsilon\in(0,1)$, 
\begin{align*}
    \mathcal{T}_q:=\lb\mO\mapsto V^{(1)}_{n,q}(\mO;{\boldsymbol{\epsilon}}):\|\boldsymbol{\epsilon}\|_\infty\leq M_{\upsilon},\|V^{(1)}_{n,q}\|_{\mP,2}\leq M_{\upsilon}\lb \lsb1+\frac{\xi_{\check{K}}}{\sqrt{n}}\rsb d_2(\mu_1)+\sqrt{\frac{\check{K}}{n}}\rb\rb.
\end{align*}

Next, we will show that, with high probability at least $1-\upsilon$, one can always find a sufficiently large $M_{\upsilon}$ and define the corresponding $\calT_q$ such that $V^{(1)}_{n,q}(\mO)\in\mathcal{T}_q$. 

First, let $c_G:=\inf_n\psi_1(\underline{\bG})>0$, as guaranteed by Assumption~\ref{asp:regularityL2}(a) for the targeting basis $\check{\bb}$. For every realization of $\widehat{\boldsymbol{\epsilon}}$, $\|\check{\bb}(\bX)^\top\widehat{\boldsymbol{\epsilon}}\|_{\mP,2}^2
=\widehat{\boldsymbol{\epsilon}}^\top\underline{\bG}\widehat{\boldsymbol{\epsilon}}
\geq c_G\|\widehat{\boldsymbol{\epsilon}}\|_2^2
\geq c_G\|\widehat{\boldsymbol{\epsilon}}\|_\infty^2$.
Consequently,
\begin{align*}
\|\widehat{\boldsymbol{\epsilon}}\|_\infty
\leq& c_G^{-1/2}\|\check{\bb}(\bX)^\top\widehat{\boldsymbol{\epsilon}}\|_{\mP,2}\\
\leq& c_G^{-1/2}\|\check{\bb}(\bX)^\top\widehat{\boldsymbol{\epsilon}}\|_{\mP,\infty}\\
\leq& c_G^{-1/2}\left(\|\widehat{\mu}_1^\ast\|_{\mP,\infty}
+\|\widehat{\mu}_1\|_{\mP,\infty}\right)\lesssim1,
\end{align*}
where the last step uses the assumed uniform boundedness of both outcome regressions. Thus one may choose a constant $M_{\upsilon,1}$ such that $\|\widehat{\boldsymbol{\epsilon}}\|_\infty\leq M_{\upsilon,1}$ on that boundedness event.

Second, we note 
\begin{align*}
    \| V^{(1)}_{n,q}(\mO;\widehat{\boldsymbol{\epsilon}})\|_{\mP,2}\lesssim& \|\widehat{\mu}_1^q-\mu_1+\check{\bb}(\bX)^\top\widehat{\boldsymbol{\epsilon}}\|_{\mP,2}\\
    \leq&\|\widehat{\mu}_1^q-\mu_1\|_{\mP,2}+ \|\check{\bb}(\bX)^\top\widehat{\boldsymbol{\epsilon}}\|_{\mP,2}\\
    \lessp&\lsb1+\frac{\xi_{\check{K}}}{\sqrt{n}}\rsb d_2(\mu_1)+\sqrt{\frac{\check{K}}{n}},
\end{align*}
where the last inequality follows from \eqref{eq:bound-mu1ast-term1} and \eqref{eq:bound-mu1ast-term2}. Thus, for all $\upsilon\in(0,1)$, one can pick $M_{\upsilon,2}$ such that $\Pr(\| V^{(1)}_{n,q}(\mO;\widehat{\boldsymbol{\epsilon}})\|_{\mP,2}\leq M_{\upsilon,2}\lsb1+{\xi_{\check{K}}}/{\sqrt{n}}\rsb d_2(\mu_1)+\sqrt{{\check{K}}/{n}})\geq 1-\upsilon$. Finally, set $M_{\upsilon}=\max(M_{\upsilon,1},M_{\upsilon,2})$.

Moreover, the class $\calT_q$ is a subset of a $(\check{K}_n+1)$-dimensional linear space with uniformly bounded coefficients. Following the arguments in Lemma 7 of \cite{van2024combining}, which uses Lemma 2.6.15 and Theorem 2.10.20 of \cite{vanderVaartWellner1996}, the class $\calT_q$ is VC-subgraph with dimension $O(\check{K})$, and its entropy integral satisfies $J(\delta,\calT_q,\|\cdot\|_{\mP,\infty})\lesssim \delta \sqrt{\check{K}\log(1/\delta)}$. Thus, we have verified all conditions required to apply Lemma \ref{lemma:local-maximam-inequality-2}. Therefore, by Markov's inequality and Lemma \ref{lemma:local-maximam-inequality-2}, we obtain
\begin{align*}
   &\sup_{g\in\mathcal{G}}\left|(\mP_{n_q}-\mP)\lbb V_{n,q}^{(1)}(\mO;\widehat{\boldsymbol{\epsilon}})\{g-\Pi(g)\}\rbb\right|\\
   \lessp&\frac{\|\mathcal{T}_q\|_{\mP,2}}{\sqrt{n}}
J
\left(
\frac{c_{\check{K}}\vee n^{-1/2}}
{\|\mathcal{T}_q\|_{\mP,2}},
\mathcal{R}_0,\|\cdot\|_{\mP,\infty}
\right)
+
l_{\check{K}}c_{\check{K}}
\sqrt{\frac{\check{K}\log n}{n}}
\left(
\|\mathcal{T}_q\|_{\mP,2}\vee
\sqrt{\frac{\check{K}\log n}{n}}
\right)\\
\lesssim&n^{-1/2}\{\|\calT_q\|_{\mP,2}(1+\Lambda_{\check{K}})\}^{1/(2\tau)}
\lsb{c_{\check{K}}\vee n^{-1/2}}\rsb^{1-1/(2\tau)}
+
l_{\check{K}}c_{\check{K}}
\sqrt{\frac{\check{K}\log n}{n}}
\left(
\|\calT_q\|_{\mP,2}\vee
\sqrt{\frac{\check{K}\log n}{n}}
\right),
\end{align*}
where the last inequality follows from \eqref{eq:bound-on-entropy-integral-residual-class}. By the definition of $\calT_q$, $\|\calT_q\|_{\mP,2}\lesssim r_{1n}$.

\textbf{Step 3.2.2: }
To bound the term involving $V^{(2)}_{n,q}(\mO)$, we note that, conditional on $\mathcal{F}_q^c$, $V^{(2)}_{n,q}$ is a deterministic function of the data $\mO$. We follow arguments analogous to those used in Step 2 and apply Lemma \ref{lemma:local-maximam-inequality-1} to obtain the bound. Conditional on $\mathcal{F}_q^c$, we consider the following class of functions:
\begin{align*}
    \mathcal{R}^\ast:=\{\mO\mapsto V^{(2)}_{n,q}(\mO)T(g):g\in\mathcal{G}\},
\end{align*}
We derive the entropy for $\mathcal{R}^\ast$. We note, for any $g_1,g_2\in\mathcal{G}$ and some sufficiently large constant $C>0$ depending only on $\epsilon_1,\epsilon_2,\epsilon_3$, 
\begin{align*}
    &\|V^{(2)}_{n,q}(\mO)T(g_1)-V^{(2)}_{n,q}(\mO)T(g_2)\|_{\mP,\infty}\\
    \leq &C\lsb\|\widehat{\pi}^q-\pi\|_{\mP,\infty}+\sum_{a=0,1}\|\widehat{f}_a^q-f_a\|_{\mP,\infty}\rsb\|T(g_1)-T(g_2)\|_{\mP,\infty}\\
    \leq &C\lb d_\infty(\pi)+\sum_{a=0,1}d_\infty(f_a)\rb(1+\Lambda_{\check{K}})\|g_1-g_2\|_{\mP,\infty}
\end{align*}
under Assumption \ref{asp:regularityL2}(d) and conditions (i) and (vii) of Theorem
\ref{prop:uniform-convergence-drift-term-targeted-learning}, which implies the operator $V^{(2)}_{n,q}(\mO)T(g)$ is $L_K^V:=C\{ d_\infty(\pi)+\sum_{a=0,1}d_\infty(f_a)\}(1+\Lambda_{\check{K}})$-Lipschitz under the sup norm. Therefore, a $\rho/L_K^V$-cover of $\mathcal{G}$ is a $\rho$-cover of $\mathcal{R}^\ast$, which suggests that
\begin{align*}
    N(\rho,\mathcal{R}^\ast,\|
\cdot\|_{\mP,\infty})\leq N(\rho/L_K^V,\mathcal{G},\|
\cdot\|_{\mP,\infty}),
\end{align*}
and 
\begin{align}
    J(\delta,\mathcal{R}^\ast,\|
\cdot\|_{\mP,\infty})
\leq&L_K^VJ(\delta/L_K^V,\mathcal{G},\|
\cdot\|_{\mP,\infty})\nonumber\\
\lesssim&\{ d_\infty(\pi)+\sum_{a=0,1}d_\infty(f_a)\}(1+\Lambda_{\check{K}})\left\{\frac{\delta}{\{d_\infty(\pi)+\sum_{a=0,1}d_\infty(f_a)\}(1+\Lambda_{\check{K}})}\right\}^{1-1/(2\tau)}\nonumber\\
=&\delta^{1-1/(2\tau)}[\{ d_\infty(\pi)+\sum_{a=0,1}d_\infty(f_a)\}(1+\Lambda_{\check{K}})]^{1/(2\tau)}.\nonumber
\end{align}

If $d_\infty(\pi)+\sum_{a=0,1}d_\infty(f_a)=0$, then $V^{(2)}_{n,q}=0$ and the corresponding entropy contribution is zero; the calculation above is for a nonzero nuisance error. Moreover, for $\mathcal{R}^\ast$, under Assumption \ref{asp:regularityL2}(d) and condition (vii), we have $\sup_{f\in\mathcal{R}^\ast}\|f\|_{\mP,2}\lesssim\{ d_\infty(\pi)+\sum_{a=0,1}d_\infty(f_a)\}c_{\check{K}}$ and $\sup_{f\in\mathcal{R}^\ast}\|f\|_{\mP,2}\lesssim\{ d_2(\pi)+\sum_{a=0,1}d_2(f_a)\}l_{\check{K}}c_{\check{K}}$, which implies $\sup_{f\in\mathcal{R}^\ast}\|f\|_{\mP,2}\lesssim r_{2n}:=[\{ d_\infty(\pi)+\sum_{a=0,1}d_\infty(f_a)\}c_{\check{K}}]\wedge [\{ d_2(\pi)+\sum_{a=0,1}d_2(f_a)\}l_{\check{K}}c_{\check{K}}]$.

To apply Lemma
\ref{lemma:local-maximam-inequality-1} conditional on $\mathcal{F}_q^c$, define, for $u\in\{2,\infty\}$,
\[
D_{u,q}
:=
\|\widehat{\pi}^q-\pi\|_{\mP,u}
+\sum_{a=0,1}\|\widehat{f}_a^q-f_a\|_{\mP,u},
\]
and $r_{2n,q}
:=
(D_{\infty,q}c_{\check{K}})
\wedge
(D_{2,q}l_{\check{K}}c_{\check{K}})$. These quantities are nonrandom conditional on $\mathcal{F}_q^c$.
The preceding entropy and $L^2$ bounds also hold with
$d_u(\pi)+\sum_{a=0,1}d_u(f_a)$ and $r_{2n}$ replaced by
$D_{u,q}$ and $r_{2n,q}$, respectively.
If $D_{\infty,q}=0$, then $V^{(2)}_{n,q}=0$ and the corresponding
empirical process term is zero.

Finally, we apply Lemma \ref{lemma:local-maximam-inequality-1}
conditionally on $\mathcal{F}_q^c$, with
$\delta=C'(r_{2n,q}+n_q^{-1/2})$
for some sufficiently large constant $C'$ to obtain
\begin{align*}
    &\bE\left[
    \sup_{g\in\mathcal{G}}
    \left|
    (\mP_{n_q}-\mP)
    \lbb V_{n,q}^{(2)}(\mO)\{g-\Pi(g)\}\rbb
    \right|
    \,\middle|\,\mathcal{F}_q^c
    \right]\\
    \lesssim&
    n_q^{-1/2}
    J\left(
    C'(r_{2n,q}+n_q^{-1/2}),
    \mathcal{R}^\ast,\|\cdot\|_{\mP,\infty}
    \right)\\
    \lesssim&
    n_q^{-1/2}
    (r_{2n,q}+n_q^{-1/2})^{1-1/(2\tau)}
    [D_{\infty,q}(1+\Lambda_{\check{K}})]^{1/(2\tau)}.
\end{align*}
Conditional Markov's inequality, together with
$D_{u,q}\leq d_u(\pi)+\sum_{a=0,1}d_u(f_a)$,
$r_{2n,q}\leq r_{2n}$, and $n_q\asymp n$ for fixed $Q$, gives
\begin{align*}
    &\sup_{g\in\mathcal{G}}
    \left|
    (\mP_{n_q}-\mP)
    \lbb V_{n,q}^{(2)}(\mO)\{g-\Pi(g)\}\rbb
    \right|\\
    \lessp&
    n^{-1/2}
    (r_{2n}+n^{-1/2})^{1-1/(2\tau)}
    [\{d_\infty(\pi)+\sum_{a=0,1}d_\infty(f_a)\}
    (1+\Lambda_{\check{K}})]^{1/(2\tau)}.
\end{align*}


\end{proof}

\subsection{Proof of Theorem \ref{thm:L2-rate-of-convergence}}\label{ss:proof-of-L2-rate}

\begin{proof}
\textbf{Step 1:} By the triangle inequality, it follows that
\begin{align*}
  \| \widehat{g}-g\|_{\mathbb{P},2}=&  \| \widehat{g}-g^\ast+g^\ast-g\|_{\mathbb{P},2}\\
  \leq&\|\widehat{g}-g^\ast\|_{\mathbb{P},2} +\|g^\ast-g\|_{\mathbb{P},2}\\
  \leq&\|\widehat{g}-g^\ast\|_{\mathbb{P},2} +c_K~~~(\text{Assumption \ref{asp:regularityL2}(b)}).
\end{align*}
Following \cite{newey1997convergence}, Assumption \ref{asp:regularityL2}(a) implies that, without loss of generality, we may orthonormalize $\bG$ to the identity matrix to simplify the exposition. Then the first term equals 
\begin{align*}
    \|\widehat{g}-g^\ast\|_{\mathbb{P},2}=&\|\bb(\bX)^\top(\widehat{\bbeta}-\bbeta^\ast)\|_{\mathbb{P},2}\\
    =&\left[\bE_{\bX}[\{\bb(\bX)^\top(\widehat{\bbeta}-\bbeta^\ast)\}^2]\right]^{1/2}\\
    =&\left[\bE_{\bX}[(\widehat{\bbeta}-\bbeta^\ast)^\top\bb(\bX)\bb(\bX)^\top(\widehat{\bbeta}-\bbeta^\ast)]\right]^{1/2}\\
    =&\lb(\widehat{\bbeta}-\bbeta^\ast)^\top\bG(\widehat{\bbeta}-\bbeta^\ast)\rb ^{1/2}\\
    =&\|\widehat{\bbeta}-\bbeta^\ast\|_2.
\end{align*}
It remains to analyze $\|\widehat{\bbeta}-\bbeta^\ast\|_2$, the $\ell_2$ estimation error of $\widehat{\bbeta}$ relative to its projection $\bbeta^\ast$.

\textbf{Step 2:} We first show that all eigenvalues of the estimated Gram matrix $\widehat{\bH}$ are bounded below by $\epsilon_1/2$ with probability approaching one. Because
\[
\bE(\phi_d| \bX)=w(\bX)+\omega'\{\pi(\bX)\}\bE\{A-\pi(\bX)| \bX\}=w(\bX),
\]
the smallest eigenvalue of the Gram matrix $\bH$ satisfies
\begin{align}
    \psi_1(\bH)=&\inf_{\|v\|_2=1} v^\top \bH v\nonumber\\
    =&{\inf_{\|v\|_2=1} \bE\left[ w(\bX)\lb v^\top\bb(\bX)\rb^2\right]}\nonumber\\
    \geq&\epsilon_1 \inf_{\|v\|_2=1}\bE\lb v^\top\bb(\bX)\bb(\bX)^\top v\rb=\epsilon_1,\label{eq:bounded-eigven-gram-H}
\end{align}
where the inequality follows from $\inf_{\bX}w(\bX)\geq\epsilon_1$ in Assumption \ref{asp:regularityL2}(d), and the last equality follows from the assumption that $\bG$ is the identity matrix after orthonormalization. By Weyl's inequality, it follows that
\begin{align*}
    \psi_1(\widehat{\bH})\geq& \psi_1(\bH)-\|\widehat{\bH}-\bH\|_\textop\\
    \geq&\epsilon_1-\|\widehat{\bH}-\bH\|_\textop.
\end{align*}
Since Theorem \ref{thm:rate-gram-norm} implies that $\|\widehat{\bH}-\bH\|_\textop=o_\mP(1)$ under Assumption \ref{asp:regularityL2}(c), it follows that 
\begin{align*}
    \lim_{n\to\infty}\Pr\left(\|\widehat{\bH}-\bH\|_\textop\leq \epsilon_1/2\right)=1.
\end{align*}
Because the event $\|\widehat{\bH}-\bH\|_\textop\leq \epsilon_1/2$ implies $\psi_1(\widehat{\bH})\geq\epsilon_1-\|\widehat{\bH}-\bH\|_\textop\geq\epsilon_1/2$, we obtain 
\begin{align*}
    \lim_{n\to\infty}\Pr\left\{\psi_1(\widehat{\bH})\geq\epsilon_1/2\right\}\geq\lim_{n\to\infty}\Pr\left(\|\widehat{\bH}-\bH\|_\textop\leq \epsilon_1/2\right)=1.
\end{align*}

Next, we analyze the primary quantity of interest remaining from Step 1, $\left\|  \widehat{\bbeta}-\bbeta^\ast\right\|_2$. By the triangle inequality and {the submultiplicative property of the spectral norm}, it follows that
\begin{align}
        \left\|  \widehat{\bbeta}-\bbeta^\ast\right\|_2
        =&\left\|  \lsb\widehat{\bH}+\lambda \bP\rsb^{-1} \{\widehat{\bh}-(\widehat{\bH}+\lambda \bP)\bbeta^\ast\}\right\|_2\nonumber\\
        \leq&\left\|  \lsb\widehat{\bH}+\lambda \bP\rsb^{-1}\right\| _{\text{op}}\left\|  \widehat{\bh}-(\widehat{\bH}+\lambda \bP)\bbeta^\ast\right\|_2\nonumber\\
        \leq&\frac{\left\|  \widehat{\bh}-\widehat{\bH}\bbeta^\ast\right\|_2+\lambda\left\|  \bP\bbeta^\ast\right\|_2}{\psi_1(\widehat{\bH}+\lambda\bP)}\nonumber\\
        \leq&\frac{\left\|  \widehat{\bh}-\widehat{\bH}\bbeta^\ast\right\|_2+\lambda\left\|  \bP\bbeta^\ast\right\|_2}{\psi_1(\widehat{\bH})+\lambda\psi_1(\bP)}\nonumber\\
        \lessp&2\frac{\left\|  \widehat{\bh}-\widehat{\bH}\bbeta^\ast\right\|_2+\lambda\left\|  \bP\bbeta^\ast\right\|_2}{\epsilon_1+2\lambda\psi_1(\bP)}.\label{eq:bound-hatbeta-beta-1}
\end{align} 
To proceed, we first bound the term $\left\|  \widehat{\bh}-\widehat{\bH}\bbeta^\ast\right\|_2$. The triangle inequality implies that
\begin{align}
     \left\|  \widehat{\bh}-\widehat{\bH}\bbeta^\ast\right\|_2=&\left\|\mP_n\left[ \bb(\bX)\lb\widehat{\phi}_n-\widehat{\phi}_d\bb(\bX)^\top\bbeta^\ast\rb\right]\right\|_2\nonumber\\
     \leq&\left\|  \mP_n\{\bb(\bX)(\widehat{\phi}_n-\phi_n)\}\right\|_2 +\label{eq:L2-term1}\\
     &\left\|  \mP_n\left[\bb(\bX)\lb\phi_n-w(\bX)g(\bX)\rb\right]\right\|_2+\label{eq:L2-term2}\\
     &\left\|  \mP_n\{\bb(\bX)w(\bX)\alpha(\bX;g,\bbeta^\ast)\}\right\|_2+\label{eq:L2-term3}\\
     &\left\|  \mP_n\left[\bb(\bX)\{w(\bX)-\phi_d\}\bb(\bX)^\top\bbeta^\ast\right]\right\|_2+\label{eq:L2-term4}\\
     &\left\|  \mP_n\left[\bb(\bX)(\phi_d-\widehat{\phi}_d)\bb(\bX)^\top\bbeta^\ast\right]\right\|_2.\label{eq:L2-term5}
\end{align}
Subsequently, we analyze these five terms. 

\textbf{Step 2.1:} Following convention in empirical process theory, we use $\mP$ to denote the expectation. By sample cross-fitting and triangle inequality, the first term in \eqref{eq:L2-term1} can be written as
\begin{align}
    \eqref{eq:L2-term1}=& \left\|  \sum_{q=1}^Q\frac{n_q}{n}\mP_{n_q}\{\bb(\bX)(\widehat{\phi}_n-\phi_n)\}\right\|_2\nonumber\\
\leq& \left\|  \sum_{q=1}^Q\frac{n_q}{n}\mP\{\bb(\bX)(\widehat{\phi}_n-\phi_n)|\widehat{\bGamma}^q\}\right\|_2+\label{eq:L2-term1-2}\\
&\left\|  \sum_{q=1}^Q\frac{n_q}{n}\left[(\mP_{n_q}-\mP)\{\bb(\bX)(\widehat{\phi}_n-\phi_n)|\widehat{\bGamma}^q\}\right]\right\|_2,\label{eq:L2-term1-1}
\end{align}
where $n_q=|\mathcal{F}_q|$ is the sample size for the testing fold $\mathcal{F}_q$ and $\mP_{n_q}(V)=n_q^{-1}\sum_{i\in\mathcal{F}_q}V_i$ denotes the empirical mean across the testing fold $\mathcal{F}_q$. To avoid cumbersome notation, we sometimes suppress dependence on the data. For example, we write $\bE(\widehat{\pi})$ to denote $\bE\{\widehat{\pi}(\bX)\}$ and simply use $f_1$ to denote $f(M|1,\bX)$. 

Following \cite{tchetgen2012semiparametric}, the expected bias of the triply robust estimator for the NIE, conditional on the training sample, is given by
\begin{align}
    &\bE\lbb\bb(\bX)\lb\widehat{\zeta}(\mO)-\kappa(\bX)\rb|\widehat{\bGamma}^q\rbb\nonumber\\
    =&\bE\lb\bb(\bX)\frac{\widehat{\pi}-\pi}{\widehat{\pi}}\int(\widehat{\mu}_1-\mu_1)f_1dm|\widehat{\bGamma}^q\rb+\bE\lb\bb(\bX)\frac{\widehat{\pi}-\pi}{\widehat{\pi}}\int\widehat{\mu}_1(\widehat{f}_1-f_1)dm|\widehat{\bGamma}^q\rb+\nonumber\\
    &\bE\lb\bb(\bX)\int(\widehat{\mu}_1-\mu_1)\frac{\widehat{f}_0}{\widehat{\pi}}(\widehat{\pi}-\pi)dm|\widehat{\bGamma}^q\rb+\bE\lb\bb(\bX)\int(\widehat{\mu}_1-\mu_1)\frac{\pi\widehat{f}_0}{\widehat{\pi}\widehat{f}_1}(\widehat{f}_1-f_1)dm|\widehat{\bGamma}^q\rb-\nonumber\\
    &\bE\lb \bb(\bX)\int(\widehat{\mu}_1-\mu_1)(\widehat{f}_0-f_0)dm|\widehat{\bGamma}^q\rb+\bE\lb \bb(\bX)\frac{\pi-\widehat{\pi}}{1-\widehat{\pi}}\int \widehat{\mu}_1(\widehat{f}_0-f_0)dm|\widehat{\bGamma}^q\rb.\label{eq:decomposition-expected-bias-nie}
\end{align}
To bound the term \eqref{eq:L2-term1-2}, we note that
\begin{align}
    &\mP\{\bb(\bX)(\widehat{\phi}_n^q-\phi_n)|\widehat{\bGamma}^q\}\nonumber\\
    =&\bE\left[\bb(\bX)\left[\omega'\{\widehat{\pi}^q(\bX)\}\left\{\pi(\bX)-\widehat{\pi}^q(\bX)\right\}\widehat{\kappa}^q(\bX)+\omega\{\widehat{\pi}^q(\bX)\}\widehat{\zeta}^q(\mO)-w(\bX)\kappa(\bX)\right]|\widehat{\bGamma}^q\right]\nonumber\\
    =&\bE\left[\bb(\bX)\left[\omega'\{\widehat{\pi}^q(\bX)\}\{\pi(\bX)-\widehat{\pi}^q(\bX)\}\{\widehat{\kappa}^q(\bX)-\kappa(\bX)\}-\frac{1}{2}\omega''\{\widetilde{\pi}^q(\bX)\}\{\widehat{\pi}^q(\bX)-\pi(\bX)\}^2\kappa(\bX)\right]|\widehat{\bGamma}^q\right]+\nonumber\\
    &\bE\left[\bb(\bX)\omega\{\widehat{\pi}^q(\bX)\}\{\widehat{\zeta}^q(\mO)-\kappa(\bX)\}|\widehat{\bGamma}^q\right],\label{eq:pre-bound22}
\end{align}
where the first equality follows from the law of total expectation and the second equality follows from a pointwise Taylor expansion with $\widetilde{\pi}^q(\bX)$ lying between $\pi(\bX)$ and $\widehat{\pi}^q(\bX)$. Combining \eqref{eq:decomposition-expected-bias-nie} and \eqref{eq:pre-bound22} yields 
\begin{align}
    \eqref{eq:L2-term1-2}\leq&Q^{-1}\sum_{q=1}^Q\left\|\bE\lb\bb(\bX)(\widehat{\phi}^q_n-\phi_n)|\widehat{\bGamma}^q\rb\right\|\nonumber\\
    =&Q^{-1}\sum_{q=1}^Q\left[\sum_{k=1}^K \left[\bE\lb b_k(\bX)(\widehat{\phi}^q_n-\phi_n)|\widehat{\bGamma}^q\rb\right]^2\right]^{1/2}\nonumber\\
    \lesssim&\sqrt{K}\max_q\left[\|\widehat{\pi}^q-\pi\|_{\mP,2}\left(\|\widehat{\mu}^q_1-\mu_1\|_{\mP,2}+\sum_{a=0}^1\|\widehat{f}^q_a-f_a\|_{\mP,2}+\|\widehat{\pi}^q-\pi\|_{\mP,2}\right)+\right.\nonumber\\
    &\left. \|\widehat{\mu}^q_1-\mu_1\|_{\mP,2}\sum_{a=0}^1\|\widehat{f}^q_a-f_a\|_{\mP,2}\right],\label{eq:error-bound-term23}
\end{align}
where the first inequality follows from the fact that $n_q/n\approx Q^{-1}$ 
and the third inequality follows from Jensen's inequality, the Cauchy-Schwarz, and Assumption \ref{asp:regularityL2}(d). Alternatively, it is possible to drop $\sqrt{K}$ in the bound \eqref{eq:error-bound-term23} by replacing the $L^2$ nuisance error with the $L^4$ nuisance error. To see this, we have
\begin{align}
    \eqref{eq:L2-term1-2}\leq&\max_q \sup_{\|v\|_2=1}\left| v^\top\bE\lb\bb(\bX)(\widehat{\phi}^q_n-\phi_n)|\widehat{\bGamma}^q\rb\right|\nonumber\\
    \lesssim&\max_q \sup_{\|v\|_2=1}\lbb \bE\lbb \lb v^\top\bb(\bX)\rb^2\rbb\rbb^{1/2}\left[\|\widehat{\pi}^q-\pi\|_{\mP,4}\left(\|\widehat{\mu}^q_1-\mu_1\|_{\mP,4}+\sum_{a=0}^1\|\widehat{f}^q_a-f_a\|_{\mP,4}+\|\widehat{\pi}^q-\pi\|_{\mP,4}\right)+\right.\nonumber\\
    &\left. \|\widehat{\mu}^q_1-\mu_1\|_{\mP,4}\sum_{a=0}^1\|\widehat{f}^q_a-f_a\|_{\mP,4}\right]\nonumber\\
    \lesssim&\max_q\left[\|\widehat{\pi}^q-\pi\|_{\mP,4}\left(\|\widehat{\mu}^q_1-\mu_1\|_{\mP,4}+\sum_{a=0}^1\|\widehat{f}^q_a-f_a\|_{\mP,4}+\|\widehat{\pi}^q-\pi\|_{\mP,4}\right)+\right.\nonumber\\
    &\left. \|\widehat{\mu}^q_1-\mu_1\|_{\mP,4}\sum_{a=0}^1\|\widehat{f}^q_a-f_a\|_{\mP,4}\right],\label{eq:error-bound-term23-alternative}    
\end{align}
where the first inequality is because the $\ell_2$ norm is self-dual, the second inequality follows from the Cauchy-Schwarz inequality, Assumption \ref{asp:regularityL2}(d), and the decompositions in  \eqref{eq:decomposition-expected-bias-nie} and \eqref{eq:pre-bound22}, and the last inequality follows from the fact that $\sup_{\|v\|_2=1}\bE\lbb \lb v^\top\bb(\bX)\rb^2\rbb=\sup_{\|v\|_2=1}|v^\top \bG v|=\|\bG\|_\textop\lesssim1$. To analyze the empirical process term \eqref{eq:L2-term1-1}, we note that 
\begin{align}
    \eqref{eq:L2-term1-1}\leq&Q^{-1}\sum_{q=1}^Q \left\|(\mP_{n_q}-\mP)\lb \bb(\bX)(\widehat{\phi}^q_n-\phi_n)|\widehat{\bGamma}^q\rb\right\|_2.\label{eq:bound-term1-2-intermediate-0}
\end{align}
Thus, by Markov's inequality, it is sufficient to bound $\bE\left[\left\|(\mP_{n_q}-\mP)\lb \bb(\bX)(\widehat{\phi}_n-\phi_n)|\widehat{\bGamma}^q\rb\right\|_2^2\right]$. Because $\bE\left\{\|(\mP_n-\mP)(\mathbf{V})\|_2^2\right\}
= n^{-1}\bE\left\{\|\mathbf{V}-\bE(\mathbf{V})\|_2^2\right\}$
for i.i.d. copies of any random vector $\mathbf{V}$, it follows that
\begin{align}
    &\bE\left[\left\|(\mP_{n_q}-\mP)\lb \bb(\bX)(\widehat{\phi}^q_n-\phi_n)|\widehat{\bGamma}^q\rb\right\|_2^2\right]\nonumber\\
    =&n_q^{-1}\bE\left[\left\|\bb(\bX)(\widehat{\phi}^q_n-\phi_n)-\bE\lb \bb(\bX)(\widehat{\phi}^q_n-\phi_n)|\widehat{\bGamma}^q\rb\right\|_2^2|\widehat{\bGamma}^q\right]\nonumber\\
    \leq&n_q^{-1}\bE\left[\left\|\bb(\bX)\right\|_2^2(\widehat{\phi}^q_n-\phi_n)^2|\widehat{\bGamma}^q\right]\nonumber\\
    \leq&n_q^{-1}\xi_K^2\|\widehat{\phi}^q_n-\phi_n\|_{\mP,2}^2,\label{eq:bound-term1-2-intermediate}
\end{align}
where the first inequality follows from the fact that the coordinatewise variance is bounded above by the corresponding coordinatewise second moment. Combining \eqref{eq:bound-term1-2-intermediate-0} and \eqref{eq:bound-term1-2-intermediate} yields
\begin{align}
    \eqref{eq:L2-term1-1}\lessp& Q^{-1}\sum_{q=1}^Qn_q^{-1/2}\xi_K\|\widehat{\phi}_n^q-\phi_n\|_{\mP,2}\nonumber\\
    \lesssim&n^{-1/2}\xi_K\max_q\|\widehat{\phi}_n^q-\phi_n\|_{\mP,2}\nonumber\\
    \lesssim&n^{-1/2}\xi_K\max_q \left(\|\widehat{\mu}^q_1-\mu_1\|_{\mP,2}+\sum_{a=0}^1\|\widehat{f}^q_a-f_a\|_{\mP,2}+\|\widehat{\pi}^q-\pi\|_{\mP,2}\right)=m_{2n}^{(1)}.\label{eq:bound-term-23}
\end{align}
where the last inequality follows from standard arguments in the debiased machine learning literature (e.g., \cite{kennedy2022semiparametric}) and Assumption \ref{asp:regularityL2}(d). Eventually, by Markov's inequality, we obtain $\eqref{eq:L2-term1}\lessp m_{1n}+m_{2n}^{(1)}$.

\textbf{Step 2.2: }Because $\bE\{\bb(\bX)\Omega\}=0$ and the observations are independent,
\begin{align*}
\bE\{\eqref{eq:L2-term2}^2\}
&=n^{-1}\bE\left[\mP_n\{\|\bb(\bX)\|_2^2\Omega^2\}\right]\\
&=n^{-1}\bE\{\|\bb(\bX)\|_2^2\Omega^2\}\\
&\leq\epsilon_2n^{-1}\operatorname{trace}(\bG)\lesssim K/n,
\end{align*}
where the first equality uses the vanishing cross terms in the squared norm of a mean of independent, centered vectors, the inequality uses the conditional second moment bound in Assumption~\ref{asp:regularityL2}(d), and the final bound uses Assumption~\ref{asp:regularityL2}(a). Thus, by Markov's inequality, we obtain $\eqref{eq:L2-term2}\lessp\sqrt{K/n}$.

\textbf{Step 2.3: }There are two alternative ways to bound the term \eqref{eq:L2-term3}. First, we can bound the term \eqref{eq:L2-term3} as follows. We note 
\begin{align*}
    \bE\{\eqref{eq:L2-term3}^2\}\leq&n^{-1}\bE\lb w(\bX)^2\alpha(\bX;g,\bbeta^\ast)^2\bb(\bX)^\top\bb(\bX)\rb\\
    \leq&\epsilon_2^2l_K^2c_K^2n^{-1}\text{trace}(\bG)\\
    \lesssim&l_K^2c_K^2n^{-1}K,
\end{align*}
where the second inequality follows from Assumptions \ref{asp:regularityL2}(b) and \ref{asp:regularityL2}(d), and the third inequality follows from Assumption \ref{asp:regularityL2}(a). Thus, by Markov's inequality, we obtain $\eqref{eq:L2-term3}\lessp l_Kc_K\sqrt{K/n}$. Second, we can also bound the term \eqref{eq:L2-term3} as follows. We note 
\begin{align*}
    \bE\{\eqref{eq:L2-term3}^2\}\leq&n^{-1}\bE\lb w(\bX)^2\alpha(\bX;g,\bbeta^\ast)^2\|\bb(\bX)\|^2_2\rb\\
    \lesssim&n^{-1}\xi_K^2\|\alpha(\bX;g,\bbeta^\ast)\|^2_{\mP,2}\\
    \lesssim&n^{-1}\xi_K^2c_K^2,
\end{align*}
where the second inequality follows from Assumption \ref{asp:regularityL2}(b). Thus, by Markov's inequality, we obtain $\eqref{eq:L2-term3}\lessp \xi_Kc_K/\sqrt{n}$. Combining the above two alternative bounds yields that 
\begin{align*}
    \eqref{eq:L2-term3}\lessp \min\left(l_Kc_K\sqrt{\frac{K}{n}},\frac{\xi_Kc_K}{\sqrt{n}}\right).
\end{align*}

\textbf{Step 2.4: } We bound the term in \eqref{eq:L2-term4}. By the Pythagorean theorem (applicable because $w(\bX)>0$ almost surely), it follows that 
\begin{align}
    \bE\lb w(\bX)g(\bX)^2\rb=&\bE\left[w(\bX)\lb \bb(\bX)^\top\bbeta^\ast\rb^2\right]+\bE\left[w(\bX)\lb g(\bX)-\bb(\bX)^\top\bbeta^\ast\rb^2\right]\nonumber\\
    \geq&\bE\left[w(\bX)\lb \bb(\bX)^\top\bbeta^\ast\rb^2\right],\label{eq:Pythagorean-theorem}
\end{align}
which, together with Assumption \ref{asp:regularityL2}(d), implies
\begin{align}
    \epsilon_1\bE \left[\lb \bb(\bX)^\top\bbeta^\ast\rb^2\right]\leq\bE\left[w(\bX)\lb \bb(\bX)^\top\bbeta^\ast\rb^2\right]\leq\bE\lb w(\bX)g(\bX)^2\rb\leq\epsilon_2 \bE\lb g(\bX)^2\rb\label{eq:eq:Pythagorean-theorem-bound}.
\end{align}
Then we obtain
\begin{align*}
    \bE\{\eqref{eq:L2-term4}^2\}=&n^{-1}\bE\left(\check{\Omega}^2\left\{\bb(\bX)^\top\bbeta^\ast\right\}^2\left\|\bb(\bX)\right\|_2^2\right)\\
    \leq&\epsilon_2^2/\epsilon_1n^{-1}\xi_K^2\bE\lb g^\ast(\bX)^2\rb\\
    \leq&\epsilon_2^2/\epsilon_1n^{-1}\xi_K^2\bE\lb g(\bX)^2\rb\lesssim n^{-1}\xi_K^2,
\end{align*}
where the first inequality follows from Assumptions \ref{asp:regularityL2}(b) and \ref{asp:regularityL2}(d), and the second inequality follows from \eqref{eq:eq:Pythagorean-theorem-bound} and Assumption \ref{asp:regularityL2}(d). Thus, by Markov's inequality, we obtain $\eqref{eq:L2-term4}\lessp \xi_K/\sqrt{n}$.

\textbf{Step 2.5: } To bound the term \eqref{eq:L2-term5}, we apply a similar argument to that used for bounding the term \eqref{eq:L2-term1}. We use the decomposition:
\begin{align}
    \eqref{eq:L2-term5}=& \left\|  \sum_{q=1}^Q\frac{n_q}{n}\mP_{n_q}\left[\bb(\bX)(\phi_d-\widehat{\phi}_d)\bb(\bX)^\top\bbeta^\ast\right]\right\|_2\nonumber\\
\leq& \left\|  \sum_{q=1}^Q\frac{n_q}{n}\mP\{\bb(\bX)(\phi_d-\widehat{\phi}_d^q)g^\ast(\bX)|\widehat{\bGamma}^q\}\right\|_2+\label{eq:L2-term5-2}\\
&\left\|  \sum_{q=1}^Q\frac{n_q}{n}\left[(\mP_{n_q}-\mP)\{\bb(\bX)(\phi_d-\widehat{\phi}_d^q)g^\ast(\bX)|\widehat{\bGamma}^q\}\right]\right\|_2.\label{eq:L2-term5-1}
\end{align}
By Equation \eqref{eq:expected-bias-phid}, it follows that
\begin{align}
    \eqref{eq:L2-term5-2}\leq&Q^{-1}\sum_{q=1}^Q\left\|\bE\lb\bb(\bX)g^\ast(\bX)(\widehat{\phi}^q_d-\phi_d)|\widehat{\bGamma}^q\rb\right\|_2\nonumber\\
    \lesssim&\sum_{q=1}^Q\left\|\bE\left[\bb(\bX)g^\ast(\bX)\frac{1}{2}\omega''\{\widetilde{\pi}^q(\bX)\}\lb\widehat{\pi}^q(\bX)-\pi(\bX)\rb^2|\widehat{\bGamma}^q\right]\right\|_2\nonumber\\
    \leq&\sum_{q=1}^Q\left[\sum_{k=1}^K\left[\bE\left[b_k(\bX)g^\ast(\bX)\frac{1}{2}\omega''\{\widetilde{\pi}^q(\bX)\}\lb\widehat{\pi}^q(\bX)-\pi(\bX)\rb^2|\widehat{\bGamma}^q\right]\right]^2\right]^{1/2}\nonumber\\
    \lesssim&\max_q\left[\sum_{k=1}^K\bE\lb g^\ast(\bX)^2\rb\bE\left[\lb\widehat{\pi}^q(\bX)-\pi(\bX)\rb^4|\widehat{\bGamma}^q\right]\right]^{1/2}\nonumber\\
    \lesssim&\sqrt{K}\max_q\|\widehat{\pi}^q-\pi\|_{\mP,4}^2\nonumber,
\end{align}
where the second inequality follows from \eqref{eq:expected-bias-phid} and the LOTE, the fourth inequality follows from Assumption \ref{asp:regularityL2}(d) and the Cauchy-Schwarz inequality, and the last inequality follows from Assumption \ref{asp:regularityL2}(b) and \eqref{eq:eq:Pythagorean-theorem-bound}. Applying a similar argument to that used for bounding the term \eqref{eq:L2-term1-1} yields:
\begin{align*}
    \eqref{eq:L2-term5-1}\lessp n^{-1/2}\xi_K\max_q \|\widehat{\pi}^q-\pi\|_{\mP,4}:=m_{2n}^{(2)},
\end{align*}
which implies that $\eqref{eq:L2-term5}\lessp m_{2n}^{(2)}+m_{3n}$.

\textbf{Step 3: }Finally, we bound the regularization bias term $\lambda\left\| \bP\bbeta^\ast\right\|_2$ as follows. By \eqref{eq:Pythagorean-theorem}, it follows that
\begin{align}
    \bE\left[w(\bX)g^\ast(\bX)^2\right]=&{\bbeta^\ast}^\top \bH\bbeta^\ast\nonumber\\
    \leq&\bE\lb w(\bX)g(\bX)^2\rb\leq\epsilon_2\bE \lb g(\bX)^2\rb.\label{eq:bounded-gast-eq1}
\end{align}
Moreover, we have
\begin{align}
    {\bbeta^\ast}^\top \bH\bbeta^\ast\geq \psi_1(\bH)\|\bbeta^\ast\|_2^2\geq\epsilon_1\|\bbeta^\ast\|_2^2,\label{eq:bounded-gast-eq2}
\end{align}
where the last inequality follows from \eqref{eq:bounded-eigven-gram-H}. Therefore, \eqref{eq:bounded-gast-eq1} and \eqref{eq:bounded-gast-eq2} imply 
\begin{align}
    \|\bbeta^\ast\|_2^2\leq \epsilon_2/\epsilon_1\bE\{g(\bX)^2\}\lesssim1.\label{eq:bounded-betaast-norm}
\end{align}
Finally, we obtain
\begin{align*}
    \lambda\| \bP\bbeta^\ast\|_2
    \leq &\lambda\|\bP\|_\textop\|\bbeta^\ast\|_2\\
    \lesssim&\lambda \psi_K(\bP),
\end{align*}
where 
the last inequality follows from \eqref{eq:bounded-betaast-norm}.

Finally, we conclude by combining the preceding results from Steps 1-3.

\end{proof}

\subsection{Proof of Proposition \ref{thm:pointwise-linearization}}\label{ss:proof-propositon2}
\begin{proof}
\textbf{Step 0: }The estimator satisfies the exact identity
\[
\sqrt{n}(\widehat{\bbeta}-\bbeta^\ast)
=(\widehat{\bH}+\lambda\bP)^{-1}
\left\{\sqrt{n}(\widehat{\bh}-\widehat{\bH}\bbeta^\ast)
-\sqrt{n}\lambda\bP\bbeta^\ast\right\}.
\]
By \eqref{eq:L2-term1}--\eqref{eq:L2-term5}, we obtain the following decomposition:
\begin{align*}
    &\sqrt{n}\widetilde{\bb}^\top\lsb\widehat{\bbeta}-\bbeta^\ast\rsb\\
    =&\sqrt{n}\widetilde{\bb}^\top\lsb\widehat{\bH}+\lambda \bP\rsb^{-1}\lbb\mP_n\lb\bb(\bX)\widehat{\phi}_n\rb-(\widehat{\bH}+\lambda \bP)\bbeta^\ast\rbb\\
    =&\widetilde{\bb}^\top\lsb\widehat{\bH}+\lambda \bP\rsb^{-1}\bigg[\sqrt{n}\mP_n\{\bb(\bX)(\widehat{\phi}_n-\phi_n)\}+\mG_n\lbb\bb(\bX)\lb\phi_n-w(\bX)g(\bX)\rb\rbb+ \\
    &\mG_n\{\bb(\bX)w(\bX)\alpha(\bX;g,\bbeta^\ast)\}+\mG_n\lbb\bb(\bX)\lb w(\bX)-\phi_d\rb g^\ast(\bX)\rbb+\\
    &\sqrt{n}\mP_n\lb \bb(\bX)\lsb \phi_d-\widehat{\phi}_d\rsb g^\ast(\bX)\rb-\sqrt{n}\lambda\bP \bbeta^\ast\bigg],
\end{align*}
which further implies 
\begin{align}
   &\text{Rem}_{1n}(\widetilde{\bb})\nonumber\\
    =&\widetilde{\bb}^\top\{(\widehat{\bH}+\lambda \bP)^{-1}-(\bH+\lambda \bP)^{-1}\}\mG_n\lbb\bb(\bX)\lb\phi_n-w(\bX)g(\bX)\rb\rbb+\label{eq:point-linear-term1}\\
    &\widetilde{\bb}^\top\{(\widehat{\bH}+\lambda \bP)^{-1}-(\bH+\lambda \bP)^{-1}\}\mG_n\{\bb(\bX)w(\bX)\alpha(\bX;g,\bbeta^\ast)\}+\label{eq:point-linear-term2}\\  
    &\widetilde{\bb}^\top\lb(\bH+\lambda \bP)^{-1}-\bH^{-1}\rb\mG_n\left[\bb(\bX)\lb\phi_n-\phi_d g^\ast(\bX)\rb\right]+\label{eq:point-linear-term3}\\
    &\widetilde{\bb}^\top\{(\widehat{\bH}+\lambda \bP)^{-1}-(\bH+\lambda \bP)^{-1}\}\mG_n\lbb\bb(\bX)\lb w(\bX)-\phi_d\rb g^\ast(\bX)\rbb+\label{eq:point-linear-term4}\\  
    &\sqrt{n}\widetilde{\bb}^\top(\bH+\lambda\bP)^{-1}\mP_n\{\bb(\bX)(\widehat{\phi}_n-\phi_n)\}+\label{eq:point-linear-term5}\\
    &\sqrt{n}\widetilde{\bb}^\top(\bH+\lambda\bP)^{-1}\mP_n\{\bb(\bX)(\phi_d-\widehat{\phi}_d)g^\ast(\bX)\}+\label{eq:point-linear-term6}\\
    &\sqrt{n}\widetilde{\bb}^\top\{(\widehat{\bH}+\lambda \bP)^{-1}-(\bH+\lambda \bP)^{-1}\}\mP_n\{\bb(\bX)(\widehat{\phi}_n-\phi_n)\}+\label{eq:point-linear-term7}\\
    &\sqrt{n}\widetilde{\bb}^\top\{(\widehat{\bH}+\lambda \bP)^{-1}-(\bH+\lambda \bP)^{-1}\}\mP_n\{\bb(\bX)(\phi_d-\widehat{\phi}_d)g^\ast(\bX)\}-\label{eq:point-linear-term8}\\
    &\sqrt{n}\lambda\widetilde{\bb}^\top\{(\widehat{\bH}+\lambda \bP)^{-1}-(\bH+\lambda \bP)^{-1}\}\bP\bbeta^\ast-\label{eq:point-linear-term9}\\
    &\sqrt{n}\lambda\widetilde{\bb}^\top(\bH+\lambda \bP)^{-1}\bP\bbeta^\ast.\label{eq:point-linear-term10}
\end{align}
It suffices to analyze the ten terms in \eqref{eq:point-linear-term1} through \eqref{eq:point-linear-term10}.

\textbf{Step 1.1: } 
We can bound the term in \eqref{eq:point-linear-term1} as follows:
\begin{align*}
    |\eqref{eq:point-linear-term1}|\leq&\sqrt{n}\left\|\widetilde{\bb}\right\|_2\left\|\lsb\widehat{\bH}+\lambda \bP\rsb^{-1}-(\bH+\lambda \bP)^{-1}\right\|_\textop\left\|\mP_n\lbb\bb(\bX)\lb\phi_n-w(\bX)g(\bX)\rb\rbb\right\|_2\\
    \lessp&\frac{\sqrt{K}m_n^{\bH}}{\lb\epsilon_1+\lambda\psi_1(\bP)\rb\lb\epsilon_1/2+\lambda\psi_1(
    \bP
    )\rb},
\end{align*}
where the last inequality follows from Lemma \ref{lemma:rate-penalized-gram} and the bound for the term in \eqref{eq:L2-term2}. 

\textbf{Step 1.2: }We can bound the term in \eqref{eq:point-linear-term2} as follows:
\begin{align*}
    |\eqref{eq:point-linear-term2}|\leq&\sqrt{n}\left\|\widetilde{\bb}\right\|_2\left\|\lsb\widehat{\bH}+\lambda \bP\rsb^{-1}-(\bH+\lambda \bP)^{-1}\right\|_\textop\left\|\mP_n\{\bb(\bX)w(\bX)\alpha(\bX;g,\bbeta^\ast)\}\right\|_2\\
    \lessp&\frac{m_n^{\bH}\min\left(l_Kc_K\sqrt{K},\xi_Kc_K\right)}{\lb\epsilon_1+\lambda\psi_1(\bP)\rb\lb\epsilon_1/2+\lambda\psi_1(
    \bP
    )\rb},
\end{align*}
where the last inequality follows from Lemma \ref{lemma:rate-penalized-gram} and the bound for the term in \eqref{eq:L2-term3}.

\textbf{Step 1.3: }To bound the term in \eqref{eq:point-linear-term3}, we note
\begin{align}
  \eqref{eq:point-linear-term3}=&  \sqrt{n}\widetilde{\bb}^\top\lb(\bH+\lambda \bP)^{-1}-\bH^{-1}\rb\mP_n\left[\bb(\bX)\phi_d\alpha(\bX;g,\bbeta^\ast)\right]+\label{eq:point-linear-term3-1}\\
  &\sqrt{n}\widetilde{\bb}^\top\lb(\bH+\lambda \bP)^{-1}-\bH^{-1}\rb\mP_n\left[\bb(\bX)w(\bX)\lb\zeta(\mO)-\kappa(\bX)\rb\right]\label{eq:point-linear-term3-2}.
\end{align}
To bound the term in \eqref{eq:point-linear-term3-1}, we note
\begin{align*}
    &\bE\lbb\left\|\mP_n\left[\bb(\bX)\phi_d\alpha(\bX;g,\bbeta^\ast)\right]\right\|_2^2\rbb\\
    =&\frac{1}{n}\bE\lb\phi_d^2\alpha(\bX;g,\bbeta^\ast)^2\left\|\bb(\bX)\right\|_2^2\rb\\
    \lesssim&\min\left(l_K^2c_K^2\frac{K}{n},\frac{\xi_K^2c_K^2}{n}\right),
\end{align*}
where the equality follows from the fact that $\bE\{\bb(\bX)\phi_d\alpha(\bX;g,\bbeta^\ast)\}=0$ and the inequality is due to \eqref{eq:bounded-root-phid} and the proof for bounding the term in \eqref{eq:L2-term3}. By the resolvent identity, we obtain $(\bH+\lambda \bP)^{-1}-\bH^{-1}=-\lambda(\bH+\lambda\bP)^{-1}\bP\bH^{-1}$, which implies that $\|(\bH+\lambda \bP)^{-1}-\bH^{-1}\|_\textop\leq\lambda \|(\bH+\lambda\bP)^{-1}\|_\textop\|\bP\|_\textop\|\bH^{-1}\|_\textop=\lambda\psi_K(\bP)/\{\psi_1(\bH+\lambda\bP)\psi_1(\bH)\}\leq\lambda\psi_K(\bP)/[\epsilon_1\{\epsilon_1+\lambda\psi_1(\bP)\}]$. For the term in \eqref{eq:point-linear-term3-2}, we have 
\begin{align*}
    &\bE\lbb\left\|\mP_n\left[\bb(\bX)w(\bX)\lb\zeta(\mO)-\kappa(\bX)\rb\right]\right\|_2^2\rbb\\
    =&\frac{1}{n}\bE\lbb\{w(\bX)\}^2\lb\zeta(\mO)-\kappa(\bX)\rb^2\left\|\bb(\bX)\right\|_2^2\rbb\\
    \lesssim&n^{-1}\text{trace}(\bG)=\frac{K}{n},
\end{align*}
where the inequality uses the bounded conditional second moment of $\zeta(\mO)-\kappa(\bX)$ implied by Assumption~\ref{asp:regularityL2}(d). Thus, combining the above yields 
\begin{align*}
    |\eqref{eq:point-linear-term3}|\lessp&\sqrt{n}\lb \min\left(l_Kc_K\sqrt{\frac{K}{n}},\frac{\xi_Kc_K}{\sqrt{n}}\right)+\sqrt{\frac{K}{n}}\rb\frac{\lambda \psi_K(\bP)}{\epsilon_1+\lambda \psi_1(\bP)}.
\end{align*}

\textbf{Step 1.4: }We can bound the term in \eqref{eq:point-linear-term4} as follows:
\begin{align*}
    |\eqref{eq:point-linear-term4}|\leq&\sqrt{n}\left\|\widetilde{\bb}\right\|_2\left\|\lsb\widehat{\bH}+\lambda \bP\rsb^{-1}-(\bH+\lambda \bP)^{-1}\right\|_\textop\left\|\mP_n\lbb\bb(\bX)\lb w(\bX)-\phi_d\rb g^\ast(\bX)\rbb\right\|_2\\
    \lessp&\frac{\xi_Km_n^{\bH}}{\lb\epsilon_1+\lambda\psi_1(\bP)\rb\lb\epsilon_1/2+\lambda\psi_1(
    \bP
    )\rb},
\end{align*}
where the last inequality follows from Lemma \ref{lemma:rate-penalized-gram} and the bound for the term in \eqref{eq:L2-term4}.

\textbf{Step 1.5: }We can bound the term in \eqref{eq:point-linear-term5} as follows:
\begin{align*}
    |\eqref{eq:point-linear-term5}|\leq&\sqrt{n}\left\|\widetilde{\bb}\right\|_2\left\|(\bH+\lambda \bP)^{-1}\right\|_\textop\left\|\mP_n\{\bb(\bX)(\widehat{\phi}_n-\phi_n)\}\right\|_2\\
    \lessp&\frac{\sqrt{n}\lsb m_{1n}+m_{2n}^{(1)}\rsb}{\epsilon_1+\lambda\psi_1(\bP)},
\end{align*}
where the last inequality follows from the bound for the term in \eqref{eq:L2-term1}.

\textbf{Step 1.6: }We can bound the term in \eqref{eq:point-linear-term6} as follows:
\begin{align*}
    |\eqref{eq:point-linear-term6}|\leq&\sqrt{n}\left\|\widetilde{\bb}\right\|_2\left\|(\bH+\lambda \bP)^{-1}\right\|_\textop\left\|\mP_n\{\bb(\bX)(\phi_d-\widehat{\phi}_d)g^\ast(\bX)\}\right\|_2\\
    \lessp&\frac{\sqrt{n}\left(m_{2n}^{(2)}+m_{3n}\right)}{\epsilon_1+\lambda\psi_1(\bP)},
\end{align*}
where the last inequality follows from the bound for the term in \eqref{eq:L2-term5}.

\textbf{Step 1.7: }We can bound the term in \eqref{eq:point-linear-term7} as follows:
\begin{align*}
    |\eqref{eq:point-linear-term7}|\leq&\sqrt{n}\left\|\widetilde{\bb}\right\|_2\left\|\lsb\widehat{\bH}+\lambda \bP\rsb^{-1}-(\bH+\lambda \bP)^{-1}\right\|_\textop\left\|\mP_n\{\bb(\bX)(\widehat{\phi}_n-\phi_n)\}\right\|_2\\
    \lessp&\frac{m_n^{\bH}\sqrt{n}\lsb m_{1n}+m_{2n}^{(1)}\rsb}{\lb\epsilon_1+\lambda\psi_1(\bP)\rb\lb\epsilon_1/2+\lambda\psi_1(
    \bP
    )\rb},
\end{align*}
where the last inequality follows from the bound for the term in \eqref{eq:L2-term1} and Lemma \ref{lemma:rate-penalized-gram}.

\textbf{Step 1.8: }We can bound the term in \eqref{eq:point-linear-term8} as follows:
\begin{align*}
    |\eqref{eq:point-linear-term8}|\leq&\sqrt{n}\left\|\widetilde{\bb}\right\|_2\left\|\lsb\widehat{\bH}+\lambda \bP\rsb^{-1}-(\bH+\lambda \bP)^{-1}\right\|_\textop\left\|\mP_n\{\bb(\bX)(\phi_d-\widehat{\phi}_d)g^\ast(\bX)\}\right\|_2\\
    \lessp&\frac{m_n^{\bH}\sqrt{n}\left(m_{2n}^{(2)}+m_{3n}\right)}{\lb\epsilon_1+\lambda\psi_1(\bP)\rb\lb\epsilon_1/2+\lambda\psi_1(
    \bP
    )\rb},
\end{align*}
where the last inequality follows from the bound for the term in \eqref{eq:L2-term5} and Lemma \ref{lemma:rate-penalized-gram}.

\textbf{Step 1.9: }We can bound the term in \eqref{eq:point-linear-term9} as follows:
\begin{align*}
    |\eqref{eq:point-linear-term9}|\leq&\sqrt{n}\lambda\left\|\widetilde{\bb}\right\|_2\left\|\lsb\widehat{\bH}+\lambda \bP\rsb^{-1}-(\bH+\lambda \bP)^{-1}\right\|_\textop\left\|\bP\bbeta^\ast\right\|_2\\
    \lessp&\frac{m_n^{\bH}\sqrt{n}\lambda\psi_K(\bP)}{\lb\epsilon_1+\lambda\psi_1(\bP)\rb\lb\epsilon_1/2+\lambda\psi_1(
    \bP
    )\rb},
\end{align*}
where the last inequality follows from \eqref{eq:bounded-betaast-norm} and Lemma \ref{lemma:rate-penalized-gram}.

\textbf{Step 1.10: }We can bound the term in \eqref{eq:point-linear-term10} as follows:
\begin{align*}
    |\eqref{eq:point-linear-term10}|\leq&\sqrt{n}\lambda\left\|\widetilde{\bb}\right\|_2\left\|(\bH+\lambda \bP)^{-1}\right\|_\textop\left\|\bP\bbeta^\ast\right\|_2\\
    \lessp&\frac{\sqrt{n}\lambda\psi_K(\bP)}{\epsilon_1+\lambda\psi_1(\bP)},
\end{align*}
where the last inequality follows from \eqref{eq:bounded-eigven-gram-H} and \eqref{eq:bounded-betaast-norm}.

Combining the bounds in Steps 1.1--1.10 yields the following bound for the remainder term. Under Assumption \ref{asp:regularityL2}, it follows that
\begin{align*}
    \text{Rem}_{1n}(\widetilde{\bb})\lessp&\frac{\sqrt{n}\sum_{j=1}^3m_{jn}+\lambda \psi_K(\bP)(\sqrt{n}+\sqrt{K}+\sqrt{n}m_{0n})}{\epsilon_1+\lambda\psi_1(\bP)}+\\
    &\frac{m^{\bH}_n\lb \xi_K+\sqrt{n}\sum_{j=0}^3m_{jn}+\sqrt{n}\lambda\psi_K(\bP)\rb}{\lb\epsilon_1+\lambda\psi_1(\bP)\rb\lb\epsilon_1/2+\lambda\psi_1(
    \bP
    )\rb}. 
\end{align*}

\end{proof}

\subsection{Proof of Proposition \ref{thm:uniform-linearization}}\label{ss:proof-proposition-uniform-linearization}
\begin{proof}
We reuse the decomposition results in \eqref{eq:point-linear-term1}--\eqref{eq:point-linear-term10}, replacing $\widetilde{\bb}$ with $\widetilde{\bb}(\bx)$, because $\widetilde{\bb}(\bx)$ has unit norm by definition.

\textbf{Step 1.1: } The objective is to bound the sup norm of the term
\begin{align*}
    \widetilde{\bb}(\bx)^\top\{(\widehat{\bH}+\lambda \bP)^{-1}-(\bH+\lambda \bP)^{-1}\}\mG_n\lbb\bb(\bX)\lb\phi_n-w(\bX)g(\bX)\rb\rbb.
\end{align*}
The proof strategy follows Step 1 of the proof of Lemma 4.2 in \cite{belloni2015some}, with some modifications. To proceed, conditional on the data, we define
\begin{align*}
    t_{\bx}(\mO_i)=\widetilde{\bb}(\bx)^\top\{(\widehat{\bH}+\lambda \bP)^{-1}-(\bH+\lambda \bP)^{-1}\}\bb(\bX_i)\lb\phi_{n}(\mO_i)-w(\bX_i)g(\bX_i)\rb.
\end{align*}
and $\bt_{\bx}=(t_{\bx}(\mO_1),\ldots,t_{\bx}(\mO_n))^\top$. Let $\calT:=\{t_\bx:\bx\in\calX\}$ be the index set. Define $\|\bt_{\bx}\|_{n,2}:=\{n^{-1}\sum_{i=1}^nt_{\bx}(\mO_i)^2\}^{1/2}$. For any $\bx'\neq \bx$, we define $\Delta_{\widetilde{\bb}}:=\widetilde{\bb}(\bx)-\widetilde{\bb}(\bx')$ and obtain
\begin{align*}
    &\|\bt_{\bx}-\bt_{\bx'}\|_{n,2}\\
    =&\lbb n^{-1}\sum_{i=1}^n \lbb \Delta_{\widetilde{\bb}}^\top\{(\widehat{\bH}+\lambda \bP)^{-1}-(\bH+\lambda \bP)^{-1}\}\bb(\bX_i)\lb\phi_{n}(\mO_i)-w(\bX_i)g(\bX_i)\rb\rbb^2\rbb^{1/2}\\
    \leq&\max_{1\leq i\leq n}|\Omega(\mO_i)|\lbb n^{-1}\sum_{i=1}^n \lbb \Delta_{\widetilde{\bb}}^\top\{(\widehat{\bH}+\lambda \bP)^{-1}-(\bH+\lambda \bP)^{-1}\}\bb(\bX_i)\rbb^2\rbb^{1/2}\\
    =&\max_{1\leq i\leq n}|\Omega(\mO_i)|\lbb \Delta_{\widetilde{\bb}}^\top\{(\widehat{\bH}+\lambda \bP)^{-1}-(\bH+\lambda \bP)^{-1}\}\widehat{\bG}\{(\widehat{\bH}+\lambda \bP)^{-1}-(\bH+\lambda \bP)^{-1}\}\Delta_{\widetilde{\bb}}\rbb^{1/2}\\
    \leq&\max_{1\leq i\leq n}|\Omega(\mO_i)|\left\|\lsb\widehat{\bH}+\lambda \bP\rsb^{-1}-(\bH+\lambda \bP)^{-1}\right\|_\textop\|\widehat{\bG}\|_\textop^{1/2}\|\Delta_{\widetilde{\bb}}\|_2\\
    \leq&\xi_K^L\max_{1\leq i\leq n}|\Omega(\mO_i)| \left\|\lsb\widehat{\bH}+\lambda \bP\rsb^{-1}-(\bH+\lambda \bP)^{-1}\right\|_\textop\|\widehat{\bG}\|_\textop^{1/2}\|\bx-\bx'\|_2,
\end{align*}
where the last inequality follows from Assumption \ref{assump:uniform-further-boundedness}. Therefore, the function $\bt_\bx$ is Lipschitz with Lipschitz constant $L_{\bt}:=\xi_K^L\max_{1\leq i\leq n}|\Omega(\mO_i)| \|(\widehat{\bH}+\lambda \bP)^{-1}-(\bH+\lambda \bP)^{-1}\|_\textop\|\widehat{\bG}\|_\textop^{1/2}$. Since $\bt(\bx)$ is Lipschitz and Assumption \ref{assump:uniform-further-boundedness} assumes that $\text{diam}(\calX)$ is uniformly bounded from above, it follows that the covering number for $\calT$ equipped with $\|\bt_\bx\|_{n,2}$ satisfies
\begin{align}\label{eq:convering-number-T}
    N(\calT,\|\cdot\|_{n,2},\epsilon)\leq \lsb \text{Constant }\epsilon^{-1}L_{\bt}\rsb^d.
\end{align}
We note 
\begin{align}
    &\bE\lsb \max_{1\leq i\leq n}|\Omega(\mO_i)||\{\bX_i\}_{i=1}^n\rsb\nonumber\\
    \leq&\bE\lb \lsb \sum_{i=1}^n|\Omega(\mO_i)|^\nu\rsb^{1/\nu}|\{\bX_i\}_{i=1}^n\rb\nonumber\\
    \leq&\lb  \sum_{i=1}^n\bE\lsb|\Omega(\mO_i)|^\nu|\bX_i\rsb\rb^{1/\nu}\nonumber\\
    \lesssim&n^{1/\nu}\label{eq:bound-on-max-omega},
\end{align}
where the second inequality is due to Jensen’s inequality and the third inequality follows from Assumption \ref{assump:uniform-further-boundedness}. 
By \eqref{eq:convering-number-T} and Dudley's inequality, we obtain that, conditional on the data,  
\begin{align*}
    &\bE_{\iota}\lb \sup_{\bx}\left|\frac{1}{\sqrt{n}}\sum_{i=1}^n\iota_it_{\bx}(\mO_i)\right||\{\mO_i\}_{i=1}^n\rb\\
    =&\bE_{\iota}\lb \sup_{t_{\bx}\in\calT}\left|\frac{1}{\sqrt{n}}\sum_{i=1}^n\iota_it_{\bx}(\mO_i)\right||\{\mO_i\}_{i=1}^n\rb\\
    \lesssim&\int_{0}^{2\sup_{t_{\bx}\in\calT}\|\bt_{\bx}\|_{n,2}} \sqrt{\log N(\calT,\|\cdot\|_{n,2},\epsilon)}d\epsilon\\
    \leq&\int_{0}^{2 \max_{1\leq i\leq n}|\Omega(\mO_i)|\|(\widehat{\bH}+\lambda \bP)^{-1}-(\bH+\lambda \bP)^{-1}\|_\textop\|\widehat{\bG}\|_\textop^{1/2}} \sqrt{\log N(\calT,\|\cdot\|_{n,2},\epsilon)}d\epsilon\\
    \lesssim&\sqrt{\log K}  \max_{1\leq i\leq n}|\Omega(\mO_i)|\|(\widehat{\bH}+\lambda \bP)^{-1}-(\bH+\lambda \bP)^{-1}\|_\textop\|\widehat{\bG}\|_\textop^{1/2}\\
    \lessp&\frac{\sqrt{\log K} n^{1/\nu}m_n^{\bH}}{\lb\epsilon_1+\lambda\psi_1(\bP)\rb\lb\epsilon_1/2+\lambda\psi_1(
    \bP
    )\rb},
\end{align*}
where $\{\iota_i\}_{i=1}^n$ are independent Rademacher random variables, taking values in $\{-1,1\}$ with equal probability, the third inequality follows from a change of variables and the fact that $\log\xi_K^L\lesssim\log K$ under Assumption \ref{assump:uniform-further-boundedness}, and the last inequality follows from Assumption \ref{asp:regularityL2}(c), Lemma \ref{lemma:rate-penalized-gram}, \eqref{eq:bound-on-max-omega}, and Markov's inequality.

\textbf{Step 1.2: }We bound the sup norm of the term as follows:
\begin{align*}
    &|\widetilde{\bb}(\bx)^\top\{(\widehat{\bH}+\lambda \bP)^{-1}-(\bH+\lambda \bP)^{-1}\}\mG_n\{\bb(\bX)w(\bX)\alpha(\bX;g,\bbeta^\ast)\}|\\
    \leq&\sqrt{n}\left\|\widetilde{\bb}(\bx)\right\|_2\left\|\lsb\widehat{\bH}+\lambda \bP\rsb^{-1}-(\bH+\lambda \bP)^{-1}\right\|_\textop\left\|\mP_n\{\bb(\bX)w(\bX)\alpha(\bX;g,\bbeta^\ast)\}\right\|_2\\
    \lessp&\frac{m_n^{\bH}l_Kc_K\sqrt{K}}{\lb\epsilon_1+\lambda\psi_1(\bP)\rb\lb\epsilon_1/2+\lambda\psi_1(
    \bP
    )\rb},
\end{align*}
Therefore, the supremum norm of the sum of the first two terms is of stochastic order $\lessp m_n^{\bH}m_{4n}/\{\epsilon_1/2+\lambda\psi_1(\bP)\}/\{\epsilon_1+\lambda\psi_1(\bP)\}$.

\textbf{Step 1.3: }Using $g=\kappa$ and the definition of $\phi_n$, we have the exact decomposition
\[
\phi_n-\phi_dg^\ast(\bX)
=\phi_d\alpha(\bX;g,\bbeta^\ast)
+w(\bX)\{\zeta(\mO)-\kappa(\bX)\}.
\]
The Cauchy--Schwarz and triangle inequalities therefore give
\begin{align*}
&\sup_{\bx}\left|\widetilde{\bb}(\bx)^\top
\{(\bH+\lambda\bP)^{-1}-\bH^{-1}\}
\mG_n[\bb(\bX)\{\phi_n-\phi_dg^\ast(\bX)\}]\right|\\
\leq&\|(\bH+\lambda\bP)^{-1}-\bH^{-1}\|_{\mathrm{op}}
\Big(\|\mG_n\{\bb(\bX)\phi_d\alpha(\bX;g,\bbeta^\ast)\}\|_2+\|\mG_n[\bb(\bX)w(\bX)\{\zeta(\mO)-\kappa(\bX)\}]\|_2\Big)\\
\lesssim_{\mP}&
\frac{\lambda\psi_K(\bP)}{\epsilon_1+\lambda\psi_1(\bP)}
\left(\sqrt{n}m_{0n}+\sqrt{K}\right),
\end{align*}
where the last inequality follows from the second moment bounds in Step 1.3 of Section~\ref{ss:proof-propositon2} and the resolvent identity.

\textbf{Step 1.4: }We note 
\begin{align}
    &\bE\lsb \max_{1\leq i\leq n}|\lb w(\bX_i)-\phi_d(A_i,\bX_i)\rb g^\ast(\bX_i)|\rsb\nonumber\\
    \leq&3\epsilon_2^2\bE\lb \lsb \sum_{i=1}^n| g^\ast(\bX_i)|^\nu\rsb^{1/\nu}\rb\nonumber\\
    \lesssim&\lb  \sum_{i=1}^n\bE(|g^\ast(\bX_i)|^\nu)\rb^{1/\nu}\nonumber\\
    \lesssim&n^{1/\nu}\label{eq:bound-on-max-w-phidtimesgast},
\end{align}
where the first inequality follows from \eqref{eq:bounded-root-phid} and Assumption \ref{asp:regularityL2}(d), the second inequality  is due to Jensen's inequality, and the last inequality follows from Assumption \ref{assump:uniform-further-boundedness}. We then apply empirical process arguments similar to those used in Step 1.1 and obtain 
\begin{align*}
    &\sup_{\bx}|\widetilde{\bb}(\bx)^\top\{(\widehat{\bH}+\lambda \bP)^{-1}-(\bH+\lambda \bP)^{-1}\}\mG_n\lbb\bb(\bX)\lb w(\bX)-\phi_d\rb g^\ast(\bX)\rbb|\\
    \lesssim&\sqrt{\log K}  \max_{1\leq i\leq n}|\lb w(\bX_i)-\phi_d(A_i,\bX_i)\rb g^\ast(\bX_i)|\|(\widehat{\bH}+\lambda \bP)^{-1}-(\bH+\lambda \bP)^{-1}\|_\textop\|\widehat{\bG}\|_\textop^{1/2}\\
    \lessp&\frac{\sqrt{\log K}n^{1/\nu}m_n^{\bH}}{\lb\epsilon_1+\lambda\psi_1(\bP)\rb\lb\epsilon_1/2+\lambda\psi_1(
    \bP
    )\rb},
\end{align*}
where the last inequality follows from \eqref{eq:bound-on-max-w-phidtimesgast} and Lemma \ref{lemma:rate-penalized-gram}.

\textbf{Step 1.5: }For the fifth term, we have that 
\begin{align*}
    &\sup_{\bx}|\sqrt{n}\widetilde{\bb}(\bx)^\top(\bH+\lambda\bP)^{-1}\mP_n\{\bb(\bX)(\widehat{\phi}_n-\phi_n)\}|\\
    \leq&\sqrt{n}\sup_{\bx}\left\|\widetilde{\bb}(\bx)\right\|_2\left\|(\bH+\lambda \bP)^{-1}\right\|_\textop\left\|\mP_n\{\bb(\bX)(\widehat{\phi}_n-\phi_n)\}\right\|_2\\
    \lessp&\frac{\sqrt{n}\lsb m_{1n}+m_{2n}^{(1)}\rsb}{\epsilon_1+\lambda\psi_1(\bP)},
\end{align*}
where the last inequality follows from Step 1.5 of the proof in Section \ref{ss:proof-propositon2}.

\textbf{Step 1.6: }For the sixth term, it follows that
\begin{align*}
    &\sup_{\bx}|\sqrt{n}\widetilde{\bb}(\bx)^\top(\bH+\lambda\bP)^{-1}\mP_n\{\bb(\bX)(\phi_d-\widehat{\phi}_d)g^\ast(\bX)\}|\\
    \leq&\sqrt{n}\sup_{\bx}\left\|\widetilde{\bb}(\bx)\right\|_2\left\|(\bH+\lambda \bP)^{-1}\right\|_\textop\left\|\mP_n\{\bb(\bX)(\phi_d-\widehat{\phi}_d)g^\ast(\bX)\}\right\|_2\\
    \lessp&\frac{\sqrt{n}\left(m_{2n}^{(2)}+m_{3n}\right)}{\epsilon_1+\lambda\psi_1(\bP)},
\end{align*}
where the last inequality follows from Step 1.6 of the proof in Section \ref{ss:proof-propositon2}.

\textbf{Step 1.7: }For the seventh term, it follows that
\begin{align*}
    &\sup_{\bx}|\sqrt{n}\widetilde{\bb}(\bx)^\top\{(\widehat{\bH}+\lambda \bP)^{-1}-(\bH+\lambda \bP)^{-1}\}\mP_n\{\bb(\bX)(\widehat{\phi}_n-\phi_n)\}|\\
    \leq&\sqrt{n}\sup_{\bx}\left\|\widetilde{\bb}(\bx)\right\|_2\left\|\lsb\widehat{\bH}+\lambda \bP\rsb^{-1}-(\bH+\lambda \bP)^{-1}\right\|_\textop\left\|\mP_n\{\bb(\bX)(\widehat{\phi}_n-\phi_n)\}\right\|_2\\
    \lessp&\frac{\sqrt{n}\lsb m_{1n}+m_{2n}^{(1)}\rsb m_n^{\bH}}{\lb\epsilon_1+\lambda\psi_1(\bP)\rb\lb\epsilon_1/2+\lambda\psi_1(
    \bP
    )\rb},
\end{align*}
where the last inequality follows from Step 1.7 of the proof in Section \ref{ss:proof-propositon2}.

\textbf{Step 1.8: }For the eighth term, it follows that
\begin{align*}
    &\sup_{\bx}|\sqrt{n}\widetilde{\bb}(\bx)^\top\{(\widehat{\bH}+\lambda \bP)^{-1}-(\bH+\lambda \bP)^{-1}\}\mP_n\{\bb(\bX)(\phi_d-\widehat{\phi}_d)g^\ast(\bX)\}|\\
    \leq&\sqrt{n}\sup_{\bx}\left\|\widetilde{\bb}(\bx)\right\|_2\left\|\lsb\widehat{\bH}+\lambda \bP\rsb^{-1}-(\bH+\lambda \bP)^{-1}\right\|_\textop\left\|\mP_n\{\bb(\bX)(\phi_d-\widehat{\phi}_d)g^\ast(\bX)\}\right\|_2\\
    \lessp&\frac{\sqrt{n}\left(m_{2n}^{(2)}+m_{3n}\right)m_n^{\bH}}{\lb\epsilon_1+\lambda\psi_1(\bP)\rb\lb\epsilon_1/2+\lambda\psi_1(
    \bP
    )\rb},
\end{align*}
where the last inequality follows from Step 1.8 of the proof in Section \ref{ss:proof-propositon2}.

\textbf{Step 1.9: }For the ninth term, it follows that
\begin{align*}
    &\sup_{\bx}|\sqrt{n}\lambda\widetilde{\bb}(\bx)^\top\{(\widehat{\bH}+\lambda \bP)^{-1}-(\bH+\lambda \bP)^{-1}\}\bP\bbeta^\ast|\\
    \leq&\sqrt{n}\lambda\sup_{\bx}\left\|\widetilde{\bb}(\bx)\right\|_2\left\|\lsb\widehat{\bH}+\lambda \bP\rsb^{-1}-(\bH+\lambda \bP)^{-1}\right\|_\textop\left\|\bP\bbeta^\ast\right\|_2\\
    \lessp&\frac{\sqrt{n}\lambda\psi_K(\bP)m_n^{\bH}}{\lb\epsilon_1+\lambda\psi_1(\bP)\rb\lb\epsilon_1/2+\lambda\psi_1(
    \bP
    )\rb},
\end{align*}
where the last inequality follows from Step 1.9 of the proof in Section \ref{ss:proof-propositon2}.

\textbf{Step 1.10: }For the tenth term, it follows that
\begin{align*}
    &\sup_{\bx}|\lambda\widetilde{\bb}(\bx)^\top(\bH+\lambda \bP)^{-1}\bP\bbeta^\ast|\\
    \leq&\lambda\sup_{\bx}\left\|\widetilde{\bb}(\bx)\right\|_2\left\|(\bH+\lambda \bP)^{-1}\right\|_\textop\left\|\bP\bbeta^\ast\right\|_2\\
    \lessp&\frac{\lambda\psi_K(\bP)}{\epsilon_1+\lambda\psi_1(\bP)},
\end{align*}
where the last inequality follows from Step 1.10 of the proof in Section \ref{ss:proof-propositon2}.

\textbf{Step 1.11: } Finally, the term $\sup_\bx|\text{Rem}_{2n}(\widetilde{\bb}(\bx))|$ can be analyzed using arguments analogous to those in Step 1.1:
\begin{align*}
    &\bE_{\iota}\lb \sup_\bx|\text{Rem}_{2n}(\widetilde{\bb}(\bx))||\{\mO_i\}_{i=1}^n\rb\\
    \lesssim&\sqrt{\log K}  \max_{1\leq i\leq n}|\phi_d(A_i,\bX_i)\alpha(\bX_i;g,\bbeta^\ast)|\|\bH^{-1}\|_\textop\|\widehat{\bG}\|_\textop^{1/2}\\
    \lessp&\sqrt{\log K}l_Kc_K\sup_{A,\bX}|\phi_d|\frac{1}{\psi_1(\bH)}\\
    \leq&\sqrt{\log K}l_Kc_K\frac{3\epsilon_2}{\epsilon_1}\\
    \lesssim&\sqrt{\log K}l_Kc_K,
\end{align*}
where the penultimate inequality follows from \eqref{eq:bounded-root-phid} and \eqref{eq:bounded-eigven-gram-H}.


\end{proof}

\subsection{Proof of Theorem \ref{thm:uniform-rate}}
\begin{proof}We first establish the following uniform bound for the linearized form in Proposition \ref{thm:uniform-linearization}:
\begin{align}
    &\sup _\bx |\widetilde{\bb}(\bx)^\top\bH^{-1}\mG_n\left[\bb(\bX)\lb\phi_n-\phi_dg(\bX)\rb\right]|\nonumber\\
    \lessp&\sqrt{\log K}+\frac{\xi_K^{\nu/(\nu-2)}\log K}{\sqrt{n}}\lesssim \sqrt{\log K}.\label{eq:proof-uniform-rate-of-convergence-preliminary}
\end{align}
The proof of \eqref{eq:proof-uniform-rate-of-convergence-preliminary} closely follows the proof of Theorem 4.3 in \cite{belloni2015some}, with some modifications. We reproduce the argument below. We will apply Lemma \ref{lemma:tight-bound-empirical-process-uniform}. To proceed, we define $\varepsilon:=\phi_n(\mO)-\phi_d(\mO)g(\bX)$ 
and consider $\mathcal{G}:=\{(\varepsilon,\bx)\mapsto\varepsilon\widetilde{\bb}(\bv)^\top\bH^{-1}\bb(\bx):\bv\in\calX\}$. It is straightforward to verify that $\bE(\varepsilon|\bX)=0$ and $\sup_{\bx}\bE(|\varepsilon|^\nu|\bX=\bx)\lesssim1$ under Assumption \ref{assump:uniform-further-boundedness}. Also, we note (i) $|\widetilde{\bb}(\bv)^\top\bH^{-1}\bb(\bx)|\leq\xi_K$; (ii) $\bE_\bX[\{\widetilde{\bb}(\bv)^\top\bH^{-1}\bb(\bX)\}^2]=\widetilde{\bb}(\bv)^\top\bH^{-1}\bG\bH^{-1}\widetilde{\bb}(\bv)=\|\widetilde{\bb}(\bv)^\top\bH^{-1}\|_2^2\leq \|\bH^{-1}\|_\textop^2=1/\psi_1(\bH)^2\leq1/\epsilon_1^2\lesssim1$; and (iii) by Assumption \ref{assump:uniform-further-boundedness}, 
\begin{align*}
    |\{\widetilde{\bb}(\bv)-\widetilde{\bb}(\bv')\}^\top\bH^{-1}\bb(\bx)\varepsilon|\leq|\varepsilon|\|\bb(\bx)\|_2\|\widetilde{\bb}(\bv)-\widetilde{\bb}(\bv')\|_2\leq |\varepsilon|\xi_K\xi_K^L\|\bv-\bv'\|_2.
\end{align*}
Furthermore, taking $G(\varepsilon,\bx):=|\varepsilon|\xi_K$ yields
\begin{align*}
    \sup_{\mathbb{Q}} N\{\mathcal{G},L^2(\mathbb{Q}),\varepsilon\|G\|_{L^2(\mathbb{Q})}\}\leq \lsb \frac{C\xi_K^L}{\varepsilon}\rsb^p.
\end{align*}
The preceding verification of the conditions permits application of Lemma \ref{lemma:tight-bound-empirical-process-uniform}, which implies
\begin{align*}
    \bE\lb \sup _\bx |\widetilde{\bb}(\bx)^\top\bH^{-1}\mG_n\left[\bb(\bX)\lb\phi_n-\phi_dg(\bX)\rb\right]|\rb\lesssim \sqrt{\log K}+\sqrt{\log K}\sqrt{\frac{\xi_K^{2\nu/(\nu-2)}\log K}{n}}\lesssim\sqrt{\log K},
\end{align*}
where the last two inequalities use the two conditions in Assumption \ref{assump:uniform-further-boundedness}: $\log \xi_K^L\lesssim\log K$ and $\xi_K^{2\nu/(\nu-2)}\log K/n\lesssim1$. Eventually, we obtain
\begin{align}
     &\sup_\bx |\widehat{g}(\bx)-g^\ast(\bx)|\nonumber\\
    =&\sup_\bx \|\bb(\bx)\|_2|\widetilde{\bb}(\bx)^\top(\widehat{\bbeta}-\bbeta^\ast)|\nonumber\\
    \leq&\xi_K\sup_\bx |\widetilde{\bb}(\bx)^\top(\widehat{\bbeta}-\bbeta^\ast)|\nonumber\\
    \leq&\xi_K/\sqrt{n}\sup_\bx |\sqrt{n}\widetilde{\bb}(\bx)^\top(\widehat{\bbeta}-\bbeta^\ast)|\nonumber\\
    \leq&\xi_K/\sqrt{n}\left\{\sup _\bx |\widetilde{\bb}(\bx)^\top\bH^{-1}\mG_n\left[\bb(\bX)\lb\phi_n^{\text{NIE}}-\phi_d^{\text{NIE}}g(\bX)\rb\right]|+\sup_\bx\sum_{j=1}^2|\text{Rem}_{jn}\{\widetilde{\bb}(\bx)\}|\right\}\nonumber\\
\lessp &\frac{\xi_K}{\sqrt{n}}\Bigg[\frac{\sqrt{n}\sum_{j=1}^3m_{jn}+\lambda \psi_K(\bP)(\sqrt{n}+\sqrt{K}+\sqrt{n}m_{0n})}{\epsilon_1+\lambda\psi_1(\bP)}+\frac{m^{\bH}_n\lb \sqrt{n}\sum_{j=0}^3m_{jn}+m_{4n}+\sqrt{n}\lambda\psi_K(\bP)\rb}{\lb\epsilon_1+\lambda\psi_1(\bP)\rb\lb\epsilon_1/2+\lambda\psi_1(
    \bP
    )\rb}+\nonumber\\
&\sqrt{\log K}(1+l_Kc_K)\Bigg],\label{eq:uniform-bound:hatg-gast}
\end{align}
where the last inequality follows from \eqref{eq:proof-uniform-rate-of-convergence-preliminary} and Proposition \ref{thm:uniform-linearization}. Finally, we obtain the uniform rate of convergence as follows:
\begin{align*}
&\sup_\bx |\widehat{g}(\bx)-g(\bx)|\\
\leq&\sup_\bx |\bb(\bx)^\top(\widehat{\bbeta}-\bbeta^\ast
)+\alpha(\bx;g,\bbeta^\ast)|\\
\leq&\sup_\bx |\bb(\bx)^\top(\widehat{\bbeta}-\bbeta^\ast
)|+\sup_\bx|\alpha(\bx;g,\bbeta^\ast)|\\
\leq&\sup_\bx |\bb(\bx)^\top(\widehat{\bbeta}-\bbeta^\ast)|+l_Kc_K\\
\lessp &\frac{\xi_K}{\sqrt{n}}\Bigg\{\frac{\sqrt{n}\sum_{j=1}^3m_{jn}+\lambda \psi_K(\bP)(\sqrt{n}+\sqrt{K}+\sqrt{n}m_{0n})}{\epsilon_1+\lambda\psi_1(\bP)}+\frac{m^{\bH}_n\lb \sqrt{n}\sum_{j=0}^3m_{jn}+m_{4n}+\sqrt{n}\lambda\psi_K(\bP)\rb}{\lb\epsilon_1+\lambda\psi_1(\bP)\rb\lb\epsilon_1/2+\lambda\psi_1(
    \bP
    )\rb}+\\
&\sqrt{\log K}(1+l_Kc_K)\Bigg\}+l_Kc_K.
\end{align*}
where the third inequality follows from Assumption \ref{asp:regularityL2}(b). 
\end{proof}

\subsection{Proof of Theorem \ref{thm:pointwise-normal}}
\begin{proof}
 Under the additional condition (i) in Theorem \ref{thm:pointwise-normal}, the remainder term in Proposition \ref{thm:pointwise-linearization} satisfies $\text{Rem}_{1n}(\widetilde{\bb})=o_{\mP}(1)$. Therefore, we obtain
\begin{align*}
    \frac{\sqrt{n}\widetilde{\bb}^\top}{\|\widetilde{\bb}^\top\mathbb{V}^{1/2}\|_2}\left(\widehat{\bbeta}-\bbeta^\ast\right)=&\frac{\widetilde{\bb}^\top}{\|\widetilde{\bb}^\top\mathbb{V}^{1/2}\|_2}\bH^{-1}\mG_n\left[\bb(\bX)\lb\phi_n-\phi_dg^\ast(\bX)\rb\right]+o_{\mP}(1)\\
    =&\frac{\widetilde{\bb}^\top}{\|\widetilde{\bb}^\top\mathbb{V}^{1/2}\|_2}\bH^{-1}\mG_n\left[\bb(\bX)\lb\widetilde{\Omega}+\phi_d\alpha(\bX;g,\bbeta^\ast)\rb\right]+o_{\mP}(1)\\
    =&\sum_{i=1}^n\underline{\omega}_{in}\lb \widetilde{\Omega}_i+\phi_d(A_i,\bX_i)\alpha(\bX_i;g,\bbeta^\ast)\rb+o_{\mP}(1),    
\end{align*}
where 
\begin{align*}
    \underline{\omega}_{in}:=\frac{\widetilde{\bb}^\top\bH^{-1}\bb(\bX_i)}{\|\widetilde{\bb}^\top\mathbb{V}^{1/2}\|_2\sqrt{n}}.
\end{align*}
The additional condition (iii) in Theorem \ref{thm:pointwise-normal} implies $1\lesssim \|\widetilde{\bb}^\top\mathbb{V}^{1/2}\|_2$. We also observe that $|\underline{\omega}_{in}|\lesssim \xi_K/\sqrt{n}$ and $|\widetilde{\Omega}_i+\phi_d(A_i,\bX_i)\alpha(\bX_i;g,\bbeta^\ast)|\leq |\widetilde{\Omega}_i|+3\epsilon_2 l_Kc_K$. Therefore, we can follow the proof of Theorem 4.2 in \cite{belloni2015some} to verify the Lindeberg condition for the CLT by setting $\omega_{in}$ in \cite{belloni2015some} to $\underline{\omega}_{in}$.

\end{proof}

\subsection{Proof of Theorem \ref{thm:uniform-strong-gaussian}}
\begin{proof}
 Under the additional condition (i) in Theorem \ref{thm:uniform-strong-gaussian}, the remainder term in Proposition \ref{thm:uniform-linearization} satisfies $\sup_{\bx}\text{Rem}_{1n}\{\widetilde{\bb}(\bx)\}=o_{\mP}(a_n^{-1})$. Therefore, we obtain
 \begin{align*}
     \sqrt{n}\widetilde{\bb}(\bx)^\top\left(\widehat{\bbeta}-\bbeta^\ast\right)=&\widetilde{\bb}(\bx)^\top\bH^{-1}\mG_n\left[\bb(\bX)\lb\widetilde{\Omega}+\phi_d\alpha(\bX;g,\bbeta^\ast)\rb\right]+o_{\mP}(a_n^{-1}).
 \end{align*}
To apply arguments analogous to those used in the proof of Theorem 4.4 of \cite{belloni2015some}, it remains to verify the following condition:
\begin{align*}
    &\bE\lbb\left\|\mathbb{V}^{-1/2}\bH^{-1}\bb(\bX)\lb\widetilde{\Omega}+\phi_d\alpha(\bX;g,\bbeta^\ast)\rb\right\|_2^3\rbb\\
    \lesssim&\bE\lbb\left\|\bb(\bX)\lb\widetilde{\Omega}+\phi_d\alpha(\bX;g,\bbeta^\ast)\rb\right\|^3\rbb\\
\lesssim&\bE\lbb\|\bb(\bX)\|^3\lb|\widetilde{\Omega}|^3+|\phi_d\alpha(\bX;g,\bbeta^\ast)|^3\rb\rbb\\
\lesssim&\bE\lbb\|\bb(\bX)\|^3\rbb(1+l_K^3c_K^3)\\
\leq&\xi_K(1+l_K^3c_K^3)\bE\lbb\|\bb(\bX)\|^2\rbb\\
\leq&\text{trace}(\bG)\xi_K(1+l_K^3c_K^3)\asymp K\xi_K(1+l_K^3c_K^3).
\end{align*}
Here, for the first inequality, write $c_\Omega:=\inf_{\bx}\bE(\widetilde{\Omega}^2\mid\bX=\bx)>0$. At the true nuisance functions, $\bE(\widetilde{\Omega}\mid A,\bX)=0$, so the cross term with $\phi_d\alpha$ vanishes conditionally on $\bX$. Hence, under the orthonormalization $\bG=\boldsymbol{I}_K$, $\mathbb{V}\succeq c_\Omega\bH^{-1}\bG\bH^{-1}
\succeq (3\epsilon_2)^{-2}c_\Omega\bG^{-1}
=(3\epsilon_2)^{-2}c_\Omega\boldsymbol{I}_K$. Together with $\|\bH^{-1}\|_{\mathrm{op}}\leq\epsilon_1^{-1}$, this bounds $\|\mathbb{V}^{-1/2}\bH^{-1}\|_{\mathrm{op}}$. In particular, the centered vector $\mathbb{V}^{-1/2}\bH^{-1}\bb(\bX)\{\widetilde{\Omega}+\phi_d\alpha(\bX;g,\bbeta^\ast)\}$ has covariance $\boldsymbol{I}_K$, as required by the Gaussian coupling argument.
 
\end{proof}

\subsection{Proof of Theorem \ref{thm:honest-cov-gauss-boot-ci}}
\begin{proof}
\textbf{Step 1:} By Proposition \ref{thm:uniform-linearization} and arguments analogous to those used in the proofs of Theorems 5.4 and 5.5 of \cite{belloni2015some}, we obtain the following approximation for the supremum of the $t$-statistics: under the same set of conditions as outlined in Theorem~\ref{thm:honest-cov-gauss-boot-ci} of the main manuscript, we obtain 
\begin{align}
\sup_{\bx}|T_n(\bx)|=_d\sup_{\bx}\left|\frac{\bb(\bx)^\top\mathbb{V}^{1/2}}{\| \mathbb{V}^{1/2}\bb(\bx)\|_2}\mathcal{N}(0,\boldsymbol{I}_K)\right|+o_{\mP}\left(\frac{1}{\sqrt{\log K}}\right).\label{eq:strong-approximation-suprema-t-statistics}
\end{align}

\textbf{Step 2:} We derive the convergence rate for the covariance matrix estimator, that is, an upper bound for $\| \widehat{\mathbb{V}}-\mathbb{V}\|_{\text{op}}$. To proceed, we first derive an upper bound for $\|\widehat{\Upsilon}-\Upsilon\|_{\textop}$, where 
\begin{align*}
    &\widehat{\Upsilon}:=\mP_n[\{\widehat{\phi}_n-\widehat{\phi}_d\widehat{g}(\bX)\}^2\bb(\bX)\bb(\bX)^\top],\\
    &\Upsilon:=\bE[\{\widetilde{\Omega}+\phi_d\alpha(\bX;g,\bbeta^\ast)\}^2\bb(\bX)\bb(\bX)^\top].
\end{align*}
Consider the decomposition:
\begin{align}
    \widehat{\Upsilon}-\Upsilon=&\mP_n[[\{\widehat{\phi}_n-\widehat{\phi}_d\widehat{g}(\bX)\}^2-\{\widetilde{\Omega}+\phi_d\alpha(\bX;g,\bbeta^\ast)\}^2]\bb(\bX)\bb(\bX)^\top]\label{eq:matrix-decomposition-1}\\
    +&(\mP_n-\mP)[\{\widetilde{\Omega}+\phi_d\alpha(\bX;g,\bbeta^\ast)\}^2\bb(\bX)\bb(\bX)^\top].\label{eq:matrix-decomposition-2}
\end{align}

Let $v_n:=\{\bE(\max_{1\leq i\leq n}\widetilde{\Omega}_i^2)\}^{1/2}$. By the proof of Theorem 4.6 in \cite{belloni2015some}, it follows that $\bE\{\eqref{eq:matrix-decomposition-2}\}\lesssim {m_n^{\bG}}^2(v_n^2+l_K^2c_K^2)+{m_n^{\bG}}\sqrt{v_n^2+l_K^2c_K^2}\|\Upsilon\|_{\textop}^{1/2}$. We then show that (i) $\|\Upsilon\|_{\textop}\lesssim1$ and (ii) ${m_n^{\bG}}^2(v_n^2+l_K^2c_K^2)=o(1)$. To show (i), because $\sqrt{\log K}l_Kc_K\lesssim\sqrt{\log K}$, we get $l_Kc_K\lesssim1$. Since $m_n^\bG=o(1)$ by Assumption \ref{asp:regularityL2}(c), we get $l_K^2c_K^2{m_n^\bG}^2=o(1)$. Since $\bE(\widetilde{\Omega}|A,\bX)=0$ and $\phi_d$ is a function of $(A,\bX)$, we obtain $\bE[\{\widetilde{\Omega}+\phi_d\alpha(\bX;g,\bbeta^\ast)\}^2|\bX]=\bE(\widetilde{\Omega}^2|\bX)+\bE(\phi_d^2|\bX)\alpha(\bX;g,\bbeta^\ast)^2\lesssim1+l_K^2c_K^2$ by Assumption \ref{asp:regularityL2}(b), Assumption \ref{asp:regularityL2}(d), and \eqref{eq:bounded-root-phid}. Thus, it follows that, for some constant $C>0$,  $\|\Upsilon\|_{\textop}=\sup_{\|v\|_2=1}|v^\top\Upsilon v|\leq C (1+l_Kc_K)^2  \|\bG\|_{\textop}\leq C(1+l_Kc_K)^2 \lesssim1$. To show (ii), because $m_n^\bG=o(1)$ under Assumption \ref{asp:regularityL2}(c) and $l_Kc_K\lesssim1$, we obtain ${m_n^\bG}^2l_K^2c_K^2=o(1)$. By \eqref{eq:bound-on-max-omega}, with $\nu$ replaced by $\nu/2$, we obtain $v_n\lesssim n^{1/\nu}$, which implies ${m_n^\bG}\nu_n\lesssim m_n^\bG n^{1/\nu}=o(1)$. Therefore, by Markov's inequality, we obtain $\eqref{eq:matrix-decomposition-2}\lessp (v_n \vee 1+l_Kc_K)m_n^\bG$.

We then analyze the term in \eqref{eq:matrix-decomposition-1}. Let $\Psi=\widetilde{\Omega}+\phi_d\alpha(\bX;g,\bbeta^\ast)$, $\widehat{\Psi}:=\widehat{\phi}_n-\widehat{\phi}_d\widehat{g}(\bX)$, and $\Delta_\phi:=\widehat{\Psi}-\Psi$. Then $\widehat{\Psi}^2-\Psi^2=\Delta_\phi^2+2\Psi\Delta_\phi$. Thus, we obtain
\begin{align*}
    &\|\eqref{eq:matrix-decomposition-1}\|_\textop=\|\mP_n\{(\widehat{\Psi}^2-\Psi^2)\bb(\bX)\bb(\bX)^\top\}\|_\textop\\
    \leq &\|\mP_n\{\Delta_\phi^2\bb(\bX)\bb(\bX)^\top\}\|_\textop+2\|\mP_n\{\Psi\Delta_\phi\bb(\bX)\bb(\bX)^\top\}\|_\textop\\
    \leq &(\max_i |\Delta_\phi(\mO_i)|^2+2\max_i|\Psi(\mO_i)|\max_i|\Delta_\phi(\mO_i)|) \|\widehat{\bG}\|_\textop\\
    \lessp&\max_i |\Delta_\phi(\mO_i)|^2+\max_i|\Psi(\mO_i)|\max_i|\Delta_\phi(\mO_i)|,
\end{align*}
where the last inequality follows from the fact that $\|\widehat{\bG}\|_\textop\lessp1$. We note $\max_i|\Psi(\mO_i)|\leq\max_i|\widetilde{\Omega}(\mO_i)|+3\epsilon_2l_Kc_K\lesssim n^{1/\nu}+l_Kc_K$. We derive an upper bound for $\max_i|\Delta_\phi(\mO_i)|$ as follows:
\begin{align*}
   &\max_i|\Delta_\phi(\mO_i)|\\
   \leq&\max_i|\widehat{\phi}_n(\mO_i)-\phi_n(\mO_i)|+\max_i|\{\widehat{\phi}_d(\mO_i)-\phi_d(\mO_i)\}\bb(\bX_i)^\top\bbeta^\ast|+\\
   &\max_i|\phi_d(\mO_i)\bb(\bX_i)^\top(\widehat{\bbeta}-\bbeta^\ast)|+\max_i|\{\widehat{\phi}_d(\mO_i)-\phi_d(\mO_i)\}\bb(\bX_i)^\top(\widehat{\bbeta}-\bbeta^\ast)|\\
   \leq&\max_i|\widehat{\phi}_n(\mO_i)-\phi_n(\mO_i)|+\max_i|\{\widehat{\phi}_d(\mO_i)-\phi_d(\mO_i)\}|(\|g-g^\ast\|_{\mP,\infty}+\|g\|_{\mP,\infty})+\\
   &(3\epsilon_2+\max_i|\widehat{\phi}_d(\mO_i)-\phi_d(\mO_i)|)\|\widehat{g}-g^\ast\|_{\mP,\infty}\\
   \leq&d_\infty^{\Sigma}+d_\infty(\pi)(l_Kc_K+\|g\|_{\mP,\infty})+\{3\epsilon_2+d_\infty(\pi)\}\|\widehat{g}-g^\ast\|_{\mP,\infty}\\
   \lesssim&d_\infty^{\Sigma}+d_\infty(\pi)(l_Kc_K+\|g\|_{\mP,\infty})+\{1+d_\infty(\pi)\}m_n^\bG\\
   \lesssim&d_\infty^{\Sigma}+m_n^\bG,
\end{align*}
where the penultimate inequality follows from \eqref{eq:uniform-bound:hatg-gast} and the conditions in Theorem \ref{thm:honest-cov-gauss-boot-ci}. Thus, we obtain $\|\eqref{eq:matrix-decomposition-1}\|_\textop\lessp (d_\infty^{\Sigma}+m_n^\bG)^2+(d_\infty^{\Sigma}+m_n^\bG)(n^{1/\nu}+l_Kc_K)\lesssim (d_\infty^{\Sigma}+m_n^\bG)(n^{1/\nu}+l_Kc_K)$. Eventually, we have $\|\widehat{\Upsilon}-\Upsilon\|_{\textop}\lessp (d_\infty^{\Sigma}+m_n^\bG)(n^{1/\nu}+l_Kc_K)$. Finally, we obtain
\begin{align*}
    &\|\widehat{\mathbb{V}}-\mathbb{V}\|_\textop\\
    \lesssim&\|(\widehat{\bH}^{-1}-\bH^{-1})\widehat{\Upsilon}\widehat{\bH}^{-1}\|_\textop+\|\bH^{-1}(\widehat{\Upsilon}-\Upsilon)\widehat{\bH}^{-1}\|_\textop+\|\bH^{-1}\Upsilon(\widehat{\bH}^{-1}-\bH^{-1})\|_{\textop}\\
    \lessp&\|\widehat{\Upsilon}-\Upsilon\|_{\textop}+\|\widehat{\bH}-\bH\|_{\textop}\lessp (d_\infty^{\Sigma}+m_n^\bG)(n^{1/\nu}+l_Kc_K)+m_n^\bG,
\end{align*}
where the last inequality follows from Theorem \ref{thm:rate-gram-norm} and the conditions in Theorem \ref{thm:honest-cov-gauss-boot-ci}.

Thus, by Equation (A.57) in \cite{belloni2015some}, we obtain 
\begin{align}
    &\sup_{\bx}\left|\frac{\| \widehat{\mathbb{V}}^{1/2}\bb(\bx)\|_2}{\| \mathbb{V}^{1/2}\bb(\bx)\|_2}-1\right|\nonumber\\
    \lessp &\| \widehat{\mathbb{V}}-\mathbb{V}\|_{\text{op}}\nonumber\\
    \lessp&(d_\infty^{\Sigma}+m_n^\bG)(n^{1/\nu}+l_Kc_K)+m_n^\bG=o_\mP(1).\label{eq:covariance-matrix-rate-int-1}
\end{align}

\textbf{Step 3:} Finally, we conclude by arguments analogous to those used in the proof of Theorem 5.6 of \cite{belloni2015some}, replacing their Theorem 5.5 and Lemma 5.1 with \eqref{eq:strong-approximation-suprema-t-statistics} and \eqref{eq:covariance-matrix-rate-int-1}, respectively.

\end{proof}

\section{Supporting information for simulation studies}\label{supp;sec:supporting-info-simulation}

The complete simulation setup is provided as follows. We consider a sample of $n=3000$ independent individuals. For each individual $i=1,\ldots,n$, we first generate three covariates
\[
\bX_i=2\Phi(\widetilde{\bX}_i)-1,
\]
where $\widetilde{\bX}_i$ follows a multivariate normal distribution with mean $\mathbf{0}$ and covariance matrix $0.9\mathbf{I}_3+0.1\mathbf{e}_3\mathbf{e}_3^\top$.
Here, $\mathbf{I}_3$ is the $3\times 3$ identity matrix, $\mathbf{e}_3$ is the three-dimensional vector of ones, and $\Phi(\cdot)$ denotes the cumulative distribution function of the standard normal distribution, applied componentwise. This construction yields covariates $\bX_i=(X_{i1},X_{i2},X_{i3})^\top$ supported on $[-1,1]^3$ with marginal uniform distributions and dependence induced through the Gaussian copula.

Let $\operatorname{expit}(x)=\{1+\exp(-x)\}^{-1}$. The binary treatment $A_i$ is generated from
\[
A_i| \bX_i \sim \operatorname{Bernoulli}\{\pi(\bX_i)\},
\]
where the treatment propensity score is specified as
\[
\pi(\bX_i)
=
\operatorname{expit}\left(
0.10
+
0.18X_{i1}
-
0.16X_{i2}
+
0.18X_{i3}
+
0.05X_{i1}X_{i3}
\right).
\]

The mediator $M_i$ is continuous. For $a\in\{0,1\}$, let $\underline{\mu}_a(\bX_i)=\mathbb{E}(M_i| A_i=a,\bX_i)$ denote the true conditional mediator mean. We first define the baseline mediator component
\begin{align*}
h_M(\bX_i)
=&
0.15
+
0.45\sin(2\pi X_{i1})
+
0.35\cos(2\pi X_{i2})
+
0.30\sin(\pi X_{i3})
+
0.28\sin(\pi X_{i1}X_{i2})\\
&\quad
+
0.24\cos(\pi X_{i2}X_{i3})
+
0.22\sin(\pi X_{i1}X_{i3})
+
0.35\operatorname{expit}\{2.0(X_{i1}+X_{i2}-0.25)\}\\
&\quad
-
0.30\operatorname{expit}\{2.0(X_{i2}-X_{i3}+0.15)\}
+
0.22\left(X_{i1}^2-\frac{1}{3}\right)
-
0.18\left(X_{i2}^2-\frac{1}{3}\right)\\
&\quad
+
0.15X_{i1}X_{i2}X_{i3}.
\end{align*}
The mediator treatment effect is
\[
\tau_M(\bX_i)
=
0.60
+
0.06X_{i1}
-
0.05X_{i2}
+
0.04X_{i3}
+
0.03X_{i1}X_{i2}
-
0.02X_{i2}X_{i3}.
\]
The two conditional mediator means are specified as
\[
\underline{\mu}_0(\bX_i)
=
h_M(\bX_i)-\frac{1}{2}\tau_M(\bX_i),
\qquad
\underline{\mu}_1(\bX_i)
=
h_M(\bX_i)+\frac{1}{2}\tau_M(\bX_i).
\]
Thus, $\underline{\mu}_1(\bX_i)-\underline{\mu}_0(\bX_i)=\tau_M(\bX_i)$. The observed mediator is generated as
\[
M_i=\underline{\mu}_{A_i}(\bX_i)+\mathcal{N}(0,1).
\]

The true CNIE is specified as the smooth nonlinear target
\begin{align*}
g(\bX_i)
=&
0.45\Bigg[
0.80\sin\left(\frac{\pi X_{i1}}{2}\right)
-
0.70\cos\left(\frac{\pi X_{i2}}{2}\right)
+
0.60\sin\left(\frac{\pi X_{i3}}{2}\right)\\
&\quad
+
0.45X_{i1}X_{i2}
-
0.40X_{i2}X_{i3}
+
0.35X_{i1}X_{i3}\\
&\quad
+
0.30\sin\left(\frac{\pi X_{i1}}{2}\right)
      \cos\left(\frac{\pi X_{i2}}{2}\right)
+
0.25\cos\left(\frac{\pi X_{i2}}{2}\right)
      \sin\left(\frac{\pi X_{i3}}{2}\right)
\Bigg].
\end{align*}
This target is smooth and nonlinear, with main nonlinear terms and two-way interaction structure.

The continuous outcome $Y_i$ is generated from
\[
Y_i=\mu_{A_i}(M_i,\bX_i)+\mathcal{N}(0,1.2^2),
\]
where, for $a\in\{0,1\}$,
\[
\mu_a(m,\bX_i)
=
h_Y(\bX_i)+a\,d(\bX_i)+\beta_M(\bX_i)m.
\]
The baseline outcome function is
\begin{align*}
h_Y(\bX_i)
=&
0.40
+
0.45X_{i1}
-
0.35X_{i2}
+
0.30X_{i3}
+
0.45\sin(\pi X_{i1})
+
0.35\cos(\pi X_{i2})\\
&\quad
+
0.30\sin(\pi X_{i3})
+
0.28\sin(\pi X_{i1})\cos(\pi X_{i3})
+
0.24\cos(\pi X_{i2})\sin(\pi X_{i3})\\
&\quad
+
0.22\operatorname{expit}\{2.0(X_{i1}+X_{i2}-0.2)\}
-
0.20\operatorname{expit}\{2.0(X_{i2}-X_{i3}+0.1)\}\\
&\quad
+
0.18X_{i1}X_{i2}
-
0.16X_{i2}X_{i3}
+
0.14X_{i1}X_{i2}X_{i3}.
\end{align*}
The direct treatment effect function is
\begin{align*}
d(\bX_i)
=&
-0.40
+
0.18X_{i1}
-
0.16X_{i2}
+
0.14X_{i3}
+
0.24\sin(\pi X_{i1})\sin(\pi X_{i3})\\
&\quad
+
0.20\cos(\pi X_{i2})\sin(\pi X_{i3})
+
0.16\operatorname{expit}\{2.0(X_{i1}+X_{i2}+X_{i3}-0.25)\}.
\end{align*}
The mediator-outcome coefficient is constructed as
\[
\beta_M(\bX_i)
=
\frac{g(\bX_i)}{\tau_M(\bX_i)}.
\]
Therefore, under the above data-generating process, the true conditional natural indirect effect is
\[
\beta_M(\bX_i)
\{\underline{\mu}_1(\bX_i)-\underline{\mu}_0(\bX_i)\}
=
\beta_M(\bX_i)\tau_M(\bX_i)
=
g(\bX_i).
\]
The pointwise metrics are calculated at four randomly chosen representative points: ${\bx_1}=(0.58,0.00,-1.00)^\top$, ${\bx_2}=(-0.83,-0.67,-0.33)^\top$, ${\bx_3}=(0.33,-0.42,0.25)^\top$, and ${\bx_4}=(0.42,-0.08,0.92)^\top$.

\section{Supporting information for the CARDIA analysis and two additional empirical applications}\label{sec:two-more-data-application}
\subsection{Additional information for the CARDIA application}
\label{supp:cardia-additional-information}

The cross-fitted probability of current smoking ranges from 0.0046 to 0.9434, with a median of 0.0250; 68.3\% of the observed profiles have a propensity below 0.05, and none have a propensity above 0.95. Thus, the overlap limitation is one-sided and concentrated near zero. Among the 402 current smokers, the effective sample size increases from 45.1 under the TR learner to 330.5 under the TW learner. Together with its favorable ISE performance and more concentrated ISE distributions in the simulation study, these findings motivate our choice to present the TTW learner in the main manuscript.

We compare the TR, TW, OW, TTR, TTW, and TOW estimates at the same 2,396 observed covariate profiles. Table~\ref{tab:cardia-six-learner-summary} reports each fitted distribution and its individual agreement with TTW. Across all 15 learner pairs, Pearson correlations range from 0.435 to 0.962, Spearman correlations range from 0.409 to 0.943, and sign agreement ranges from 66.0\% to 88.2\%. Agreement is strongest within corresponding untargeted--targeted pairs: the Spearman correlations are 0.863, 0.822, and 0.943 for TR--TTR, TW--TTW, and OW--TOW, respectively, with corresponding sign agreement of 83.8\%, 82.5\%, and 88.2\%. All six learners assign the same sign to 44.1\% of profiles, including 28.2\% positive under every learner and 15.9\% negative under every learner. The modest variation in agreement across methods may reflect statistical uncertainty arising from the low treatment prevalence and limited overlap: only 16.8\% of participants were current smokers at Year 20, and more than 60\% had estimated propensity scores below 0.05.
Relative to TR, TW, and OW, targeting reduces the mean pointwise band width by 60.9\%, 45.8\%, and 50.6\%, respectively, and reduces the mean uniform band width by 59.9\%, 43.9\%, and 49.5\%. 

\begin{table}[ht!]
\centering
\caption{Distribution of the estimated CNIEs and agreement with the TTW learner
across the 2,396 observed CARDIA covariate profiles.}
\label{tab:cardia-six-learner-summary}

\begingroup
\small
\setlength{\tabcolsep}{5pt}
\renewcommand{\arraystretch}{1.08}
\sisetup{table-number-alignment=center}

\begin{tabular}{
@{}l
S[table-format=1.3]
S[table-format=1.3]
c
S[table-format=2.1]
S[table-format=1.3]
S[table-format=3.1]
@{}
}
\toprule
& \multicolumn{4}{c}{Estimated CNIEs}
& \multicolumn{2}{c}{Agreement with TTW} \\
\cmidrule(lr){2-5}
\cmidrule(l){6-7}
Learner
& {Mean}
& {Median}
& {IQR}
& {Positive (\%)}
& {Spearman $\rho$}
& {Same sign (\%)} \\
\midrule
TR  & 0.544 & 0.389 & {$[-1.052,\,1.991]$} & 57.1 & 0.409 & 66.4 \\
TW  & 1.192 & 0.682 & {$[-1.745,\,3.414]$} & 58.1 & 0.822 & 82.5 \\
OW  & 0.610 & 0.085 & {$[-2.228,\,2.757]$} & 51.2 & 0.586 & 68.5 \\
TTR & 0.044 & 0.045 & {$[-0.904,\,0.980]$} & 51.5 & 0.456 & 66.4 \\
TTW & 0.562 & 0.623 & {$[-0.855,\,1.945]$} & 61.6 & 1.000 & 100.0 \\
TOW & 0.288 & 0.100 & {$[-1.433,\,1.876]$} & 51.7 & 0.647 & 72.6 \\
\bottomrule
\end{tabular}
\endgroup
\end{table}


The shallow fit-the-fit summaries identify related low-dimensional modifiers. TR first splits on household income and then baseline high-density lipoprotein cholesterol; TW splits on age and baseline BMI; OW splits on baseline hypertension history and age; TTR splits on household income and prior smoking; TTW splits on prior smoking and household income; and TOW splits on baseline hypertension history and household income. Household income and smoking history recur across multiple learners.

\subsection{Application to the PSACR-002 study}\label{subsec:data-application-PSACR-002}

We illustrate the proposed method using data from the Psychological Science Accelerator's COVID-19 Rapid-Response cognitive-reappraisal experiment (PSACR-002) \citep{wang2021multicountry}. For illustration, we restrict attention to U.S. participants and analyze $n=1238$ participants with complete mediator and outcome data. The treatment \(A\) is assignment to reconstrual, with active control as the comparison condition. Reconstrual instructed participants to reinterpret the COVID-19 situation, whereas active control asked them to reflect on their thoughts and feelings while providing a comparable amount of instruction and engagement. The mediator \(M\) is the mean negative-emotion rating across ten COVID-19-related photographs, and the outcome \(Y\) is the mean of five negative-emotion items assessed after the photograph-viewing task. Both variables are measured on the original 1--5 scale, with lower values indicating less negative emotion. We consider baseline demographic, socioeconomic, emotional, and COVID-19-related characteristics as plausible effect modifiers and mediator--outcome confounders. In the analytic sample, the mean age is 21.92 years, 72.4\% of participants are female, and the mean baseline negative-emotion score is 2.53. The scientific question is whether part of the effect of reconstrual on post-task negative emotion is transmitted through immediate emotional responses to the photographs and whether this pathway differs across participants. 

For illustration, we implement all six orthogonal learners with sieve smoothers satisfying $\check{\mathcal G}_n=\mathcal G_n$. The nuisance functions are estimated using Super Learner with a library comprising \texttt{SL.mean}, \texttt{SL.glm}, \texttt{SL.glmnet}, \texttt{SL.earth}, \texttt{SL.ranger}, and \texttt{SL.rpart}. The second-stage sieve models each of the seven continuous covariates using a cubic penalized B-spline basis of dimension five and includes categorical main effects. To accommodate scientifically plausible heterogeneity while limiting potential overfitting, the sieve includes pairwise tensor-product interactions between baseline negative emotion and baseline positive emotion, emotional worry, physical worry, and perceived manageability of restrictions. It also includes factor-by-smooth terms allowing the baseline-negative-emotion curve to differ under partial and full lockdown. Penalty selection follows the recommendations from the simulation study.

Supplementary Figures~\ref{fig:psacr-prediction-ow}--\ref{fig:psacr-prediction-tpsw} display the ordered CNIE estimates, together with pointwise and uniform confidence bands, for the OW, TOW, TR, TTR, TW, and TTW learners. Notably, the supremum-based critical value makes the uniform bands wider rather than smoother. For the TOW learner, the estimated CNIEs have a mean of \(-0.269\), a median of \(-0.274\), and an interquartile range of \([-0.429,-0.119]\), with 88.4\% being negative. Pointwise intervals lie entirely below zero for 38.9\% of the observed profiles, whereas the uniform band lies entirely below zero for only 1.7\%. Overall, the estimated mediated pathway is negative for most fitted profiles. The three targeted learners yield similar mean CNIEs and substantially stabilize the fitted surface near the boundaries relative to the untargeted learners.

To interpret the fitted heterogeneity, Supplementary Material Figure \ref{fig:psacr-tree-tow} summarizes the TOW-estimated CNIE surface using a shallow regression tree, an approach also known as \emph{fit-the-fit}. The first split is physical worry, suggesting that concern about the physical consequences of COVID-19 is the dominant modifier of the fitted mediation pathway. Among participants with physical-worry scores of at least 3.5, baseline positive emotion provides an additional split, with a more negative estimated CNIE among those with lower baseline positive emotion. The corresponding leaf means are \(-0.397\) and \(-0.229\), compared with \(-0.058\) among participants with physical-worry scores below 3.5. Supplementary Figures \ref{fig:psacr-tree-ow}--\ref{fig:psacr-tree-tpsw} present fit-the-fit summaries for the other learners, in which physical worry and baseline positive emotion also emerge as prominent modifiers. 

A population average analysis suggests a negative mediated pathway from reconstrual through immediate negative emotion to post-task negative emotion. The conditional estimates reveal a more nuanced pattern: the fitted pathway is weak for some baseline profiles and appreciably more negative among participants reporting greater physical worry, particularly those with lower baseline positive emotion.

\subsection{Application to the STAR study}\label{subsec:data-application-STAR}
We apply the proposed method to data from the Tennessee Student/Teacher Achievement Ratio (STAR) experiment \citep{krueger1999experimental}. 
For illustration, we proceed with a complete-case analysis of $n=4020$ students. The treatment ($A$) is assignment to a small kindergarten class, with regular and regular-with-aide classes combined as the comparison group. The mediator ($M$) is kindergarten listening achievement, measured after class assignment, and the outcome ($Y$) is first-grade reading achievement. 
{The baseline covariate vector $\bX$ comprises the pretreatment
characteristics prespecified as plausible effect modifiers and
mediator--outcome confounders: age at kindergarten entry, sex,
race/ethnicity, free-lunch status, and school locale.}
In the analytic sample, the mean age at kindergarten entry is 5.43 years; approximately 50.1\% of students are female, 69.8\% are White or Asian, and 44.3\% receive free lunch. The scientific question is whether part of the effect of small kindergarten classes on later reading achievement is transmitted through early listening achievement, and whether this pathway differs across students. To address this question, we estimate the CNIE on the original reading-score scale. For illustration, and following the recommendation from the simulation study, we implement the six orthogonal learners with sieve smoothers satisfying $\check{\mathcal{G}}_n=\mathcal{G}_n$. The sieve space is constructed using an order-6 spline for age at kindergarten entry, categorical main effects for sex, race/ethnicity, free-lunch status, and school locale, all pairwise categorical interactions, and factor-by-smooth age interactions for each categorical covariate. The penalty is selected by the GCV score. This specification allows nonlinear age-related heterogeneity and subgroup-specific age patterns while avoiding high-dimensional tensor-product overfitting.



Supplementary Material Figures \ref{fig:star-prediction-ow}--\ref{fig:star-prediction-tpsw} display the ordered estimated CNIEs, along with pointwise and uniform confidence bands, for all students under the OW, TOW, TR, TTR, TW, and TTW learners, respectively. Across the six learners, most estimated CNIEs oscillate around zero, whereas the upper tail is clearly positive and the lower tail is moderately negative. Thus, the listening-mediated pathway appears weak for many students but substantially stronger for a subset. All learners yield similar fitted patterns, while the targeted learners generally prevent extreme estimated values but appear to produce slightly wider bands. Notably, the uniform confidence bands based on the Gaussian bootstrap in Algorithm \ref{alg:gaussian-threshold} use a supremum-based critical value, which makes them wider, not smoother; their smoothness is determined by the estimator and standard errors. To interpret the fitted heterogeneity, Supplementary Material Figure \ref{fig:star-tree-tow} summarizes the TOW-estimated CNIE surface using a shallow regression tree, an approach also known as \emph{fit-the-fit}. The first split is school locale, suggesting that school context is the dominant low-dimensional modifier of the estimated mediation pathway. The fitted CNIE is smaller on average among students in inner-city or rural schools and larger among students in urban or suburban schools. Within the latter group, free-lunch status is an important additional modifier, with the largest estimated CNIEs appearing among students receiving free lunch. This pattern suggests that early listening achievement may be an especially important pathway through which small classes improve later reading for economically disadvantaged students in urban or suburban school settings. Supplementary Material Figures \ref{fig:star-tree-ow}--\ref{fig:star-tree-tpsw} present fit-the-fit plots based on the other learners, and the results are similar. 

A single population average indirect effect would suggest a positive but modest mediated pathway. The finer conditional estimates reveal a more nuanced pattern: for many students, the estimated listening-mediated component is small, whereas for particular baseline profiles it is substantially larger. For example, the analysis suggests a coherent and interpretable policy implication: small kindergarten classes may improve later reading partly by strengthening early listening skills, with this pathway most pronounced in school and socioeconomic contexts where early classroom learning may be especially consequential.

\clearpage

\section{Additional tables and figures}\label{supp;sec:additional-figure}

\subsection{Simulation experiments}
\begin{figure}[ht!]
    \centering
    \includegraphics[width=0.93\linewidth]{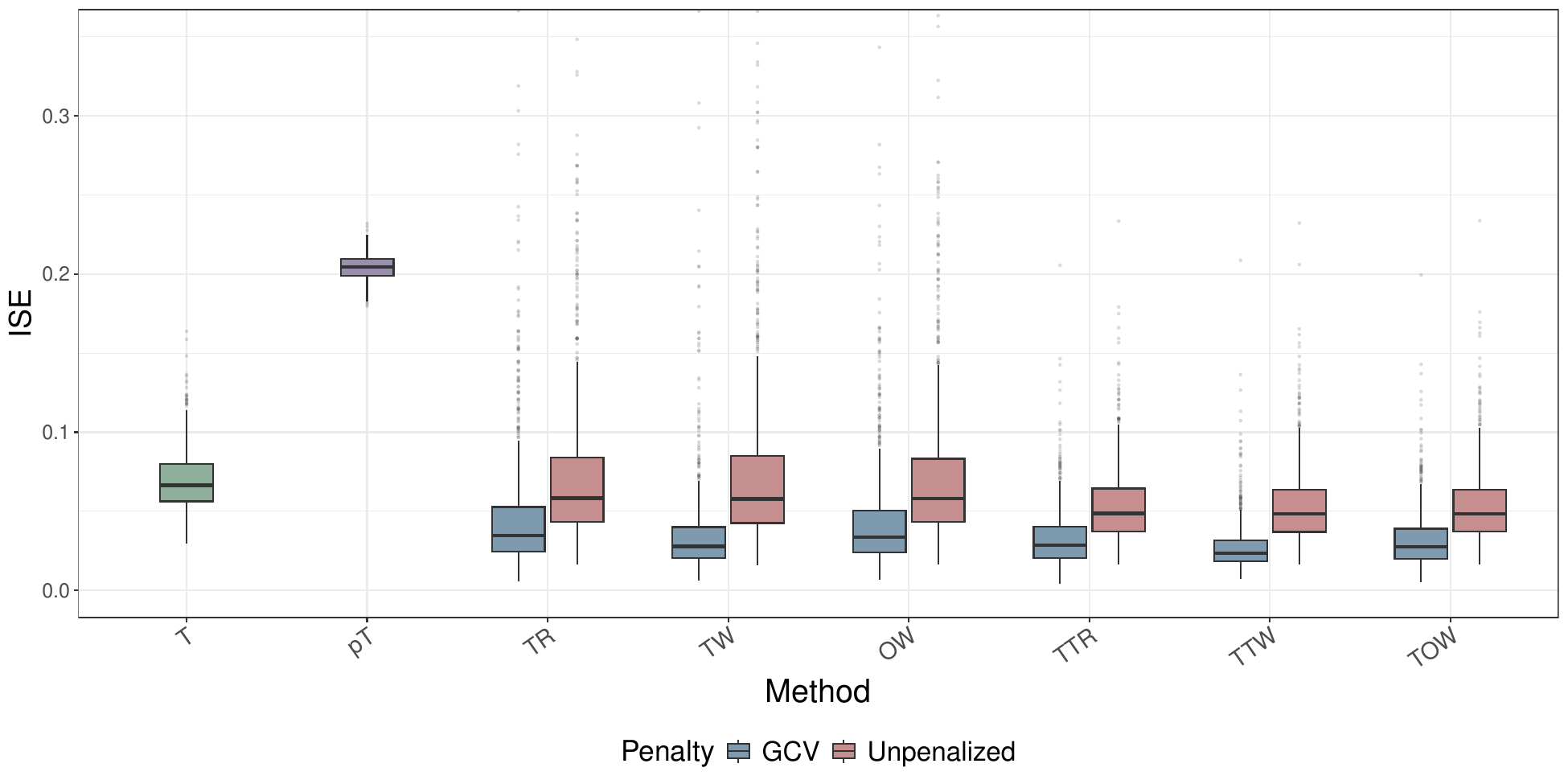}
    \caption{Simulation results presenting box plots of the integrated squared error (ISE) across $10^3$ iterations for the parametric T-learner in \cite{zhao2025estimation}, the nonparametric T-learner, and the six orthogonal learners in Section \ref{sec:6-repre-ortho-learners}. Each orthogonal learner is implemented using $(K_1,K_2)=(5,3)$ and two penalty strategies.}
    \label{fig:sim-results-ise-11,4}
\end{figure}

\begin{figure}[ht!]
    \centering
    \includegraphics[width=0.93\linewidth]{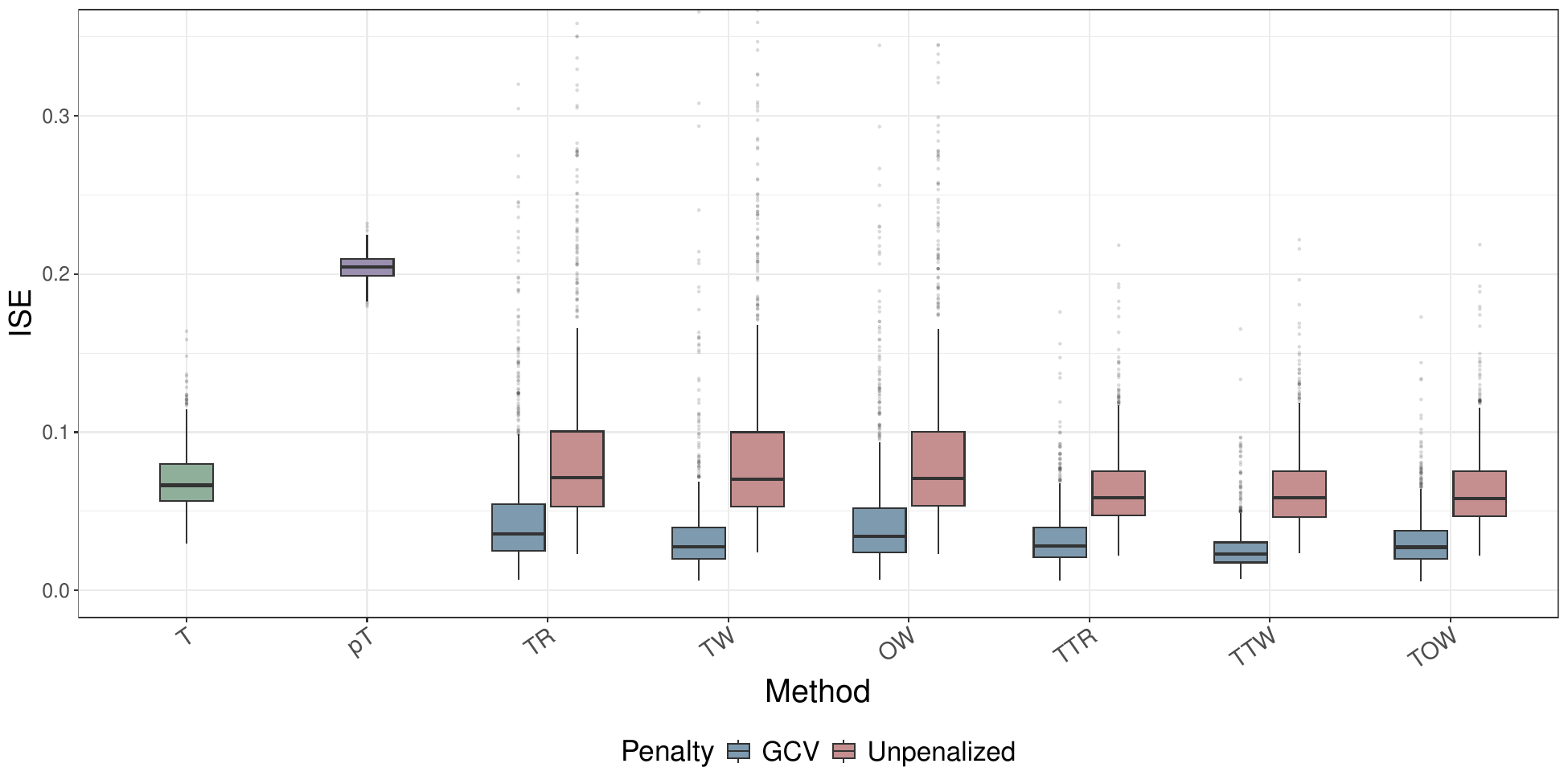}
    \caption{Simulation results presenting box plots of the integrated squared error (ISE) across $10^3$ iterations for the parametric T-learner in \cite{zhao2025estimation}, the nonparametric T-learner, and the six orthogonal learners in Section \ref{sec:6-repre-ortho-learners}. Each orthogonal learner is implemented using $(K_1,K_2)=(7,3)$ and two penalty strategies.}
    \label{fig:sim-results-ise-14,5}
\end{figure}

\begin{table}[ht!]
\centering
\caption{Pointwise bias, Monte Carlo standard deviation (MCSD), average estimated standard error (AESE), and empirical pointwise coverage for the parametric T-learner based on \cite{zhao2025estimation} at four representative points.}
\label{tab:pointwise-parametric-T}
\begingroup
\setlength{\tabcolsep}{6.8pt}
\renewcommand{\arraystretch}{0.70}
\setlength{\aboverulesep}{0.20ex}
\setlength{\belowrulesep}{0.21ex}
\begin{tabular}[t]{crrrr}
\toprule 
$\bx$ & BIAS & MCSD & AESE & CP(\%)\\
\midrule
$\bx_1$ & 0.27 & 0.03 & 0.03 & 0.0\\
$\bx_2$ & 0.50 & 0.03 & 0.03 & 0.0\\
$\bx_3$ & -0.26 & 0.03 & 0.03 & 0.0\\
$\bx_4$ & -0.58 & 0.03 & 0.03 & 0.0\\
\bottomrule
\end{tabular}
\endgroup
\end{table}

\begin{table}[ht!]
\centering
\caption{Pointwise bias, Monte Carlo standard deviation (MCSD), average estimated standard error (AESE), and empirical pointwise coverage for the TR learner at four representative points.
}\label{tab:pointwise-TR}
\begingroup
\setlength{\tabcolsep}{6.8pt}
\renewcommand{\arraystretch}{0.70}
\setlength{\aboverulesep}{0.20ex}
\setlength{\belowrulesep}{0.21ex}
\begin{tabular}[t]{ccrrrrrrrr}
\toprule
\multicolumn{2}{c}{ } & \multicolumn{4}{c}{GCV} & \multicolumn{4}{c}{Unpenalized} \\
\cmidrule(l{3pt}r{3pt}){3-6} \cmidrule(l{3pt}r{3pt}){7-10}
$K$ & $\bx$ & BIAS & MCSD & AESE & CP(\%) & BIAS & MCSD & AESE & CP(\%)\\
\midrule
\multirow{4}{*}{$(4,3)$} & $\bx_1$ & -0.04 & 0.22 & 0.26 & 98.0 & -0.02 & 0.29 & 0.26 & 92.6\\
 & $\bx_2$ & 0.01 & 0.18 & 0.22 & 98.4 & 0.02 & 0.21 & 0.22 & 96.5\\
 & $\bx_3$ & -0.03 & 0.10 & 0.13 & 99.1 & 0.01 & 0.13 & 0.13 & 95.1\\
 & $\bx_4$ & 0.02 & 0.16 & 0.20 & 97.9 & 0.02 & 0.21 & 0.19 & 94.1\\
\cmidrule{1-10}
\multirow{4}{*}{$(5,3)$} & $\bx_1$ & -0.05 & 0.22 & 0.29 & 98.4 & -0.02 & 0.31 & 0.29 & 94.0\\
 & $\bx_2$ & 0.01 & 0.18 & 0.23 & 98.5 & 0.01 & 0.22 & 0.23 & 96.3\\
 & $\bx_3$ & -0.04 & 0.10 & 0.14 & 98.2 & 0.00 & 0.14 & 0.14 & 95.7\\
 & $\bx_4$ & 0.03 & 0.16 & 0.21 & 98.4 & 0.03 & 0.22 & 0.21 & 94.0\\
\cmidrule{1-10}
\multirow{4}{*}{$(7,3)$} & $\bx_1$ & -0.03 & 0.23 & 0.34 & 99.5 & 0.02 & 0.37 & 0.33 & 93.1\\
 & $\bx_2$ & 0.01 & 0.19 & 0.25 & 99.0 & 0.02 & 0.25 & 0.25 & 96.3\\
 & $\bx_3$ & -0.04 & 0.10 & 0.17 & 99.3 & -0.02 & 0.18 & 0.17 & 94.0\\
 & $\bx_4$ & 0.02 & 0.16 & 0.22 & 98.8 & 0.03 & 0.24 & 0.22 & 95.2\\
\bottomrule
\end{tabular}
\endgroup
\end{table}

\begin{table}[ht!]
\centering
\caption{Pointwise bias, Monte Carlo standard deviation (MCSD), average estimated standard error (AESE), and empirical pointwise coverage for the TW learner at four representative points. 
}\label{tab:pointwise-PSW}
\begingroup
\setlength{\tabcolsep}{6.8pt}
\renewcommand{\arraystretch}{0.70}
\setlength{\aboverulesep}{0.20ex}
\setlength{\belowrulesep}{0.21ex}
\begin{tabular}[t]{ccrrrrrrrr}
\toprule
\multicolumn{2}{c}{ } & \multicolumn{4}{c}{GCV} & \multicolumn{4}{c}{Unpenalized} \\
\cmidrule(l{3pt}r{3pt}){3-6} \cmidrule(l{3pt}r{3pt}){7-10}
$K$ & $\bx$ & BIAS & MCSD & AESE & CP(\%) & BIAS & MCSD & AESE & CP(\%)\\
\midrule
\multirow{4}{*}{$(4,3)$} & $\bx_1$ & -0.04 & 0.16 & 0.26 & 99.6 & -0.02 & 0.30 & 0.26 & 93.4\\
 & $\bx_2$ & 0.00 & 0.16 & 0.22 & 99.1 & 0.02 & 0.21 & 0.22 & 96.0\\
 & $\bx_3$ & -0.04 & 0.07 & 0.13 & 99.1 & 0.01 & 0.13 & 0.13 & 95.8\\
 & $\bx_4$ & 0.04 & 0.13 & 0.19 & 99.1 & 0.02 & 0.21 & 0.19 & 93.8\\
\cmidrule{1-10}
\multirow{4}{*}{$(5,3)$} & $\bx_1$ & -0.04 & 0.16 & 0.29 & 99.8 & -0.01 & 0.32 & 0.29 & 93.7\\
 & $\bx_2$ & 0.00 & 0.16 & 0.23 & 99.0 & 0.01 & 0.22 & 0.23 & 96.2\\
 & $\bx_3$ & -0.04 & 0.07 & 0.14 & 99.4 & 0.00 & 0.14 & 0.14 & 95.7\\
 & $\bx_4$ & 0.04 & 0.12 & 0.21 & 99.4 & 0.03 & 0.22 & 0.21 & 94.3\\
\cmidrule{1-10}
\multirow{4}{*}{$(7,3)$} & $\bx_1$ & -0.04 & 0.16 & 0.35 & 99.9 & 0.02 & 0.38 & 0.34 & 92.9\\
 & $\bx_2$ & 0.01 & 0.16 & 0.26 & 99.3 & 0.02 & 0.25 & 0.26 & 96.1\\
 & $\bx_3$ & -0.04 & 0.07 & 0.17 & 99.7 & -0.02 & 0.18 & 0.17 & 94.6\\
 & $\bx_4$ & 0.03 & 0.13 & 0.23 & 99.6 & 0.03 & 0.24 & 0.23 & 94.8\\
\bottomrule
\end{tabular}
\endgroup
\end{table}

\begin{table}[ht!]
\centering
\caption{Pointwise bias, Monte Carlo standard deviation (MCSD), average estimated standard
error (AESE), and empirical pointwise coverage for the OW learner at four representative points.
}\label{tab:pointwise-OW}
\begingroup
\setlength{\tabcolsep}{6.8pt}
\renewcommand{\arraystretch}{0.70}
\setlength{\aboverulesep}{0.20ex}
\setlength{\belowrulesep}{0.21ex}
\begin{tabular}[t]{ccrrrrrrrr}
\toprule
\multicolumn{2}{c}{ } & \multicolumn{4}{c}{GCV} & \multicolumn{4}{c}{Unpenalized} \\
\cmidrule(l{3pt}r{3pt}){3-6} \cmidrule(l{3pt}r{3pt}){7-10}
$K$ & $\bx$ & BIAS & MCSD & AESE & CP(\%) & BIAS & MCSD & AESE & CP(\%)\\
\midrule
\multirow{4}{*}{$(4,3)$} & $\bx_1$ & -0.05 & 0.21 & 0.26 & 98.0 & -0.02 & 0.29 & 0.26 & 92.4\\
 & $\bx_2$ & 0.01 & 0.18 & 0.22 & 98.3 & 0.02 & 0.21 & 0.22 & 96.4\\
 & $\bx_3$ & -0.04 & 0.10 & 0.13 & 98.9 & 0.01 & 0.13 & 0.13 & 95.0\\
 & $\bx_4$ & 0.02 & 0.16 & 0.20 & 98.1 & 0.02 & 0.21 & 0.19 & 94.2\\
\cmidrule{1-10}
\multirow{4}{*}{$(5,3)$} & $\bx_1$ & -0.05 & 0.22 & 0.28 & 98.4 & -0.02 & 0.31 & 0.28 & 93.8\\
 & $\bx_2$ & 0.01 & 0.18 & 0.23 & 98.5 & 0.01 & 0.22 & 0.23 & 96.3\\
 & $\bx_3$ & -0.04 & 0.10 & 0.14 & 98.3 & 0.00 & 0.14 & 0.14 & 95.5\\
 & $\bx_4$ & 0.03 & 0.16 & 0.21 & 98.5 & 0.03 & 0.22 & 0.21 & 93.9\\
\cmidrule{1-10}
\multirow{4}{*}{$(7,3)$} & $\bx_1$ & -0.05 & 0.22 & 0.34 & 99.3 & 0.02 & 0.37 & 0.34 & 93.5\\
 & $\bx_2$ & 0.01 & 0.19 & 0.26 & 98.8 & 0.02 & 0.25 & 0.26 & 96.3\\
 & $\bx_3$ & -0.05 & 0.10 & 0.17 & 99.4 & -0.02 & 0.18 & 0.17 & 94.2\\
 & $\bx_4$ & 0.03 & 0.16 & 0.23 & 99.2 & 0.03 & 0.24 & 0.23 & 95.2\\
\bottomrule
\end{tabular}
\endgroup
\end{table}

\begin{table}[ht!]
\centering
\caption{Pointwise bias, Monte Carlo standard deviation (MCSD), average estimated standard error (AESE), and empirical pointwise coverage for the TTR learner at four representative points. 
}\label{tab:pointwise-TTR}
\begingroup
\setlength{\tabcolsep}{6.8pt}
\renewcommand{\arraystretch}{0.70}
\setlength{\aboverulesep}{0.20ex}
\setlength{\belowrulesep}{0.21ex}
\begin{tabular}[t]{ccrrrrrrrr}
\toprule
\multicolumn{2}{c}{ } & \multicolumn{4}{c}{GCV} & \multicolumn{4}{c}{Unpenalized} \\
\cmidrule(l{3pt}r{3pt}){3-6} \cmidrule(l{3pt}r{3pt}){7-10}
$K$ & $\bx$ & BIAS & MCSD & AESE & CP(\%) & BIAS & MCSD & AESE & CP(\%)\\
\midrule
\multirow{4}{*}{$(4,3)$} & $\bx_1$ & -0.04 & 0.20 & 0.24 & 97.7 & -0.01 & 0.28 & 0.24 & 91.8\\
 & $\bx_2$ & 0.01 & 0.17 & 0.21 & 98.5 & 0.02 & 0.20 & 0.21 & 96.5\\
 & $\bx_3$ & -0.03 & 0.09 & 0.12 & 98.6 & 0.02 & 0.12 & 0.12 & 95.0\\
 & $\bx_4$ & 0.01 & 0.15 & 0.18 & 98.3 & 0.01 & 0.19 & 0.18 & 93.2\\
\cmidrule{1-10}
\multirow{4}{*}{$(5,3)$} & $\bx_1$ & -0.04 & 0.21 & 0.27 & 98.7 & -0.01 & 0.30 & 0.27 & 92.3\\
 & $\bx_2$ & 0.01 & 0.17 & 0.22 & 98.6 & 0.01 & 0.21 & 0.22 & 96.3\\
 & $\bx_3$ & -0.04 & 0.09 & 0.13 & 98.7 & 0.01 & 0.13 & 0.13 & 95.6\\
 & $\bx_4$ & 0.00 & 0.14 & 0.19 & 98.6 & 0.01 & 0.20 & 0.19 & 93.7\\
\cmidrule{1-10}
\multirow{4}{*}{$(7,3)$} & $\bx_1$ & -0.03 & 0.21 & 0.32 & 99.4 & 0.02 & 0.35 & 0.32 & 92.3\\
 & $\bx_2$ & 0.01 & 0.17 & 0.24 & 99.5 & 0.02 & 0.23 & 0.24 & 96.8\\
 & $\bx_3$ & -0.04 & 0.09 & 0.16 & 99.5 & -0.00 & 0.16 & 0.16 & 94.1\\
 & $\bx_4$ & 0.01 & 0.14 & 0.21 & 99.3 & 0.01 & 0.22 & 0.21 & 94.7\\
\bottomrule
\end{tabular}
\endgroup
\end{table}

\begin{table}[ht!]
\centering
\caption{Pointwise bias, Monte Carlo standard deviation (MCSD), average estimated standard error (AESE), and empirical pointwise coverage for the TOW learner at four representative points. 
}\label{tab:pointwise-TOW}
\begingroup
\setlength{\tabcolsep}{6.8pt}
\renewcommand{\arraystretch}{0.70}
\setlength{\aboverulesep}{0.20ex}
\setlength{\belowrulesep}{0.21ex}
\begin{tabular}[t]{ccrrrrrrrr}
\toprule
\multicolumn{2}{c}{ } & \multicolumn{4}{c}{GCV} & \multicolumn{4}{c}{Unpenalized} \\
\cmidrule(l{3pt}r{3pt}){3-6} \cmidrule(l{3pt}r{3pt}){7-10}
$K$ & $\bx$ & BIAS & MCSD & AESE & CP(\%) & BIAS & MCSD & AESE & CP(\%)\\
\midrule
\multirow{4}{*}{$(4,3)$} & $\bx_{1}$ & -0.04 & 0.20 & 0.24 & 97.7 & -0.00 & 0.28 & 0.24 & 91.8\\
 & $\bx_{2}$ & 0.01 & 0.17 & 0.21 & 98.6 & 0.02 & 0.20 & 0.21 & 96.3\\
 & $\bx_{3}$ & -0.04 & 0.09 & 0.12 & 99.0 & 0.02 & 0.12 & 0.12 & 94.9\\
 & $\bx_{4}$ & 0.01 & 0.14 & 0.18 & 98.4 & 0.01 & 0.19 & 0.18 & 93.4\\
\cmidrule{1-10}
\multirow{4}{*}{$(5,3)$} & $\bx_{1}$ & -0.05 & 0.20 & 0.27 & 98.6 & -0.01 & 0.30 & 0.27 & 92.5\\
 & $\bx_{2}$ & 0.01 & 0.17 & 0.22 & 98.6 & 0.01 & 0.21 & 0.22 & 96.2\\
 & $\bx_{3}$ & -0.04 & 0.09 & 0.13 & 98.9 & 0.01 & 0.13 & 0.13 & 95.3\\
 & $\bx_{4}$ & 0.00 & 0.14 & 0.19 & 98.9 & 0.01 & 0.20 & 0.19 & 93.7\\
\cmidrule{1-10}
\multirow{4}{*}{$(7,3)$} & $\bx_{1}$ & -0.04 & 0.20 & 0.32 & 99.4 & 0.02 & 0.35 & 0.31 & 92.5\\
 & $\bx_{2}$ & 0.01 & 0.17 & 0.24 & 99.6 & 0.02 & 0.23 & 0.24 & 96.6\\
 & $\bx_{3}$ & -0.04 & 0.09 & 0.16 & 99.6 & -0.00 & 0.16 & 0.16 & 93.8\\
 & $\bx_{4}$ & 0.01 & 0.14 & 0.21 & 99.6 & 0.01 & 0.22 & 0.21 & 94.6\\
\bottomrule
\end{tabular}
\endgroup
\end{table}

\begin{table}[ht!]
\centering
\caption{Empirical coverage of uniform confidence bands across $25^3$ equally spaced grid points in the cube $[-1,1]^3$. EP denotes empirical simultaneous coverage probability, reported as a percentage. GNC denotes grid noncoverage, defined as $10^4 \times (1-\bar{C}_{\mathrm{grid}})$, where $\bar{C}_{\mathrm{grid}}$ is the average fraction of evaluation-grid points covered by the uniform band. Smaller GNC indicates better grid-point coverage.}
\label{tab:uniform-coverage}

\begingroup
\setlength{\tabcolsep}{6.8pt}        
\renewcommand{\arraystretch}{0.60}   
\setlength{\aboverulesep}{0.20ex}    
\setlength{\belowrulesep}{0.21ex}    

\begin{tabular}{cclcccccc}
\toprule
$(K_1,K_2)$ & Penalty & Metric & TR & TW & OW & TTR & TTW & TOW \\
\midrule
\multirow{4}{*}{$(4,3)$}
& \multirow{2}{*}{GCV}
& EP  & 95.2\% & 94.6\% & 95.4\% & 97.8\% & 98.3\% & 97.9\% \\
& & GNC & 1.644 & 5.349 & 1.646 & 0.304 & 0.178 & 0.349 \\[0.5ex]
& \multirow{2}{*}{Unpenalized}
& EP  & 89.7\% & 89.5\% & 89.5\% & 89.3\% & 89.3\% & 88.8\% \\
& & GNC & 3.253 & 3.166 & 3.452 & 3.744 & 3.560 & 3.980 \\
\cmidrule(lr){2-9}

\multirow{4}{*}{$(5,3)$}
& \multirow{2}{*}{GCV}
& EP  & 96.2\% & 96.4\% & 96.9\% & 99.6\% & 99.9\% & 99.4\% \\
& & GNC & 1.761 & 4.842 & 1.620 & 0.066 & 0.003 & 0.070 \\[0.5ex]
& \multirow{2}{*}{Unpenalized}
& EP  & 90.0\% & 90.4\% & 89.5\% & 89.6\% & 89.3\% & 89.5\% \\
& & GNC & 2.462 & 2.304 & 2.551 & 2.930 & 2.644 & 3.081 \\
\cmidrule(lr){2-9}

\multirow{4}{*}{$(7,3)$}
& \multirow{2}{*}{GCV}
& EP  & 97.4\% & 97.1\% & 97.3\% & 99.9\% & 100.0\% & 99.9\% \\
& & GNC & 1.367 & 3.532 & 1.219 & 0.006 & 0.000 & 0.003 \\[0.5ex]
& \multirow{2}{*}{Unpenalized}
& EP  & 90.2\% & 89.9\% & 90.3\% & 88.2\% & 89.1\% & 88.9\% \\
& & GNC & 1.126 & 0.998 & 1.171 & 1.420 & 1.199 & 1.473 \\
\bottomrule
\end{tabular}
\endgroup
\end{table}

\clearpage

\subsection{CARDIA application}

\begin{figure}[ht!]
   \centering
   \includegraphics[width=0.88\linewidth]
   {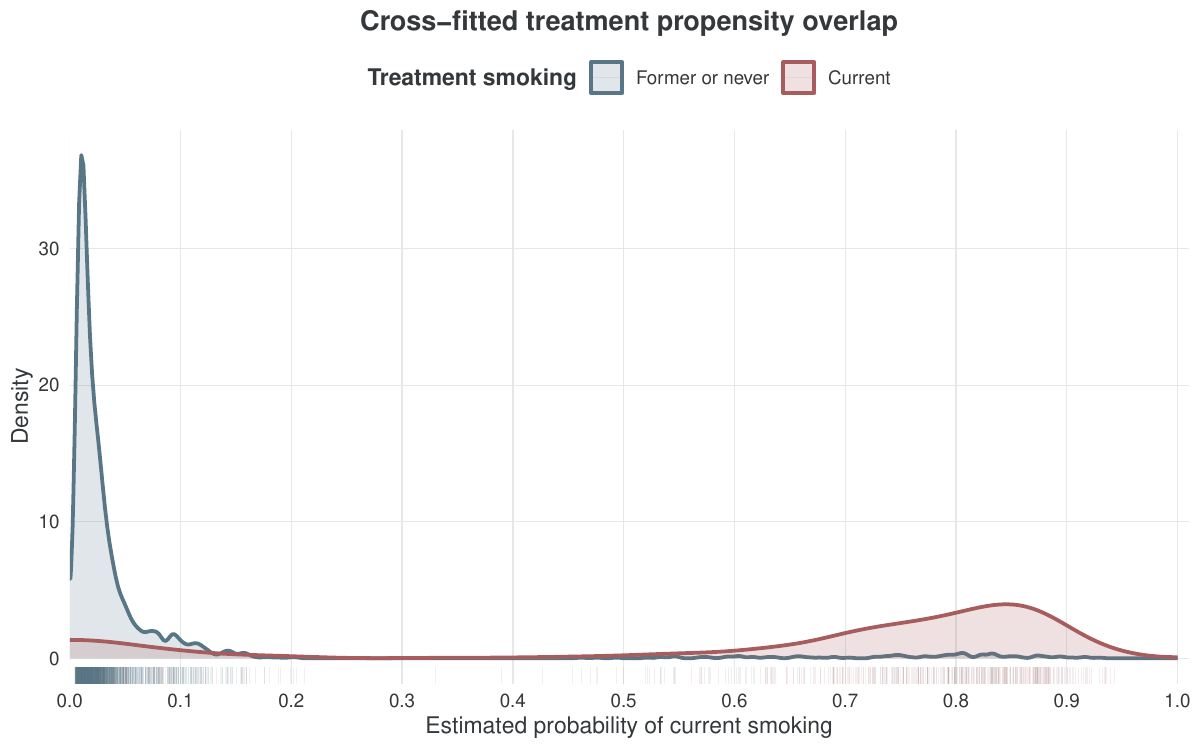}
   \caption{Cross-fitted treatment propensity overlap by Year-20 smoking status. The horizontal axis is the estimated probability of current smoking; rug marks show the observed covariate profiles.}
   \label{fig:cardia-imat-propensity-overlap}
\end{figure}

\begin{figure}[ht!]
   \centering
   \includegraphics[width=0.95\linewidth]
   {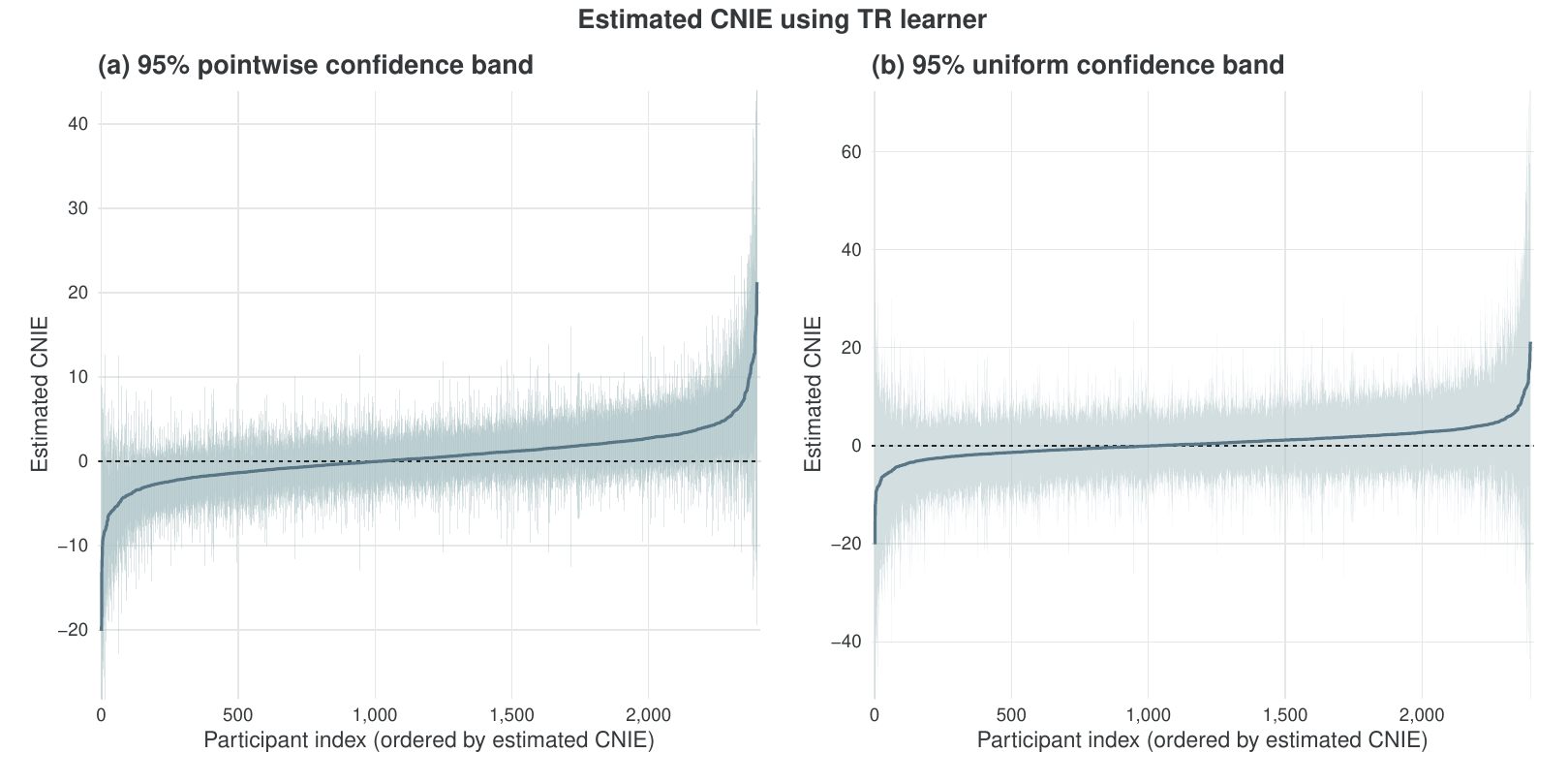}
   \caption{CNIE estimates at the observed covariate profiles in the CARDIA application, obtained using the TR learner, along with 95\% pointwise and uniform confidence bands.}
   \label{fig:cardia-imat-prediction-tr}
\end{figure}

\begin{figure}[ht!]
   \centering
   \includegraphics[width=0.95\linewidth]
   {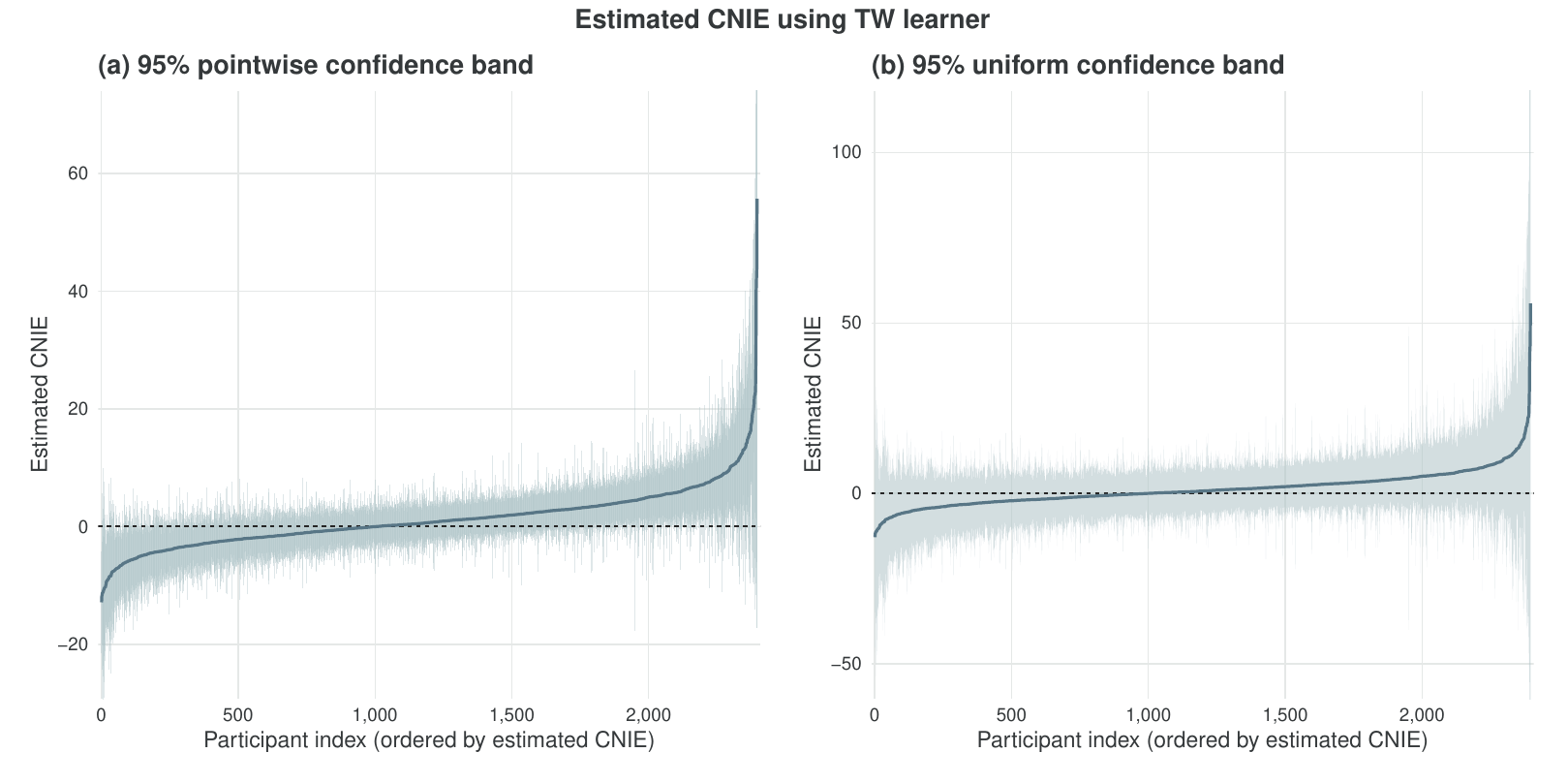}
   \caption{CNIE estimates at the observed covariate profiles in the CARDIA application, obtained using the TW learner, along with 95\% pointwise and uniform confidence bands.}
   \label{fig:cardia-imat-prediction-tw}
\end{figure}

\begin{figure}[ht!]
   \centering
   \includegraphics[width=0.95\linewidth]
   {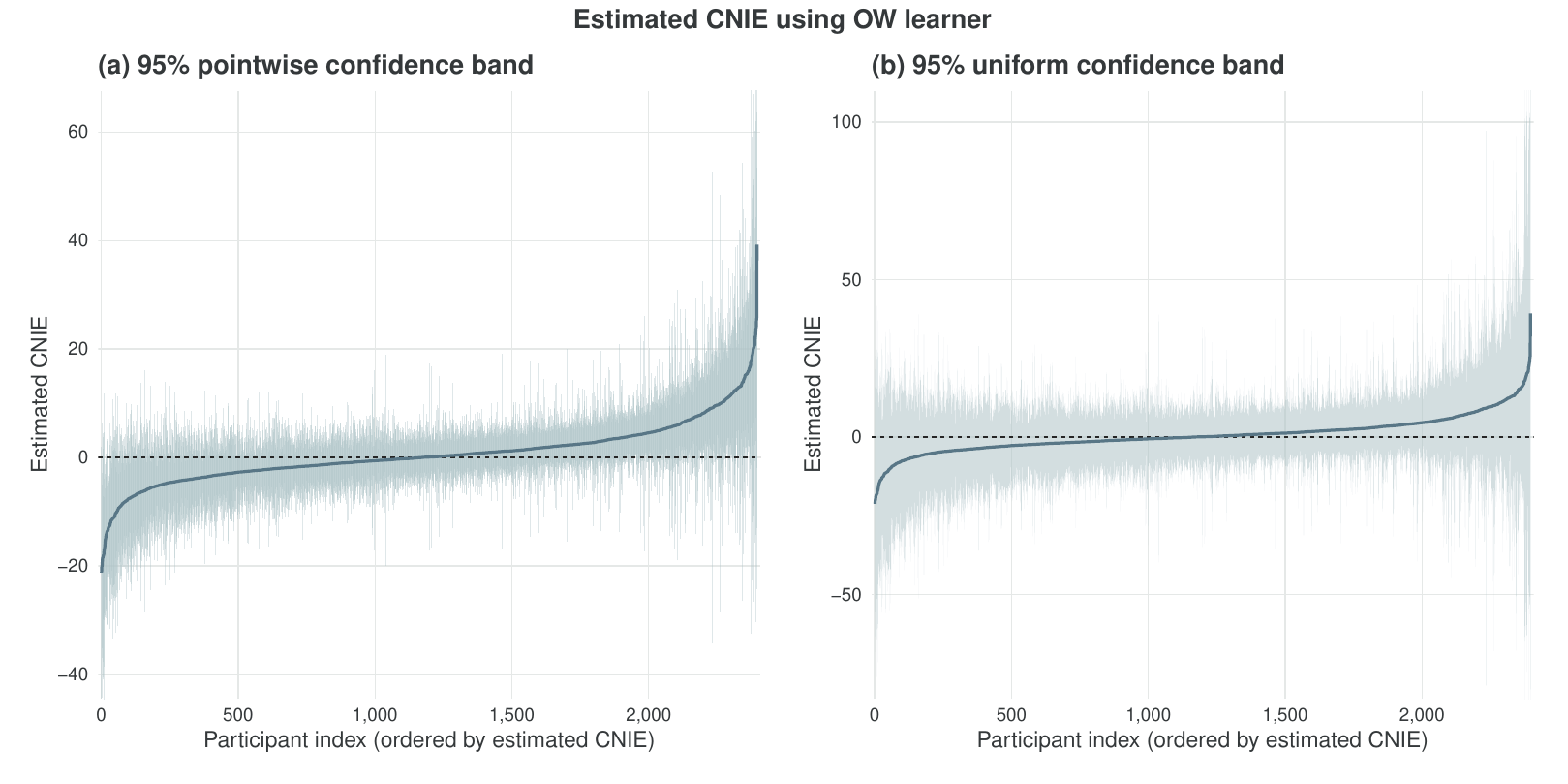}
   \caption{CNIE estimates at the observed covariate profiles in the CARDIA application, obtained using the OW learner, along with 95\% pointwise and uniform confidence bands.}
   \label{fig:cardia-imat-prediction-ow}
\end{figure}

\begin{figure}[ht!]
   \centering
   \includegraphics[width=0.95\linewidth]
   {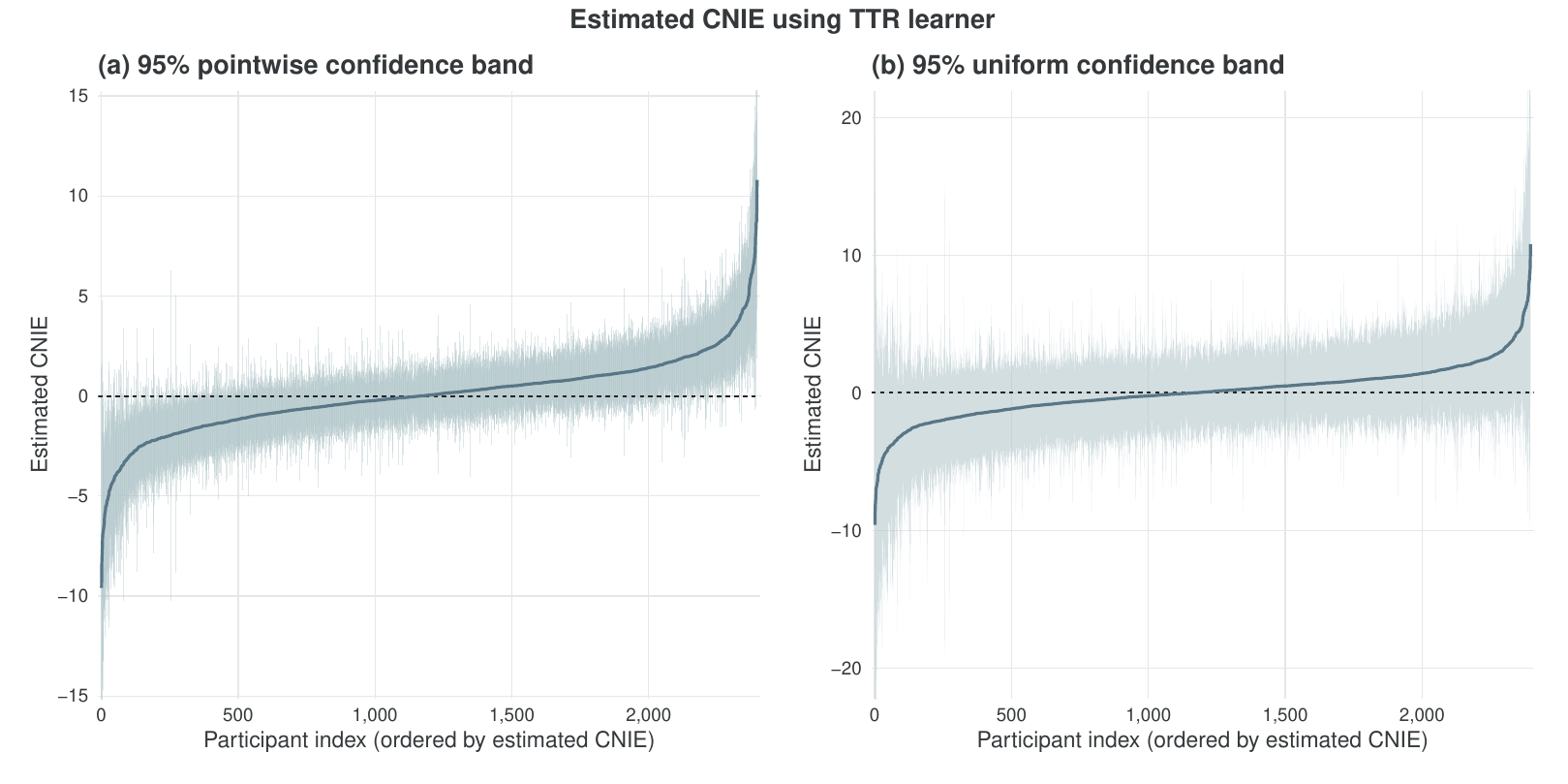}
   \caption{CNIE estimates at the observed covariate profiles in the CARDIA application, obtained using the TTR learner, along with 95\% pointwise and uniform confidence bands.}
   \label{fig:cardia-imat-prediction-ttr}
\end{figure}

\begin{figure}[ht!]
    \centering
    \includegraphics[width=0.95\linewidth]
    {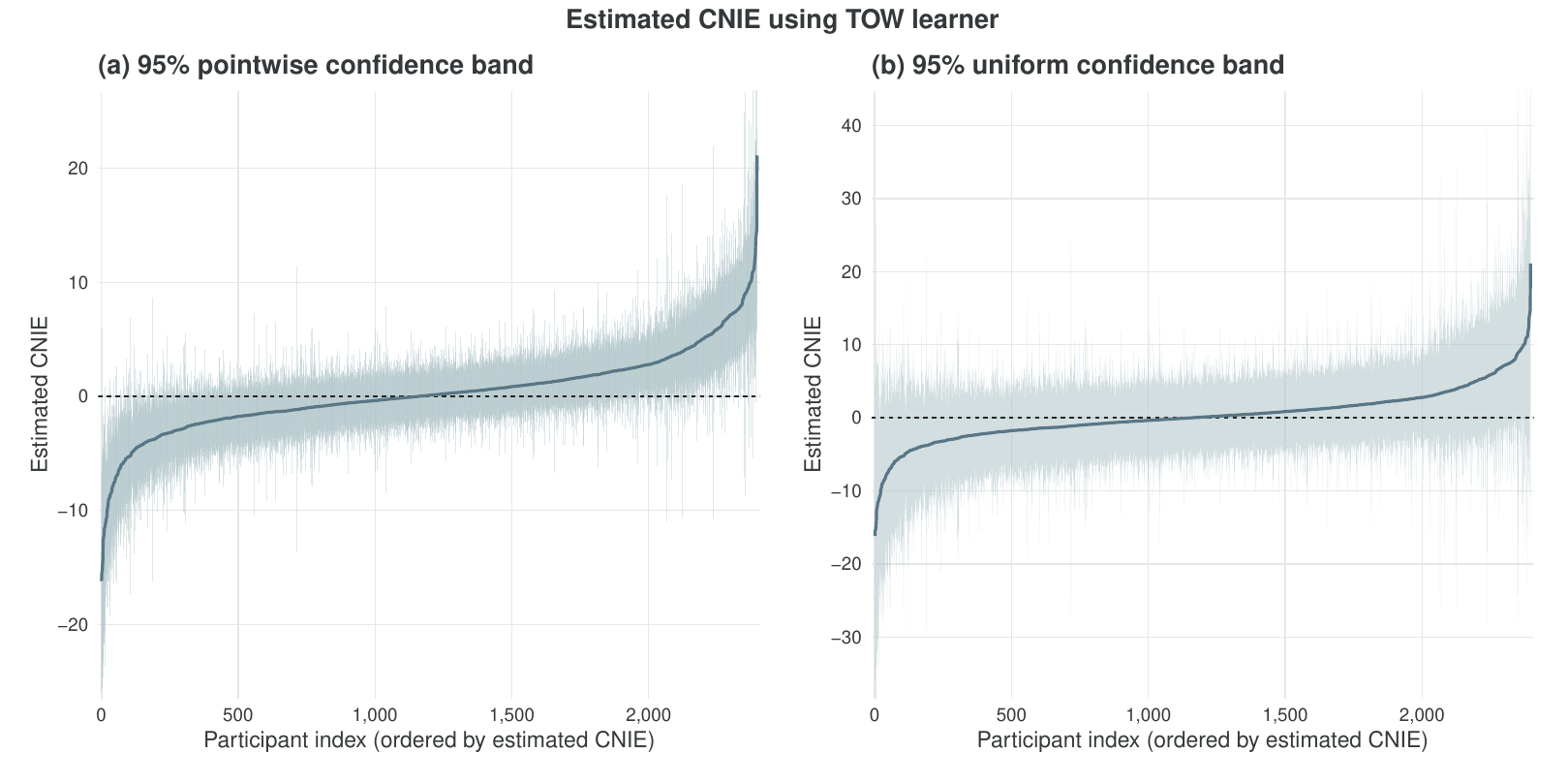}
    \caption{CNIE estimates at the observed covariate profiles in the CARDIA application, obtained using
    the TOW learner, along with 95\% pointwise confidence intervals and a 95\% simultaneous confidence band.}
    \label{fig:cardia-imat-prediction-tow}
\end{figure}

\begin{figure}[ht!]
   \centering
   \includegraphics[width=0.85\linewidth]
   {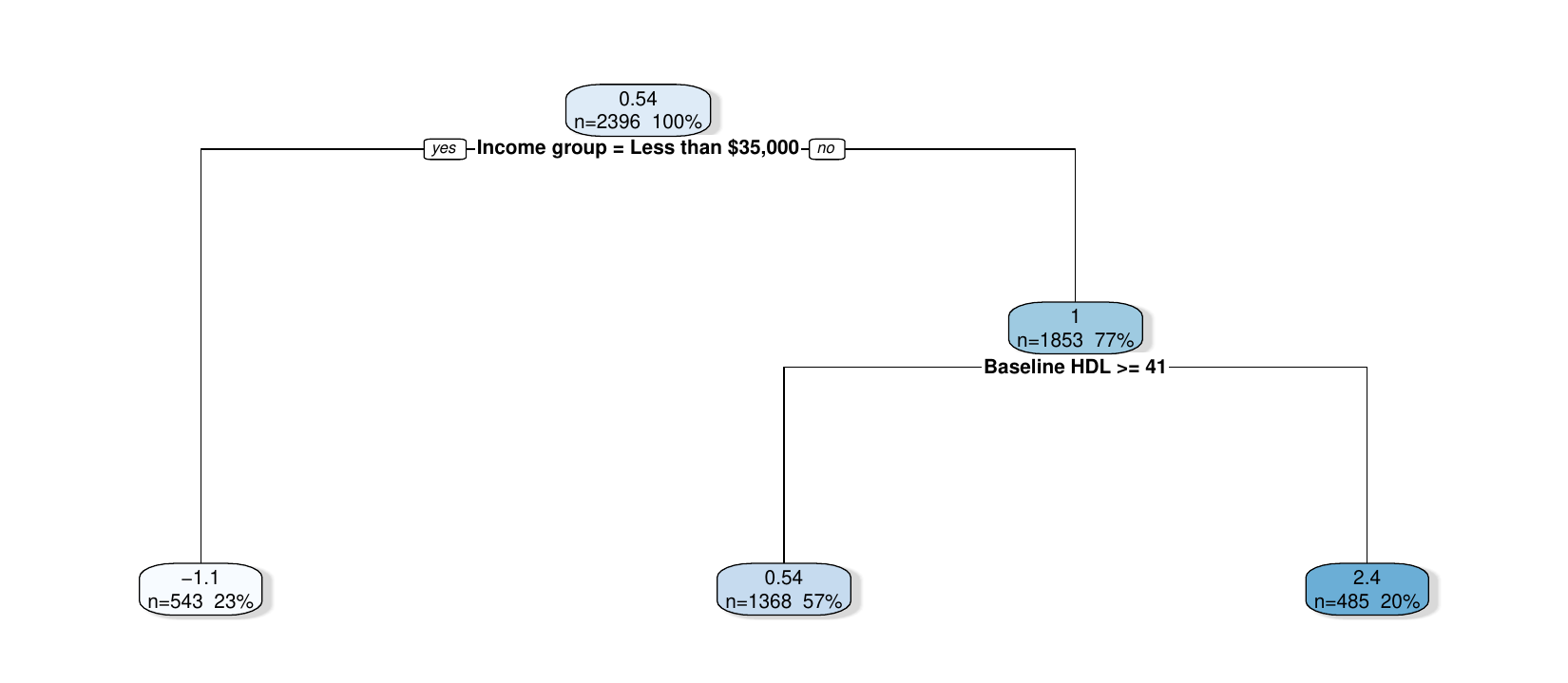}
   \caption{Fit-the-fit plot based on a decision tree summary of the estimated CNIEs from the TR learner in the CARDIA application.}
   \label{fig:cardia-imat-tree-tr}
\end{figure}

\begin{figure}[ht!]
   \centering
   \includegraphics[width=0.85\linewidth]
   {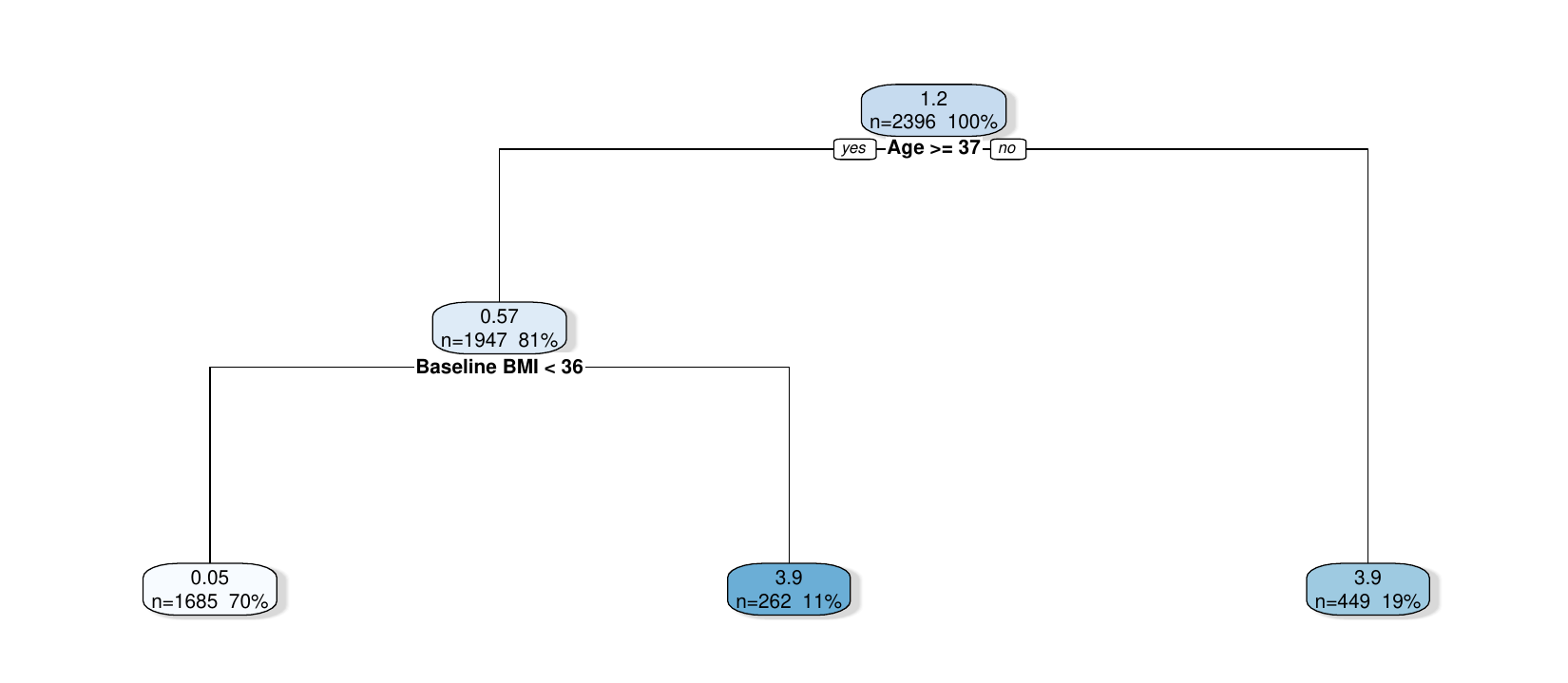}
   \caption{Fit-the-fit plot based on a decision tree summary of the estimated CNIEs from the TW learner in the CARDIA application.}
   \label{fig:cardia-imat-tree-tw}
\end{figure}

\begin{figure}[ht!]
   \centering
   \includegraphics[width=0.85\linewidth]
   {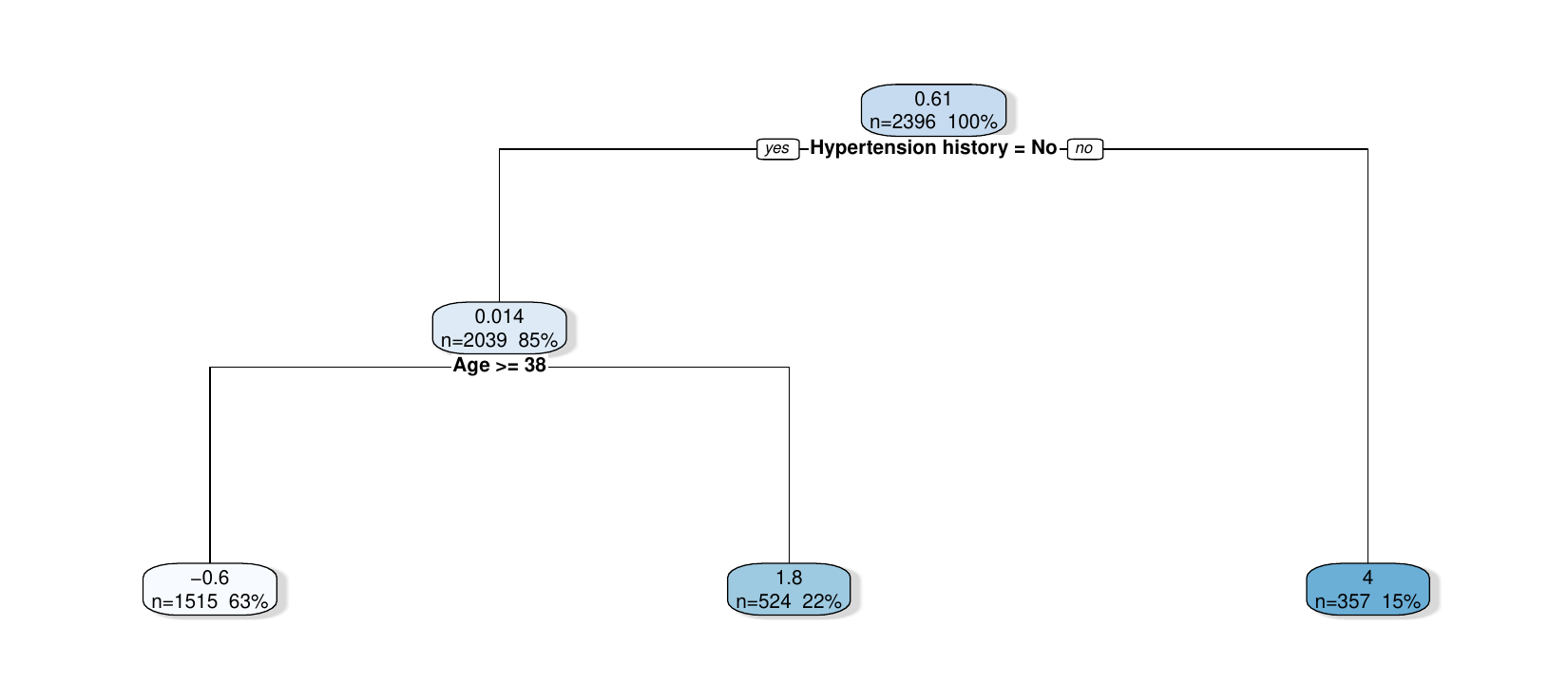}
   \caption{Fit-the-fit plot based on a decision tree summary of the estimated CNIEs from the OW learner in the CARDIA application.}
   \label{fig:cardia-imat-tree-ow}
\end{figure}

\begin{figure}[ht!]
   \centering
   \includegraphics[width=0.85\linewidth]
   {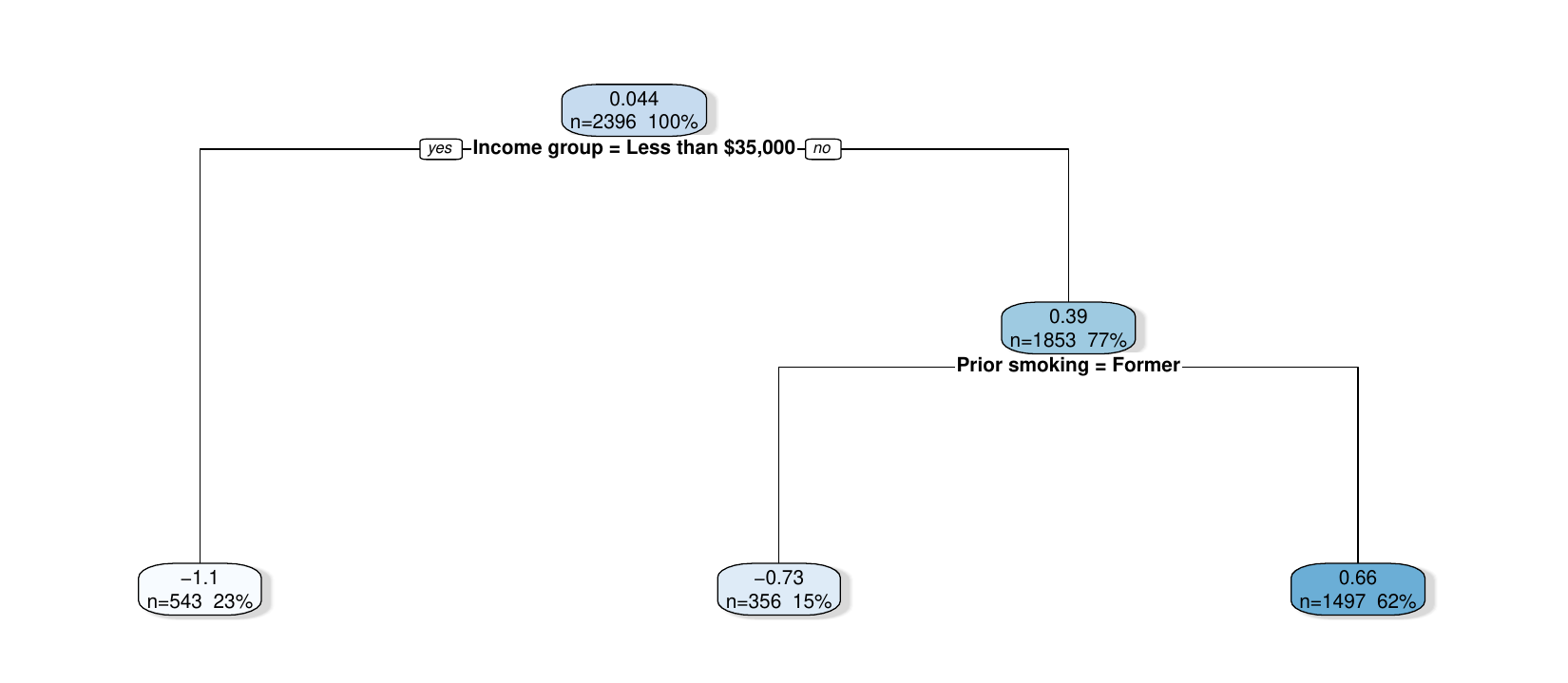}
   \caption{Fit-the-fit plot based on a decision tree summary of the estimated CNIEs from the TTR learner in the CARDIA application.}
   \label{fig:cardia-imat-tree-ttr}
\end{figure}

\begin{figure}[ht!]
    \centering
    \includegraphics[width=0.85\linewidth]
    {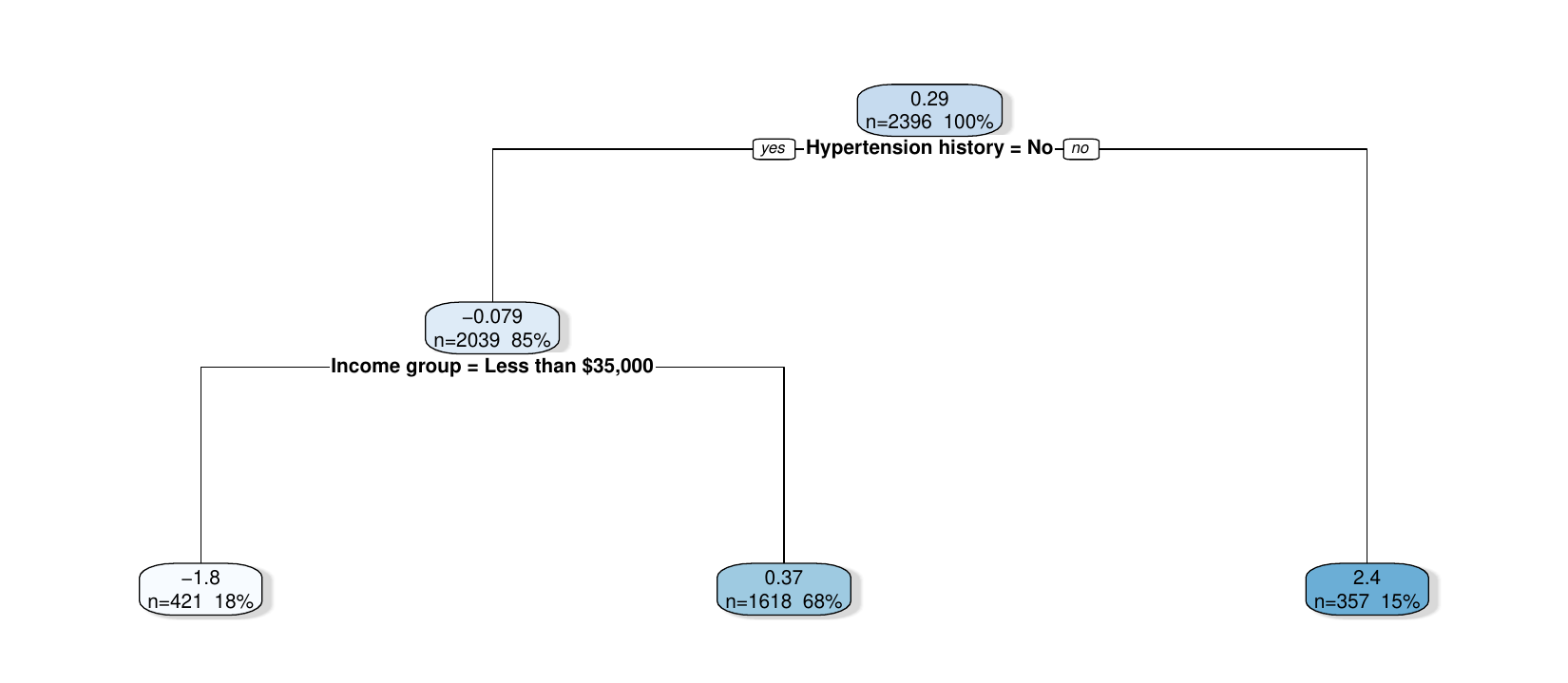}
    \caption{Fit-the-fit plot based on a decision tree summary of the estimated
    CNIEs from the TOW learner in the CARDIA application.}
    \label{fig:cardia-imat-tree-tow}
\end{figure}
\clearpage

\subsection{PSACR application}

\begin{figure}[ht!]
    \centering
    \includegraphics[width=0.95\linewidth]{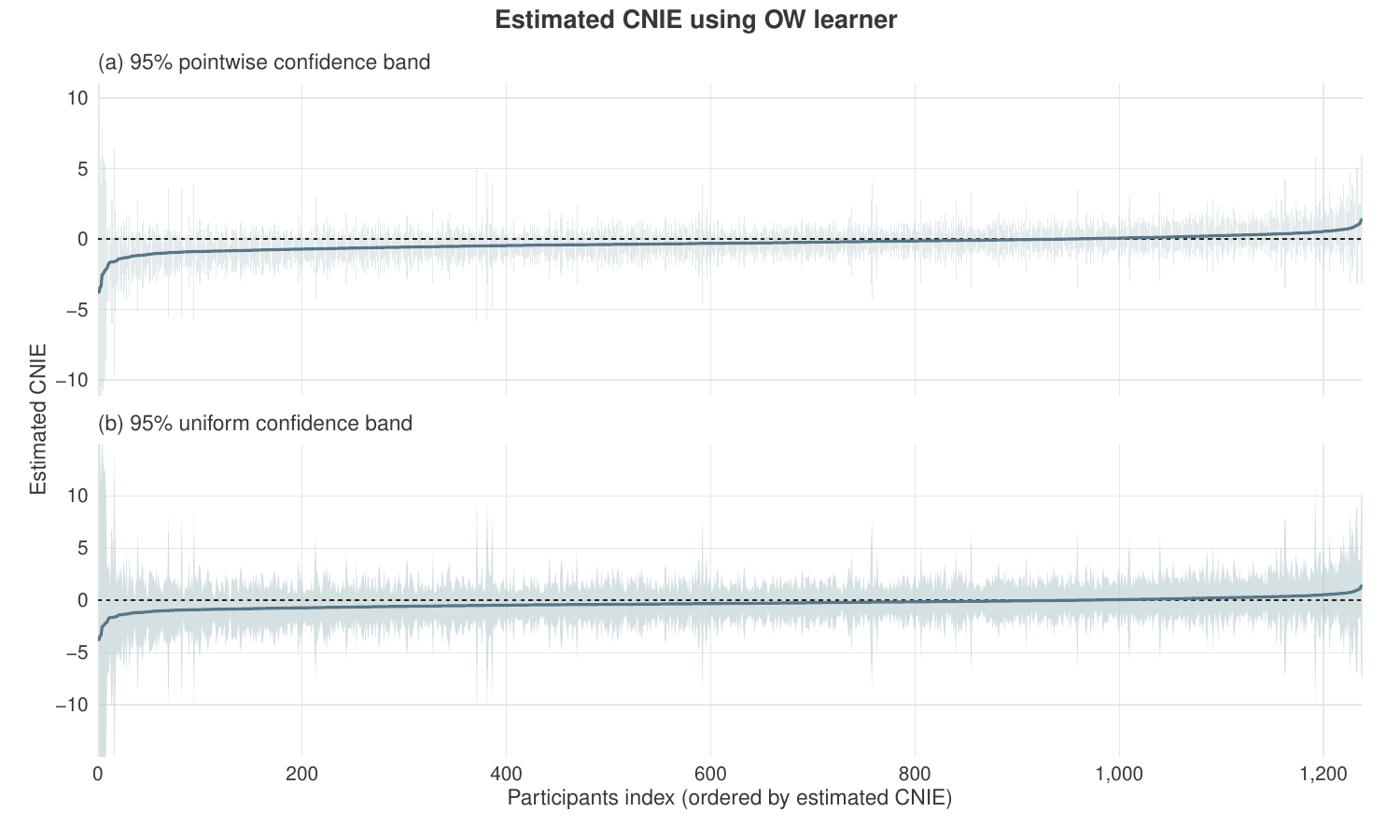}
    \caption{CNIE estimates at the observed covariate profiles in the PSACR application, obtained using
    the OW learner, along with pointwise and uniform confidence bands.}
    \label{fig:psacr-prediction-ow}
\end{figure}

\begin{figure}[ht!]
    \centering
    \includegraphics[width=0.95\linewidth]
    {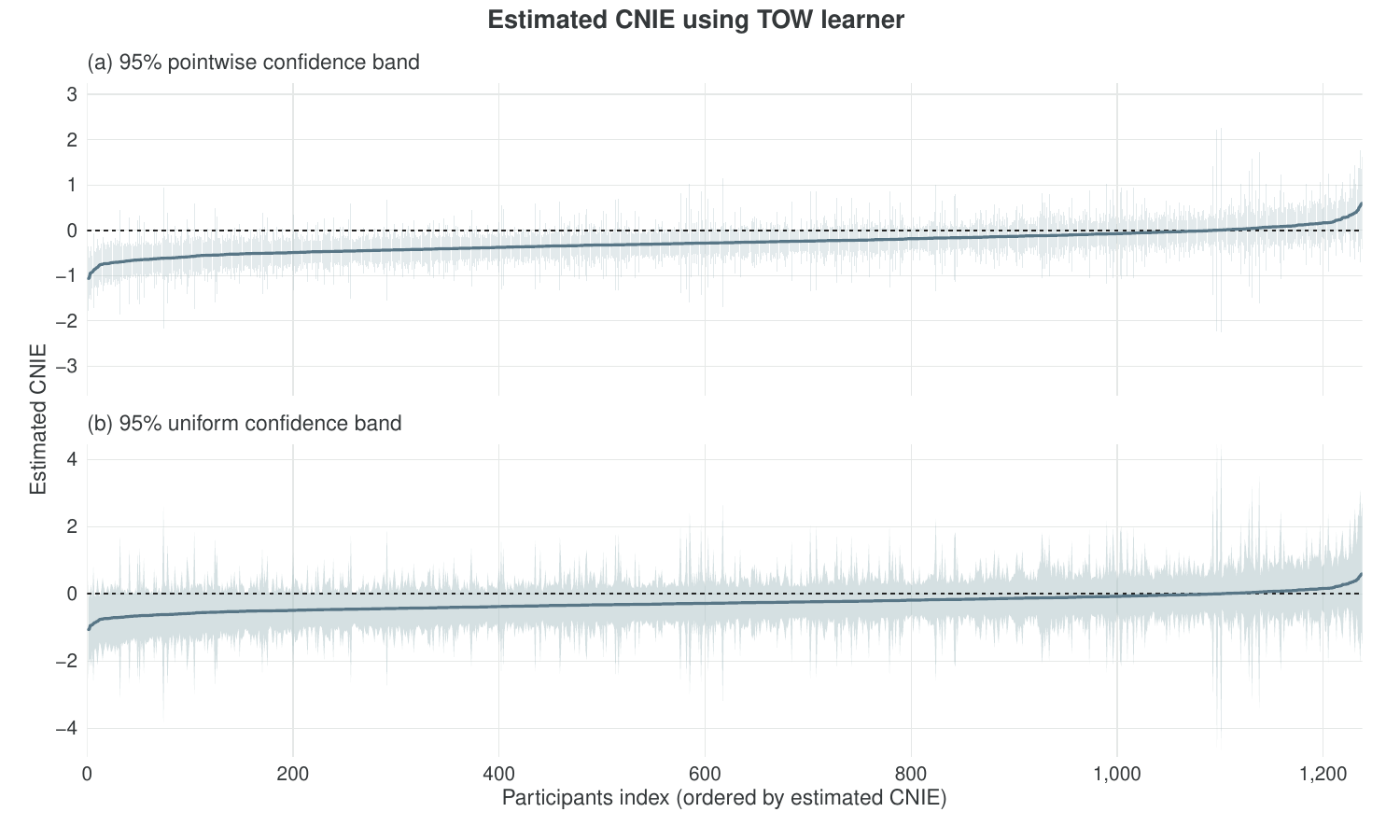}
    \caption{CNIE estimates at the observed covariate profiles in the PSACR application, obtained using
    the TOW learner, along with pointwise and uniform confidence bands.}
    \label{fig:psacr-prediction-tow}
\end{figure}

\begin{figure}[ht!]
    \centering
    \includegraphics[width=0.95\linewidth]{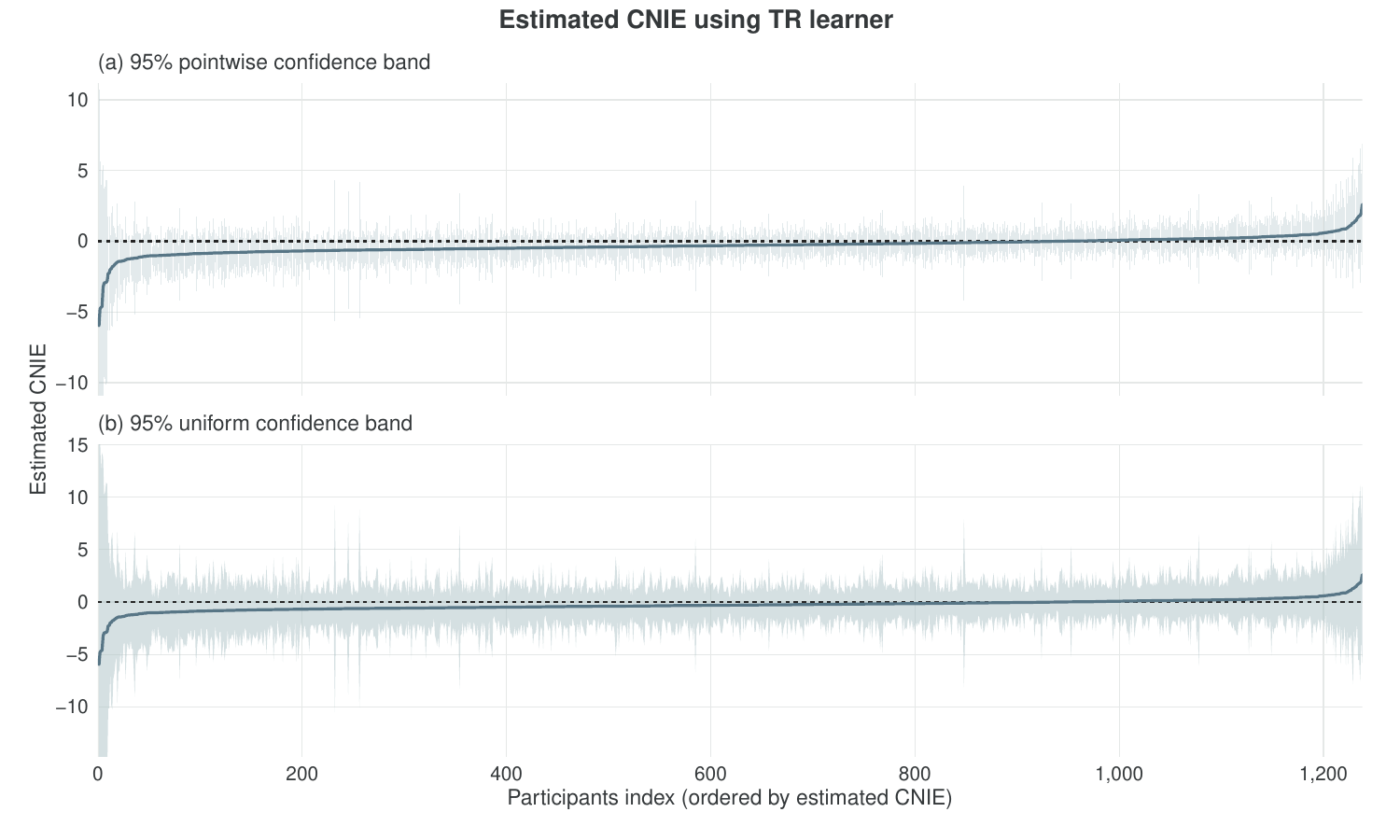}
    \caption{CNIE estimates at the observed covariate profiles in the PSACR application, obtained using
    the TR learner, along with pointwise and uniform confidence bands.}
    \label{fig:psacr-prediction-tr}
\end{figure}

\begin{figure}[ht!]
    \centering
    \includegraphics[width=0.95\linewidth]{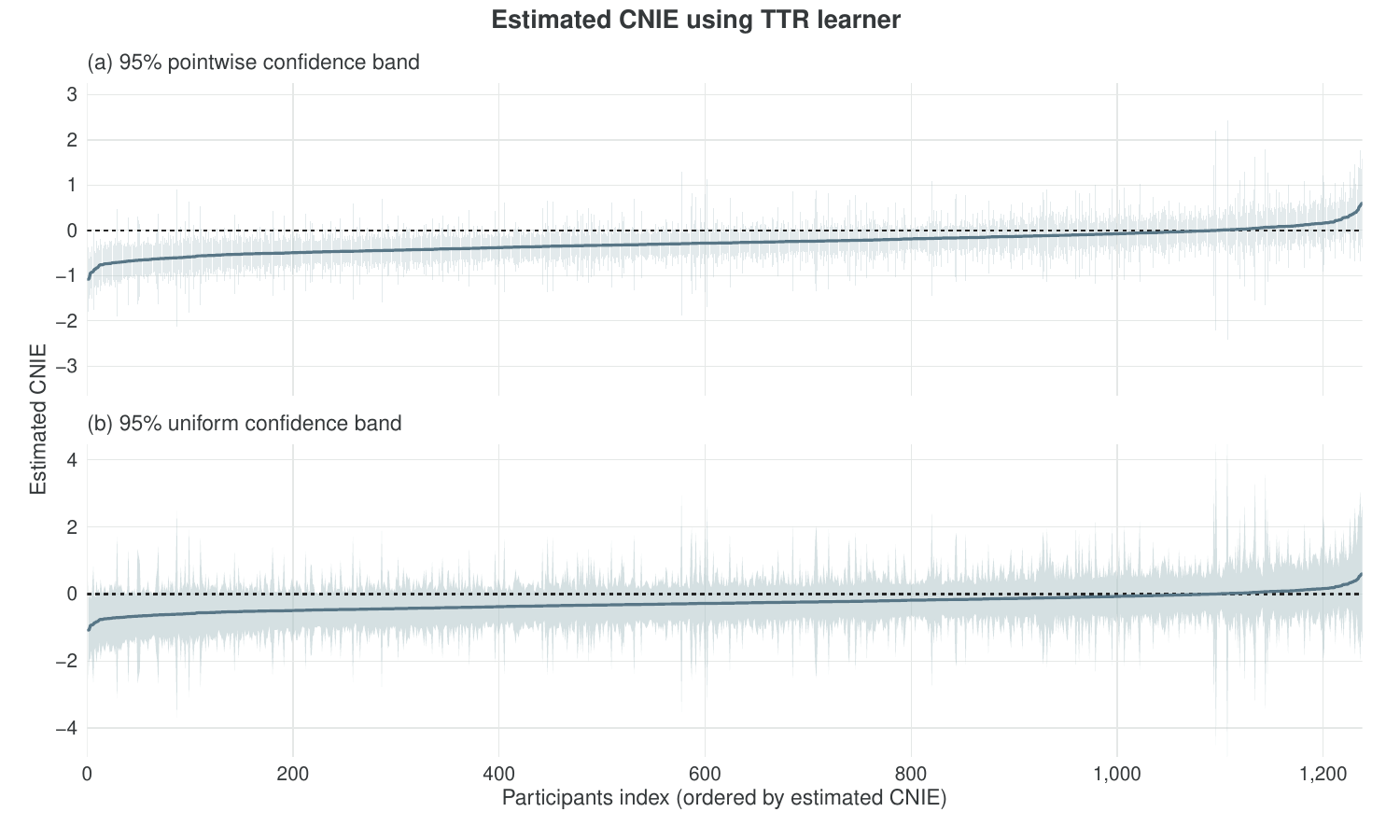}
    \caption{CNIE estimates at the observed covariate profiles in the PSACR application, obtained using
    the TTR learner, along with pointwise and uniform confidence bands.}
    \label{fig:psacr-prediction-ttr}
\end{figure}

\begin{figure}[ht!]
    \centering
    \includegraphics[width=0.95\linewidth]{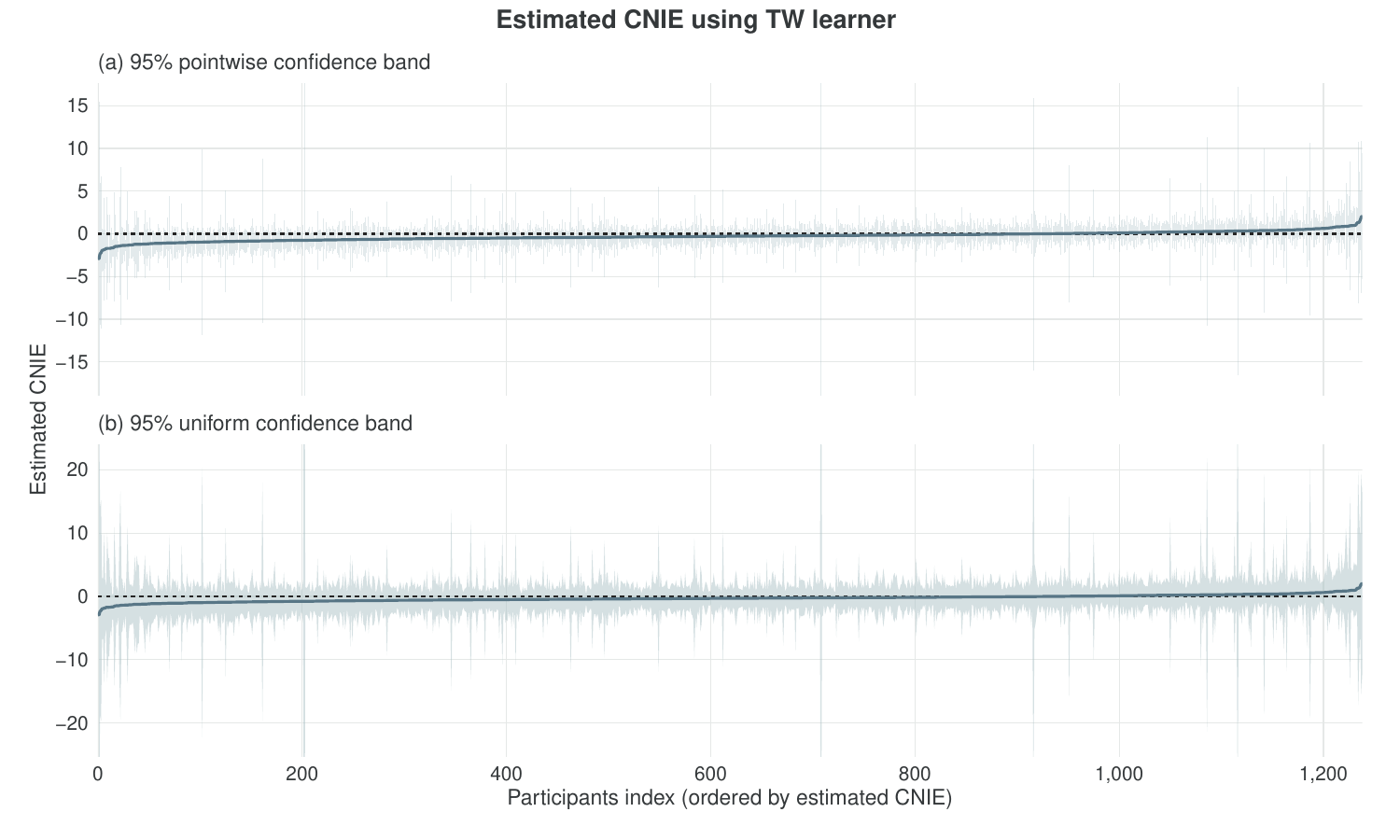}
    \caption{CNIE estimates at the observed covariate profiles in the PSACR application, obtained using
    the TW learner, along with pointwise and uniform confidence bands.}
    \label{fig:psacr-prediction-psw}
\end{figure}

\begin{figure}[ht!]
    \centering
    \includegraphics[width=0.95\linewidth]{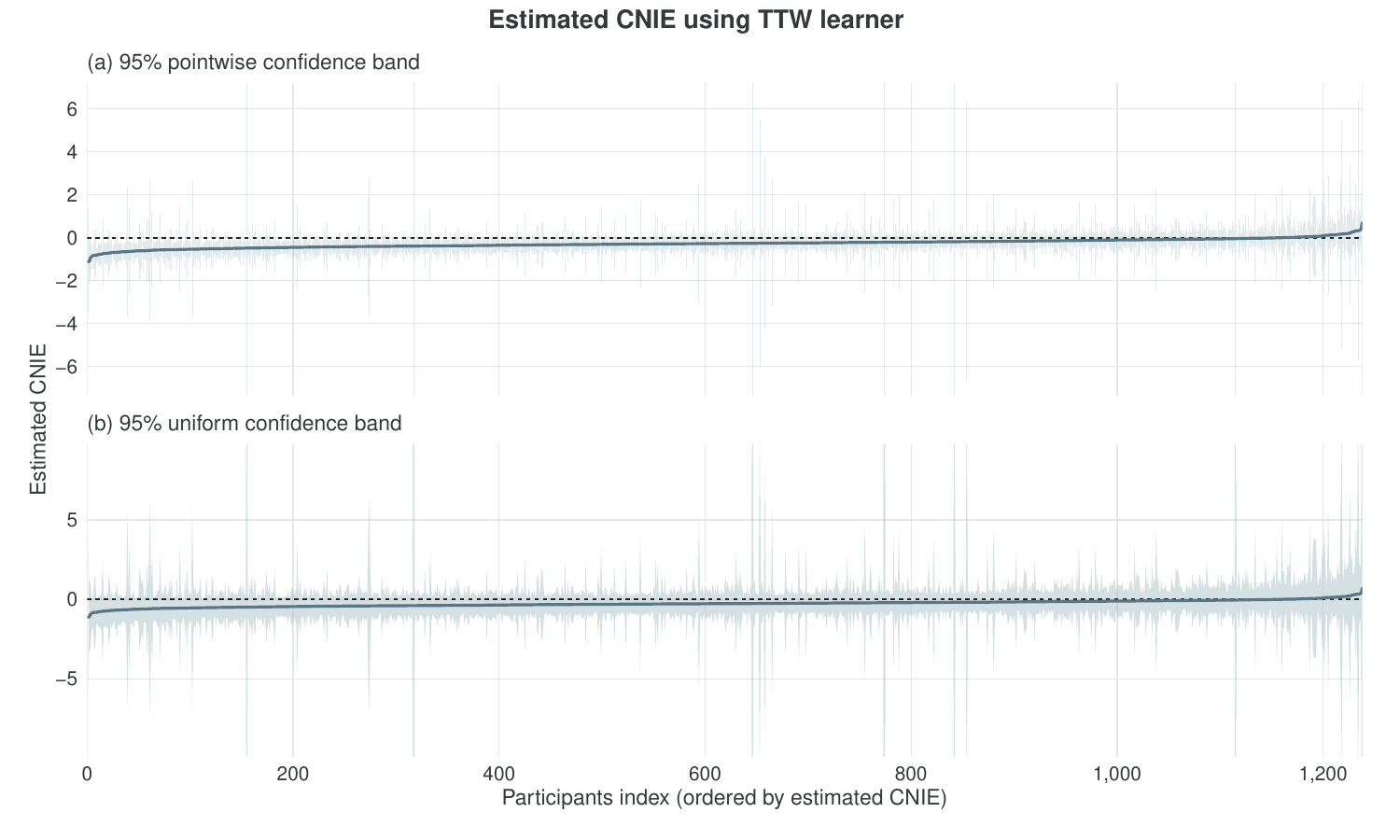}
    \caption{CNIE estimates at the observed covariate profiles in the PSACR application, obtained using
    the TTW learner, along with pointwise and uniform confidence bands.}
    \label{fig:psacr-prediction-tpsw}
\end{figure}

\begin{figure}[ht!]
    \centering
    \includegraphics[width=0.85\linewidth]
    {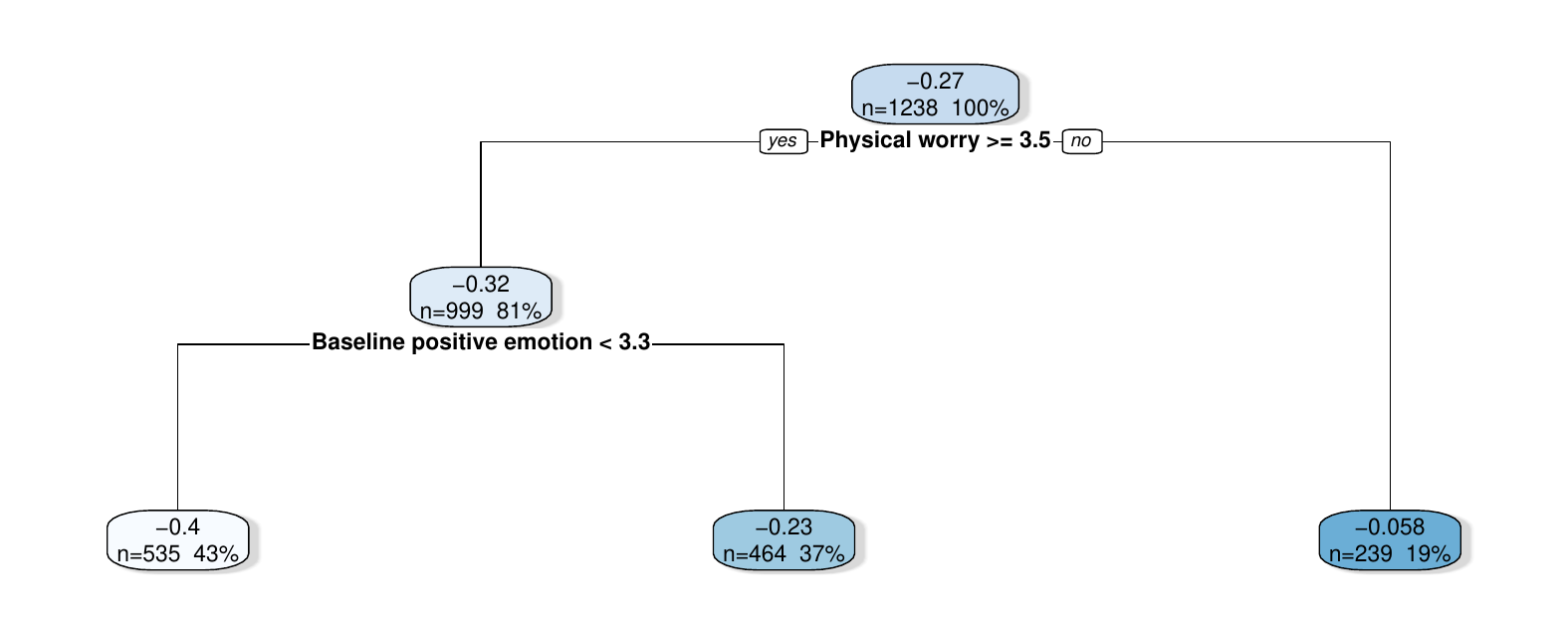}
    \caption{Fit-the-fit plot based on a decision tree summary of the estimated
    CNIEs from the TOW learner in the PSACR application.}
    \label{fig:psacr-tree-tow}
\end{figure}

\begin{figure}[ht!]
    \centering
    \includegraphics[width=0.85\linewidth]{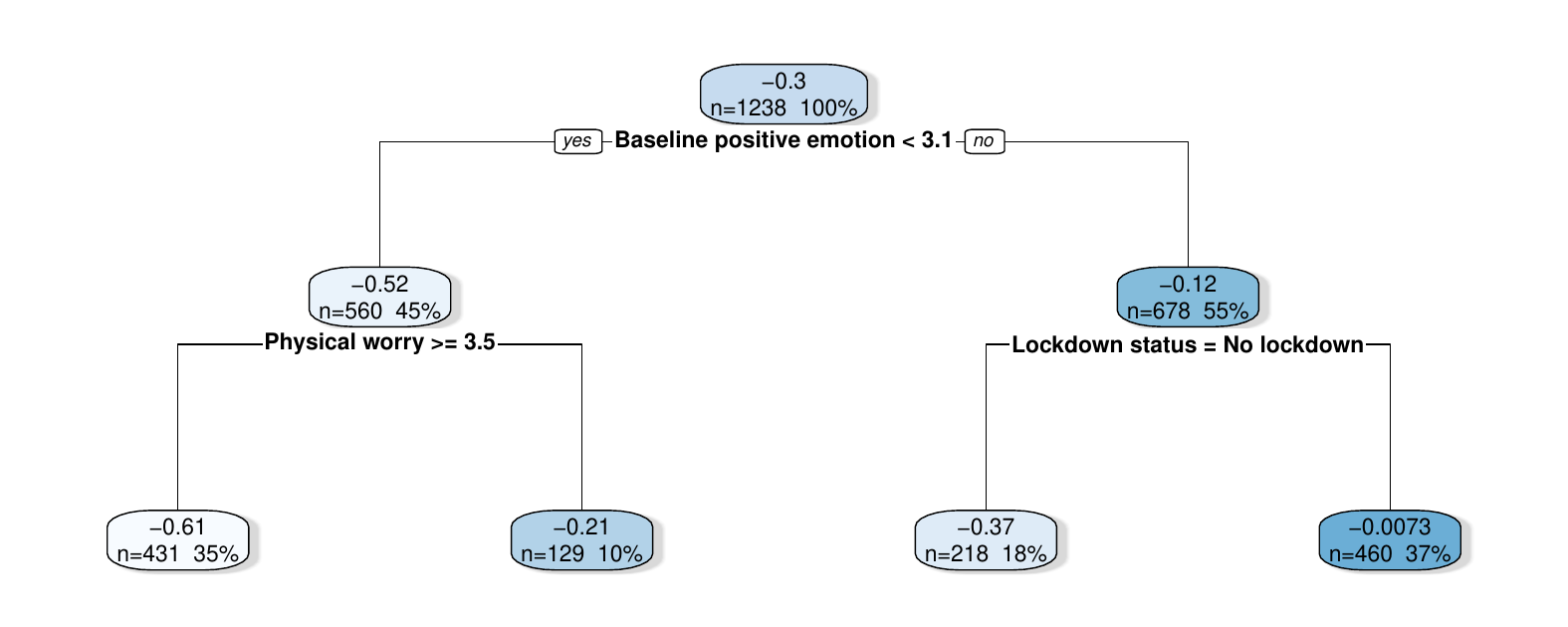}
    \caption{Fit-the-fit plot based on a decision tree summary of the estimated
    CNIEs from the OW learner in the PSACR application.}
    \label{fig:psacr-tree-ow}
\end{figure}

\begin{figure}[ht!]
    \centering
    \includegraphics[width=0.85\linewidth]{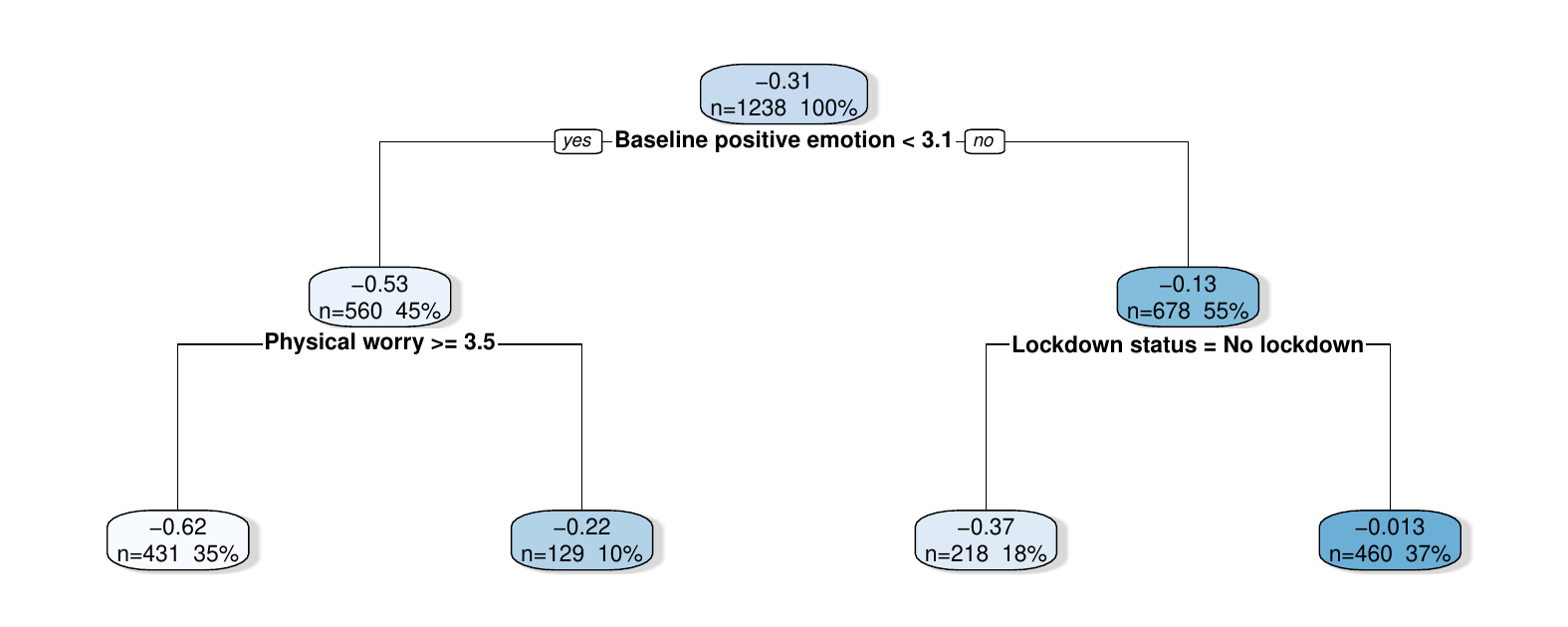}
    \caption{Fit-the-fit plot based on a decision tree summary of the estimated
    CNIEs from the TR learner in the PSACR application.}
    \label{fig:psacr-tree-tr}
\end{figure}

\begin{figure}[ht!]
    \centering
    \includegraphics[width=0.85\linewidth]{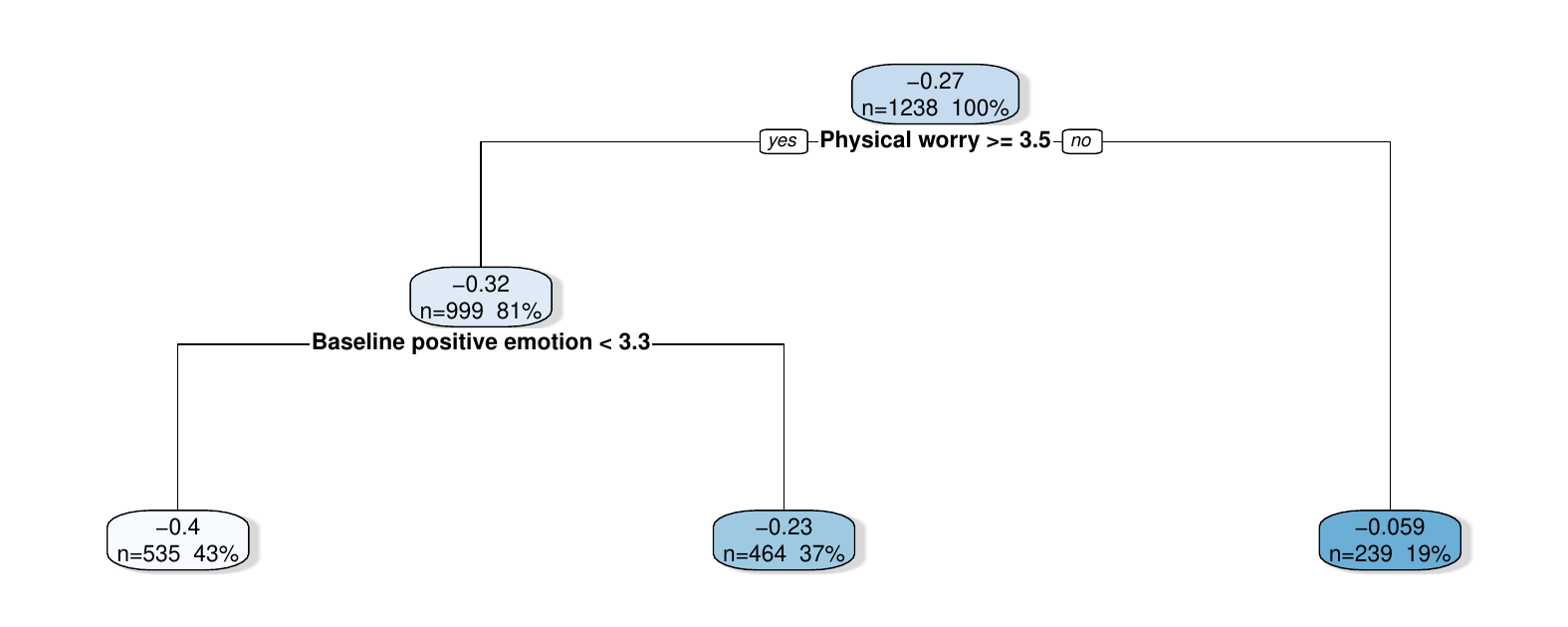}
    \caption{Fit-the-fit plot based on a decision tree summary of the estimated
    CNIEs from the TTR learner in the PSACR application.}
    \label{fig:psacr-tree-ttr}
\end{figure}

\begin{figure}[ht!]
    \centering
    \includegraphics[width=0.85\linewidth]{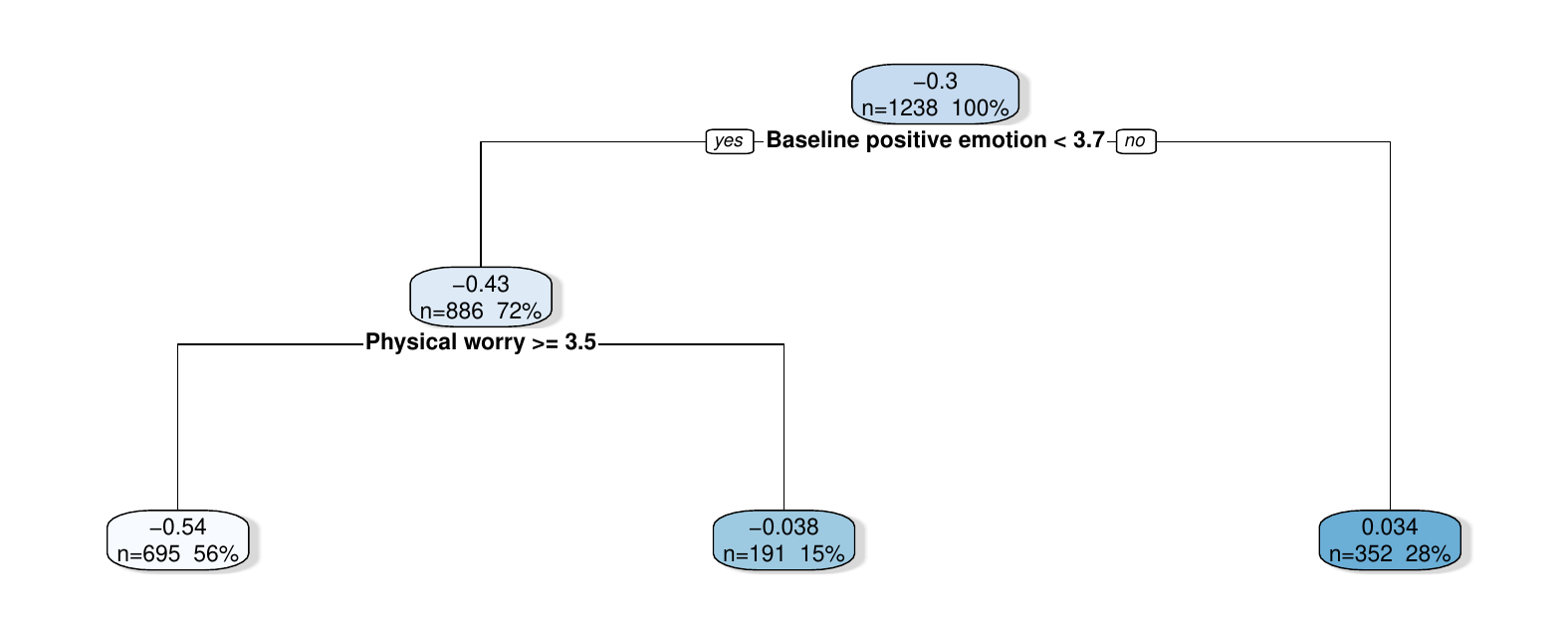}
    \caption{Fit-the-fit plot based on a decision tree summary of the estimated
    CNIEs from the TW learner in the PSACR application.}
    \label{fig:psacr-tree-psw}
\end{figure}

\begin{figure}[ht!]
    \centering
    \includegraphics[width=0.85\linewidth]{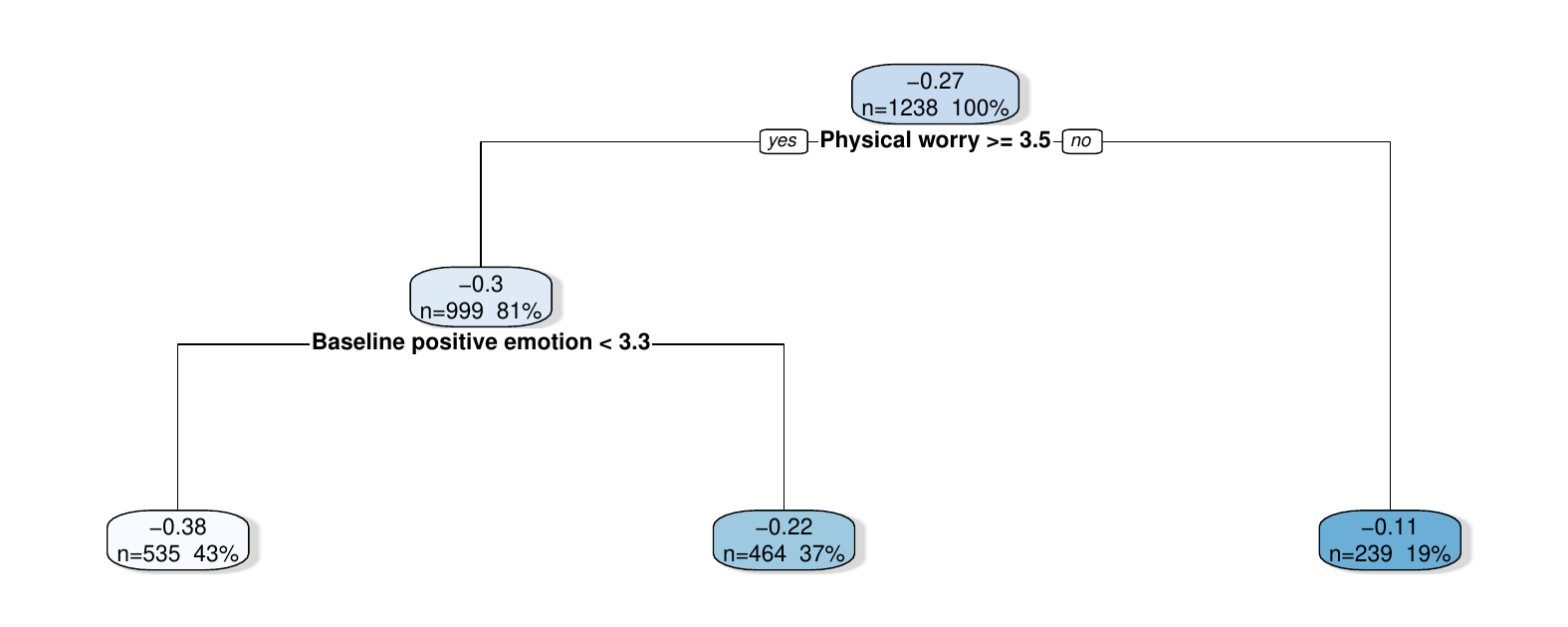}
    \caption{Fit-the-fit plot based on a decision tree summary of the estimated
    CNIEs from the TTW learner in the PSACR application.}
    \label{fig:psacr-tree-tpsw}
\end{figure}

\clearpage

\subsection{STAR application}

\begin{figure}[ht!]
    \centering
    \includegraphics[width=0.95\linewidth]{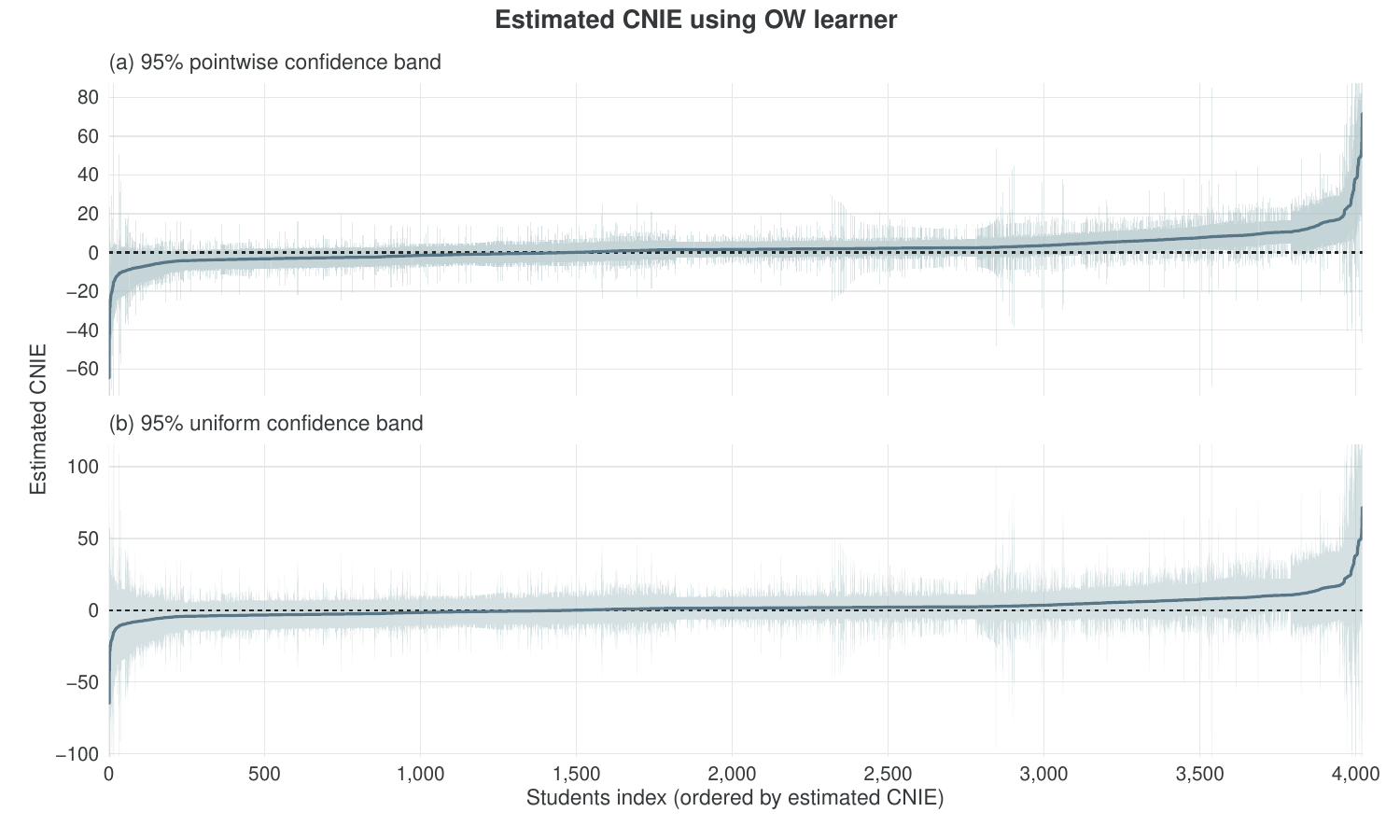}
    \caption{CNIE estimates at the observed covariate profiles in the STAR application, obtained using the OW learner, with pointwise and uniform confidence bands.}
    \label{fig:star-prediction-ow}
\end{figure}

\begin{figure}[ht!]
    \centering
    \includegraphics[width=0.95\linewidth]{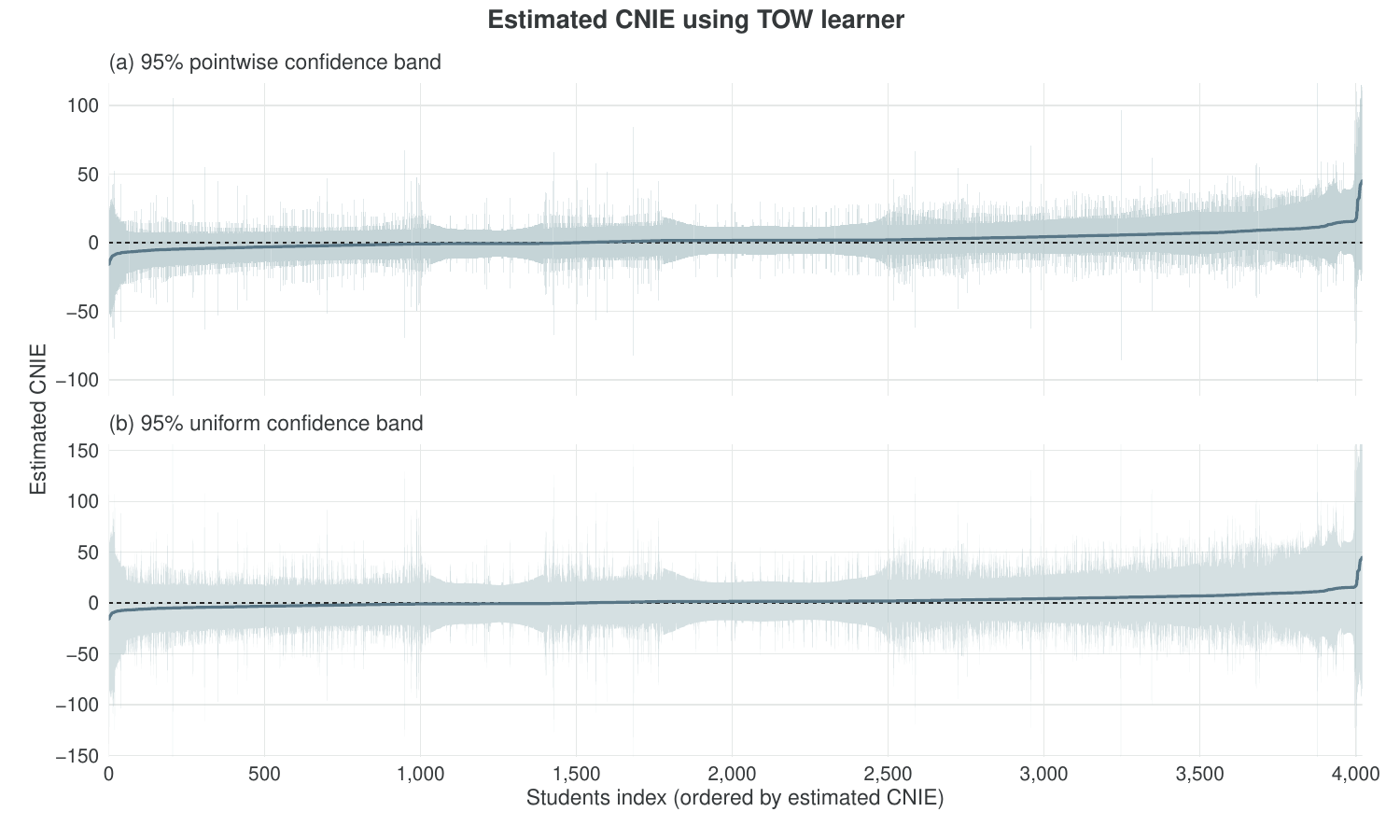}
    \caption{CNIE estimates at the observed covariate profiles in the STAR application, obtained using the TOW learner, with pointwise and uniform confidence bands.}
    \label{fig:star-prediction-tow}
\end{figure}

\begin{figure}[ht!]
    \centering
    \includegraphics[width=0.95\linewidth]{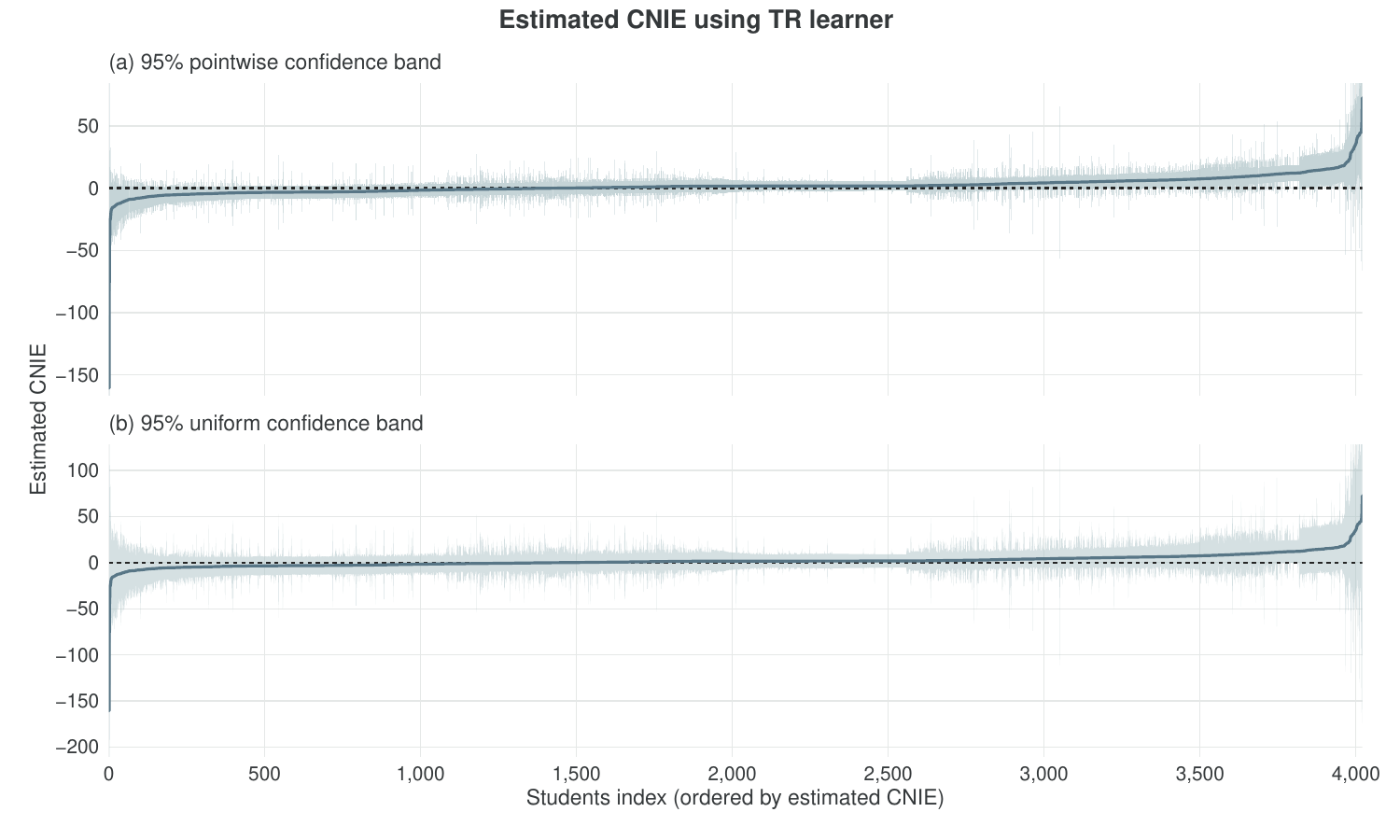}
    \caption{CNIE estimates at the observed covariate profiles in the STAR application, obtained using the TR learner, with pointwise and uniform confidence bands.}
    \label{fig:star-prediction-tr}
\end{figure}

\begin{figure}[ht!]
    \centering
    \includegraphics[width=0.95\linewidth]{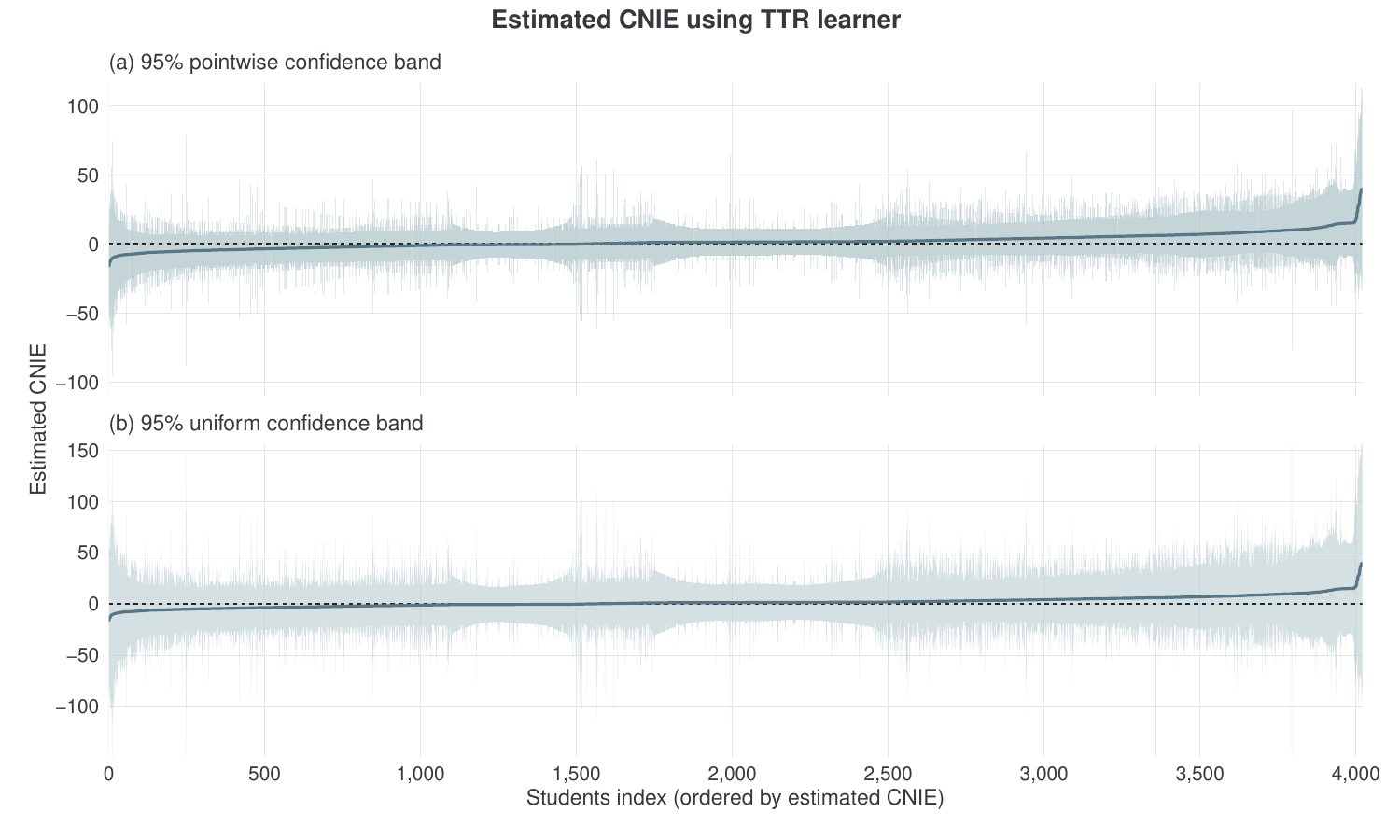}
    \caption{CNIE estimates at the observed covariate profiles in the STAR application, obtained using the TTR learner, with pointwise and uniform confidence bands.}
    \label{fig:star-prediction-ttr}
\end{figure}

\begin{figure}[ht!]
    \centering
    \includegraphics[width=0.95\linewidth]{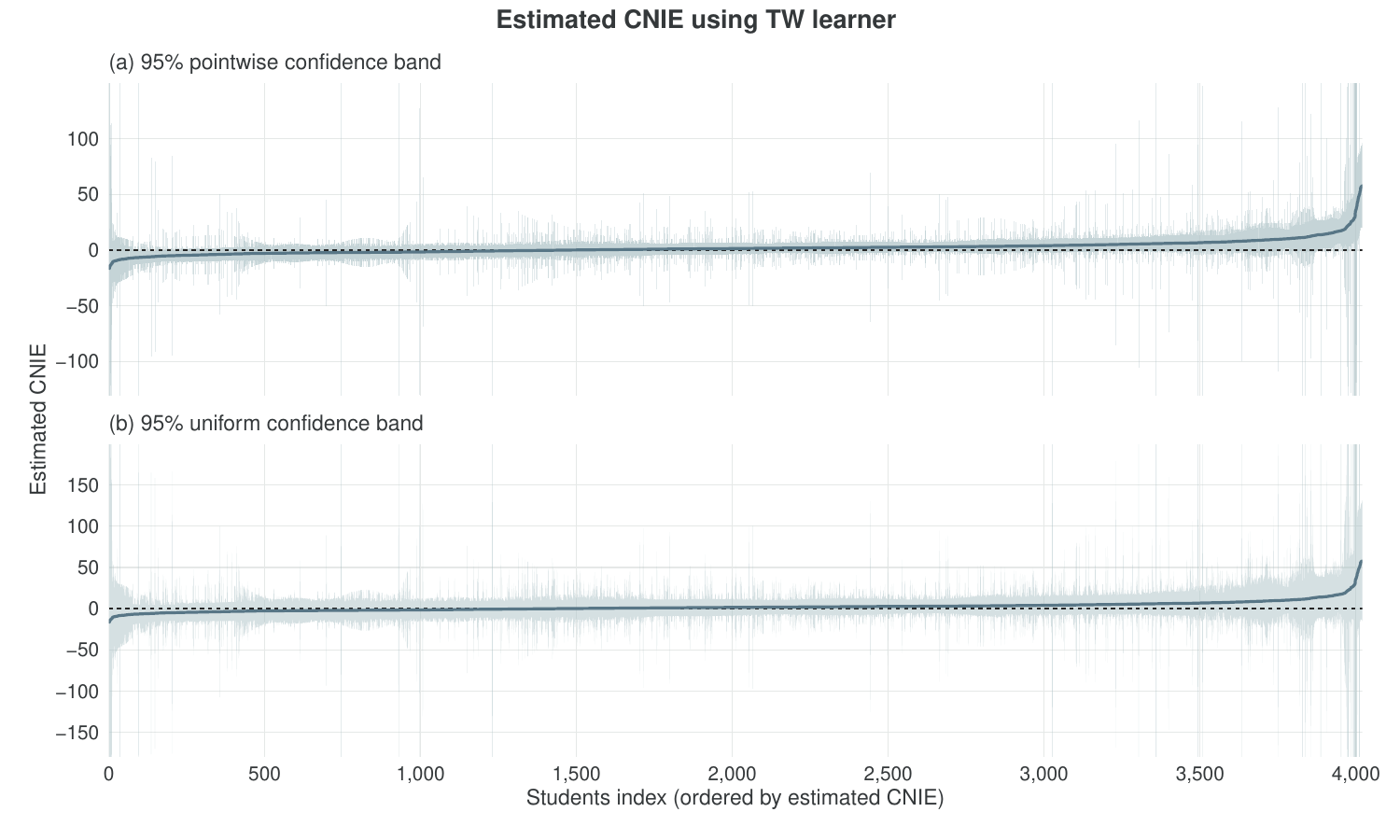}
    \caption{CNIE estimates at the observed covariate profiles in the STAR application, obtained using the TW learner, with pointwise and uniform confidence bands.}
    \label{fig:star-prediction-psw}
\end{figure}

\begin{figure}[ht!]
    \centering
    \includegraphics[width=0.95\linewidth]{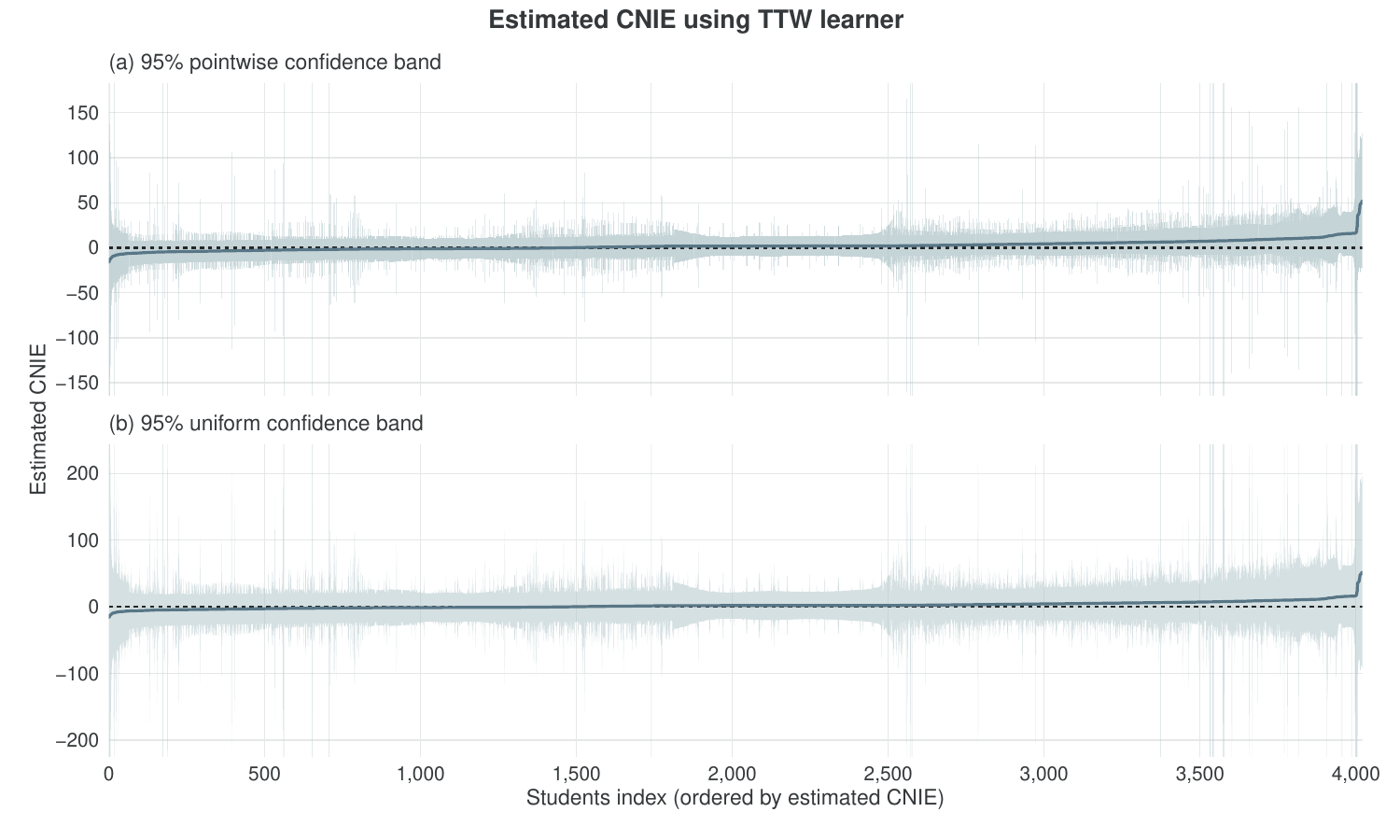}
    \caption{CNIE estimates at the observed covariate profiles in the STAR application, obtained using the TTW learner, with pointwise and uniform confidence bands.}
    \label{fig:star-prediction-tpsw}
\end{figure}

\begin{figure}[ht!]
    \centering
    \includegraphics[width=0.85\linewidth]{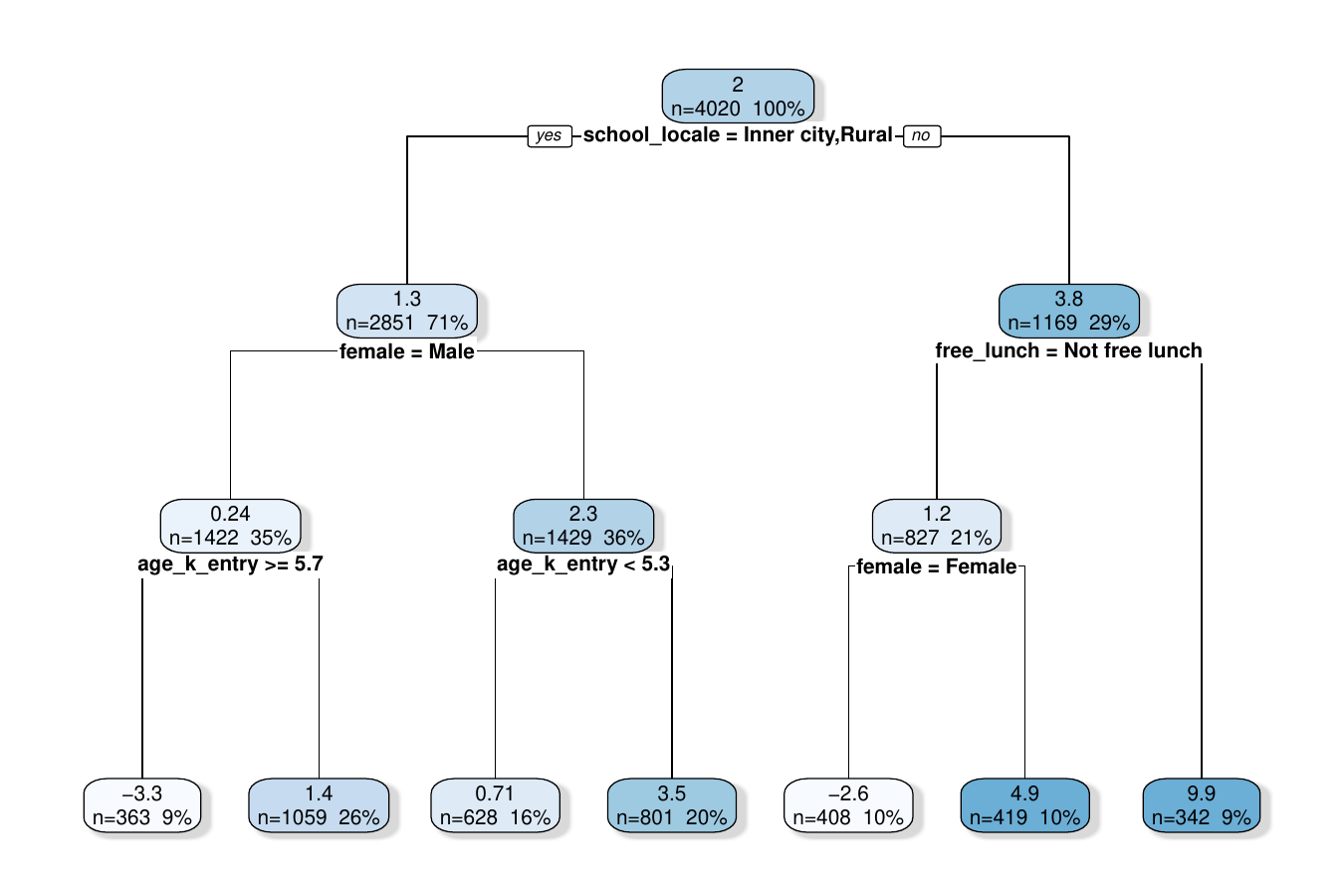}
    \caption{Fit-the-fit plot based on a decision tree summary of the CNIEs estimated using the TOW learner in the STAR application.}
    \label{fig:star-tree-tow}
\end{figure}

\begin{figure}[ht!]
    \centering
    \includegraphics[width=0.85\linewidth]{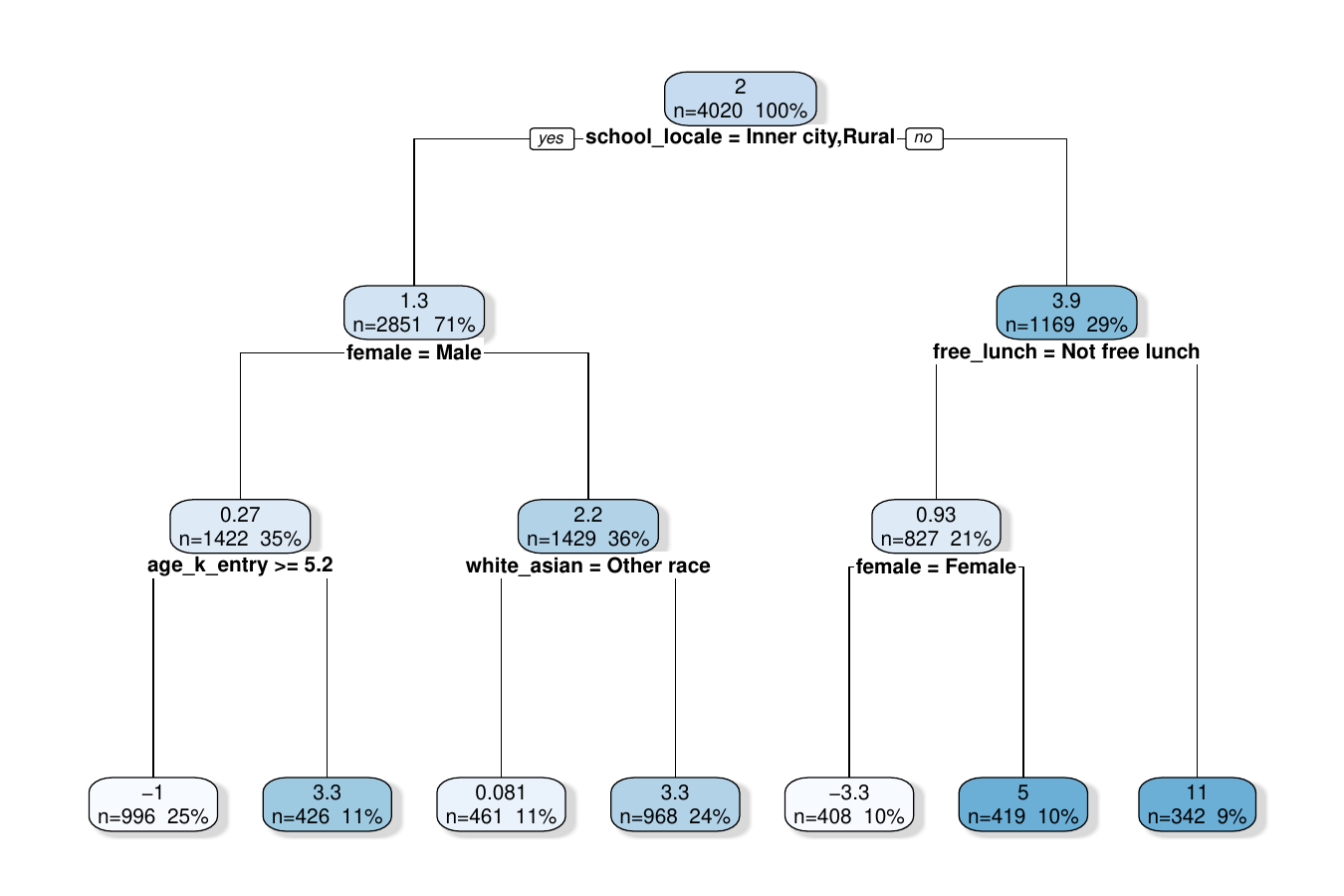}
    \caption{Fit-the-fit plot based on a decision tree summary of the CNIEs estimated using the OW learner in the STAR application.}
    \label{fig:star-tree-ow}
\end{figure}

\begin{figure}[ht!]
    \centering
    \includegraphics[width=0.85\linewidth]{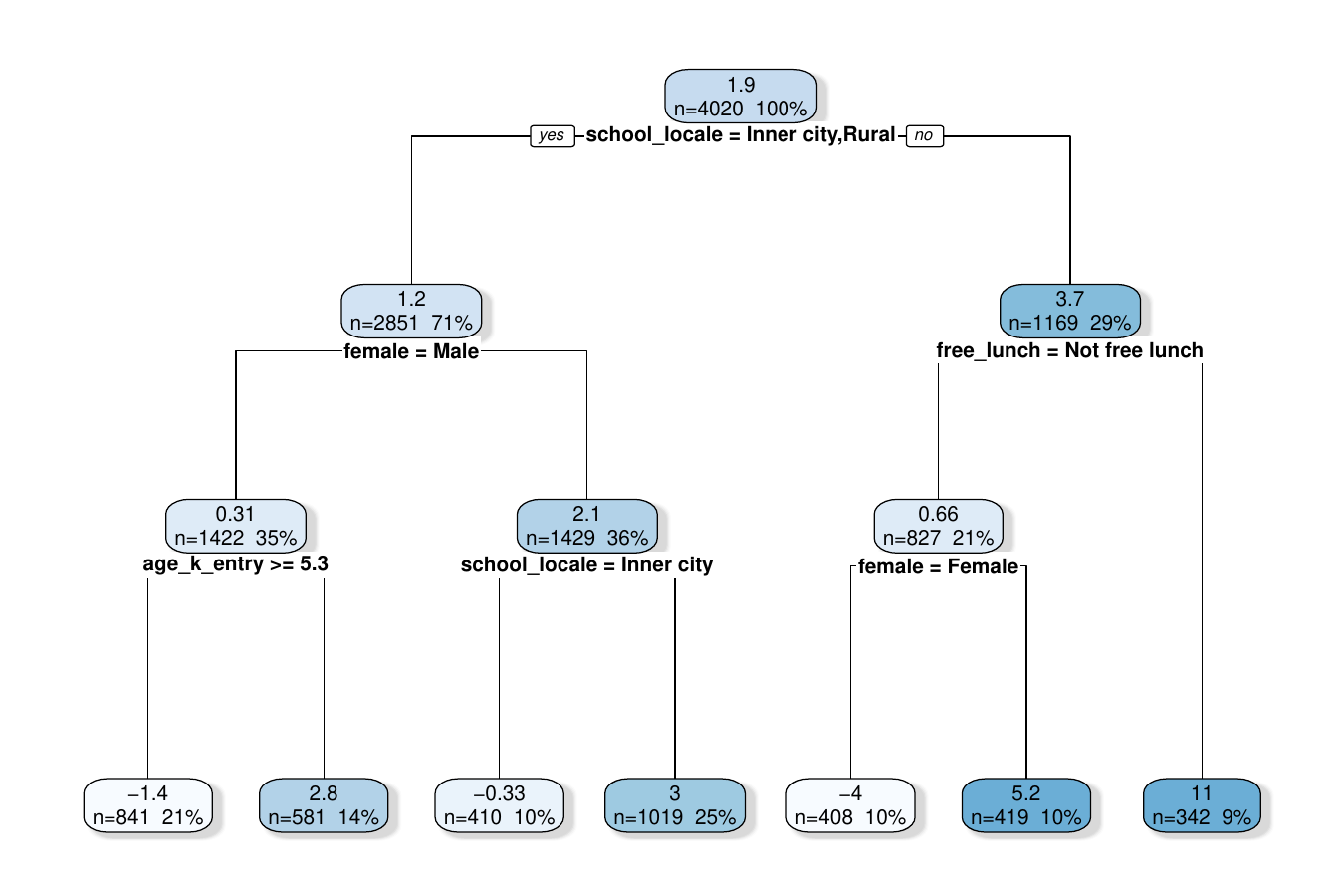}
    \caption{Fit-the-fit plot based on a decision tree summary of the CNIEs estimated using the TR learner in the STAR application.}
    \label{fig:star-tree-tr}
\end{figure}

\begin{figure}[ht!]
    \centering
    \includegraphics[width=0.85\linewidth]{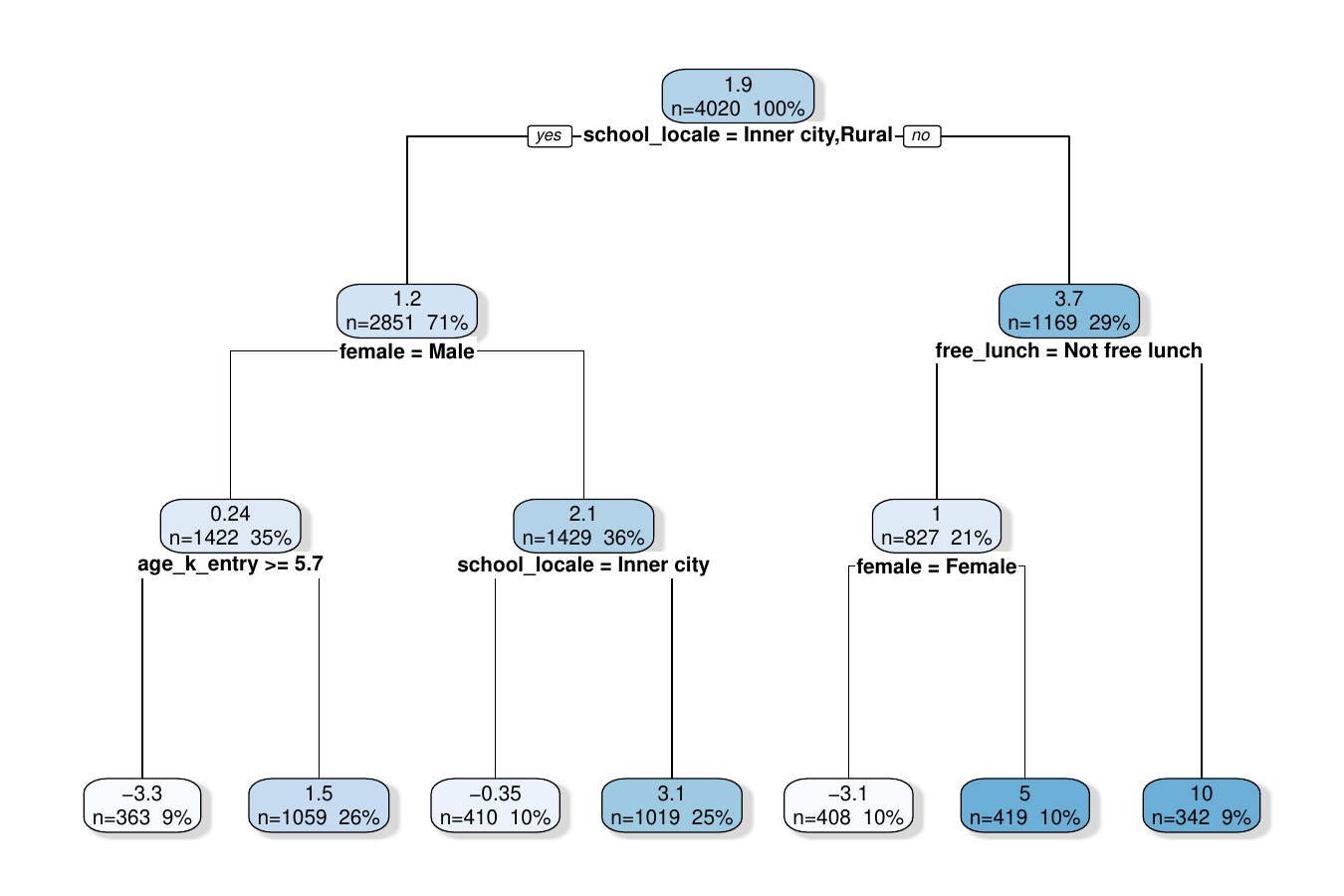}
    \caption{Fit-the-fit plot based on a decision tree summary of the CNIEs estimated using the TTR learner in the STAR application.}
    \label{fig:star-tree-ttr}
\end{figure}

\begin{figure}[ht!]
    \centering
    \includegraphics[width=0.85\linewidth]{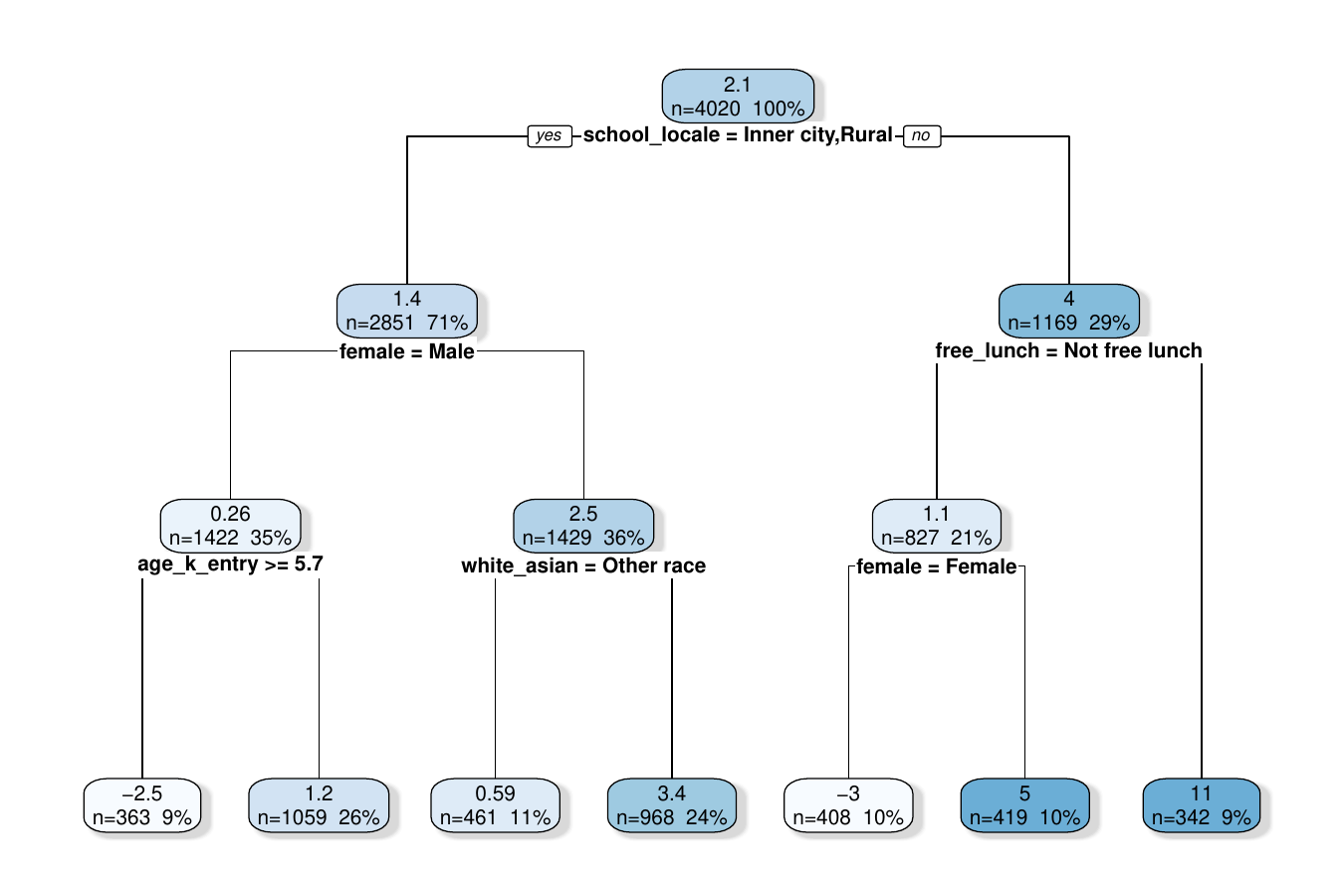}
    \caption{Fit-the-fit plot based on a decision tree summary of the CNIEs estimated using the TW learner in the STAR application.}
    \label{fig:star-tree-psw}
\end{figure}

\begin{figure}[ht!]
    \centering
    \includegraphics[width=0.85\linewidth]{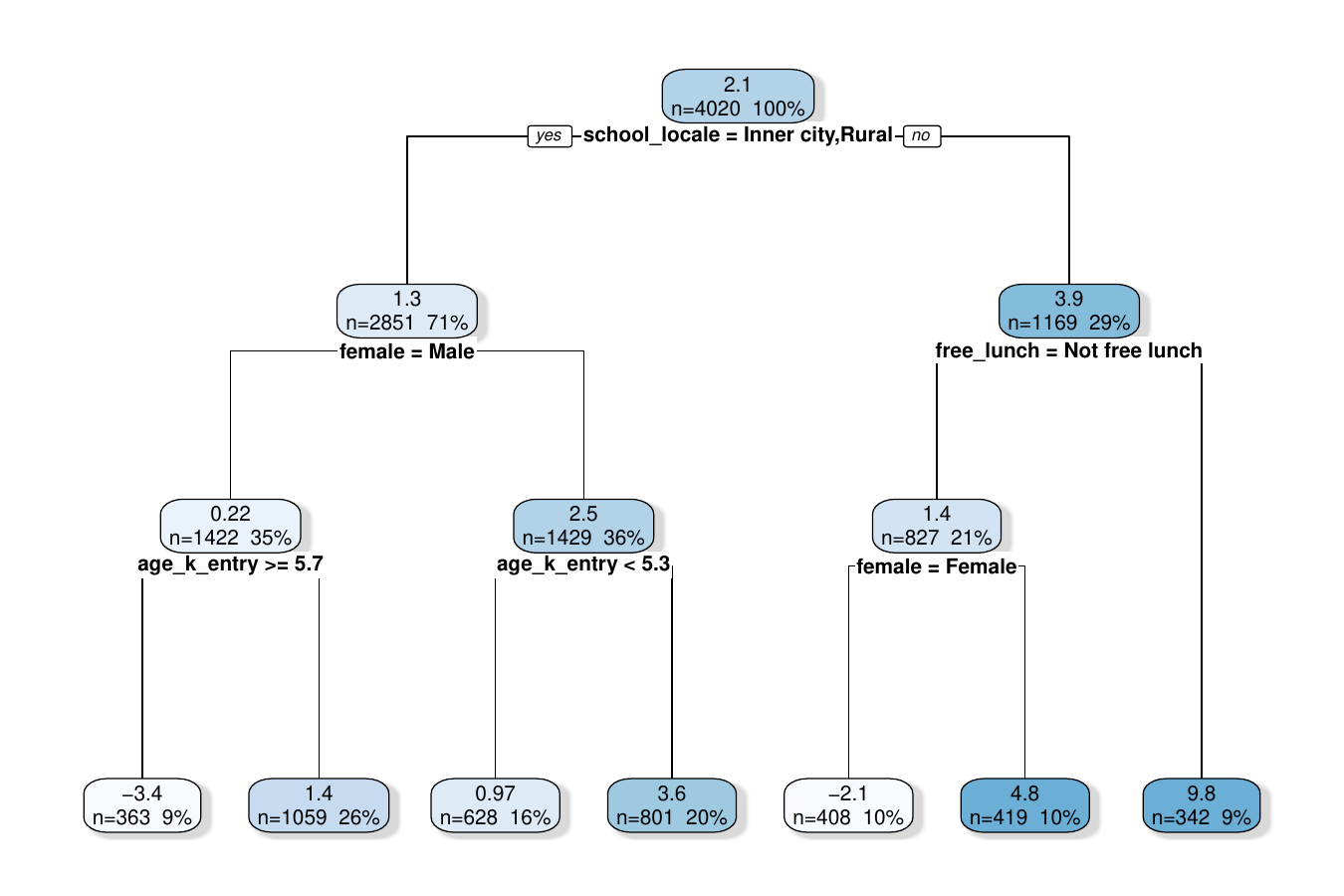}
    \caption{Fit-the-fit plot based on a decision tree summary of the CNIEs estimated using the TTW learner in the STAR application.}
    \label{fig:star-tree-tpsw}
\end{figure}

\clearpage


\end{document}